%% file: pl-csp.tex
\documentclass[11pt]{article}

\usepackage{titling}
\usepackage{cancel}
\usepackage{enumitem} 
\setlist{noitemsep,topsep=0pt,parsep=0pt} 

\usepackage[margin=1cm]{caption} 
\usepackage{subfig}
\usepackage{mathtools} 

\usepackage{tikz}
\usetikzlibrary{arrows,backgrounds,calc,fit,decorations.pathreplacing,decorations.markings,shapes.geometric}

\tikzset{every fit/.append style=text badly centered}

\usepackage{tyson}

\usepackage[bookmarks=true,hypertexnames=false]{hyperref}
\hypersetup{colorlinks=true, citecolor=blue, linkcolor=red, urlcolor=blue}

\newcommand{\Holant}{\operatorname{Holant}}
\newcommand{\PlHolant}{\operatorname{Pl-Holant}}
\newcommand{\holant}[2]{\ensuremath{\Holant\left(#1\mid #2\right)}}
\newcommand{\plholant}[2]{\ensuremath{\PlHolant\left(#1\mid #2\right)}}

\newcommand{\CSP}{\operatorname{\#CSP}}
\newcommand{\PlCSP}{\operatorname{Pl-\#CSP}}

\newcommand{\trans}[4]{\ensuremath{\left[\begin{smallmatrix} #1 & #2 \\ #3 & #4 \end{smallmatrix}\right]}}

{\par\medskip}

\def\borderColor{blue!60}
\def\scale{0.6}
\def\nodeDist{1.4cm}
\tikzstyle{internal} = [draw, fill, shape=circle]
\tikzstyle{external} = [shape=circle]
\tikzstyle{square}   = [draw, fill, rectangle]
\tikzstyle{triangle} = [draw, fill, regular polygon, regular polygon sides=3, inner sep=3pt]
\tikzstyle{pentagon} = [draw, fill, regular polygon, regular polygon sides=5, inner sep=2pt, minimum size=14pt]

\begin{document}
\title{{\bf Holographic Algorithm with Matchgates Is Universal for Planar \#CSP Over Boolean Domain}}

\vspace{0.3in}
\author{Jin-Yi Cai\thanks{University of Wisconsin-Madison.
 {\tt jyc@cs.wisc.edu}}
\and Zhiguo Fu\thanks{School of Mathematics, Jilin University. {\tt
fuzg@jlu.edu.cn}}}

\date{}
\maketitle

\bibliographystyle{plain}

\begin{abstract}
We prove a complexity classification theorem that classifies all
counting constraint satisfaction problems (\#CSP) over Boolean variables
into exactly three categories:
(1) Polynomial-time tractable; (2) \#P-hard for general
instances, but solvable in polynomial-time over planar graphs;
and (3) \#P-hard over planar graphs.
The classification applies to all sets of local, \emph{not necessarily
symmetric}, constraint
functions on Boolean variables that take complex values.
It is shown that Valiant's holographic algorithm with matchgates
is a \emph{universal} strategy for all
problems in category (2).
\end{abstract}

\input{1introduction}
\input{2Preliminaries}
\input{3no-parity}
\input{4pl-csp2}
\input{5parity}

\input{6main}
\input{7ref}

\input{8appendix}

\end{document}

%% file: 1introduction.tex
\section{Introduction}
Half a century ago, the
Fisher-Kasteleyn-Temperley (FKT) algorithm
was discovered~\cite{TF61, Kasteleyn1961, Kasteleyn1967}. The FKT algorithm can
count the number of perfect matchings (dimers) over planar graphs
in polynomial time. This is a milestone in the long history
in statistical physics starting with Lenz,
Ising, Onsager, Yang,
Lee, Fisher, Temperley, Kasteleyn, Baxter, Lieb, Wilson
etc~\cite{ising1925beitrag,onsager1944crystal,yang1952spontaneous,yang1952statistical,lee1952statistical,TF61,Kasteleyn1961,Kasteleyn1967,baxter1982exactly,lieb1981general}, with beautiful contributions from many others.
The central question is what constitutes an Exactly Solved Model.
The basic
conclusion from physicists is that for some ``systems''
their partition functions
 are ``exactly solvable'' for planar structures, but
appear intractable for higher dimensions.
However, exactly what does it  mean to be intractable?
Physicists did not have a formal notion of
intractability.

This notion is supplied by complexity theory.
Following the P vs.~NP theory, in 1979
L.~Valiant~\cite{Valiant-permanent-paper79-TCS} defined the class
\#P for counting problems. Most counting problems of a combinatorial
nature are included in this broad class. Sum-of-Product computations,
such as partition functions studied in physics and counting constraint
satisfaction problems are included in \#P (or by a  P-time reduction when the output is not an integer),
 and \#P-hardness
is at least as hard as NP-hardness. In particular,
counting perfect matchings in general graphs is \#P-complete.

But are there other surprises like the FKT-algorithm? If so, can they
solve any \#P-hard problems?
In two seminal papers~\cite{Val02a,Val06},
L.~Valiant introduced \emph{matchgates} and \emph{holographic algorithms}.
These holographic algorithms
use a quantum-like superposition
to achieve fantastic cancellations, which
produce polynomial time algorithms to solve a number of concrete problems
that would seem to require exponential time to compute.
The first ingredient of holographic algorithms is the FKT algorithm.
The second ingredient  is
a tensor theoretic transformation that establishes a quantitative
equivalence of two seemingly different counting problems.
This holographic reduction in general does not
preserve solutions between the two problems in a 1-1 fashion.
These transformations establish a
duality similar in spirit to the Fourier transform and its
inverse.

As these novel algorithms
solve problems that appear so close to being \#P-hard, they naturally
raise the question whether they can solve \#P-hard problems in P-time.
In the past 10 to 15 years significant progress was made
in the understanding of these remarkable algorithms~\cite{Cai-Fu-Guo-W, caiguowilliams13,  art-sc, asymmetric-sig, Cai-Lu-Xia-real, Guo-Williams,
 landsberg-morton, Val02b, string23, Val06}.
In an interesting twist, it turns out
that the idea of a holographic reduction is not only
a powerful technique to design new and unexpected algorithms,
but also  an indispensable tool to classify the inherent
complexity of counting problems, in particular, to
understand the limit and scope of holographic
algorithms~\cite{clx-focs-2008, clx-holant, ghlx-stacs-2011, HL12, glva-2013, clx-soda-2013,
caiguowilliams13, Guo-Williams, cgw-focs-2014, Cai-Fu-Guo-W}.
This study has produced
a number of complexity dichotomy theorems.
These classify \emph{every} problem expressible
 in a framework as either solvable in P or \#P-hard,
with nothing in between.

One such framework is called \#CSP problems.
A \#CSP problem on Boolean variables
is specified  by a set of local constraint
functions $\mathcal{F}$. Each function $f \in \mathcal{F}$ has an arity $k$,
and maps $\{0,1\}^k \rightarrow \mathbb{C}$. (For consideration of models
of computation, we restrict function values to be algebraic numbers.
Unweighted \#CSP problems are defined by 0-1 valued constraint
functions.)
An instance of \#CSP($\mathcal{F}$) is specified by
a finite set of Boolean variables $X = \{x_1, x_2, \ldots, x_n\}$,
and a finite sequence of constraints $\mathcal{S}$ from $\mathcal{F}$, each
applied to  an ordered sequence of variables from $X$.
Every instance can be described by a bipartite graph
where  LHS nodes  are variables $X$,
RHS  nodes are constraints $\mathcal{S}$,
and the connections between them specify occurrences of variables
in constraints in the input instance.
The output of this instance is $\sum_{\sigma}
 \prod_{f \in \mathcal{S}} f|_{\sigma}$, a sum over all
$\sigma: X \rightarrow \{0,1\}$, of products
of all constraints in $\mathcal{S}$ evaluated according to $\sigma$.
In the unweighted 0-1 case, each such product contributes a 1
if  $\sigma$ satisfies all constraints in $\mathcal{S}$, and 0 otherwise.
In the general case, the output  is a weighted sum of $2^n$ terms.
\#CSP is a very expressive framework for locally specified counting
problems. A spin system is a special case  where there is one
single binary constraint in $\mathcal{F}$, and possibly one or more
unary constraints when there are ``external fields''.

We prove in this paper that, holographic algorithms with matchgates
form a \emph{universal} strategy for problems expressible in
this framework that are \#P-hard in general but solvable
in polynomial time on planar graphs.
More specifically we prove the following classification theorem.
\begin{theorem}\label{theorem-main-intro}
For any set of constraint functions $\mathcal{F}$ over Boolean variables,
each  taking complex values and not necessarily symmetric,
\#CSP($\mathcal{F}$) belongs to exactly one of three categories
according to $\mathcal{F}$:
(1) It is P-time solvable;
(2) It is P-time solvable over planar graphs but \#P-hard over general graphs;
(3) It is \#P-hard over planar graphs.
Moreover,
category (2) consists precisely of those problems that are holographically reducible to the FKT algorithm.
\end{theorem}
This theorem finally settles the full reach of the power of
Valiant's holographic algorithms in the \#CSP framework over Boolean
variables. Several results preceded this.
The most direct three predecessors are as follows:
(I) In~\cite{Cai-Lu-Xia-real} it is shown that
Theorem~\ref{theorem-main-intro} holds if every function in $\mathcal{F}$ is real-valued
and \emph{symmetric}. The value of a symmetric function
is invariant when the input values are permuted. This is quite
a stringent restriction. A constraint function on $n$ Boolean
variables requires $2^n$ output values to specify, while a symmetric
one needs only $n+1$ values.
(II) Guo and Williams~\cite{Guo-Williams} generalize \cite{Cai-Lu-Xia-real}
to the case where functions
in $\mathcal{F}$ are complex-valued,
but they must still be \emph{symmetric}.
Complex numbers form the natural setting
to discuss the power of these problems. Many problems, even
though they are real-valued, are shown to be equivalent under
a holographic reduction which goes through $\mathbb{C}$, and their
inherent complexity is only understood by an  analysis in $\mathbb{C}$
on quantities such as eigenvalues.
(III) If one ignores  planarity, \cite{CLX14} proves
 a complexity dichotomy.
This result itself generalizes
previous results by Creignou-Hermann~\cite{Creignou-Hermann} for the
case when all constraint functions are 0-1 valued,
by Dyer-Goldberg-Jerrum~\cite{Dyer-Goldberg-Jerrum} for non-negative
valued constraint functions, and  by Bulatov et. al.~\cite{Bulatov-etal}
for real-valued constraint functions of mixed signs.

The classification in Theorem~\ref{theorem-main-intro},
especially the claim that holographic reductions followed by the FKT
are universal for category (2), is by no means self-evident.
In fact such a sweeping claim should invite skepticism.
Nowhere in  complexity theory of
decision problems are we aware of such a provable universal
algorithmic strategy for a broad class
of problems. Moreover, in the study of holographic algorithms, an even
broader class than \#CSP of locally specified Sum-of-Product computations
has been introduced~\cite{clx-holant}, called Holant problems.
It turns out that counting perfect matchings
is naturally expressible as a  Holant problem, but not
as a \#CSP problem. Very recently we discovered that for
planar Holant problems a corresponding universality statement
 as in
Theorem~\ref{theorem-main-intro} is \emph{false}~\cite{Cai-Fu-Guo-W}.
For planar Holant problems, a holographic reduction to the FKT is \emph{not}
universal; there are other \#P-hard problems that become
 tractable on planar structures, and they are not holographically
reducible to the FKT.

The class of  Holant problems turns out to be more than
just a separate framework providing a cautionary reference to
Theorem~\ref{theorem-main-intro}.
In fact they form the main arena we carry out the proof of
Theorem~\ref{theorem-main-intro}. A basic idea in this proof
is a holographic transformation between the \#CSP
setting and the Holant setting via the Hadamard transformation
$H_2 = \frac{1}{\sqrt{2}} \left[\begin{smallmatrix} 1 & 1 \\ 1 & -1 \end{smallmatrix}\right]$.
This transformation is similar to the Fourier transform.
Certain properties are easier to handle in one setting
while others are easier after a transform. We will go back and
forth.

In subsection~\ref{outline-subsection}
we give a more detailed account of the strategies used, and a
proof outline.
Among the techniques used are a derivative operator $\partial$,
a Tableau Calculus, and arity reduction.
An overall philosophy is that various tractable constraint
functions of different families cannot mix. Then the truth of
Theorem~\ref{theorem-main-intro} itself, precisely because it is such
a complete statement without any exceptions, guides the choices
made in various constructions.
As a proof strategy,
this is pretty dicey or at least self-serving.
Essentially we want the validity of the very
statement we want to prove to provide its own guarantee of
success in every step in its proof.
If there were other
tractable problems, e.g., as in the case of planar
 Holant~\cite{Cai-Fu-Guo-W} where different classes of
tractable constraints can indeed mix,  then we would be stuck.
Luckily, the vision is correct for planar \#CSP. And therefore,
the self-serving plan becomes a reliable guide
to the proof, a bit self-fulfilling.

%% file: 2Preliminaries.tex
\section{Preliminaries} \label{sec:preliminaries}

\subsection{Problems and Definitions}
In this paper, ${\frak i}$ denotes a square root of $-1$, i.e., ${\frak i}^2=-1$.

Even though our focus in this paper is on planar counting CSP
problems, most of the proof need to be carried out in
the framework of Holant problems~\cite{clx-holant,CLX11}.
A Holant problem is specified by a set of local constraint functions,
also called \emph{signatures}.
In this paper,
we investigate complex-valued planar \#CSP problems over
Boolean variables,
and thus all signatures in the corresponding  Holant problems are
of the form $\{0,1\}^n \to \mathbb{C}$.
For consideration of models of computation,
functions take complex algebraic numbers.

Graphs may have self-loops and parallel edges.
A graph without self-loops or parallel edges is a \emph{simple} graph.
Fix a set of local constraint functions $\mathcal{F}$.
A \emph{signature grid} $\Omega = (G, \pi)$ consists of a graph $G = (V,E)$,
and a mapping $\pi$ which maps each vertex $v \in V$ to
some $f_v \in \mathcal{F}$ of arity $\deg(v)$,
and its incident edges $E(v)$ to the input variables of $f_v$.
We say that $\Omega$ is a \emph{planar signature grid} if $G$ is a plane graph,
where the variables of $f_v$ are ordered counterclockwise starting from an edge specified by $\pi$.
The Holant problem on instance $\Omega$ is to evaluate
\[\Holant(\Omega; \mathcal{F}) = \sum_{\sigma: E \to \{0,1\}}
\prod_{v \in V} f_v(\sigma \mid_{E(v)}),\]
where $\sigma \mid_{E(v)}$ denotes the restriction of $\sigma$ to $E(v)$.
We write $G$ in place of $\Omega$ when $\pi$ is clear from context.

A signature $f$ of arity $n$
 can be specified by listing its values in lexicographical
order as in a truth table,
which is a vector in $\mathbb{C}^{2^{n}}$,
or as a tensor in $(\mathbb{C}^{2})^{\otimes n}$.
A symmetric signature $f$  of arity $n$
takes values depending only on the Hamming weight of
the input, and can be expressed as $[f_0,f_1,\dotsc,f_n]$,
where $f_w$ is the value of $f$ on inputs of Hamming weight $w$.
An example is the \textsc{Equality} signature $(=_n)=[1, 0, \ldots, 0, 1]$ of arity $n$.
Another example is the \textsc{Exact-One} signature
$[0, 1, \ldots, 0, 0]$ corresponding to the \textsc{Perfect Matching}
constraint.

A Holant problem is parametrized by a set of signatures.

\begin{definition}
 Given a set of signatures $\mathcal{F}$,
 we define the counting problem $\Holant(\mathcal{F})$ as:

 Input: A \emph{signature grid} $\Omega = (G, \pi)$;

 Output: $\Holant(\Omega; \mathcal{F})$.

The problem $\PlHolant(\mathcal{F})$ is defined similarly using a planar signature grid.
\end{definition}

A signature $f$ of arity $n$ is \emph{degenerate} if there exist unary signatures $u_j \in \mathbb{C}^2$ ($1 \le j \le n$)
such that $f = u_1 \otimes \cdots \otimes u_n$.
Using a degenerate signature
is equivalent to replacing
 it by its $n$ unary signatures, each on its corresponding edge.
A symmetric degenerate signature has the form $u^{\otimes n}$.
Replacing a signature $f \in \mathcal{F}$ by a constant multiple $c f$,
where $c \ne 0$,
does not change the complexity of $\Holant(\mathcal{F})$.
It introduces a factor $c^{m}$
to $\Holant(\Omega; \mathcal{F})$, where $f$ occurs $m$ times in $\Omega$.

We allow $\mathcal{F}$ to be an infinite set.
For $\Holant(\mathcal{F})$
or $\PlHolant(\mathcal{F})$ to be tractable,
the problem must be computable in polynomial time even when the description of the signatures in the input $\Omega$ are included in the input size.
On the other hand,
we say  $\Holant(\mathcal{F})$
or  $\PlHolant(\mathcal{F})$ is $\SHARPP$-hard if there exists a finite subset of $\mathcal{F}$ for which the problem is $\SHARPP$-hard.
In this paper we focus on planar problems, and so
we say a signature set $\mathcal{F}$ is tractable (resp.~$\SHARPP$-hard)
if the corresponding counting problem
 $\PlHolant(\mathcal{F})$ is tractable (resp.~$\SHARPP$-hard).
For a signature $f$,
we say $f$ is tractable (resp.~$\SHARPP$-hard) if $\{f\}$ is.
We follow the usual conventions about polynomial time Turing reduction $\le_T$ and polynomial time Turing equivalence $\equiv_T$.

\subsection{Holographic Reduction}

To introduce the idea of holographic reductions,
it is convenient to consider bipartite graphs.
For a general graph,
we can always transform it into a bipartite graph while preserving the Holant value,
as follows.
For each edge in the graph,
we replace it by a path of length two.
(This operation is called the \emph{2-stretch} of the graph and yields the edge-vertex incidence graph.)
Each new vertex is assigned the binary \textsc{Equality} signature $(=_2) = [1,0,1]$.

We use $\holant{\mathcal{F}}{\mathcal{G}}$ to denote the Holant problem over signature grids with a bipartite graph $H = (U,V,E)$,
where each vertex in $U$ or $V$ is assigned a signature in $\mathcal{F}$ or $\mathcal{G}$,
respectively.
Signatures in $\mathcal{F}$ are considered as row vectors (or covariant tensors);
signatures in $\mathcal{G}$ are considered as column vectors (or contravariant tensors).
Similarly,
$\plholant{\mathcal{F}}{\mathcal{G}}$ denotes the Holant problem over signature grids with a planar bipartite graph.

For an invertible $2$-by-$2$ matrix $T \in {\rm GL}_2({\mathbb{C}})$
 and a signature $f$ of arity $n$, written as
a column vector (contravariant tensor) $f \in \mathbb{C}^{2^n}$, we denote by
$T^{-1}f = (T^{-1})^{\otimes n} f$ the transformed signature.
  For a signature set $\mathcal{F}$,
define $T^{-1} \mathcal{F} = \{T^{-1}f \mid  f \in \mathcal{F}\}$ the set of
transformed signatures.
For signatures written as
 row vectors (covariant tensors) we define $\mathcal{F} T$ similarly.
Whenever we write $T^{-1} f$ or $T^{-1} \mathcal{F}$,
we view the signatures as column vectors;
similarly for $f T$ or $\mathcal{F} T$ as row vectors.
In the special case of the Hadamard matrix
$H_2 = \frac{1}{\sqrt{2}} \left[\begin{smallmatrix} 1 & 1 \\ 1 & -1 \end{smallmatrix}\right]$,
we also define $\widehat{\mathcal{F}} = H_2  \mathcal{F}$.
Note that $H_2$ is orthogonal.
Since constant factors are immaterial, for convenience we sometime
drop the factor $\frac{1}{\sqrt{2}}$ when using $H_2$.

Let $T \in {\rm GL}_2({\mathbb{C}})$.
The holographic transformation defined by $T$ is the following operation:
given a signature grid $\Omega = (H, \pi)$ of $\holant{\mathcal{F}}{\mathcal{G}}$,
for the same bipartite graph $H$,
we get a new grid $\Omega' = (H, \pi')$ of $\holant{\mathcal{F} T}{T^{-1} \mathcal{G}}$ by replacing each signature in
$\mathcal{F}$ or $\mathcal{G}$ with the corresponding signature in $\mathcal{F} T$ or $T^{-1} \mathcal{G}$.

\begin{theorem}[Valiant's Holant Theorem~\cite{string23}]
 For any $T \in {\rm GL}_2({\mathbb{C}})$,
  \[\Holant(\Omega; \mathcal{F} \mid \mathcal{G}) = \Holant(\Omega'; \mathcal{F} T \mid T^{-1} \mathcal{G}).\]
\end{theorem}

Therefore,
a holographic transformation does not change the complexity of the Holant problem in the bipartite setting.

\subsection{Counting Constraint Satisfaction Problems and $\Holant(\widehat{\mathcal{EQ}}, \widehat{\mathcal{F}})$}

Counting constraint satisfaction problems (\#CSP)
can be defined as a special case of Holant problems.
An instance of $\CSP(\mathcal{F})$ is presented
as a bipartite graph.
There is one node for each variable and for each occurrence
of constraint functions respectively.
Connect a constraint node to  a variable node if the
variable appears in that occurrence
of constraint, with a labeling on the edges
for the order of these variables.
This bipartite graph is also known as the \emph{constraint graph}.
If we attach each variable node with an \textsc{Equality} function,
and consider every edge as a variable, then
the \#CSP is just the Holant problem on this bipartite graph.
Thus
$\CSP(\mathcal{F}) \equiv_T \holant{\mathcal{EQ}}{\mathcal{F}}$,
where $\mathcal{EQ} = \{{=}_1, {=}_2, {=}_3, \dotsc\}$ is the set of \textsc{Equality} signatures of all arities.
By restricting to planar constraint graphs,
we have the planar \#CSP framework,
which we denote by $\PlCSP$.
The construction above also shows that $\PlCSP(\mathcal{F}) \equiv_T \plholant{\mathcal{EQ}}{\mathcal{F}}$.

For any positive integer $d$,
the problem $\CSP^d(\mathcal{F})$ is the same as $\CSP(\mathcal{F})$ except that every variable appears a multiple of $d$ times.
Thus,
$\PlCSP^d(\mathcal{F}) \equiv_T \plholant{\mathcal{EQ}_d}{\mathcal{F}}$,
where $\mathcal{EQ}_d = \{{=}_d, {=}_{2 d}, {=}_{3 d}, \dotsc\}$ is the set of \textsc{Equality} signatures of arities that are multiples of $d$.
For $d=1$, we have just $\CSP$ problems. For $d=2$,
these are $\CSP$ problems where every variable appears an even number of times.
If $d \in \{1,2\}$,
then we further have
\begin{equation} \label{eqn:prelim:PlCSPd_equiv_Holant}
 \PlCSP^d(\mathcal{F}) \equiv_T \plholant{\mathcal{EQ}_d}{\mathcal{F}} \equiv_T \PlHolant(\mathcal{EQ}_d, \mathcal{F}).
\end{equation}
The first equivalence is by definition.
The reduction from left to right in the second equivalence is trivial.
For the other direction,
we take a signature grid for  $\PlHolant(\mathcal{EQ}_d, \mathcal{F})$
and create a bipartite signature grid for
$\plholant{\mathcal{EQ}_d}{\mathcal{F}}$
such that both signature grids have the same Holant value up to an easily computable factor.
If two signatures in $\mathcal{F}$ are assigned to adjacent vertices,
then we subdivide all edges between them and assign the binary \textsc{Equality} signature $({=}_2) \in \mathcal{EQ}_d$ to all new vertices.
Suppose \textsc{Equality} signatures $({=}_n), ({=}_m) \in \mathcal{EQ}_d$ are assigned to adjacent vertices connected by $k$ edges.
If $n = m = k$,
then we simply remove these two vertices.
The Holant of the resulting signature grid differs from the original by a factor of~$2$.
Otherwise,
we contract all $k$ edges, merge the two
vertices, and assign $({=}_{n+m-2k}) \in \mathcal{EQ}_d$ to the new vertex.

By the holographic transformation defined by the matrix $H_2 = \frac{1}{\sqrt{2}} \trans{1}{1}{1}{-1}$
(or equivalently without the nonzero factor $\frac{1}{\sqrt{2}}$ since this does not affect the complexity), we have
\begin{equation}\label{csp-holant}
\PlCSP(\mathcal{F})\equiv_T\PlHolant(\widehat{\mathcal{EQ}}, \widehat{\mathcal{F}}),
\end{equation}
where
 $\widehat{\mathcal{EQ}}=\{[1, 0], [1, 0, 1], [1, 0, 1, 0], \ldots\}$
(where we ignore nonzero factors) and $\widehat{\mathcal{F}}=H_2\mathcal{F}$.
 This equivalence (\ref{csp-holant}) plays a central role in our proof.

The next lemma shows that if we have $(=_4)$ in $\PlHolant(\widehat{\mathcal{EQ}}, \widehat{\mathcal{F}})$,
then we can construct $(=_{2k})$ for any $k\in\mathbb{Z}^+$.
 \begin{lemma}\label{equality-4-to-all-even-equality}\label{constructing-even-equality-by-equality-4}
 \begin{equation}
\PlHolant(\widehat{\mathcal{EQ}}, \mathcal{EQ}_2, \widehat{\mathcal{F}})
\equiv_{\rm T} \PlHolant(\widehat{\mathcal{EQ}}, (=_4), \widehat{\mathcal{F}}),
\end{equation}
 \end{lemma}
 \begin{proof}
One direction is trivial, since $(=_4) \in \mathcal{EQ}_2$.
For the other direction
 we prove by induction.
For $k=1$, we have $(=_2)\in\widehat{\mathcal{EQ}}$.
 For $k=2$, we have $(=_4)$ given.
 Assume that we have $(=_{2(k-1)})$.
Then connecting $(=_{2(k-1)})$ and $(=_4)$ by one edge we get $(=_{2k})$.
 \end{proof}
 By (\ref{eqn:prelim:PlCSPd_equiv_Holant}), we have
 \begin{equation*}
 \PlCSP^2(\widehat{\mathcal{EQ}}, \mathcal{F}) \equiv_T \PlHolant(\widehat{\mathcal{EQ}}, \mathcal{EQ}_2, \mathcal{F}).
\end{equation*}
Thus Lemma~\ref{equality-4-to-all-even-equality} implies that
\begin{equation*}
 \PlCSP^2(\widehat{\mathcal{EQ}}, \mathcal{F}) \leq_T\PlHolant(\widehat{\mathcal{EQ}}, (=_4), \widehat{\mathcal{F}}).
\end{equation*}

 \begin{definition}\label{def:crossover}
 The crossover function $\mathfrak{X}$ is a signature
of arity 4 which satisfies $f_{0000}=f_{1111}=f_{0101}=f_{1010}=1$ and
 $f_{\alpha}=0$ for all other $\alpha\in\{0, 1\}^4$.
 \end{definition}

The crossover function $\mathfrak{X}$ on $(x_1, x_2, x_3, x_4)$
is the tensor product of two
binary {\sc Equality} functions $(=_2)$ on $(x_1, x_3)$ and
on $(x_2, x_4)$.
 If we can obtain $\mathfrak{X}$ (by some construction or
reduction) in
 $\PlCSP(\mathcal{F})$,
then we can reduce $\CSP(\mathcal{F})$ to $\PlCSP(\mathcal{F})$.
The same is true
for $\Holant(\widehat{\mathcal{EQ}}, \widehat{\mathcal{F}})$
and  $\PlHolant(\widehat{\mathcal{EQ}}, \widehat{\mathcal{F}})$.
 Moreover, note that $H_2^{\otimes 4}(\mathfrak{X})=\mathfrak{X}$,
because an orthogonal transformation does not change a
binary {\sc Equality} function $(=_2)$.
 So we can obtain $\mathfrak{X}$ in $\PlCSP(\mathcal{F})$ iff we can
obtain
 $\mathfrak{X}$ in $\PlHolant(\widehat{\mathcal{EQ}}, \widehat{\mathcal{F}})$.

\subsection{Realization}

One basic notion used throughout the paper is realization.
We say a signature $f$ is \emph{realizable} or \emph{constructible} from a signature set $\mathcal{F}$
if there is a gadget with some dangling edges such that each vertex is assigned a signature from $\mathcal{F}$,
and the resulting graph,
when viewed as a black-box signature with inputs on the dangling edges,
is exactly $f$.
If $f$ is realizable from a set $\mathcal{F}$,
then we can freely add $f$ into $\mathcal{F}$ while preserving the complexity.
\input{f-gate}

Formally,
this notion is defined by an $\mathcal{F}$-gate.
An $\mathcal{F}$-gate is similar to a signature grid $(G, \pi)$ for $\Holant(\mathcal{F})$ except that $G = (V,E,D)$ is a graph with some dangling edges $D$.
The dangling edges define external variables for the $\mathcal{F}$-gate.
(See Figure~\ref{fig:Fgate} for an example.)
We denote the regular edges in $E$ by $1, 2, \dotsc, m$ and the dangling edges in $D$ by $m+1, \dotsc, m+n$.
Then we can define a function $f$ for this $\mathcal{F}$-gate as
\[
f(y_1, \dotsc, y_n) = \sum_{x_1, \dotsc, x_m \in \{0, 1\}} H(x_1, \dotsc, x_m, y_1, \dotsc, y_n),
\]
where $(y_1, \dotsc, y_n) \in \{0, 1\}^n$ is an assignment on the dangling edges
and $H(x_1, \dotsc, x_m, y_1, \dotsc, y_n)$ is the value of the signature grid on an assignment of all edges in $G$,
which is the product of evaluations at all vertices in $V$.
We also call this function $f$ the signature of the $\mathcal{F}$-gate.

An $\mathcal{F}$-gate is planar if the underlying graph $G$ is a planar graph,
and the dangling edges,
ordered counterclockwise corresponding to the order of the input variables,
are in the outer face in a planar embedding.
A planar $\mathcal{F}$-gate can be used in a planar signature grid as if it is just a single vertex with the particular signature.

Using  planar $\mathcal{F}$-gates,
we can reduce one planar Holant problem to another.
Suppose $g$ is the signature of some planar $\mathcal{F}$-gate.
Then $\PlHolant(\mathcal{F}, g) \leq_T \PlHolant(\mathcal{F})$.
The reduction is simple.
Given an instance of $\PlHolant(\mathcal{F}, g)$,
by replacing every occurrence of $g$ by the $\mathcal{F}$-gate,
we get an instance of $\PlHolant(\mathcal{F})$.
Since the signature of the $\mathcal{F}$-gate is $g$,
the Holant values for these two signature grids are identical.

When a gadget has an asymmetric signature,
we place a diamond on the edge corresponding to the first input.
The remaining inputs are ordered counterclockwise around the vertex.

\begin{definition}[Derivative] \label{derivative}
Let $f$ and
 $g$ be two signatures of arities $n$ and $m$ respectively,
and $n > m$.
We connect all $m$ input edges $(1 \le j \le m$) of  $g$ to
 $m$ consecutive edges of $f$ in a clockwise order,
indexed $i-1, \ldots, i-m \pmod n$.
The  \emph{derivative signature}  $\partial^{\{i-1, \ldots, i-m\}}_g(f)$
is the signature of this planar $\{f, g\}$-gate of arity $n-m$,
whose variables are the unmatched variables of $f$
in the original counterclockwise order starting with the first
unmatched variable.
(The clockwise order of edges of $f$ to be matched with
the counterclockwise order of edges of $g$ ensures planarity.)

If $f$ is symmetric, we will simply write $\partial_g(f)$
 since the derivative signature is independent of
the choice of $i$ in this case.
Moreover, if $kn<m$ and we connect $k$ copies of $g$ to $f$,
which is the same as forming $\partial_g(f)$ sequentially
$k$ times, the resulting \emph{repeated derivative signature} is
 denoted by $\partial^k_g(f)$.

For a unary signature $u$,
we can connect a copy of $u$
to each edge of $f$ indexed by a subset $S \subset [n]$, and
we also denote the resulting signature by $\partial^{S}_u(f)$.

For convenience, we use $f^{x_i=0}$ to denote $\partial^{\{i\}}_{[1, 0]}{f}$ and $f^{x_i=1}$ to denote $\partial^{\{i\}}_{[0, 1]}{f}$.

\end{definition}

We use the signature matrix
\[M_{x_1, x_2x_3}(f)=\begin{bmatrix}
f_{000} & f_{001} & f_{010} & f_{011}\\
f_{100} & f_{101} & f_{110} & f_{111}
\end{bmatrix}\]
to denote a ternary signature and
use the signature matrix
\[M_{x_1x_2, x_4x_3}(f)=\begin{bmatrix}
f_{0000} & f_{0010} & f_{0001} & f_{0011}\\
f_{0100} & f_{0110} & f_{0101} & f_{0111}\\
f_{1000} & f_{1010} & f_{1001} & f_{1011}\\
f_{1100} & f_{1110} & f_{1101} & f_{1111}
\end{bmatrix}\]
to denote a signature of arity 4.
Note that in $M_{x_1x_2, x_4x_3}(f)$, the rows are indexed by $x_1x_2$,
 and the columns are indexed by $x_4x_3$ (not $x_3x_4$),
both in lexicographic order.
This reversal of column index ensures
 that the  signature matrix of linking two arity 4 signatures
is simply the matrix
product of the two signature matrices.

For example, the signature matrix of the crossover function $\mathfrak{X}$ is
\[M_{x_1x_2, x_4x_3}(\mathfrak{X})=\begin{bmatrix}
1 & 0 & 0 & 0\\
0 & 0 & 1 & 0\\
0 & 1 & 0 & 0\\
0 & 0 & 0 & 1
\end{bmatrix}.\]

When we rotate a signature, the transformation of its signature matrix
is depicted in Figure~\ref{fig:rotate_asymmetric_signature}.
\begin{figure}[ht]
 \centering
 \def\capWidth{6cm}
 \captionsetup[subfigure]{width=\capWidth}
 \tikzstyle{entry} = [internal, inner sep=2pt]
 \subfloat[A counterclockwise rotation]{
  \begin{tikzpicture}[scale=\scale,transform shape,node distance=1.7 * \nodeDist,semithick]
   \node[internal]  (0)                    {};
   \node[external]  (1) [above  left of=0] {};
   \node[external]  (2) [above right of=0] {};
   \node[external]  (3) [below  left of=0] {};
   \node[external]  (4) [below right of=0] {};
   \node[external]  (5) [      right of=0] {};
   \node[external]  (6) [      right of=5] {};
   \node[internal]  (7) [      right of=6] {};
   \node[external]  (8) [above  left of=7] {};
   \node[external]  (9) [above right of=7] {};
   \node[external] (10) [below  left of=7] {};
   \node[external] (11) [below right of=7] {};
   \path (0) edge[postaction={decorate, decoration={
                                         markings,
                                         mark=at position 0.25 with {\arrow[>=diamond,white] {>}; },
                                         mark=at position 0.25 with {\arrow[>=open diamond]  {>}; },
                                         mark=at position 0.65 with {\arrow[>=diamond,white] {>}; },
                                         mark=at position 0.65 with {\arrow[>=open diamond]  {>}; } } }] (1)
             edge (2)
             edge (3)
             edge (4)
    (5.west) edge[->, >=stealth] (6.east)
         (7) edge[postaction={decorate, decoration={
                                         markings,
                                         mark=at position 0.65 with {\arrow[>=diamond,white] {>}; },
                                         mark=at position 0.65 with {\arrow[>=open diamond]  {>}; } } }] (8)
             edge (9)
             edge [postaction={decorate, decoration={
                                         markings,
                                         mark=at position 0.25 with {\arrow[>=diamond,white] {>}; },
                                         mark=at position 0.25 with {\arrow[>=open diamond]  {>}; } } }] (10)
             edge (11);
   \begin{pgfonlayer}{background}
    \node[draw=\borderColor,thick,rounded corners,fit = (0),inner sep=16pt] {};
    \node[draw=\borderColor,thick,rounded corners,fit = (7),inner sep=16pt] {};
   \end{pgfonlayer}
  \end{tikzpicture}}
 \qquad
 \subfloat[Movement of signature matrix entries]{
  \makebox[\capWidth][c]{
   \begin{tikzpicture}[scale=\scale,transform shape,>=stealth,node distance=\nodeDist,semithick]
    \node[entry] (11)               {};
    \node[entry] (12) [right of=11] {};
    \node[entry] (13) [right of=12] {};
    \node[entry] (14) [right of=13] {};
    \node[entry] (21) [below of=11] {};
    \node[entry] (22) [right of=21] {};
    \node[entry] (23) [right of=22] {};
    \node[entry] (24) [right of=23] {};
    \node[entry] (31) [below of=21] {};
    \node[entry] (32) [right of=31] {};
    \node[entry] (33) [right of=32] {};
    \node[entry] (34) [right of=33] {};
    \node[entry] (41) [below of=31] {};
    \node[entry] (42) [right of=41] {};
    \node[entry] (43) [right of=42] {};
    \node[entry] (44) [right of=43] {};
    \node[external] (nw) [above left  of=11] {};
    \node[external] (ne) [above right of=14] {};
    \node[external] (sw) [below left  of=41] {};
    \node[external] (se) [below right of=44] {};
    \path (13) edge[<-, dotted]                (12)
          (12) edge[<-, dotted]                (21)
          (21) edge[<-, dotted]                (31)
          (31) edge[<-, dotted,out=65,in=-155] (13)
          (42) edge[<-, dashed]                (43)
          (43) edge[<-, dashed]                (34)
          (34) edge[<-, dashed]                (24)
          (24) edge[<-, dashed,out=-115,in=25] (42)
          (14) edge[<-, very thick]            (22)
          (22) edge[<-, very thick]            (41)
          (41) edge[<-, very thick]            (33)
          (33) edge[<-, very thick]            (14)
          (23) edge[<->]                      (32);
    \path (nw.west) edge (sw.west)
          (ne.east) edge (se.east)
          (nw.west) edge (nw.east)
          (sw.west) edge (sw.east)
          (ne.west) edge (ne.east)
          (se.west) edge (se.east);
   \end{tikzpicture}}}
 \caption{The movement of the entries in the signature matrix of a
 signature of arity 4 under a counterclockwise rotation of the input edges.
  Entires of Hamming weight~$1$ are in the dotted cycle,
  entires of Hamming weight~$2$ are in the two solid cycles (one has length~$4$ and the other one is a swap),
  and entries of Hamming weight~$3$ are in the dashed cycle.}
 \label{fig:rotate_asymmetric_signature}
\end{figure}

\subsection{Tractable Signature Sets}

We define some signatures that are known
 to be tractable~\cite{caiguowilliams13,Guo-Williams}.
These form three families: affine signatures,
product-type signatures,
and
matchgate signatures.

\paragraph{Affine Signatures}

\begin{definition}
For a signature $f$ of arity $n$,
the support of $f$ is
\[\operatorname{supp}(f)=\{(x_1, x_2, \ldots,  x_n)
\in \mathbb{Z}_2^n \mid f(x_1, x_2, \ldots, x_n)\neq 0\}.\]
\end{definition}

\begin{definition}
Let $f$ be a signature of arity $n$.
We say $f$ has affine support of dimension $k$ if
$\operatorname{supp}(f)$ is an affine subspace of
dimension $k$ over $\mathbb{Z}_2$, i.e.,
there is a matrix $A$  over $\mathbb{Z}_2$
such that
 $f(x_1, x_2, \ldots, x_n)\neq 0$ iff $AX=0$, where $X = (x_1, x_2, \ldots,
 x_n, 1)$ and
the affine space $\{(x_1, x_2, \ldots, x_n)
\in \mathbb{Z}_2^n \mid AX=0\}$ has dimension $k$.
\end{definition}

For a signature  of arity $n$  with affine support of dimension $k$,
let $X=\{x_{i_1}, x_{i_2}, \ldots, x_{i_k}\}$ be a set of free variables, where $i_1< i_2< \ldots < i_k$.
Then on $\operatorname{supp}(f)$,
 $f(x_1, x_2, \ldots, x_n)$ is uniquely determined by the input
on $X$.
\begin{definition}
If $f$ has affine support of dimension $k$, and
$X$ is a set of free variables,
then  $\underline{f_{X}}$ is the \emph{compressed signature}
of $f$ for $X$ such that
$\underline{f_X}(x_{i_1}, \ldots, x_{i_k})=f(x_1, x_2, \ldots, x_n)$,
where $(x_1, x_2, \ldots, x_n) \in \operatorname{supp}(f)$.
When it is clear from the context, we omit $X$ and use $\underline{f}$
 to denote $\underline{f_X}$.
\end{definition}

\begin{definition}\label{definition-affine}
 A signature $f(x_1, \ldots, x_n)$ of arity $n$
is \emph{affine} if it has the form
 \[
  \lambda \cdot \chi_{A X = 0} \cdot {\frak i} ^{Q(X)},
 \]
 where $\lambda \in \mathbb{C}$,
 $X = (x_1, x_2, \dotsc, x_n, 1)$,
 $A$ is a matrix over $\mathbb{Z}_2$,
 $Q(x_1, x_2, \ldots, x_n)\in \mathbb{Z}_4[x_1, x_2, \ldots, x_n]$
is a quadratic (total degree at most 2) multilinear polynomial
 with the additional requirement that the coefficients of all
 cross terms are even, i.e., $Q$ has the form
 \[Q(x_1, x_2, \ldots, x_n)=a_0+\displaystyle\sum_{k=1}^na_kx_k+\displaystyle\sum_{1\leq i<j\leq n}2b_{ij}x_ix_j,\]
 and $\chi$ is a 0-1 indicator function
 such that $\chi_{AX = 0}$ is~$1$ iff $A X = 0$.
 We use $\mathscr{A}$ to denote the set of all affine signatures.
\end{definition}

In \cite{CLX14}, there is an alternative definition for affine signatures.
\begin{definition}\label{affine-definition-by-vector}
A signature $f(x_1, x_2, \dotsc, x_n)$ of arity $n$
 is \emph{affine} if it has the form
 \[
  \lambda \cdot \chi_{A X = 0} \cdot {\frak i}^{\sum_{j=1}^k\langle\mathbf{v}_j, X\rangle},
 \]
 where $\lambda \in \mathbb{C}$,
 $X = (x_1, x_2, \dotsc, x_n, 1)$,
 $A$ is a matrix over $\mathbb{Z}_2$,
 $\mathbf{v}_j$ is a vector over $\mathbb{Z}_2$,
 and $\chi$ is a 0-1 indicator signature such that $\chi_{AX = 0}$ is~$1$ iff $A X = 0$.
 Note that the dot product $\langle \mathbf{v}_j, X \rangle$ is calculated over $\mathbb{Z}_2$ with a 0-1 output in $\mathbb{Z}$,
 while the summation $\sum_{j=1}^k$ on the exponent of ${\frak i} = \sqrt{-1}$ is evaluated as a sum mod~$4$ of 0-1 terms.
\end{definition}
Definition~\ref{definition-affine} and Definition~\ref{affine-definition-by-vector}
are equivalent.
To see this we observe that each $\langle \mathbf{v}_j, X \rangle$
 as an integer sum$\pmod 2$ can be replaced
by  $(\langle \mathbf{v}_j, X \rangle)^2$ as an integer sum$\pmod 4$
since $N \equiv 0, 1\pmod 2$ iff
$N^2\equiv 0, 1\pmod 4$ respectively, for any integer $N$.
After this, all cross terms have even coefficients and all square terms
$x_s^2$ can be replaced by $x_s$ since  $x_s=0, 1$.
 Conversely, we can express $Q \pmod 4$ as a sum of
squares of affine forms of $X$,
using the condition that all cross terms have even coefficients.

The following lemma shows that for
a $\{\pm 1, \pm {\frak i}\}$-valued signature
of arity $k$,
there exists a unique
 multilinear polynomial $P(x_1, \ldots, x_k)\mod 4$ such that
$f(x_1, \ldots x_k)={\frak i}^{P(x_1, \ldots, x_k)}$.
Thus if there exists a multilinear polynomial $P(x_1, \ldots, x_k)$
such that
$f(x_1, \ldots x_k)={\frak i}^{P(x_1, \ldots, x_k)}$,
and $P$ has total degree
 greater than 2 or has a cross term with an odd coefficient,
then $f \not \in \mathscr{A}$.
\begin{lemma}\label{multilinear-polynomial-unique-affine-definition}
Let $f$ be a signature of arity $k$ taking values in
 $\{\pm 1, \pm {\frak i}\}$.
Then there exists a unique multilinear polynomial $P(x_1, \ldots, x_k)
\in\mathbb{Z}_4[x_1, \ldots, x_k]$ such that
$f(x_1,  \ldots x_k)={\frak i}^{P(x_1, \ldots, x_k)}$.

Similarly, if $f$ is a signature of arity $k$ taking values in
 $\{\pm 1\}$, then
there exists a unique multilinear polynomial $P(x_1,  \ldots, x_k)\in\mathbb{Z}_2[x_1, \ldots, x_k]$ such that
$f(x_1,  \ldots x_k)=(-1)^{P(x_1,  \ldots, x_k)}$.
\end{lemma}
\begin{proof}
We prove the first statement.
 The proof for the second is similar and we omit it here.

For any input $\alpha=a_1\ldots a_k\in\{0, 1\}^k$,
there exists $r_{\alpha}\in\{0, 1, 2, 3\}$
such that $f_{\alpha}={\frak i}^{r_{\alpha}}$
since $f$ takes values in $\{\pm 1, \pm {\frak i}\}$.
Let $P(x_1, \ldots, x_k)
=\displaystyle\sum_{\alpha\in\{0, 1\}^k}r_{\alpha}
\displaystyle\prod_{i=1}^k \widetilde{x_i} \in\mathbb{Z}_4[x_1, \ldots, x_k]$,
where $\alpha=a_1\ldots a_k$,
$\widetilde{x_i} = x_i$ if $a_i =1$ and $1-x_i$ if $a_i =0$.
Then $f(x_1,  \ldots x_k)=i^{P(x_1, \ldots, x_k)}$.

Now we prove that $P(x_1, \ldots, x_k)$ is unique.
It is equivalent to prove that if $f$ is the constant 1
function then $P(x_1, \ldots, x_k) =0$ in
$\mathbb{Z}_4[x_1, \ldots, x_k]$.
For a contradiction suppose $r \displaystyle\prod_{i\in S} x_i$
is a nonzero term in $P(x_1, \ldots, x_k)$ with minimum $|S|$.
Set $x_i =1$ for all $i \in S$, and all other $x_i=0$.
Then $P$ evaluates to $r \not =0$ in $\mathbb{Z}_4$,
and $f$ evaluates to ${\frak i}^r \not = 1$.
This is a contradiction.
\end{proof}

By Lemma~\ref{multilinear-polynomial-unique-affine-definition}, we directly have the following Corollary.
\begin{corollary}\label{corollary-compressed-multilinear-unique}
Let $f  \in \mathscr{A}$ be a signature of arity $n$
with support of dimension $k$.
Let $X = \{x_{i_1}, \ldots, x_{i_k}\}$ be a set of free variables.
Then there exists a unique $Q(X)
\in\mathbb{Z}_4[X]$
such that \[f(x_1, x_2, \ldots, x_n)={\frak i}^{Q(X)}\]
for $(x_1, x_2, \ldots, x_n)\in \operatorname{supp}(f)$,
where $Q(X)$ is a quadratic multilinear polynomial
and the coefficients of  cross terms are even.
\end{corollary}

\begin{corollary}\label{f-affine-iff-f*-affine}
Let $f$ be  a signature of arity $n$ having
affine support of dimension $k$.
Suppose $f$ takes values in $\{0, \pm 1, \pm {\frak i}\}$,
and $X=\{x_{i_1}, \ldots, x_{i_k}\}$ is a set of free variables.
Then $f \in \mathscr{A}$ iff $\underline{f} \in \mathscr{A}$,
where $\underline{f}$ is the compressed signature of $f$ for $X$.
\end{corollary}
\begin{proof}
Note that if $\underline{f}(x_{i_1}, \ldots, x_{i_k})
={\frak i}^{Q(x_{i_1}, \ldots, x_{i_k})}$,
then $f(x_1, \ldots, x_n)={\frak i}^{Q(x_{i_1}, \ldots, x_{i_k})}$
for $(x_1, \ldots, x_n)\in  \operatorname{supp}(f)$.
So $f  \in \mathscr{A}$ iff $\underline{f} \in \mathscr{A}$.
\end{proof}

The next two lemmas allow us to easily determine
 if a binary or ternary signature is affine.

\begin{lemma}\label{binary-affine-compressed function}
 Let $f$ be a binary signature and $M_{x_1, x_2}(f)=\left[\begin{smallmatrix}
 f_{00} & f_{01}\\
 f_{10} & f_{11}
 \end{smallmatrix}\right]
 =\left[\begin{smallmatrix}
 1 & b\\
 c & d
 \end{smallmatrix}\right]=\left[\begin{smallmatrix}
 1 & {\frak i}^s\\
 {\frak i}^t & d
 \end{smallmatrix}\right]$,
  where $r, s\in\{0, 1, 2, 3\}$. Then $f\in\mathscr{A}$  iff $d=\pm {\frak i}^{r+s}$.
In particular, if $b, c, d\in\{1, -1\}$ then $f \in\mathscr{A}$.
\end{lemma}
\begin{proof}
Let  $Q(x_1, x_2)=sx_1+rx_2$
if $d={\frak i}^{r+s}$, and
$sx_1+rx_2+2x_1x_2$ if $d=-{\frak i}^{r+s}$.
Then $f(x_1, x_2)={\frak i}^{Q(x_1, x_2)}$. Thus $f\in\mathscr{A}$
by Definition~\ref{definition-affine}.
Conversely, if
$f(x_1, x_2)={\frak i}^{Q(x_1, x_2)}$,
for some
 $Q(x_1, x_2)=a_0+a_1x_1+a_2x_2+2b_{12}x_1x_2\in\mathbb{Z}_4[x_1, x_2]$,
then we have  $a_0=0$ by $f_{00}=1$, $a_1=s$ by $f_{10}={\frak i}^s$, and
 $a_2=r$ by  $f_{01}={\frak i}^r$.
Thus $f_{11}={\frak i}^{r+s+2b_{12}}=\pm {\frak i}^{r+s}$.
\end{proof}

\begin{lemma}\label{ternary-affine-compressed function}
Let $f$ be a ternary signature and
\[M_{x_1, x_2x_3}(f)=\begin{bmatrix}
f_{000} & f_{001} & f_{010} & f_{011}\\
f_{100} & f_{101} & f_{110} & f_{111}
\end{bmatrix}=\begin{bmatrix}
1 & {\frak i}^r & {\frak i}^s & \epsilon_1 {\frak i}^{r+s}\\
{\frak i}^t & \epsilon_2 {\frak i}^{r+t} & \epsilon_3 {\frak i}^{s+t} & \epsilon_4 {\frak i}^{r+s+t}
\end{bmatrix},\]
where $r, s, t\in\{0, 1, 2, 3\}$, $\epsilon_i\in\{1, -1\}$ for $1\leq i\leq 4$.
Then
 $f \in\mathscr{A}$  iff $\epsilon_1\epsilon_2\epsilon_3\epsilon_4=1$.
\end{lemma}
\begin{proof}
We can write $\epsilon_i=(-1)^{a_i}$ for $a_i\in\{0, 1\}$,  and let
\[Q=tx_1+sx_2+rx_3+2a_1x_2x_3+2a_2x_3x_1+2a_3x_1x_2+2(a_1+a_2+a_3+a_4)x_1x_2x_3
\in \mathbb{Z}_4[x_1, x_2, x_3].\]
Then $f(x_1, x_2, x_3)={\frak i}^{Q(x_1, x_2, x_3)}$.

By Lemma~\ref{multilinear-polynomial-unique-affine-definition},
 $f \in \mathscr{A}$ iff $Q$ is a
multilinear
quadratic polynomial and the coefficients of the cross terms are even.
Thus $f  \in \mathscr{A}$  iff $2(a_1+a_2+a_3+a_4)\equiv 0\pmod 4$.
This is equivalent to  $a_1+a_2+a_3+a_4\equiv 0\pmod 2$, i.e.,
 $\epsilon_1\epsilon_2\epsilon_3\epsilon_4=1$.
\end{proof}

In addition, we
often use the following facts that can be derived from Definition~\ref{definition-affine} directly.
\begin{proposition}\label{A-has-same-norm-etc}
The following hold by definition.
\begin{itemize}
\item For any signature $f \in \mathscr{A}$, up to a
nonzero factor, all nonzero entries are powers of ${\frak i}$.
In particular they have the same norm.

\item If a signature has only one nonzero entry, then it is affine.
In particular $[1, 0], [0, 1] \in \mathscr{A}$.

\item If a signature $f$ has only two nonzero entries $f_{\alpha}$
and $f_{\beta}$, then
the support of $f$ is affine. Moreover,
in this case $f \in \mathscr{A}$ iff $f_{\alpha}^4=f_{\beta}^4$.

\item
$[1, a], [1, 0, a]$ are affine iff $a^4=0, 1$.
\end{itemize}
\end{proposition}

The following lemma is useful in proving arity reductions
for non-affine signatures.

\begin{lemma}\label{[1,0]-[0,1]-pinning-implies-affine-arity-4}
Let $f$ be a signature of arity $n$  with
affine support of dimension $k\geq 4$. If
$f^{x_i=0} \in \mathscr{A}$ and $f^{x_i=1} \in \mathscr{A}$
 for $1\leq i\leq n$,
then $f \in \mathscr{A}$.
\end{lemma}
\begin{proof}
Let $X=\{x_{i_1}, x_{i_2}, \ldots, x_{i_k}\}$ be a set of free variables
of $f$ and let $\underline{f}$ be the compressed signature of $f$ for $X$.
Since both  $f^{x_{i_1}=0}$ and $f^{x_{i_1}=1}$ are affine,
$\underline{f}^{x_{i_1}=0}$ and $\underline{f}^{x_{i_1}=1}$
are affine by Corollary~\ref{f-affine-iff-f*-affine}.
By Corollary~\ref{corollary-compressed-multilinear-unique}
 there exist $Q_0(x_{i_2}, \ldots, x_{i_k})$
and $Q_1(x_{i_2}, \ldots, x_{i_k})$
such that $\underline{f}(0, x_{i_2}, \ldots, x_{i_k})={\frak i}^{Q_0(x_{i_2}, \ldots, x_{i_k})}$
and $\underline{f}(1, x_{i_2}, \ldots, x_{i_k})={\frak i}^{Q_1(x_{i_2}, \ldots, x_{i_k})}$,
where $Q_0$ and $Q_1$ are quadratic multilinear polynomials
in $\mathbb{Z}_4[x_{i_2}, \ldots, x_{i_k}]$,
 and the coefficients of all cross terms are even.

Let $Q(x_{i_1}, x_{i_2}, \ldots, x_{i_k})
=(1 - x_1)Q_0+x_1Q_1$, then $\underline{f}(x_{i_1}, x_{i_2}, \ldots, x_{i_k})={\frak i}^{Q(x_{i_1}, x_{i_2}, \ldots, x_{i_k})}$.
Note that the total degree of $Q$ is at most 3.
If $Q$ has total degree at most 2
 and the coefficients of all cross terms are even, then $\underline{f}$
is affine. Thus $f$ is affine by Corollary~\ref{f-affine-iff-f*-affine} and we are done.

Otherwise, either there is a cross term $x_{i_s}x_{i_t}$
($1 \le s < t \le k$) with odd coefficient $a_{st}$
 or there is a term
$x_1x_{i_s}x_{i_t}$ ($2 \le s < t \le k$)
 with coefficient $a_{1st}\neq 0$ in $Q$.
Since $k\geq 4$, there exists some $r\in[k]\setminus\{1, s, t\}$.
Then $Q^{x_{i_r} =0}
= Q(x_{i_1},\ldots,  x_{i_{r-1}}, 0 , x_{i_{r+1}}, \ldots , x_{i_k})$ has
a cross term $x_{i_s}x_{i_t}$ with odd coefficient $a_{st}$ or a term
$x_1x_{i_s}x_{i_t}$ with coefficient $a_{1st}\neq 0$.
Note that $\underline{f}^{x_{i_r}=0}={\frak i}^{Q(x_{i_1},\ldots,  x_{i_{r-1}}, 0 , x_{i_{r+1}}, \ldots , x_{i_k})}$.
Thus $\underline{f}^{x_{i_r}=0}$ is not affine by
Lemma~\ref{multilinear-polynomial-unique-affine-definition}.
So $f^{x_{i_r}=0}$ is not affine. This is a contradiction.
\end{proof}

Let
\begin{align*}
 \mathscr{F}_1 &= \left\{\lambda \left([1,0]^{\otimes k} + {\frak i}^r [0, 1]^{\otimes k}\right) | \lambda \in \mathbb{C}, k = 1, 2, \ldots, r = 0, 1, 2, 3\right\},\\
 \mathscr{F}_2 &= \left\{\lambda \left([1,1]^{\otimes k} + {\frak i}^r [1,-1]^{\otimes k}\right) | \lambda \in \mathbb{C}, k = 1, 2, \ldots, r = 0, 1, 2, 3\right\}, \\
 \mathscr{F}_3 &= \left\{\lambda \left([1,{\frak i}]^{\otimes k} + {\frak i}^r [1,-{\frak i}]^{\otimes k}\right) | \lambda \in \mathbb{C}, k = 1, 2, \ldots, r = 0, 1, 2, 3\right\}.
\end{align*}
It is known that the set of non-degenerate symmetric signatures in
$\mathscr{A}$ consists of precisely the nonzero signatures
in  $\mathscr{F}_1 \bigcup \mathscr{F}_2 \bigcup \mathscr{F}_3$ with arity at least~$2$ ($\lambda \neq 0$).

\paragraph{Product-Type Signatures}
\begin{definition}
\label{definition-product-2}
 A signature on a set of variables $X$
 is of \emph{product type} if it can be expressed as a
product of unary functions,
 binary equality functions $([1,0,1])$,
and binary disequality functions $([0,1,0])$, each on one or two
variables of $X$.
 We use $\mathscr{P}$ to denote the set of product-type functions.
\end{definition}

 A symmetric signature of the form $[a, 0, \ldots, 0, b]$ is called a
{\sc Generalized Equality}.
\begin{proposition}(cf.~Lemma~A.1 in the full version of~\cite{HL12})\label{symmetric:prod:type:signatures}
 If $f$ is a symmetric signature in $\mathscr{P}$,
then $f$ is either degenerate,
binary \textsc{Disequality} $(\neq_2) = [0,1,0]$,
or $[a,0,\dotsc,0,b]$ for some $a, b \in \mathbb{C}$.
\end{proposition}

\begin{corollary}\label{[1,0,1,0]:not:prod}
$[1, 0, 1, 0]\notin\mathscr{P}$.
\end{corollary}

We will use Corollary~\ref{[1,0,1,0]:not:prod} in the proof of Theorem~\ref{dichotomy-csp-2}.

Definition~\ref{definition-product-2} is succinct
and is from \cite{CLX14}. But to deal with asymmetric signatures,
an alternative definition of $\mathscr{P}$ given in~\cite{CLX11a} is
useful. This is given below in Definition~\ref{definition-product-1}.
To state it we need some notations.

Suppose $f$ is a signature of arity $n$ and $\mathcal{I}=\{I_1, I_2, \ldots, I_k\}$ is a partition of $[n]$.
If $f(X)=\displaystyle\prod_{j=1}^k f_j(X|_{I_j})$ for some signatures $f_1, f_2, \ldots, f_k$, where $X=\{x_1, x_2, \ldots, x_n\}$
and $X|_{I_j}=\{x_s|s\in I_j\}$ (we also denote it by $X_j$), then we
say $f$ can be decomposed as a tensor product of $f_1, f_2, \ldots, f_k$.
We denote such a function
 by $f=\bigotimes_{\mathcal{I}}(f_1, f_2, \ldots, f_k)$.
If each $f_j$ is the signature of some
$\mathcal{F}$-gate, then $\bigotimes_{\mathcal{I}}(f_1, f_2, \ldots, f_k)$ is
the signature
of the $\mathcal{F}$-gate which is the disjoint union of
 the $\mathcal{F}$-gates for $f_j$, with variables
renamed and ordered according to $\mathcal{I}$.
(In general this is not a planar $\mathcal{F}$-gate
even when the $\mathcal{F}$-gates for all $f_j$ are planar,
unless the sets $I_1, I_2, \ldots, I_k$ partition $[n]$ in order.)
 When
the indexing is clear, we also use
the notation $f_1\otimes f_2\otimes\cdots\otimes f_k$.
Note that this tensor product notation
$\otimes$ is consistent with tensor product of matrices.
We say a signature set $\mathcal{F}$ is closed under tensor product,
if for any partition  $\mathcal{I}=\{I_1, I_2\}$, and any  $f, g\in\mathcal{F}$
on $X_1$ and $X_2$ respectively,
we have $\bigotimes_{\mathcal{I}}(f, g)\in\mathcal{F}$.
The
tensor closure $\langle \mathcal{F}\rangle$
of $\mathcal{F}$
is the minimum set containing $\mathcal{F}$, closed under tensor product.

\begin{definition}\label{definition-product-1}
Let $\mathcal{E}$ be the set of all signatures $f$ such that
$\operatorname{supp}(f)$ is contained in two antipodal points, i.e.,
if $f$ has arity $n$,
then $f$ is zero except on (possibly) two inputs $\alpha
= (a_1, a_2, \ldots, a_n)$
and $\overline{\alpha} =
(\overline{a}_1, \overline{a}_2, \ldots, \overline{a}_n)=(1- a_1, 1- a_2,
 \ldots, 1 - a_n)$.
Then $\mathscr{P}=\langle \mathcal{E}\rangle$.
\end{definition}
We claim that Definition~\ref{definition-product-2} is equivalent to Definition~\ref{definition-product-1}.
If $f\in\mathcal{E}$, then its support
$\operatorname{supp}(f) \subseteq \{\alpha, \bar{\alpha}\}$,
for some  $\alpha=a_1a_2\cdots a_n\in\{0, 1\}^n$.
We may assume that $a_1=0$.
Then
\[f=[f_{\alpha}, f_{\bar{\alpha}}](x_1)\displaystyle\prod_{i=1}^{n-1}f_i(x_i, x_{i+1}),\]
where $[f_{\alpha}, f_{\bar{\alpha}}](x_1)$
is a unary function on $x_1$,
and for all $1 \le i \le n-1$,
 $f_{i}=[1, 0, 1]$ if $a_i=a_{i+1}$ and $f_{i}=[0, 1, 0]$ if $a_i\neq a_{i+1}$.
This implies that $f$ is a product of the unary
function $[f_{\alpha}, f_{\bar{\alpha}}]$, and
binary functions $[1, 0, 1]$ and $[0, 1, 0]$.
Thus all functions in $\langle\mathcal{E}\rangle$ are products of  unary functions, and binary functions
$[1, 0, 1]$ and $[0, 1, 0]$.
Conversely, if $f$ has arity $n$ and
 is a product of  unary functions, and binary functions
 $[1, 0, 1]$ and $[0, 1, 0]$,
then there exist $S\subseteq\{(i, j) \mid i, j\in[n], \mbox{ and } i<j\}$
and $S' \subseteq [n]$, such that
 $f=\displaystyle\prod_{(i, j)\in S}h_{ij}(x_i, x_j)
\displaystyle\prod_{\ell\in S' }u_{\ell}(x_{\ell})$,
where $h_{ij}=[1, 0, 1]$ or $[0, 1, 0]$, and $u_{\ell}$ are unary functions.
Let  $\mathcal{I}=\{I_1, I_2, \ldots, I_k\}$ be the partition of $[n]$
 such that
for any $i, j\in[n]$, $i, j$ are in the same $I_c$ iff $i$ and $j$ belong
to the same connected component of the graph $([n], S)$.
Let \[H_c=\displaystyle\prod_{i, j\in I_c, \/ (i,j) \in S}h_{ij}(x_i, x_j)
\displaystyle\prod_{\ell\in I_c \cap S'}u_{\ell}(x_{\ell})\]
for $1\leq c\leq k$.
Then $H_c\in\mathcal{E}$ and
\[f=\displaystyle\bigotimes_{\mathcal{I}}(H_1, H_2, \ldots, H_k).\]
This implies that $f\in\langle\mathcal{E}\rangle$.

Given a function $f(X)$, if it is the product of two
functions $g$ and $h$ on disjoint proper subsets of variables of $X$,
then $f = g \otimes h$. Clearly every function $f(X)$
has a decomposition as a tensor product $f_1 \otimes f_2  \otimes
\cdots \otimes f_k$ where each $f_i$ is not further expressible
as a tensor product of functions on disjoint proper subsets.
If $f$ is not identically 0, then
such a \emph{primitive decomposition} is unique up to a nonzero constant factor.
To see this, suppose $f = f_1 \otimes f_2  \otimes
\cdots \otimes f_k = g_1 \otimes g_2 \otimes  \cdots \otimes g_{\ell}$
are  two such decompositions. Since $f$ is not
identically 0, all $f_i$ are not identically 0.
For any $i$, there is a partial
 assignment for $f$ to all variables in $X$ except those in $f_i$,
such that the resulting function is
a nonzero constant multiple of $f_i$.
   This gives an expression
$c_i f_i(X_i) = g'_1 \otimes g'_2 \otimes  \cdots \otimes g'_{\ell}$
where $c_i \not =0$ and each $g'_j$ is on a disjoint subset $X_{ij}$ of $X_i$.
By the assumption on $f_i$,
the only possibility is that all but one $X_{ij} = \emptyset$.
It follows that there is one (unique) $j$ such that $X_i$
is a subset of the variables of $g_j$.
By symmetry, for every $j$, the set of variable of $g_j$ is
a subset of the variables of some $f_{i'}$. As the $X_i$ are
disjoint, $i'=i$. Hence there is
a 1-1 correspondence of these subsets, and so $k= \ell$,
 and the corresponding
subsets are equal. After renaming these functions and subsets,
there are nonzero
constants $c_i'$ such that $f_i = c_i' g_i$ ($1 \le i \le k$).

We will consider the primitive decomposition of signatures in
$\mathscr{P}$. We claim that for any function
$f \in \mathcal{E}$ of arity at least 2,
$f$ is nondegenerate iff $|\operatorname{supp}(f)| = 2$.
One direction is trivial:  if $|\operatorname{supp}(f)| = 0$ or $1$,
then $f$ is identically 0 or  is a product
of unary functions, thus degenerate.
Conversely, suppose $f$ is degenerate,
 $f = u_1(x_1) \otimes \cdots  \otimes u_n(x_n)$.
If any $u_i$ is identically 0, then $|\operatorname{supp}(f)| = 0$.
Otherwise, if every $u_i$ is a multiple of $[1,0]$ or $[0,1]$,
then  $|\operatorname{supp}(f)| = 1$.
Otherwise, some  $u_i$ has the form $[a,b]$ with $ab \not =0$.
As $n \ge 2$ there is another $u_j= [c,d]$, where $c$ or $d \not =0$.
Without loss of generality, $c \not =0$.
Then there are two points $a_1 a_2 \ldots a_n$ and $a'_1 a'_2 \ldots a'_n
\in \operatorname{supp}(f)$, where $a_i=0, a_j =0$
and $a'_i =1, a'_j =0$. This contradicts $f \in \mathcal{E}$.

It follows that for any $f \in \mathscr{P}$ not identically 0,
its primitive decomposition exists and is unique up to  constant
factors, and is a product of unary functions
and nondegenerate functions in $\mathcal{E}$ with
$|\operatorname{supp}(f)| = 2$.

\begin{definition}
Let $f\in\mathscr{P}$, where $f$ is not identically zero.
There exist a partition $\mathcal{I}=\{I_1, I_2, \ldots, I_k\}$ of $[n]$
and signatures $f_1, f_2, \ldots, f_k\in\mathcal{E}$, where
each $f_i$ is a unary signature or $f_i$ is nondegenerate,
  such that
  \begin{equation}\label{primitive:decomposition}
  f(X)=\displaystyle\prod_{j=1}^k f_j(X|_{I_j}).
  \end{equation}
We call (\ref{primitive:decomposition}) a
 primitive decomposition of $f$.
\end{definition}

To define a compatibility relation on functions  in $\mathscr{P}$
we need to first define a notion of compatible partitions of $[n]$.
\begin{definition}
Two partitions $\mathcal{I}=\{I_1, I_2, \ldots, I_k\}$ and $\mathcal{J}=\{J_1, J_2, \ldots, J_{\ell}\}$ are \emph{compatible} if
for any $i\in[k]$, $j\in[\ell]$, $I_{i}$ and $J_j$ satisfy one of the following conditions:
\begin{itemize}
\item $I_i=J_j$;
\item $I_i\bigcap J_j=\emptyset$;
\item $|I_i|=1$ and $I_i\subseteq J_j$;
\item $|J_j|=1$ and $J_j\subseteq I_i$.
\end{itemize}
\end{definition}
An equivalent condition is that if $I_i \bigcap J_j \not = \emptyset$,
and $|I_i| \ge 2$ and $|J_j| \ge 2$, then $I_i=J_j$.

\begin{definition}\label{def:same-type}
For $f, g\in\mathscr{P}$ not identically zero,
we say $f, g$ have \emph{compatible type} if
in the primitive decompositions of $f$ and $g$,
\[f(X)=\displaystyle\prod_{i=1}^k f_i(X|_{I_i}),~~~~g(X)=\displaystyle\prod_{j=1}^{\ell} g_j(X|_{J_j})\]
\begin{itemize}
\item the partitions $\mathcal{I}
= \{I_1, I_2, \ldots, I_k\}$ and $\mathcal{J}
= \{J_1, J_2, \ldots, J_{\ell}\}$ are compatible;
\item for $I_i$ with $|I_i|\geq 2$ and so supp$(f_i)=\{\alpha, \bar{\alpha}\}$,
either
there  exists $J_j$ such that $I_i=J_j$ and $f_i$, $g_j$ have the same support,
or there exist $\{J_{j_1}, J_{j_2}, \ldots, J_{j_{|I_i|}}\}$ such that
$I_i = \displaystyle\bigcup_{t=1}^{|I_i|}J_{j_t}$,
$|J_{j_t}|=1$ for $1\leq t\leq |I_i|$, and the support of
$\displaystyle\prod_{t=1}^{|I_i|} g_{j_t}$
is the singleton set $\{\alpha\}$ or $\{\bar{\alpha}\}$;
\item for $J_j$ with $|J_j|\geq 2$ and so supp$(g_j)=\{\beta, \bar{\beta}\}$,
either
there exists $I_i$ such that $J_j=I_i$ and $f_i$, $g_j$ have the same support,
or there exist $\{I_{i_1}, I_{i_2}, \ldots, I_{i_{|J_j|}}\}$ such that
$J_j = \displaystyle\bigcup_{s=1}^{|J_j|}I_{i_s}$,
$|I_{i_s}|=1$ for $1\leq s\leq |J_j|$, and the support of
 $\displaystyle\prod_{s=1}^{|J_j|} f_{i_s}$
is  the singleton set  $\{\beta\}$ or $\{\bar{\beta}\}$.
\end{itemize}
\end{definition}

As primitive decompositions are unique up to constant factors,
Definition~\ref{def:same-type} is well-defined; it does not
depend on these constant factors.

\begin{lemma}\label{same-type}
Suppose $f, g, h\in\mathscr{P}$ and any  two  of them
 have compatible type.
Then
there exist a partition $\mathcal{L}=\{L_1, L_2, \ldots, L_{\ell}\}$ of $[n]$
and signatures $f_1, f_2, \ldots, f_{\ell}$, $g_1, g_2, \ldots, g_{\ell}$ and $h_1, h_2, \ldots, h_{\ell}$,
where $f_i, g_i, h_i\in\mathcal{E}$ for $1\leq i\leq {\ell}$
and $\operatorname{supp}(f_i), \operatorname{supp}(g_i), \operatorname{supp}(h_i)\subseteq\{\alpha_i, \bar{\alpha_i}\}$ for some $\alpha_i\in\{0, 1\}^{|L_i|}$,
such that $f(X)=\displaystyle\prod_{i=1}^{\ell} f_i(X|_{L_i})$, $g(X)=\displaystyle\prod_{i=1}^{\ell} g_i(X|_{L_i})$
and $h(X)=\displaystyle\prod_{i=1}^{\ell} h_i(X|_{L_i})$.
\end{lemma}
\begin{proof}
We prove by induction on $n$.
For $n=1$, $f, g, h$ are all unary signatures. The lemma is true trivially.
Inductively assume the lemma is true for $n'<n$
 and we prove the lemma for $n\geq 2$.

Let $f(X)=\displaystyle\prod_{j=1}^{p} F_j(X|_{I_j})$,
 $g(X)=\displaystyle\prod_{j=1}^{q} G_j(X|_{J_j})$
and $h(X)=\displaystyle\prod_{j=1}^{r} H_j(X|_{K_j})$ be the primitive decompositions of $f, g, h$ respectively,
where $\mathcal{I}=\{I_1, I_2, \ldots, I_p\}$, $\mathcal{J}=\{J_1, J_2, \ldots, J_q\}$ and $\mathcal{K}=\{K_1, K_2, \ldots, K_r\}$
are three partitions of $[n]$.
If all $|I_i|=|J_j|=|K_k|=1$  ($1\leq i\leq p, 1\leq j\leq q, 1\leq k\leq r$),
then $p=q=r=n$. We can  rename the sets so that $I_i=J_i = K_i = \{i\}$,
and let $\mathcal{L}=\mathcal{I}$ and $f_i=F_i, g_i=G_i, h_i=H_i$
for $1\leq i\leq n$.
Otherwise, without loss of generality, we assume that $|I_1|\geq 2$.
Since $f$ and $g$ have compatible type,
by the definition of primitive decomposition,
either there exists $j\in[q]$ such that $J_j=I_1$
and the support of $G_j$ is the same as $F_1$,
 or  there exist
 $J_{j_s}$ such that $|J_{j_s}|=1$ ($1 \le s \le |I_1|$)
and $I_1 = \displaystyle\bigcup_{s=1}^{|I_1|}J_{j_s}$,
 and the support of $\displaystyle\prod_{s=1}^{|I_1|}G_{j_s}$ is
a singleton subset  of the support of $F_1$.
 Then we let $G'_1=G_j$ or  $G'_1=\displaystyle\prod_{s=1}^{|I_1|}G_{j_s}$
 according to the two cases.

Similarly, since  $f$ and $h$ have compatible type,
either there exists $k\in[r]$ such that $K_k=I_1$
and the support of $H_k$ is the same as $F_1$,
 or  there exist
 $K_{k_t}$ such that $|K_{k_t}|=1$  ($1 \le t  \le |I_1|$)
and $I_1 = \displaystyle\bigcup_{t=1}^{|I_1|}K_{k_t}$,
 and the support of $\displaystyle\prod_{t=1}^{|I_1|}H_{k_t}$
is a singleton subset  of the support of $F_1$.
 Then we let $H'_1=H_k$ or  $H'_1=\displaystyle\prod_{t=1}^{|I_1|}H_{k_t}$
 according to the two cases.

Let $f'$, $g'$ and $h'$ be defined by the product
of those factors other than those of ${F_1}$, ${G'_1}$ and ${H'_1}$
in the respective primitive decompositions of $f$, $g$ and $h$.
 Then each pair of $f', g', h'$ have compatible type, and all
 have arity $n-|I_1|$.
 By induction, there exist a partition $\mathcal{L}=\{L_2, \ldots, L_{\ell}\}$ of $[n]\setminus I_1$
and signatures $f_2, \ldots, f_{\ell}$,
 $g_2, \ldots, g_{\ell}$ and $h_2, \ldots, h_{\ell}$,
where $f_i, g_i, h_i\in\mathcal{E}$ for $2\leq i\leq {\ell}$
and $\operatorname{supp}(f_i), \operatorname{supp}(g_i), \operatorname{supp}(h_i)\subseteq\{\alpha_i, \bar{\alpha_i}\}$ for some $\alpha_i\in\{0, 1\}^{|L_i|}$,
such that $f'(X)=\displaystyle\prod_{i=2}^{\ell} f_i(X|_{L_i})$, $g'(X)=\displaystyle\prod_{i=2}^{\ell} g_i(X|_{L_i})$
and $h'(X)=\displaystyle\prod_{i=2}^{\ell} h_i(X|_{L_i})$.
Then we finish the proof by letting $f_1=F_1$, $g_1=G'_1$, $h_1=H'_1$ and $L_1=I_1$.
\end{proof}

\vspace{.2in}

\noindent
{\bf Matchgate Signatures}

\vspace{.1in}
Matchgates were introduced by Valiant~\cite{Val02a, Val02b} to give polynomial-time algorithms for a collection of counting problems over planar graphs.
As the name suggests,
problems expressible by matchgates can be reduced to computing a weighted sum of perfect matchings.
The latter problem is tractable over planar graphs by Kasteleyn's algorithm~\cite{Kasteleyn1967},
a.k.a.~the FKT algorithm~\cite{TF61,Kasteleyn1961}.
These counting problems are naturally expressed in the Holant framework using \emph{matchgate signatures}.
We use $\mathscr{M}$ to denote the set of all matchgate signatures;
thus $\PlHolant(\mathscr{M})$ is tractable.

The parity of a signature is even (resp.~odd) if its support is on entries of even (resp.~odd) Hamming weight.
We say a signature satisfies the even (resp. odd) Parity Condition
if all entries of odd  (resp. even) weight are zero.
The matchgate signatures are characterized by the following two sets of conditions
(see~\cite{jinyi-aaron} for a self-contained proof).
 (1) The Parity Condition: either all
even entries are 0 or all odd entries are 0. This is due to perfect matchings. (2) A set of Matchgate Identities (MGI) defined as
follows: A pattern $\alpha$ is an $n$-bit string, i.e., $\alpha\in\{0, 1\}^{n}$. A position vector $P=\{p_{i} \mid i\in [\ell]\}$,
is a subsequence of $\{1, 2 \ldots, n\}$,
i.e., $p_{i}\in[n]$ and $p_{1}<p_{2}< \ldots<p_{\ell}$. We also use $P$ to denote the pattern, whose $p_{1}, p_{2}, \ldots, p_{\ell}$-th bits are 1 while the
others are 0. Let  $e_{i}\in\{0, 1\}^{n}$ be the pattern with 1 in the $i$-th bit and 0 elsewhere.
For $\alpha, \beta \in \{0, 1\}^{n}$,
let $\alpha \oplus\beta$ be the bitwise XOR of $\alpha$ and $\beta$.
Then any pattern $\alpha\in\{0, 1\}^{n}$ and any position vector
$P=\{p_i \mid i\in[\ell]\}$ define a MGI
\begin{equation}
\displaystyle\sum_{i=1}^{\ell}(-1)^{i}f_{\alpha\oplus e_{p_{i}}}f_{\alpha\oplus P\oplus e_{p_{i}}}=0.
\end{equation}
Actually in \cite{jinyi-aaron} it is shown that
MGI implies the Parity Condition. But in practice, it is easier
to apply the Parity Condition first.

\begin{proposition}\label{unary-matchgate}
 A unary signature $[a, b]\in\mathscr{M}$ iff it is $[1, 0]$ or $[0, 1]$ up to a scalar.
\end{proposition}
 We will use Proposition~\ref{unary-matchgate} in the proof of Theorem~\ref{main-theorem-for-no-parity}.

\begin{lemma}(cf.~Lemma~2.3, Lemma~2.4 in \cite{caifu-collapse})\label{matchgate-identity-for-arity-4}
If $f$ has arity $\leq 3$, then $f\in\mathscr{M}$ iff $f$ satisfies
the Parity Condition.

If $f$ has arity 4 and $f$ satisfies even Parity Condition, i.e.,
\[M_{x_1x_2, x_4x_3}(f)=\begin{bmatrix}
f_{0000} & 0 & 0 & f_{0011}\\
0 & f_{0110} & f_{0101} & 0\\
0 & f_{1010} & f_{1001} & 0\\
f_{1100} & 0 & 0 & f_{1111}\\
\end{bmatrix},\]
then
$f\in\mathscr{M}$ iff
\[f_{0000}f_{1111}-f_{1100}f_{0011}+f_{1010}f_{0101}-f_{1001}f_{0110}=0\]
(this is the MGI with $\alpha = 1000$ and $P = \{1,2,3,4\}$),
equivalently,
\[\det\begin{bmatrix}
f_{0000} & f_{0011}\\
f_{1100} & f_{1111}
\end{bmatrix}=
\det\begin{bmatrix}
f_{0110} & f_{0101}\\
f_{1010} & f_{1001}
\end{bmatrix}.\]
\end{lemma}

 \subsection{Transformable Signature Sets}

An important definition involving a holographic transformation is the notion of a signature set being transformable.

\begin{definition} \label{def:prelim:trans}
 We say a pair of signature sets $(\mathcal{G}, \mathcal{F})$
is $\mathscr{C}$-transformable for  $\holant{\mathcal{G}}{\mathcal{F}}$
 if there exists $T \in \mathbf{GL}_2(\mathbb{C})$ such that
$\mathcal{G}T \subseteq  \mathscr{C}$ and $T^{-1} {\mathcal{F}}
\subseteq \mathscr{C}$.

For $\mathcal{G} = \{(=_2)\}$,
$\holant{(=_2)}{\mathcal{F}} \equiv_{T} \Holant(\mathcal{F})$,
 we say simply that $\mathcal{F}$ is
$\mathscr{C}$-transformable. For  $\mathcal{G} = \mathcal{EQ}$,
$ \holant{\mathcal{EQ}}{\mathcal{F}}  \equiv_T \CSP(\mathcal{F})$,
we say that $\mathcal{F}$ is
$\mathscr{C}$-transformable for \#CSP.
We define similarly for \#CSP$^d$ when   $\mathcal{G} = \mathcal{EQ}_d$.
The definitions also work in the planar case.
\end{definition}

Notice that if
$\PlHolant(\mathscr{C})$ is tractable, and
$(\mathcal{G}, \mathcal{F})$
is $\mathscr{C}$-transformable,
then $\plholant{\mathcal{G}}{\mathcal{F}}$ is tractable
by a holographic transformation.
For example, consider $H_2= \frac{1}{\sqrt{2}}
\left[\begin{smallmatrix} 1 & 1 \\
1&-1 \end{smallmatrix}\right]$, with
 $H_2^{-1} = H_2$. Recall
the notation $\widehat{\mathcal{F}} = H_2 {\mathcal{F}}$.
Note that $\widehat{\mathcal{EQ}} \subset
\mathscr{M}$, and $H_2{\widehat{\mathscr{M}}} = \mathscr{M}$.
Thus Pl-\#{\rm CSP}$(\widehat{\mathscr{M}})$
is tractable, since $\PlHolant(\mathscr{M})$ tractable.
We list some important families of signatures
specific to the Pl-\#CSP and Pl-$\#${\rm CSP}$^2$ frameworks.
First we have
\[\widehat{\mathscr{P}}=H_2\mathscr{P}
~~~~
\mbox{and}
~~~~
\widehat{\mathscr{M}}=H_2\mathscr{M}.\]
Note that $\mathscr{A}$ is unchanged under the transformation
by $H_2$,
thus there is no  need to define
$\widehat{\mathscr{A}}$.
We have
$\widehat{\mathcal{EQ}} \subset \mathscr{A} \cap \mathscr{M}$.
Thus $\mathscr{A}$ is $\mathscr{A}$-transformable and
 $\widehat{\mathscr{M}}$ is
$\mathscr{M}$-transformable
 respectively for {\rm Pl}-$\#${\rm CSP}.

\begin{definition}
Let
$\mathcal{T}_k=
\left\{
\left[\begin{smallmatrix} 1 & 0 \\
0 & \omega \end{smallmatrix}\right]
\mid \omega^k = 1
\right\}$ be a set of diagonal matrices of order dividing $k$
and $\mathscr{T}_k=\mathcal{T}_{2k}\setminus\mathcal{T}_k
= \left\{
\left[\begin{smallmatrix} 1 & 0 \\
0 & \omega \end{smallmatrix}\right]
\mid \omega^k = -1
\right\}$.
Let
$\mathscr{A}^\dagger=
\mathscr{T}_4 \mathscr{A}$ and
$\widehat{\mathscr{M}}^\dagger
= \mathscr{T}_2 \widehat{\mathscr{M}}$
be the sets of signatures transformed by $\mathscr{T}_4$ from the
affine family $\mathscr{A}$ and transformed by $\mathscr{T}_2$
from $\widehat{\mathscr{M}}$,
respectively.
Define
\[\widetilde{\mathscr{A}} = \mathscr{A} \cup \mathscr{A}^\dagger
~~~~
\mbox{and}
~~~~
\widetilde{\mathscr{M}} = \widehat{\mathscr{M}} \cup \widehat{\mathscr{M}}^\dagger.\]
\end{definition}

Note that $\mathscr{P}$ is unchanged under
any diagonal matrix.
Thus there is no need to define $\mathscr{P}^\dagger$.
We note that
$\widetilde{\mathscr{A}}$ and $\widetilde{\mathscr{M}}$
are $\mathscr{A}$-transformable and ${\mathscr{M}}$-transformable
for {\rm Pl}-$\#${\rm CSP}$^2$,
respectively.
For $T = \left[\begin{smallmatrix} 1 & 0 \\
0 & \omega \end{smallmatrix}\right] \in \mathcal{T}_4$ with $\omega^4 = 1$,
$T \mathscr{A} =  \mathscr{A}$. Thus
$\widetilde{\mathscr{A}}$ is $\mathscr{A}$
under transformations by
$T
= \left[\begin{smallmatrix} 1 & 0 \\
0 & \omega \end{smallmatrix}\right]
\in \mathcal{T}_8$. For such $T$, we have
$(=_{2n}) T^{\otimes 2n} \in \mathscr{A}$.
Hence
$\widetilde{\mathscr{A}}$ is
$\mathscr{A}$-transformable for $\PlCSP^2$.
Similarly, for $T = \left[\begin{smallmatrix} 1 & 0 \\
0 & \pm 1 \end{smallmatrix}\right]$,
$TH_2 =
\left[\begin{smallmatrix} 1 & 1\\
\pm 1 & \mp 1 \end{smallmatrix}\right]
$ is either $H_2$ or $H_2
\left[\begin{smallmatrix} 0 & 1\\
1 & 0 \end{smallmatrix}\right]$,
and
$\left[\begin{smallmatrix} 0 & 1\\
1 & 0 \end{smallmatrix}\right] \mathscr{M} = \mathscr{M}$.
Thus $T \widehat{\mathscr{M}} = \widehat{\mathscr{M}}$, and
$\widetilde{\mathscr{M}}$ is $\mathscr{M}$ transformed under
$TH_2$
for all $T  \in
\mathcal{T}_4$.
Also note that for all such $T$,
we have $(=_{2n}) (TH_2)^{\otimes 2n} \in \mathscr{M}$.
Hence
$\widetilde{\mathscr{M}}$
is ${\mathscr{M}}$-transformable
for $\PlCSP^2$.

Note that the set of non-degenerate symmetric signatures in $\mathscr{A}^{\dagger}$ is precisely the nonzero signatures
($\lambda \neq 0$) in $\mathscr{F}^{\dagger}_1 \bigcup \mathscr{F}^{\dagger}_2$ with arity at least~$2$,
where $\mathscr{F}^{\dagger}_1$ and $\mathscr{F}^{\dagger}_2$ are two families of signatures defined as
\begin{align*}
 \mathscr{F}^{\dagger}_1 &= \left\{\lambda \left([1,0]^{\otimes k} + {\frak i}^r [0, 1]^{\otimes k}\right) | \lambda \in \mathbb{C}, k = 1, 2, \ldots, r = 0, 1, 2, 3\right\},
  \text{ and}\\
 \mathscr{F}^{\dagger}_2 &= \left\{\lambda \left([1,\alpha]^{\otimes k} + {\frak i}^r [1,-\alpha]^{\otimes k}\right) | \lambda \in \mathbb{C},
  \alpha^4=-1, k = 1, 2, \ldots, r = 0, 1, 2, 3\right\}.
\end{align*}

\begin{proposition}\label{matchgate:affine:hat}
The following hold:
\begin{itemize}
\item A unary signature is in $\widehat{\mathscr{M}}$ iff it is $\lambda[1, \pm 1], \lambda\in\mathbb{C}$.

\item A unary signature is in $\widehat{\mathscr{M}}^{\dagger}$ iff it is $\lambda[1, \pm {\frak i}], \lambda\in\mathbb{C}$.

\item $[1, 0, 1, 0]\notin\mathscr{A}^{\dagger}$.
\end{itemize}
\end{proposition}
We will use Proposition~\ref{matchgate:affine:hat} in the proof of Theorem~\ref{dichotomy-csp-2}.

\subsection{Some Known Dichotomies}

Here we list several known dichotomies.
The first is for \#CSP without planarity.
The other two are about planar \#CSP (and \#CSP$^2$)
but restricted to \emph{symmetric} signatures.

\begin{theorem}[Theorem~3.1 in~\cite{CLX14}]\label{non-planar-csp-dichotomy}
 Let $\mathcal{F}$ be any set of complex-valued signatures in Boolean variables.
 Then $\#\operatorname{CSP}(\mathcal{F})$ is $\SHARPP$-hard unless
 $\mathcal{F} \subseteq \mathscr{A}$ or
 $\mathcal{F} \subseteq \mathscr{P}$, in which case the problem is computable in polynomial time.
\end{theorem}

By (\ref{csp-holant}), it can be restated for $\Holant$ problems.
\newtheorem*{specialtheorem}{Theorem {\thetheorem}$'$}

\begin{specialtheorem}
 Let $\widehat{\mathcal{F}}$ be any set of complex-valued signatures in Boolean variables.
 Then $\Holant(\widehat{\mathcal{EQ}}, \widehat{\mathcal{F}})$ is $\SHARPP$-hard unless
 $\widehat{\mathcal{F}} \subseteq \mathscr{A}$ or
 $\widehat{\mathcal{F}} \subseteq \widehat{\mathscr{P}}$, in which case the problem is computable in polynomial time.
\end{specialtheorem}

The next theorem is a dichotomy for $\PlCSP$ problems
over symmetric signatures.

\begin{theorem}[Theorem~19 in~\cite{Guo-Williams}] \label{heng-tyson-dichotomy-pl-csp}
 Let $\mathcal{F}$ be any set of \emph{symmetric}, complex-valued signatures in Boolean variables.
 Then $\PlCSP(\mathcal{F})$ is $\SHARPP$-hard unless
 $\mathcal{F} \subseteq \mathscr{A}$,
 $\mathcal{F} \subseteq \mathscr{P}$, or
 $\mathcal{F} \subseteq \widehat{\mathscr{M}}$,
 in which case the problem is computable in polynomial time.
\end{theorem}

By (\ref{csp-holant}), it can be restated for $\PlHolant$ problems.
\begin{specialtheorem} 
 Let $\widehat{\mathcal{F}}$ be any set of \emph{symmetric}, complex-valued signatures in Boolean variables.
 Then $\PlHolant(\widehat{\mathcal{EQ}}, \widehat{\mathcal{F}})$ is $\SHARPP$-hard unless
 $\widehat{\mathcal{F}} \subseteq \mathscr{A}$,
 $\widehat{\mathcal{F}} \subseteq \widehat{\mathscr{P}}$, or
 $\widehat{\mathcal{F}} \subseteq \mathscr{M}$,
 in which case the problem is computable in polynomial time.
\end{specialtheorem}

The following theorem is a dichotomy theorem for
$\PlCSP^2$ problems over \emph{symmetric} signatures.
By (\ref{eqn:prelim:PlCSPd_equiv_Holant}) for $d=2$,
we have $\PlCSP^2(\mathcal{F}) \equiv_{\rm T} \PlHolant(\mathcal{EQ}_2,
\mathcal{F})$. Thus the theorem can be equivalently stated for
$ \PlHolant(\mathcal{EQ}_2,
\mathcal{F})$. Note that this equivalence is not by a holographic
transformation.
However, when we apply it later in this paper,  we actually use it
on the right hand side of the equivalence by
a holographic transformation
$\PlCSP(\mathcal{F})
 \equiv_{\rm T} \PlHolant(\widehat{\mathcal{EQ}}, \widehat{\mathcal{F}})$,
when  we can construct $\mathcal{EQ}_2$ in the right
hand side.

\begin{theorem}[Theorem~A.2 in \cite{Cai-Fu-Guo-W}]\label{heng-tyson-dichotomy-pl-csp2}
Let $\mathcal{F}$ be any set of \emph{symmetric}, complex-valued signatures in Boolean variables.
 Then $\PlCSP^2(\mathcal{F})$,
equivalently $\PlHolant(\mathcal{EQ}_2,
\mathcal{F})$,  is \#P-hard unless
 $\mathcal{F} \subseteq \mathscr{P}$,
 $\mathcal{F} \subseteq \mathscr{A}$,
 $\mathcal{F} \subseteq \mathscr{A}^{\dagger}$,
 $\mathcal{F} \subseteq \widehat{\mathscr{M}}$,
 or
 $\mathcal{F} \subseteq \widehat{\mathscr{M}^{\dagger}}$,
 in which case the problem is computable in polynomial time.
\end{theorem}

Note that Theorem~\ref{heng-tyson-dichotomy-pl-csp}
(and Theorem~\ref{heng-tyson-dichotomy-pl-csp2})
are applicable only for symmetric signatures.
The main theorem of
the present paper is to generalize
Theorem~\ref{heng-tyson-dichotomy-pl-csp} to be valid for all,
not necessarily symmetric, signatures over Boolean variables.

\subsection{Some Lemmas}
In this subsection, we prove some simple lemmas.
The next lemma shows that flipping any input variable
of a signature $f$ does not change its membership in
$\mathscr{P}$, or   $\mathscr{A}$, or $\mathscr{M}$.
\begin{lemma}\label{[0,1,0]-not-change-tractable}
Let $g(x_1, \ldots, x_{i-1}, x_i, x_{i+1}, \ldots, x_n)=f(x_1, \ldots, x_{i-1}, \overline{x_i}, x_{i+1}, \ldots, x_n)$,
then for $\mathscr{C} \in \{ \mathscr{P}, \mathscr{A}, \mathscr{M}\}$,
 $f\in\mathscr{C}$ iff $g\in\mathscr{C}$.
\end{lemma}
\begin{proof}
Note that $[0, 1, 0]\in \mathscr{P} \cap \mathscr{A}\cap\mathscr{M}$,
and $f$ is obtained from $g$ by flipping $x_i$.
It follows easily by definition of
$\mathscr{C} \in \{ \mathscr{P}, \mathscr{A}, \mathscr{M}\}$ that
$f\in\mathscr{C}$ iff $g\in\mathscr{C}$.
\end{proof}

The following lemma shows how
to use $[0, 1]^{\otimes 2}$ and $[1, 0, 1, 0] \in \widehat{\mathcal{EQ}}$
 to get $[0, 1, 0]^{\otimes 2}$, then
 to flip any two variables that are not necessarily
 adjacent,
while preserving planarity.

\begin{lemma}\label{how-to-flip-two-bits-by-[0,1,0]-tensor-2}
In $\operatorname{Pl-Holant}(\widehat{\mathcal{EQ}}, [0, 1]^{\otimes 2}, f)$,
if $f$ has arity $n$, then
for any $s  \not = t\in[n]$,
we can construct $g$ such that $g(x'_1, x'_2, \ldots, x'_n)=f(x_1, x_2, \ldots, x_n)$,
where $x'_k=\overline{x_k}$ for $k\in\{s, t\}$ and $x'_k=x_k$ otherwise.
Moreover,
for  $\mathscr{C} \in \{ \mathscr{P}, \mathscr{A}, \mathscr{M}\}$,
 $f\in\mathscr{C}$ iff $g\in\mathscr{C}$.
\end{lemma}
\begin{proof}
That $f\in\mathscr{C}$ iff $g\in\mathscr{C}$ follows from
Lemma~\ref{[0,1,0]-not-change-tractable}.

Note that we have $[1, 0, 1, 0] \in \widehat{\mathcal{EQ}}$, and
 $[0, 1]^{\otimes 2}$. Since
$\partial_{[0, 1]}([1, 0, 1, 0])=[0, 1, 0]$,
 by connecting  $[0, 1]^{\otimes 2}$ to two
disjoint copies of $[1, 0, 1, 0]$, we get
  $[0, 1, 0]^{\otimes 2}$.
Note that this is a planar gadget where two adjacent pairs of variables
are flipped. This function is $(x_1 \not = x_2)
\wedge (x_3 \not = x_4)$.  After a rotation of $90^{\circ}$
we also get $(x_4 \not = x_1)
\wedge (x_2 \not = x_3)$, which we will denote as $D^2$.

Without loss of generality, we assume that $t>s$.
If $t-s=1$, then $x_s$ and $x_{s+1}$ are adjacent variables
and we can directly apply $D^2$ to flip both $x_s$ and $x_{s+1}$.
In general (See Figure~\ref{fig:two:variables} for an illustration),
 we let $h^{(0)}=f$ and
for $1 \le j \le t-s$,
define $h^{(j)}(x^{(j)}_1, x^{(j)}_2, \ldots, x^{(j)}_n)=h^{(j-1)}(x^{(j-1)}_1, x^{(j-1)}_2, \ldots, x^{(j-1)}_n)$,
where $x^{(j)}_i= \overline{x^{(j-1)}_i}$ for $i\in \{s+j-1, s+j\}$ and
$x^{(j)}_i=x^{(j-1)}_i$ for all others.
Then we are done by letting $g=h^{(t-s)}$.
In effect, all variables $x_i$ with $s < i < t$ are flipped twice.
\end{proof}

\begin{figure}[htpb]
  \centering
  \begin{tikzpicture}[scale=\scale,transform shape,node distance=1.5*\nodeDist,semithick]
    \node[internal]  (0) {};
    \node[external, above of=0]  (1) {};
    \node[external, above of=1]  (2) {};
    \node[external, left of=1]  (4) {};
    \node[external, above of=4]  (3) {};
    \node[external, right of=1]  (6) {};
    \node[external, above of=6]  (5) {};
    \path (0)  edge node[square] (21) {}  (1)
          (0)  edge node[square] (11) {} (3)
          (0)  edge (2)
          (11)  edge[densely dashed] (21)
          (0)  edge node[square] (31) {} (5)
          (1)  edge node[square] (22) {} (2)
          (31)  edge[densely dashed] (22);
    \begin{pgfonlayer}{background}
      \node[draw=\borderColor,thick,rounded corners,fit = (0), inner sep=8pt,transform shape=false] {};
    \end{pgfonlayer}
  \end{tikzpicture}
  \caption{
  Flipping two variables of $f$ that are not adjacent by $
[0, 1, 0]^{\otimes 2}$ while preserving planarity.
  The circle vertex is labeled $f$ and squares are $[0, 1, 0]$.
A pair of squares connected by a dashed line forms $[0, 1, 0]^{\otimes 2}$.}
  \label{fig:two:variables}
\end{figure}

The following lemma implies that in $\operatorname{Pl-Holant}(\widehat{\mathcal{EQ}}, f)$,
where $f\notin\mathscr{A}$ or $f\notin\mathscr{M}$, we can assume that
$f_{00\cdots 0}=1$.
Moreover, if $f$ satisfies the parity condition, we can assume it satisfies
the \emph{even} parity condition.

\begin{lemma}\label{[0,1]-EQ-hat-wight-0-neq-0}
For $\mathscr{C} = \mathscr{A}$ or $\mathscr{M}$,
if $\widehat{\mathcal{F}}$
 contains a signature $f \notin\mathscr{C}$ of arity $n$,
then we can construct a function
$f'\notin\mathscr{C}$ of arity $n$
with $f'_{00\cdots 0}=1$  such that
\[\operatorname{Pl-Holant}(\widehat{\mathcal{EQ}}, f',
\widehat{\mathcal{F}})
\le_{\rm T}
\operatorname{Pl-Holant}(\widehat{\mathcal{EQ}},
\widehat{\mathcal{F}}).\]
Moreover, if $f$ satisfies the Parity Condition,
then $f'$ satisfies the even Parity Condition, and
if $f$ takes values in $\{0, 1\}$ $(\{0, 1, -1\})$, then $f'$
also takes values in $\{0, 1\}$ $(\{0, 1, -1\})$.
\end{lemma}
\begin{proof}
If  $f_{00\cdots 0} \not = 0$, then we simply normalize
$f$ by setting $f' = f/f_{00\cdots 0}$.
So we suppose $f_{00\cdots 0} = 0$.
By $f \not \in \mathscr{A}$ or $\mathscr{M}$, clearly $f$ is not identically 0.
Let ${\rm wt}(\alpha)=\displaystyle\min_{\beta\in\{0, 1\}^n}\{\operatorname{wt}(\beta)|f_{\beta}\neq 0\}$. Let
$S = \{i \mid 1\leq i\leq n, ~\mbox{the $i$-th bit of $\alpha$ is 0}\}$.
Since we have $[1,0] \in \widehat{\mathcal{EQ}}$ we can get
$\partial_{[1,0]}^S(f) = [0,1]^{\otimes {\rm wt}(\alpha)}$.
Depending on whether ${\rm wt}(\alpha)$ is odd or even,
we can take $\partial_{=_2}$ on $[0,1]^{\otimes {\rm wt}(\alpha)}$
 repeatedly and obtain
either $[0,1]$ or $[0,1]^{\otimes 2}$, respectively.
Since we have $[1,0,1,0] \in  \widehat{\mathcal{EQ}}$,
we can get either $\partial_{[0,1]}([1,0,1,0]) = [0,1,0]$
or $[0,1,0]^{\otimes 2}$.

If ${\rm wt}(\alpha)$ is odd, and this includes the case when
$f$ satisfies the odd Parity Condition, we have $[0,1,0]$
and can flip any variable of $f$ individually. By flipping
all variables in $[n] \setminus S$, and normalizing,
we obtain $f'$ with the required
property. In particular if $f$ satisfies the  odd Parity Condition,
then $f'$ satisfies the even Parity Condition.

If  ${\rm wt}(\alpha)$ is even, and this includes the case when
$f$ satisfies the even Parity Condition, we have
 $[0,1]^{\otimes 2}$ and $[0,1,0]^{\otimes 2}$.
By Lemma~\ref{how-to-flip-two-bits-by-[0,1,0]-tensor-2}
we can flip any two variables of $f$. By applying the construction
in Lemma~\ref{how-to-flip-two-bits-by-[0,1,0]-tensor-2}
simultaneously on ${\rm wt}(\alpha)/2$ pairs of variables
of $f$, we can transform $f$ to $f'$ by a planar construction
so that $f'_{00\cdots 0} = f_{\alpha} \not = 0$.
By normalizing, we obtain the required $f'$
with $f'_{00\cdots 0} =1$.
In particular if $f$ satisfies the even Parity Condition,
then $f'$ also satisfies the even Parity Condition.
We get $f'$ from $f$
 by flipping some variables in all cases. Thus if $f$ takes values
in  $\{0, 1\}$ $(\{0, 1, -1\})$, then $f'$  also takes values
in $\{0, 1\}$ $(\{0, 1, -1\})$.
\end{proof}

\subsection{Interpolation}
Polynomial interpolation is a powerful tool in the study of counting problems.
In this subsection, we give the following two lemmas by polynomial interpolation.

\begin{lemma}\label{interpolation-unary}
If $|x|\neq 0, 1$, then for any $a, b\in\mathbb{C}$, we have
\[\operatorname{Pl-Holant}(\mathcal{EQ}, [a, b], \mathcal{F})
\le_{\rm T}
\operatorname{Pl-Holant}(\mathcal{EQ}, [1, x], \mathcal{F}).\]
\end{lemma}
\begin{proof}
Note that for any $k\in\mathbb{Z}^+$, we have $\partial_{[1, x]}^k(=_{k+1})=[1, x^k]$.
Consider an instance $\Omega$ of Pl-Holant$(\mathcal{EQ}, [a, b], \mathcal{F})$.
Let $S$ be the subset of vertices assigned $[a, b]$ and suppose that $|S|=n$.
By replacing each occurrence of $[a, b]$ with $[1, x^k]$, we construct a sequence of instances $\Omega_k$
of  $\operatorname{Pl-Holant}(\mathcal{EQ}, [1, x], \mathcal{F})$.

We stratify the assignments in $\Omega$ based on the assignment to $[a, b]$.
Let $c_{\ell}$ be the sum over all assignments of products of evaluations at all vertices other than those from $S$
such that exactly $\ell$ occurrences of $[a, b]$ with their
respective incident edges assigned 1 (and
the other $n-\ell$ assigned 0). Then
\[\operatorname{Pl-Holant}(\Omega)=\displaystyle\sum_{0\leq \ell\leq n}a^{n-\ell}b^{\ell}c_{\ell}\]
and the value of the planar Holant on $\Omega_k$, for $1\leq k\leq n+1,$ is
\[\operatorname{Pl-Holant}(\Omega_k)=\displaystyle\sum_{0\leq \ell\leq n}x^{k\ell}c_{\ell}.\]
This is a linear system with unknowns $c_{\ell}$ and a Vandermonde coefficient matrix.
Since $|x|\neq 0, 1$, $x^k$ are all distinct ($1\leq k\leq n+1$),
 which implies that the Vandermonde matrix has full rank.
Therefore, we can solve the linear system for the unknown $c_{\ell}$'s and obtain the value of Pl-Holant$(\Omega)$.
\end{proof}

\begin{lemma}\label{interpolation-equality-4}
Suppose $\mathcal{F}$ contains a signature $f$ of arity 4 with
\[M_{x_1x_2, x_4x_3}(f)=\begin{bmatrix}
a & 0 & 0 & b\\
0 & 0 & 0 & 0\\
0 & 0 & 0 & 0\\
c & 0 & 0 & d
\end{bmatrix}\]
or
\[M_{x_1x_2, x_4x_3}(f)=\begin{bmatrix}
a & 0 & 0 & 0\\
0 & b & 0 & 0\\
0 & 0 & c & 0\\
0 & 0 & 0 & d
\end{bmatrix},\]
where $\begin{bmatrix}
a & b\\
c & d
\end{bmatrix}$ has full rank. Then
\[\operatorname{Pl-Holant}(=_4, \mathcal{F})
\le_{\rm T}
\operatorname{Pl-Holant}(\mathcal{F}).\]
\end{lemma}
\begin{proof}
If \[M_{x_1x_2, x_4x_3}(f)=\begin{bmatrix}
a & 0 & 0 & 0\\
0 & b & 0 & 0\\
0 & 0 & c & 0\\
0 & 0 & 0 & d
\end{bmatrix},\]
then after a rotation (See Fig.~\ref{fig:rotate_asymmetric_signature})
we have the signature
\[M_{x_4x_1,x_3x_2}(f)=\begin{bmatrix}
a & 0 & 0 & b\\
0 & 0 & 0 & 0\\
0 & 0 & 0 & 0\\
c & 0 & 0 & d
\end{bmatrix}.\]
Therefore we only need to  prove the lemma for
 \[M_{x_1x_2, x_4x_3}(f)=\begin{bmatrix}
a & 0 & 0 & b\\
0 & 0 & 0 & 0\\
0 & 0 & 0 & 0\\
c & 0 & 0 & d
\end{bmatrix}.\]
\begin{figure}[ht]
 \centering
 \captionsetup[subfigure]{labelformat=empty}
 \subfloat[$N_1$]{
  \begin{tikzpicture}[scale=\scale,transform shape,node distance=\nodeDist,semithick]
   \node[external] (0)                    {};
   \node[external] (1) [right       of=0] {};
   \node[internal] (2) [below right of=1] {};
   \node[external] (3) [below left  of=2] {};
   \node[external] (4) [left        of=3] {};
   \node[external] (5) [above right of=2] {};
   \node[external] (6) [right       of=5] {};
   \node[external] (7) [below right of=2] {};
   \node[external] (8) [right       of=7] {};
   \path (0) edge[out=   0, in=135, postaction={decorate, decoration={
                                                           markings,
                                                           mark=at position 0.4   with {\arrow[>=diamond, white] {>}; },
                                                           mark=at position 0.4   with {\arrow[>=open diamond]   {>}; },
                                                           mark=at position 0.999 with {\arrow[>=diamond, white] {>}; },
                                                           mark=at position 1.0   with {\arrow[>=open diamond]   {>}; } } }] (2)
         (2) edge[out=-135, in=  0] (4)
             edge[out=  45, in=180] (6)
             edge[out= -45, in=180] (8);
   \begin{pgfonlayer}{background}
     \node[draw=\borderColor,thick,rounded corners,fit = (1) (3) (5) (7),inner sep=0pt,transform shape=false,] {};
   \end{pgfonlayer}
  \end{tikzpicture}
 }
 \qquad
 \subfloat[$N_2$]{
  \begin{tikzpicture}[scale=\scale,transform shape,node distance=\nodeDist,semithick]
   \node[external]  (0)                    {};
   \node[external]  (1) [right       of=0] {};
   \node[internal]  (2) [below right of=1] {};
   \node[external]  (3) [below left  of=2] {};
   \node[external]  (4) [left        of=3] {};
   \node[external]  (5) [right       of=2] {};
   \node[internal]  (6) [right       of=5] {};
   \node[external]  (7) [above right of=6] {};
   \node[external]  (8) [right       of=7] {};
   \node[external]  (9) [below right of=6] {};
   \node[external] (10) [right       of=9] {};
   \path (0) edge[out=   0, in=135, postaction={decorate, decoration={
                                                           markings,
                                                           mark=at position 0.4   with {\arrow[>=diamond, white] {>}; },
                                                           mark=at position 0.4   with {\arrow[>=open diamond]   {>}; },
                                                           mark=at position 0.999 with {\arrow[>=diamond, white] {>}; },
                                                           mark=at position 1.0   with {\arrow[>=open diamond]   {>}; } } }] (2)
         (2) edge[out=-135, in=  0]  (4)
             edge[bend left,        postaction={decorate, decoration={
                                                           markings,
                                                           mark=at position 0.999 with {\arrow[>=diamond, white] {>}; },
                                                           mark=at position 1.0   with {\arrow[>=open diamond]   {>}; } } }]  (6)
             edge[bend right]  (6)
         (6) edge[out=  45, in=180]  (8)
             edge[out= -45, in=180] (10);
   \begin{pgfonlayer}{background}
     \node[draw=\borderColor,thick,rounded corners,fit = (1) (3) (7) (9),inner sep=0pt,transform shape=false,] {};
   \end{pgfonlayer}
  \end{tikzpicture}
 }
 \qquad
 \subfloat[$N_{s}$]{
  \begin{tikzpicture}[scale=\scale,transform shape,node distance=\nodeDist,semithick]
   \node[external]  (0)                     {};
   \node[external]  (1) [above left  of=0]  {};
   \node[external]  (2) [below left  of=0]  {};
   \node[external]  (3) [below left  of=1]  {};
   \node[external]  (4) [below left  of=3]  {};
   \node[external]  (5) [above left  of=3]  {};
   \node[external]  (6) [left        of=4]  {};
   \node[external]  (7) [left        of=5]  {};
   \node[external]  (8) [right       of=0]  {};
   \node[internal]  (9) [right       of=8]  {};
   \node[external] (10) [above right of=9]  {};
   \node[external] (11) [below right of=9]  {};
   \node[external] (12) [right       of=10] {};
   \node[external] (13) [right       of=11] {};
   \path let
          \p1 = (1),
          \p2 = (2)
         in
          node[external] at (\x1, \y1 / 2 + \y2 / 2) {\Huge $N_{s-1}$};
   \path let
          \p1 = (0)
         in
          node[external] (14) at (\x1 + 2, \y1 + 10) {};
   \path let
          \p1 = (0)
         in
          node[external] (15) at (\x1 + 2, \y1 - 10) {};
   \path let
          \p1 = (3)
         in
          node[external] (16) at (\x1 - 2, \y1 + 10) {};
   \path let
          \p1 = (3)
         in
          node[external] (17) at (\x1 - 2, \y1 - 10) {};
   \path (7) edge[out=   0, in=135, postaction={decorate, decoration={
                                                           markings,
                                                           mark=at position 0.43  with {\arrow[>=diamond, white] {>}; },
                                                           mark=at position 0.43  with {\arrow[>=open diamond]   {>}; },
                                                           mark=at position 0.999 with {\arrow[>=diamond, white] {>}; },
                                                           mark=at position 1.0   with {\arrow[>=open diamond]   {>}; } } }] (16)
        (17) edge[out=-135, in=  0]  (6)
        (14) edge[out=  35, in=135, postaction={decorate, decoration={
                                                           markings,
                                                           mark=at position 0.999 with {\arrow[>=diamond, white] {>}; },
                                                           mark=at position 1.0   with {\arrow[>=open diamond]   {>}; } } }]  (9)
         (9) edge[out=-135, in=-35]  (15)
             edge[out=  45, in=180] (12)
             edge[out= -45, in=180] (13);
   \begin{pgfonlayer}{background}
     \node[draw=\borderColor,thick,densely dashed,rounded corners,fit = (0) (1.south) (2.north) (3),inner sep=0pt,transform shape=false,] {};
     \node[draw=\borderColor,thick,rounded corners,fit = (4) (5) (10) (11),inner sep=0pt,transform shape=false,] {};
   \end{pgfonlayer}
  \end{tikzpicture}
 }
 \caption{Linear recursive construction used for interpolation.}
 \label{fig:gadget:arity4:linear_interpolation}
\end{figure}

Consider an instance $\Omega$ of $\operatorname{Pl-Holant}(=_4, \mathcal{F})$.
Suppose $=_4$ appears $n$ times in $\Omega$. We construct from $\Omega$ a sequence of instances $\Omega_s$
of $\operatorname{Pl-Holant}(\mathcal{F})$ indexed by $s\geq 1$.
We obtain $\Omega_s$ from $\Omega$ by replacing each occurrence of $=_4$ with the gadget $N_s$
 in Figure~\ref{fig:gadget:arity4:linear_interpolation} with $f$ assigned to all vertices.
In $\Omega_s$, the edge corresponding to the $i$-th variable
 of $N_s$ connects to the same
location as the edge corresponding to the $i$-th variable of $=_4$ in $\Omega$.
In
Figure~\ref{fig:gadget:arity4:linear_interpolation}, we place a
diamond on the edge corresponding to the first variable. The remaining
variables  are ordered counterclockwise around the vertex.

By the Jordan normal form of $\begin{bmatrix}
a & b\\
c & d
\end{bmatrix}$,
there exists
 $P=\begin{bmatrix}
p_{00} & p_{01}\\
p_{10} & p_{11}
\end{bmatrix}$ such that
\[\begin{bmatrix}
a & b\\
c & d
\end{bmatrix}=P\begin{bmatrix}
\lambda_1 & 0\\
0 & \lambda_2
\end{bmatrix}P^{-1},\]
or, when there is a double root, we can normalize it to 1,
and then
  \[\begin{bmatrix}
a & b\\
c & d
\end{bmatrix}=P\begin{bmatrix}
1 & \lambda\\
0 & 1
\end{bmatrix}P^{-1},\]
where all of $\lambda_1, \lambda_2, \lambda$ are nonzero.
This implies that
\begin{equation}\label{jordan-1}
M_{x_1x_2, x_4x_3}(f)=\begin{bmatrix}
p_{00} & 0 & 0 & p_{01}\\
0 & 1 & 0 & 0\\
0 & 0 & 1 & 0\\
p_{10} & 0 & 0 & p_{11}
\end{bmatrix}
\begin{bmatrix}
\lambda_1 & 0 & 0 & 0\\
0 & 0 & 0 & 0\\
0 & 0 & 0 & 0\\
0 & 0 & 0 & \lambda_2
\end{bmatrix}
\begin{bmatrix}
p_{00} & 0 & 0 & p_{01}\\
0 & 1 & 0 & 0\\
0 & 0 & 1 & 0\\
p_{10} & 0 & 0 & p_{11}
\end{bmatrix}^{-1}\end{equation}
or
\begin{equation}\label{jordan-2}
M_{x_1x_2, x_4x_3}(f)=\begin{bmatrix}
p_{00} & 0 & 0 & p_{01}\\
0 & 1 & 0 & 0\\
0 & 0 & 1 & 0\\
p_{10} & 0 & 0 & p_{11}
\end{bmatrix}
\begin{bmatrix}
1 & 0 & 0 & \lambda\\
0 & 0 & 0 & 0\\
0 & 0 & 0 & 0\\
0 & 0 & 0 & 1
\end{bmatrix}
\begin{bmatrix}
p_{00} & 0 & 0 & p_{01}\\
0 & 1 & 0 & 0\\
0 & 0 & 1 & 0\\
p_{10} & 0 & 0 & p_{11}
\end{bmatrix}^{-1}.\end{equation}
Let $T=\begin{bmatrix}
p_{00} & 0 & 0 & p_{01}\\
0 & 1 & 0 & 0\\
0 & 0 & 1 & 0\\
p_{10} & 0 & 0 & p_{11}
\end{bmatrix}$ and let $f_s$ be  the signature of the gadget $N_s$.

For (\ref{jordan-1}), $M_{x_1x_2, x_4x_3}(f_s)=T\Lambda^s T^{-1}$,
where
\[\Lambda=\begin{bmatrix}
\lambda_1 & 0 & 0 & 0\\
0 & 0 & 0 & 0\\
0 & 0 & 0 & 0\\
0 & 0 & 0 & \lambda_2
\end{bmatrix} ~~~~ \mbox{and}~~~~
\Lambda^s=\begin{bmatrix}
\lambda_1^s & 0 & 0 & 0\\
0 & 0 & 0 & 0\\
0 & 0 & 0 & 0\\
0 & 0 & 0 & \lambda_2^s
\end{bmatrix}.\]
If there exists $d\in\mathbb{Z}^+$ such that $(\frac{\lambda_2}{\lambda_1})^d=1$, then $f_d$ is $=_4$ after a nozero scalar and we are done.
Otherwise, For any $i, j\in\mathbb{Z}$, $(\frac{\lambda_2}{\lambda_1})^i\neq (\frac{\lambda_2}{\lambda_1})^j$ if $i\neq j$.
We can view our construction of $\Omega_s$ as replacing $=_4$  by 3 signatures, with matrices
$T, \Lambda^s, T^{-1}$ respectively. This does not change the Holant value.
The Holant value on  $\Omega$ is also unchanged by  replacing
$=_4$  with $T, (=_4), T^{-1}$ in sequence.
We stratify assignments in $\Omega$ based on assignment values to
the $n$ occurrences of the new $(=_4)$, each
sandwiched between $T$ and $T^{-1}$.
Note that we only need to consider the assignments to $(=_4)$ that assign
\begin{itemize}
\item (0, 0, 0, 0) $i$ many times,
\item (1, 1, 1, 1) $j$ many times
\end{itemize}
such that $i+j=n$,
 since
any other assignment contributes 0 to the Holant sum.
Let $c_{ij}$ be the sum over all
such assignments of the products of evaluations
(including the contributions from $T, T^{-1}$) in $\Omega$.
Then we have
\[\operatorname{Pl-Holant}(\Omega)=\displaystyle\sum_{i+j=n}c_{ij},\]
and
\[\operatorname{Pl-Holant}(\Omega_s)=\displaystyle\sum_{i+j=n}c_{ij}\lambda_1^{i s}\lambda_2^{js}=
\lambda_1^{ns}\displaystyle\sum_{i+j=n}c_{ij}\left(\frac{\lambda_2}{\lambda_1}\right)^{js}.\]
Note that the same set of values  $c_{ij}$
occur in $\operatorname{Pl-Holant}(\Omega_s)$ independent of $s$.
Then we get a Vandermonde  system with unknowns $c_{n-j, j}$.
Since $(\frac{\lambda_2}{\lambda_1})^{j}\neq (\frac{\lambda_2}{\lambda_1})^{j'}$ if $j\neq j'$, this
 coefficient matrix has full rank.
Therefore, we can solve the linear system in polynomial time and obtain the value of Holant$(\Omega)$.
This implies that
\[\operatorname{Pl-Holant}(=_4, \mathcal{F})
\le_{\rm T}
\operatorname{Pl-Holant}(\mathcal{F}).\]

For (\ref{jordan-2}),
\begin{equation}\label{jordan-3}
M_{x_1x_2, x_4x_3}(f_s)=T\Lambda^s T^{-1},
\end{equation}
where
$\Lambda=\begin{bmatrix}
1 & 0 & 0 & \lambda\\
0 & 0 & 0 & 0\\
0 & 0 & 0 & 0\\
0 & 0 & 0 & 1
\end{bmatrix}$ and
$\Lambda^s=\begin{bmatrix}
1 & 0 & 0 & s\lambda\\
0 & 0 & 0 & 0\\
0 & 0 & 0 & 0\\
0 & 0 & 0 & 1
\end{bmatrix}.$
Similarly,
we can view our construction of $\Omega_s$ as replacing $=_4$  by 3 signatures, with matrices
$T, \Lambda^s, T^{-1}$ respectively. This does not change the Holant value.
We also consider replacing each  $=_4$ by $T (=_4) T^{-1}$ in $\Omega$.
We stratify assignments in $\Omega$
according to the new $(=_4)$'s are assigned
\begin{itemize}
\item (0, 0, 0, 0) or (1, 1, 1, 1) $i$ many times,
\item (0, 0, 1, 1) $j$ many times
\end{itemize}
such that $i+j=n$.
Let $c'_{ij}$ be the sum over all such assignments
of the products of evaluations
(including the contributions from $T, T^{-1}$) in $\Omega$.
Then we have
\[\operatorname{Pl-Holant}(\Omega)=c'_{n0},\]
and
\[\operatorname{Pl-Holant}(\Omega_s)=\displaystyle\sum_{i+j=n}c'_{ij}(s\lambda)^j.\]
Again note that the same set of values $c'_{ij}$ occur in
$\operatorname{Pl-Holant}(\Omega_s)$, independent of $s$.
Then we get a Vandermonde  system with unknown $c''_{j}$,
 where $c''_{j}=c'_{n-j,j}\lambda^{j}$ (for $0 \le j \le n$).
The coefficient matrix $(s^j)$ has full rank.
Therefore, we can solve the linear system in polynomial time and obtain the value of Holant$(\Omega)$.
This shows that
\[\operatorname{Pl-Holant}(=_4, \mathcal{F})
\le_{\rm T}
\operatorname{Pl-Holant}(\mathcal{F})\]
 and we finish the proof.
\end{proof}

\subsection{Outline of the Proof}\label{outline-subsection}
We now give an outline of the proof of the main dichotomy,
Theorem~\ref{main-dichotomy-thm},
and also explain some overall vision that guided our
proof.
An important technique is to view our counting problems
in the dual perspectives of planar \#CSP \emph{and} planar
Holant problems, i.e., we make essential use of
the equivalence
 $\PlCSP(\mathcal{F}) \equiv_T
\PlHolant(\widehat{\mathcal{EQ}}, \widehat{\mathcal{F}})$.
Some questions are easier to handle in one framework, while
others are  easier in the other.

We aim to prove Theorem~\ref{main-dichotomy-thm}. Our overall vision
is that
the classification in Theorem~\ref{heng-tyson-dichotomy-pl-csp}
should be valid for general, not
 necessarily symmetric, signatures. Thus we want to show
that either $\mathcal{F} \subseteq \mathscr{A}$,
or  $\mathcal{F} \subseteq \mathscr{P}$, or
 $\mathcal{F} \subseteq \widehat{\mathscr{M}}$,
or else $\PlCSP(\mathcal{F})$ is $\SHARPP$-hard.
In the $\PlHolant(\widehat{\mathcal{EQ}}, \widehat{\mathcal{F}})$
setting, the tractability condition is expressed as
$\widehat{\mathcal{F}} \subseteq \mathscr{A}$, or
 $\widehat{\mathcal{F}} \subseteq \widehat{\mathscr{P}}$, or
 $\widehat{\mathcal{F}} \subseteq \mathscr{M}$.

Note that $\mathscr{A}$ is invariant under the transformation,
i.e., $\widehat{\mathscr{A}} = \mathscr{A}$.
 However, $\widehat{\mathscr{P}}$ is more difficult to reason about
than $\mathscr{P}$,
while $\mathscr{M}$ is easier than $\widehat{\mathscr{M}}$ to handle.
The former suggests that we carry our proof in the $\PlCSP$ framework,
while the latter  suggests the opposite,
 that we do so in the $\PlHolant$ framework instead.

One necessary condition for  $\mathscr{M}$ is
the Parity Condition. If some signature in $\widehat{\mathcal{F}}$
violates the Parity Condition, then we have eliminated
one possibility $\widehat{\mathcal{F}} \subseteq \mathscr{M}$.
In this case if we prove in the $\PlCSP$ framework,
we can avoid discussing the more difficult
class $\widehat{\mathscr{M}}$.
On the other hand, if $\widehat{\mathcal{F}}$ satisfies
the Parity Condition, then we have the lucky situation
(Proposition~\ref{parity-product-affine})
that all signatures in $\mathcal{F} \cap \mathscr{P}$
are already in $\mathscr{A}$. This is equivalent to
$\widehat{\mathcal{F}} \cap \widehat{\mathscr{P}}
 \subseteq \mathscr{A}$, and therefore
$\widehat{\mathcal{F}} \subseteq \widehat{\mathscr{P}}$
already implies
$\widehat{\mathcal{F}} \subseteq \mathscr{A}$,
with the consequence that we do not need to specifically discuss
the tractability  condition
$\widehat{\mathcal{F}} \subseteq \widehat{\mathscr{P}}$.
Thus in this case we can avoid discussing the irksome
class $\widehat{\mathscr{P}}$.

Therefore we break the proof into two main cases according to whether
 $\widehat{\mathcal{F}}$
satisfies
the Parity Condition or not.
If not, we want to show that $\PlCSP(\mathcal{F})$ is $\SHARPP$-hard
unless  $\mathcal{F} \subseteq \mathscr{A}$
or  $\mathcal{F} \subseteq \mathscr{P}$
(Theorem~\ref{main-theorem-for-no-parity}).
If yes,  we want to  prove in the $\PlHolant$ setting
for $\mathscr{A}$ and
$\mathscr{M}$, namely
$\PlHolant(\widehat{\mathcal{EQ}}, \widehat{\mathcal{F}})$ is
 $\SHARPP$-hard
unless $\widehat{\mathcal{F}} \subseteq  \mathscr{A}$
or $\widehat{\mathcal{F}} \subseteq  \mathscr{M}$
(Theorem~\ref{With-arity-4-non-matchgate-signature}).

In the first main case where  $\widehat{\mathcal{F}}$ fails
the Parity Condition,
from any signature in $\widehat{\mathcal{F}}$ violating the Parity Condition,
we can construct the simplest such signature,
namely a unary signature $[1, w]$ with $w\neq 0$,
in the Holant framework
$\operatorname{Pl-Holant}(\widehat{\mathcal{EQ}}, \widehat{\mathcal{F}})$.
Any signature that violates the Parity Condition is a witness that
 $\widehat{\mathcal{F}}\nsubseteq\mathscr{M}$,
or equivalently $\mathcal{F} \nsubseteq \widehat{\mathscr{M}}$.
  If ${\mathcal{F}}\subseteq\mathscr{A}$ or
${\mathcal{F}}\subseteq{\mathscr{P}}$,
then the problem $\PlCSP(\mathcal{F})$ is tractable. Otherwise,
there exist some signatures  $f, g \in \mathcal{F}$ such that
$f \not \in \mathscr{A}$  and $g \not \in \mathscr{P}$.
We would like to construct some \emph{symmetric} signatures
 from these that  are also non-affine and non-product type, respectively,
and then apply Theorem~\ref{heng-tyson-dichotomy-pl-csp}.
For the non-product type we will do so in the
 $\PlCSP(\mathcal{F})$ setting, to avoid having to deal with
$\widehat{\mathscr{P}}$. For the non-affine signatures,
we can do so in either the $\PlCSP$ framework or the $\PlHolant$ framework
as $\mathscr{A}$ is invariant $\widehat{\mathscr{A}} = \mathscr{A}$.

However the main difficulty in this plan is that it is generally
difficult to construct \emph{symmetric} signatures
from \emph{asymmetric} signatures in a \emph{planar} fashion,
especially for arity greater than 3. Therefore, a main engine
of the proof is \emph{arity reduction}. Starting from a
non-product type signature of arity $n > 3$, we construct
in the $\PlCSP$ setting a non-product type signature of arity $n-1$.
Then in an arduous proof (given in the \textcolor{red}{Appendix})
we show how to construct, from any non-product type signature
of arity $3$, \emph{either} a  binary non-product type signature
\emph{or} a \emph{symmetric} and non-product type signature of arity 3.
Lemma~\ref{binary-product-asymmetric-symmetric}
turns a binary non-product type signature into a \emph{symmetric
and non-product type} signature.
These constructions will need suitable unary signatures
which will be constructed, starting with that $[1,w]$
in $\operatorname{Pl-Holant}(\widehat{\mathcal{EQ}}, \widehat{\mathcal{F}})$.
The derivative operator (Definition~\ref{derivative}) will be used
throughout.

 For the construction of non-affine signatures,
we will introduce a \emph{Tableau Calculus}.
Again we will carry out an arity reduction proof, this time all the way
down to arity one.
We prove that with the help of unary signatures
 $[1, 0], [0, 1], [1, x]$ with $x\neq 0$,
we can get a unary non-affine signature from any non-affine
signature of higher arity in the $\PlHolant$ setting
(Lemma~\ref{arity-reduction-affine}).
This proof heavily  depends on the new Tableau Calculus.
  Then we construct $[1, 0], [0, 1], [1, x]$
   by shuttling between
 Pl-\#CSP$(\mathcal{F})$
  and Pl-Holant$(\widehat{\mathcal{EQ}}, \widehat{\mathcal{F}})$.
  There is an exceptional case where
 all signatures in $\mathcal{F}$ are $\{0, 1\}$-valued
 in Pl-\#CSP$(\mathcal{F})$.
  In this case, we cannot construct $[1, 0], [0, 1], [1, x]$ simultaneously.
  We resolve this case separately.
For  $\{0, 1\}$-valued $\mathcal{F}$, we actually also cannot
construct all the unary signatures in the arity reduction
proof for non-product type, if we \emph{only} assume the existence of
some $g \in  \mathcal{F} \setminus \mathscr{P}$.
However if we have both $g \in \mathcal{F} \setminus \mathscr{P}$ and some
$f \in \mathcal{F} \setminus \mathscr{A}$,
we can use $f$ to produce the needed unary signatures
to help the arity reduction on $g$.
All these  use Tableau Calculus.

The second  main case
 is  when all signatures in $\widehat{\mathcal{F}}$ satisfy
the Parity Condition.
In this case,
 if $\widehat{\mathcal{F}}\subseteq\mathscr{A}$, or $\widehat{\mathcal{F}}\subseteq\widehat{\mathscr{P}}$, or $\widehat{\mathcal{F}}\subseteq\mathscr{M}$,
 then the problem is tractable.
These are the exact tractability criteria according to
the dichotomy theorem to be proved,
 Theorem~\ref{main-dichotomy-thm}.
However due to the Parity Condition, there are really
only
 two kinds of containment here, $\widehat{\mathcal{F}}\subseteq\mathscr{A}$
or $\widehat{\mathcal{F}}\subseteq\mathscr{M}$; the containment
$\widehat{\mathcal{F}}\subseteq\widehat{\mathscr{P}}$
is subsumed by $\widehat{\mathcal{F}}\subseteq\mathscr{A}$.
Therefore we want to prove that if
$\widehat{\mathcal{F}}\not \subseteq\mathscr{A}$ and
$\widehat{\mathcal{F}}\not \subseteq\mathscr{M}$,
then
$\operatorname{Pl-Holant}(\widehat{\mathcal{EQ}}, \widehat{\mathcal{F}})$
is $\SHARPP$-hard.

Again a natural idea is to construct  non-affine and
 non-matchgate
\emph{symmetric} signatures from any such asymmetric signatures,
and then we can apply the known dichotomy
Theorem~\ref{heng-tyson-dichotomy-pl-csp}.
 The main difficulty of this approach lies in dealing
with non-matchgate signatures.
 Note that both $\widehat{\mathcal{F}}$
and   $\widehat{\mathcal{EQ}}$ (being a subset of $\mathscr{M}$)
 satisfy the Parity Condition,
and therefore the signature of any construction from
an $(\widehat{\mathcal{EQ}}
\cup \widehat{\mathcal{F}})$-gate must also satisfy
the Parity Condition.
By Lemma~\ref{matchgate-identity-for-arity-4},
any signature of arity at most 3 is a matchgate signature
 iff it satisfies the Parity Condition.
Hence all constructible non-matchgate signatures  have
arity $\geq 4$.
But it is difficult to construct a
 symmetric signature from any asymmetric signature of arity
 $\geq 4$ while preserving planarity.

 So we take an alternative approach.
For a given non-matchgate signature,
we first prove that we can get a non-matchgate signature $f$ of arity 4.
Then we can construct either the crossover function $\mathfrak{X}$ or
 $(=_4)$ from $f$.
If we have the crossover function $\mathfrak{X}$,
we can finish the proof by the non-planar
\#CSP dichotomy Theorem~\ref{non-planar-csp-dichotomy}.
If we have $(=_4)$, then we can construct $(=_{2k})$
for any $k\in\mathbb{Z}^+$ in
Pl-Holant$(\widehat{\mathcal{EQ}}, \widehat{\mathcal{F}})$.
Thus we get all $\mathcal{EQ}_2$.
This implies that (by Eqn.~(\ref{eqn:prelim:PlCSPd_equiv_Holant}))
\begin{equation}\label{cognitive-dissonance}
\operatorname{Pl-Holant}(\widehat{\mathcal{EQ}}, \widehat{\mathcal{F}})
\equiv_{\rm T}
\operatorname{Pl-Holant}({\mathcal{EQ}_2},
\widehat{\mathcal{EQ}}, \widehat{\mathcal{F}})
\equiv_{\rm T}
\operatorname{Pl-\#CSP}^2(\widehat{\mathcal{EQ}}, \widehat{\mathcal{F}}).
\end{equation}

Now comes a ``cognitive dissonance''. By (\ref{cognitive-dissonance}),
what used to be the ``right-hand-side''
in the equivalence
$\operatorname{Pl-\#CSP}(\mathcal{F})
\equiv_{\rm T}
\operatorname{Pl-Holant}(\widehat{\mathcal{EQ}}, \widehat{\mathcal{F}})$
will be treated as a $\operatorname{Pl-\#CSP}^2$ problem
with function set $\widehat{\mathcal{EQ}} \cup \widehat{\mathcal{F}}$.
\begin{align*}
\operatorname{Pl-\#CSP}(\mathcal{F})
\equiv_{\rm T}
\operatorname{Pl-Holant}(\widehat{\mathcal{EQ}}, \widehat{\mathcal{F}})
\equiv_{\rm T} &
\operatorname{Pl-Holant}({\mathcal{EQ}_2},
\widehat{\mathcal{EQ}}, \widehat{\mathcal{F}})\\
 &  ~~~~~~~~ |||_{\rm T}  \\
 &\operatorname{Pl-\#CSP}^2(\widehat{\mathcal{EQ}}, \widehat{\mathcal{F}})
\end{align*}
A $\operatorname{Pl-\#CSP}^2$ problem is more in line with
 a $\operatorname{Pl-\#CSP}$ problem. For  $\operatorname{Pl-\#CSP}^2$
problems over symmetric signatures,
Theorem~\ref{heng-tyson-dichotomy-pl-csp2} says that there are
five tractability classes
$\mathscr{P}, \mathscr{A}, \mathscr{A}^{\dagger},
\widehat{\mathscr{M}}$ and $\widehat{\mathscr{M}^{\dagger}}$.
But now we will apply these on the ``dual side'' $\widehat{\mathcal{EQ}}
\cup \widehat{\mathcal{F}}$,
instead of the ``primal side'' $\mathcal{F}$.
The ``cognitive dissonance'' is that, the transformation from
$(\mathcal{EQ}, \mathcal{F}) \mapsto
(\widehat{\mathcal{EQ}}, \widehat{\mathcal{F}})$ is precisely
for the purpose of
 transforming ${\mathcal{EQ}}$ to be a subset
 $\widehat{\mathcal{EQ}}$ of $\mathscr{M}$,
but now we will subject  $\widehat{\mathcal{EQ}}$ to tractability
tests  including
$\widehat{\mathscr{M}}$ and $\widehat{\mathscr{M}^{\dagger}}$.
But clearly  $\widehat{\mathcal{EQ}}$ contains both $[1, 0] \not \in
\widehat{\mathscr{M}} \cup \widehat{\mathscr{M}^{\dagger}}$,
and $[1,0,1,0] \not \in \mathscr{P} \cup \mathscr{A}^{\dagger}$,
therefore the only remaining possibility for tractability
is $\mathscr{A}$.

Of course if $\widehat{\mathcal{F}} \subseteq \mathscr{A}$,
then $\operatorname{Pl-Holant}(\widehat{\mathcal{EQ}}, \widehat{\mathcal{F}})$
is tractable. Suppose $\widehat{\mathcal{F}} \not \subseteq \mathscr{A}$,
we want to construct a \emph{symmetric} non-affine signature.
We produce such a signature of arity 2 if possible,
by arity reduction. From any $f \in \widehat{\mathcal{F}}
\setminus \mathscr{A}$, which satisfies
the Parity Condition, we can first get a non-affine signature
satisfying the even Parity Condition.
Then every signature constructible from that
 has even parity, as
$\widehat{\mathcal{EQ}}$ also has even parity.
Any non-affine binary signature satisfying the even Parity Condition
is automatically symmetric.
This part of the proof is the content of Section~\ref{sec-csp2}
(Theorem~\ref{dichotomy-csp-2}).

A technical  difficulty is that when
$\widehat{\mathcal{F}}$ satisfies the even
Parity Condition, it is impossible to construct $[0, 1]$.
Instead we find that we
can try to construct $[0, 1]^{\otimes 2}$ and
prove that $[0, 1]^{\otimes 2}$ is almost as good as $[0, 1]$
 with the help of $[1, 0, 1, 0] \in \widehat{\mathcal{EQ}}$.
Then
there are three cases. (1) If it is  not the case that every function in
$\widehat{\mathcal{F}}$ takes values in $\{0, 1, -1\}$ up to
a constant, then we can construct $[0, 1]^{\otimes 2}$ and $[1,0,-1]$
and complete the proof.
(2) If every function in
$\widehat{\mathcal{F}}$ takes values in $\{0, 1, -1\}$ up to
a constant but not  every function in
$\widehat{\mathcal{F}}$ takes values in  $\{0, 1\}$
up to
a constant, then we can construct $[1,0,-1]$,
and complete the proof.
(3) If every function in
$\widehat{\mathcal{F}}$ takes values in $\{0, 1\}$
then we prove it separately.
In all cases, we use our Tableau Calculus.

This completes an outline of the proof guided by an overall
vision that (A) there is a dichotomy,
and (B) the right form of this dichotomy is as stated in
Theorem~\ref{main-dichotomy-thm}.

Of course as a proof strategy,
logically this is a bit self-serving.
Essentially we want the validity of the very
statement we want to prove to provide its own guarantee of
success in every step in its proof.
Given the fact that there are
other tractable classes for Pl-Holant
problems~\cite{Cai-Fu-Guo-W}
 not encompassed in the list given in
Theorem~\ref{main-dichotomy-thm}, the validity of this vision
for $\PlCSP$ problems is at least not obvious.
Luckily, this vision is correct. And therefore,
the self-serving plan becomes a reliable guide
to the proof, a bit self-fulfilling.
Sometimes the statement of a theorem helps its own proof.

%% file: f-gate.tex
\begin{figure}[t]
 \centering
 \begin{tikzpicture}[scale=\scale,transform shape,node distance=\nodeDist,semithick]
  \node[external]  (0)                     {};
  \node[internal]  (1) [below right of=0]  {};
  \node[external]  (2) [below left  of=1]  {};
  \node[internal]  (3) [above       of=1]  {};
  \node[internal]  (4) [right       of=3]  {};
  \node[internal]  (5) [below       of=4]  {};
  \node[internal]  (6) [below right of=5]  {};
  \node[internal]  (7) [right       of=6]  {};
  \node[internal]  (8) [below       of=6]  {};
  \node[internal]  (9) [below       of=8]  {};
  \node[internal] (10) [right       of=9]  {};
  \node[internal] (11) [above right of=6]  {};
  \node[internal] (12) [below left  of=8]  {};
  \node[internal] (13) [left        of=8]  {};
  \node[internal] (14) [below left  of=13] {};
  \node[external] (15) [left        of=14] {};
  \node[internal] (16) [below left  of=5]  {};
  \path let
         \p1 = (15),
         \p2 = (0)
        in
         node[external] (17) at (\x1, \y2) {};
  \path let
         \p1 = (15),
         \p2 = (2)
        in
         node[external] (18) at (\x1, \y2) {};
  \node[external] (19) [right of=7]  {};
  \node[external] (20) [right of=10] {};
  \path (1) edge                             (5)
            edge[bend left]                 (11)
            edge[bend right]                (13)
            edge node[near start] (e1) {}   (17)
            edge node[near start] (e2) {}   (18)
        (3) edge                             (4)
        (4) edge[out=-45,in=45]              (8)
        (5) edge[bend right, looseness=0.5] (13)
            edge[bend right, looseness=0.5]  (6)
        (6) edge[bend left]                  (8)
            edge[bend left]                  (7)
            edge[bend left]                 (14)
        (7) edge node[near start] (e3) {}   (19)
       (10) edge[bend right, looseness=0.5] (12)
            edge[bend left,  looseness=0.5] (11)
            edge node[near start] (e4) {}   (20)
       (12) edge[bend left]                 (16)
       (14) edge node[near start] (e5) {}   (15)
            edge[bend right]                (12)
       (16) edge[bend left,  looseness=0.5]  (9)
            edge[bend right, looseness=0.5]  (3);
  \begin{pgfonlayer}{background}
   \node[draw=\borderColor,thick,rounded corners,fit = (3) (4) (9) (e1) (e2) (e3) (e4) (e5),inner sep=0pt,transform shape=false] {};
  \end{pgfonlayer}
 \end{tikzpicture}
 \caption{An $\mathcal{F}$-gate with 5 dangling edges.}
 \label{fig:Fgate}
\end{figure}

%% file: 3no-parity.tex
\setcounter{MaxMatrixCols}{20}

\section{When $\widehat{\mathcal{F}}$ Does Not Satisfy Parity}
The following Lemma shows that if there is a signature in
$\widehat{\mathcal{F}}$ that does not satisfy the parity condition, then in
$\operatorname{Pl-Holant}(\widehat{\mathcal{EQ}}, \widehat{\mathcal{F}})$,
we can construct a unary signature which does
not satisfy the parity condition.

\begin{lemma}\label{constructing-[1,a]}
If $\widehat{\mathcal{F}}$ contains a signature $f$
that does not satisfy the Parity Condition, then
we can construct a unary signature  $[1, w]$
with $w\neq 0$
in $\operatorname{Pl-Holant}(\widehat{\mathcal{EQ}}, \widehat{\mathcal{F}})$,
such that
\[\operatorname{Pl-Holant}(\widehat{\mathcal{EQ}}, [1, w], \widehat{\mathcal{F}})
\le_{\rm T}
\operatorname{Pl-Holant}(\widehat{\mathcal{EQ}}, \widehat{\mathcal{F}}).\]
\end{lemma}
\begin{proof}
Let $f$ have arity $n \ge 1$.
Since $f$ does not satisfy the Parity Condition,
there exists $\alpha\in\{0, 1\}^{n}$ such that $f_{\alpha}\neq 0$,
${\rm wt}(\alpha)$ is odd, and
wt$(\alpha)=\displaystyle\min_{\eta\in\{0, 1\}^{n}}\{{\rm wt}(\eta)
\mid f_{\eta}\neq 0, ~{\rm wt}(\eta)$ is odd$\}$.
Since $[1,0] \in \widehat{\mathcal{EQ}}$,
we can construct the signature $\partial_{[1, 0]}^{S}(f)=(f_{00\cdots0}, \ldots, f_{\alpha})$, where $S=\{k \mid \mbox{the $k$-th bit of $\alpha$ is 0}\}$.
Note that the signature $\partial_{[1, 0]}^{S}(f)$
has an odd arity ${\rm wt}(\alpha) = n - |S|$, where every entry having
an odd weight index is 0 except for $f_{\alpha}$, by the minimality
of ${\rm wt}(\alpha)$ among all nonzero entries of $f$ of odd weight.
So we can construct
$\partial_{=_2}^{\left(\frac{{\rm wt}(\alpha)-1}{2}\right)}(\partial_{[1, 0]}^{S}(f))
=[a, f_{\alpha}]$, by connecting all variables $x_2, \ldots,
x_{{\rm wt}(\alpha)}$ of
$\partial_{[1, 0]}^{S}(f)$ in adjacent pairs in a planar way,
 where $a$ is determined by entries of $\partial_{[1, 0]}^{S}(f)$
with even index. Note that the entry $f_{\alpha}$ remains unchanged,
since in forming
 $\partial_{=_2}^{\left(\frac{{\rm wt}(\alpha)-1}{2}\right)}(\partial_{[1, 0]}^{S}(f))$,
only entries of $f$ with lower odd indices, which are all zero,
are combined with $f_{\alpha}$.
If $a\neq 0$, then we already  have
the desired $[1, w]$ by normalization, where $w=\frac{f_{\alpha}}{a}$.

If $a=0$, then we have $[0, 1]$ up to the nonzero scalar $f_{\alpha}$.
Since $f$ does not satisfy the parity condition, there exist $\beta$
and $\gamma\in\{0, 1\}^n$, satisfying the following:
$f_{\beta}\neq 0, f_{\gamma}\neq 0$, wt$(\beta)$ and wt$(\gamma)$ have
opposite parity and
\[d={\rm wt}(\beta\oplus\gamma)=\displaystyle\min_{\zeta, \eta\in\{0, 1\}^{n}}\{{\rm wt}(\zeta\oplus\eta) ~\mid~
\mbox{wt$(\zeta)$ and wt$(\eta)$
have  opposite parity,}
 ~ f_{\zeta}\neq 0,~ f_{\eta}\neq 0\}.\]
Then we have $g=\partial_{[1, 0]}^{S_0}[\partial_{[0, 1]}^{S_1}(f)]$, where
for $b \in \{0, 1\}$,
\[S_b=\{k~\mid~ \mbox{the $k$-th bits of both $\beta$
{\rm and} $\gamma$ are $b$}\}.\]
Note that the arity of $g$ is $d$.

By deleting all bits in  $S_0\cup S_1$ from $\beta$ and $\gamma$,
we get two bit strings $\beta', \gamma' \in \{0, 1\}^d$ respectively.
We
have  $g_{\beta'}=f_{\beta}$ and $g_{\gamma'}=f_{\gamma}$ and all
other entries of $g$
are 0.
Note that $\beta'$ and $\gamma'$
 have opposite parity since $\beta$ and $\gamma$ have opposite parity.
Without loss of generality, assume that wt$(\beta')$ is odd and wt$(\gamma')$
is even.
Then $\partial_{g}([1, 0, 1, \ldots, 0, 1])=[g_{\beta'}, g_{\gamma'}]
=[f_{\beta}, f_{\gamma}]$
and we are done,
where $[1, 0, 1, \ldots, 0, 1]=\frac{1}{2}\{[1, 1]^{\otimes d+1}+[1, -1]^{\otimes d+1}\} \in  \widehat{\mathcal{EQ}}$.
\end{proof}

The next lemma is a simple fact from linear algebra~\cite{Harris-ag}.
It will be used in the proof Lemma~\ref{binary-product-asymmetric-symmetric}.
\begin{lemma}\label{tensor-rank}
Let $\bf{a, b, c, d} \in \mathbb{C}^2$
 and suppose $\bf{c, d}$ are linearly independent.
Suppose for some $n\geq 3$
we have
 ${\bf a}^{\otimes n}+{\bf b}^{\otimes n}=
{\bf c}^{\otimes n}+{\bf d}^{\otimes n}$.
Then ${\bf a}=\phi{\bf c}$,
${\bf b}=\psi{\bf d}$, or
${\bf a}=\phi{\bf d}$,
${\bf b}=\psi{\bf c}$
for some $\phi^n=\psi^n=1$.
\end{lemma}

\subsection{Arity Reduction for Non-Product-Type Signatures}

Our plan is to
use the dichotomy theorems for symmetric signatures.
For that we have to construct symmetric signatures from asymmetric signatures.
For example, starting from a signature not in $\mathscr{P}$,
we want to construct a symmetric signature not in $\mathscr{P}$.
It is generally difficult to construct symmetric signatures from
asymmetric signatures in a planar construction, especially when the arity
is high.
So one of our main techniques  is arity reduction. We want to reduce
the arity of a  signature while keeping it outside $\mathscr{P}$.
Every unary signature is in $\mathscr{P}$. So the lowest arity
outside $\mathscr{P}$ is two.
If we obtain a binary signature $f=(f_{00}, ~f_{01}, ~f_{10}, ~f_{11})
= (a, ~b, ~c, ~d) \not \in \mathscr{P}$,
we can take 3 copies of $f$ and connect the first input of each $f$
to an edge of $(=_3)$ and leave the second input of the 3 copies of $f$
 as 3 dangling edges. This planar gadget has the symmetric
 signature  $[a, b]^{\otimes 3}+[c, d]^{\otimes 3}$.
The following lemma says that this ternary symmetric signature
does not  belong to $\mathscr{P}$.  Our main construction for a symmetric
signature not in $\mathscr{P}$ will be an induction on arity $n$ with
a base case at $n=3$. The reason we start at $n=3$ is because
certain steps for $n \ge 3$ will not work for $n=2$. However
 Lemma~\ref{binary-product-asymmetric-symmetric} implies that if we have
a  binary signature that is not in $\mathscr{P}$, then we
can construct in $\PlCSP(f)$
 a symmetric signature that is not in $\mathscr{P}$.

\begin{lemma}\label{binary-product-asymmetric-symmetric}
For any binary signature  $f=(a, ~b, ~c, ~d)$,
$f\in\mathscr{P}$ iff
$g=[a, b]^{\otimes 3}+[c, d]^{\otimes 3}\in\mathscr{P}$.
\end{lemma}
\begin{proof}
If $f\in\mathscr{P}$, then either
 $f$ is degenerate or $a=d=0$ or $b=c=0$ by definition.
If $f$ is degenerate, then $[a, b], [c, d]$ are linearly dependent. Then $g$ is degenerate and $g\in\mathscr{P}$.
If $a=d=0$ or $b=c=0$, then $g$ is a \textsc{Generalized Equality} and $g\in\mathscr{P}$.

Conversely, if $g\in\mathscr{P}$, then either
 $g$ is degenerate or $g$ is a \textsc{Generalized Equality}
since $g$ is symmetric.

If $g$ is degenerate, then there exists a vector $[e, f]$ such that
\[g=[a, b]^{\otimes 3}+[c, d]^{\otimes 3}=[e, f]^{\otimes 3}.\]
To use Lemma~\ref{tensor-rank}, we rewrite it as
\[g=[a, b]^{\otimes 3}+[c, d]^{\otimes 3}=[e, f]^{\otimes 3}+[0, 0]^{\otimes 3}.\]
If $[a, b], [c, d]$ are linearly independent, then
$[a, b]=[0, 0]$ or $[c, d]=[0, 0]$. This contradicts that $[a, b], [c, d]$ are linearly independent.
Thus $[a, b], [c, d]$ are linearly dependent.
This implies that
 $f$ is degenerate. Thus $f\in\mathscr{P}$.

If $g$ is a \textsc{Generalized Equality}, then there exists $x, y$
such that
\[g=[a, b]^{\otimes 3}+[c, d]^{\otimes 3}=[x, 0]^{\otimes 3}+[0, y]^{\otimes 3}.\]
If $[a, b], [c, d]$ are linearly dependent, then $f$ is degenerate. Thus $f\in\mathscr{P}$.
If $[a, b], [c, d]$ are linearly independent,
By Lemma~\ref{tensor-rank}, there exists $\omega_1, \omega_2$
such that $[a, b]=\omega_1[x, 0], [c, d]=\omega_2[0, y]$ or $[a, b]=\omega_1[0, y], [c, d]=\omega_2[x, 0]$.
This implies that $b=c=0$ or  $a=d=0$.
Thus $f\in\mathscr{P}$.
\end{proof}

We will use the next lemma in the proof of Theorem~\ref{arity-reduction-product}.
\begin{lemma}\label{arity-3-not-product}
Suppose $f$ is a signature of arity $n\geq 3$, and $[1,a]$ and $[1,b]$
are two unary signatures. Let $P = \partial_{[1, a]}^{\{1\}}(f)$
and $Q  =\partial_{[1, b]}^{\{1\}}(f)$.
Suppose
 both  $P \in\mathscr{P}$ and $Q\in\mathscr{P}$  and both are not identically zero.
If $P$ and $Q$
do not have compatible type, then $f\notin \mathscr{P}$.
\end{lemma}
\begin{proof}
For a contradiction suppose
$f\in \mathscr{P}$. Then there exists a primitive decomposition of $f
= \displaystyle\prod_{i=1}^k F_i(X|_{I_i})$
with partition
 $\mathcal{I}=\{I_1, I_2, \ldots, I_k\}$ of $[n]$
and signatures $F_1, F_2, \ldots, F_k$.
Without loss of generality, suppose  $1\in I_1$.
If $|I_1|=1$, then $\partial_{[1, a]}^{\{1\}}(F_1)$ is a nonzero
constant; it is nonzero because $P
= \partial_{[1, a]}^{\{1\}}(F_1) \cdot \displaystyle\prod_{i=2}^k F_i(X|_{I_i})$
is not identically zero. Similarly
$Q = \partial_{[1, b]}^{\{1\}}(F_1) \cdot \displaystyle\prod_{i=2}^k F_i(X|_{I_i})$
is also a nonzero constant multiplied by the same
decomposition. Hence  $P$ and $Q$ have the same primitive
decomposition.
Thus they have compatible type. This is a contradiction.

So we may assume $|I_1|\geq 2$. Let
$F'_1 = \partial_{[1, a]}^{\{1\}}F_1$ and $F''_1 = \partial_{[1, b]}^{\{1\}}F_1$.
Being a factor in a primitive decomposition of arity at least 2,
$F_1$ is non-degenerate, and supp$(F_1)$
consists of  two antipodal points
 $\{0\alpha, 1\bar{\alpha}\}$, for some $\alpha\in\{0, 1\}^{|I_1|-1}$.
If $ab\neq 0$, then
both $F'_1$ and $F''_1$ still have support consisting of
two antipodal points
 $\{\alpha, \bar{\alpha}\}$, e.g.,
$F'_1(\bar{\alpha}) = 1 F_1(0\bar{\alpha}) + a F_1(1\bar{\alpha}) =
a F_1(1\bar{\alpha})  \not = 0$.
In particular both $F'_1\in \mathcal{E}$ and $F''_1\in \mathcal{E}$.
Thus
  $F'_1 \displaystyle\prod_{i=2}^k F_i(X|_{I_i})$
and $F''_1 \displaystyle\prod_{i=2}^k F_i(X|_{I_i})$
are the primitive decompositions of $P$ and $Q$ respectively.
Thus $P$ and $Q$ have compatible type. This is a contradiction.

Now suppose $ab=0$.
Since  $P$ and $Q$
do not have compatible type, certainly $a \not =b$.
Without loss of generality, we may assume that $a=0$ and $b \neq 0$.
In this case,
$F'_1$ is further decomposed as a product of unary signatures,
$F'_1 = \displaystyle\prod_{j=1}^{|I_1|-1}G_j$, where
each $G_j$ is a nonzero scalar
multiple
of $[1,0]$ or $[0,1]$.
The primitive decomposition of $P$ is $\displaystyle\prod_{j=1}^{|I_1|-1}G_j
\displaystyle\prod_{i=2}^k F_i(X|_{I_i})$.
The support of $\displaystyle\prod_{j=1}^{|I_1|-1}G_j$ is a
singleton point $\{\alpha\}$, a proper subset of the support of $F''_1$,
which is $\{\alpha, \bar{\alpha}\}$.
Thus $P$ and $Q$ still have compatible type. This is a contradiction.
\end{proof}

\begin{theorem}\label{arity-reduction-product}
Suppose $\mathcal{F}$  contains a signature $f \not \in \mathscr{P}$
of arity $n\geq 3$. Let $[1, a], [1,b], [1,c]$ be three unary
signatures that are pairwise linearly independent.
Then  there exists a symmetric signature $g  \not \in \mathscr{P}$ such that
\begin{equation}\label{eqn-in-thm3.5}
\operatorname{Pl-\#CSP}(g, [1, a], [1, b], [1, c], \mathcal{F})
\le_{\rm T}
\operatorname{Pl-\#CSP}([1, a], [1, b], [1, c], \mathcal{F}).
\end{equation}
\end{theorem}
\begin{proof}
We prove by induction on $n$.
The base case is $n=3$. By Theorem~\ref{arity-3-nonproduct} we can produce
$g \not \in \mathscr{P}$ satisfying (\ref{eqn-in-thm3.5}) such that
either $g$ has arity 2 or $g$ is symmetric and has arity 3.
If $g$ has arity 2, we
use Lemma~\ref{binary-product-asymmetric-symmetric}
to produce a symmetric  $g'  \not \in \mathscr{P}$ of arity 3.

Now assume $n \ge 4$, and the theorem is true for $n-1$.
We show how to construct some $g \not \in \mathscr{P}$
of arity $n-1$ satisfying (\ref{eqn-in-thm3.5}).
Define
\[P = \partial_{[1,a]}^{\{1\}}(f),~~~~
Q= \partial_{[1,b]}^{\{1\}}(f), ~~~~ R = \partial_{[1,c]}^{\{1\}}(f).\]
If any of $P, Q$ or $R \notin\mathscr{P}$, then we are done by induction.
So we may assume $P, Q$  and $R$ all belong to $\mathscr{P}$.
If two of them are linearly dependent, we claim that  $f\in\mathscr{P}$.
Without loss of generality, we assume that $P$ and $Q$ are linearly dependent.
Note that each of $P, Q$  and $R$ is a linear combination of
$f^{x_1=0}$ and $f^{x_1=1}$.
From
\begin{equation}\label{P-Q-equation}
\left[
\begin{matrix}
P \\ Q
\end{matrix}
\right]=
\left[ \begin{matrix}
1 & a\\
1 & b
\end{matrix} \right]
\left[
\begin{matrix}
f^{x_1=0}\\
f^{x_1=1}
\end{matrix}
\right],
\end{equation}
since $a \not = b$,
we have
\begin{equation}\label{P-Q-equation-solved}
\left[
\begin{matrix}
f^{x_1=0}\\
f^{x_1=1}
\end{matrix}
\right] =
\frac{1}{b-a}
\left[
\begin{matrix}
b & -a\\
-1 & 1
\end{matrix}
\right]
\left[
\begin{matrix}
P\\
Q
\end{matrix}
\right].
\end{equation}
If both $P$ and $Q$ are identically zero,
then both $f^{x_1=0}$ and $f^{x_1=1}$ are identically zero,
and so $f$ is identically zero. This contradicts that $f\notin\mathscr{P}$.
So we may assume that $P\neq 0$. Then there exists
a constant $\lambda$ such that $Q=\lambda P$.
This implies that
\[
f^{x_1=0}=\frac{b-a\lambda}{b-a}P,~~~~
f^{x_1=1}=\frac{\lambda-1}{b-a}P.
\]
So $f=\frac{1}{b-a}[b-a\lambda, \lambda-1]\otimes P$.
This implies that $f\in\mathscr{P}$. This is a contradiction.

In the following,  $P$, $Q$ and $R$ all belong to $\mathscr{P}$
 and they are pairwise linearly independent.
In particular none of them is identically zero.

\begin{itemize}
\item If all three pairs  from $\{P, Q, R\}$
 (pairwise) have compatible types,
then by Lemma~\ref{same-type},
there exist a common
 partition $\mathcal{I}=\{I_1, I_2, \ldots, I_k\}$ of $[n]\setminus\{1\}$
and signatures $P_1, P_2, \ldots, P_k$, $Q_1, Q_2, \ldots, Q_k$ and $R_1, R_2, \ldots, R_k$
such that
\[P(X)=\displaystyle\prod_{i=1}^k P_i(X|_{I_i}),~~~~
Q(X)=\displaystyle\prod_{i=1}^k Q_i(X|_{I_i}), ~~~~
R(X)=\displaystyle\prod_{i=1}^k R_i(X|_{I_i})\]
 where $P_i, Q_i, R_i\in\mathcal{E}$
and there exist $\alpha_i\in\{0, 1\}^{|I_i|}$ such that supp$(P_i)$, supp$(Q_i)$, supp$(R_i)\subseteq\{\alpha_i, \overline{\alpha}_i\}$  for $1\leq i\leq k$.

Since $P$ and $Q$ are linearly independent, there is at least
one $1 \le i \le k$ such that $P_i$ and $Q_i$ are linearly independent.
We claim that there exists exactly one $i$ such that $P_i$ and $Q_i$ are linearly independent.
Otherwise, without loss of generality, we can assume that
both pairs $P_1, Q_1$ and $P_2, Q_2$ are linearly independent respectively.
Because $P$ and $Q$ are linearly independent, $P_i$ and $Q_i$ are not the zero signature for any $i\in[k]$.
Choose any unary signature $u \in \{[1,a], [1,b], [1,c]\}$
that is not $[1, 0]$. This is clearly possible because at most one
of them can be $[1, 0]$.

For any $i\in[k]\setminus\{1, 2\}$, we shrink
both  $P_i$ and $Q_i$ to a nonzero constant  in two steps as follows:
Step 1: if the arity $|I_i|$ of both $P_i$ and $Q_i$  is greater than 1
(skip Step 1 if $|I_i|=1$), we combine $|I_i|-1$ copies of $u$ to both $P_i$ and $Q_i$.
Since  $P_i$ (respectively  $Q_i$)  is not identically zero
and has either a single point in supp$(P_i)$ (respectively  supp$(Q_i)$)
or a pair of antipodal points, and both entries of the unary signature $u$
are nonzero,
  this operation shrinks  $P_i$ (respectively  $Q_i$)
 to a nonzero unary signature $[c_1, d_1]$ (respectively $[c_2, d_2]$),
where either $c_1$ or $d_1 \not =0$ (respectively $c_2$ or $d_2 \not =0$).
Step 2:
since we have three unary signatures that are pairwise linearly independent,
there exists at least  one $u' \in \{[1,a], [1,b], [1,c]\}$
such that
$\partial_{u'}([c_1, d_1])$ and $\partial_{u'}([c_2, d_2])$
are both nonzero constants.
So we combine $u'$ to $[c_1, d_1]$ and $[c_2, d_2]$
and we have shrunken both $P_i$ and $Q_i$ to nonzero constants.

By
(\ref{P-Q-equation-solved})
we have
\[R=f^{x_1=0}+cf^{x_1=1}=\frac{1}{b-a}[(b-c)P+(c-a)Q].\]
After shrinking $P_i$ and $Q_i$
for every $i\in[k]\setminus\{1, 2\}$,
 there exist constants $c_P, c_Q$ and $c_R$ such that
\[c_R R_1\otimes R_2=c_P P_1\otimes P_2+c_Q Q_1\otimes Q_2,\]
where $c_P \neq 0$ and $c_Q\neq 0$.
If we write $R_1\otimes R_2$ in its matrix form
as a matrix in $\mathbb{C}^{2^{|I_1|} \times 2^{|I_2|}}$
where the rows and columns are indexed by
assignments to the variables in $I_1$ and $I_2$ respectively,
 it has rank at most one,  because it is expressible as
the product of a column vector times a row vector
$R_1^{\texttt{T}} R_2$,
where $R_i$ is a row vector in $\mathbb{C}^{2^{|I_i|}}$.
But
the signature matrix of $c_P P_1\otimes P_2+c_Q Q_1\otimes Q_2$ is
\[\left[\begin{matrix}
P_1^{\texttt{T}} & Q_1^{\texttt{T}}
\end{matrix}\right]
\left[\begin{matrix}
c_P & 0 \\
 0    & c_Q
\end{matrix}\right]
\left[\begin{matrix}
P_2 \\ Q_2
\end{matrix}\right]
\]
which has rank 2 because
 $P_i, Q_i$ are linearly independent for $i =1,2$,
and $\left[\begin{matrix}
P_1^{\texttt{T}} & Q_1^{\texttt{T}}
\end{matrix}\right]
\in \mathbb{C}^{2^{|I_1|} \times 2}$
and $\left[\begin{matrix}
P_2 \\ Q_2
\end{matrix}\right]
\in \mathbb{C}^{2 \times 2^{|I_2|}}$ both have rank 2.
This is a contradiction.

Now without loss of generality, we may assume that $P_1, Q_1$ are linearly independent
and $P_i, Q_i$ are linearly dependent for $i=2, \ldots, k$.
Thus each $Q_i$ is a nonzero multiple of $P_i$, for $i=2, \ldots, k$.
By replacing $Q_1$ with a nonzero multiple of $Q_1$,
we may assume $Q_i=P_i$ for $i=2, \ldots, k$.
Since we have three unary signatures that are linearly independent pairwise,
by a similar argument, we can connect the variables of $f$ in
$I_2, \ldots, I_k$ to some unary signatures such that
$P_2, \ldots, P_k$ each contributes a nonzero constant
factor.
Let the resulting signature be $h$  on
variables from $\{x_s \mid s \in \{1\} \cup I_1\}$.
Note that
$\partial_{[1,a]}^{\{1\}}(h)=\lambda P_1$, $\partial_{[1,b]}^{\{1\}}(h)=\lambda Q_1$,
 where $\lambda$ is a nonzero constant, as the following
diagrams commute.
\begin{equation*}
\begin{tikzpicture}[every node/.style={midway}]
  \matrix[column sep={8em,between origins}, row sep={2em}] at (0,0) {
    \node(f) {$f$}  ; & \node(h) {$h$}; \\
    \node(P) {$P$}; & \node (P1) {$\lambda P_1$};\\
  };
  \draw[<-] (P) -- (f) node[anchor=east]  {$\partial^{\{1\}}_{[1,a]}$};
  \draw[->] (f) -- (h) node[anchor=south] {$\partial^{I_2}\cdots \partial^{I_k}$};
  \draw[->] (h) -- (P1) node[anchor=west] {$\partial^{\{1\}}_{[1,a]}$};
  \draw[->] (P) -- (P1) node[anchor=north] {$\partial^{I_2}\cdots \partial^{I_k}$};
\end{tikzpicture}
\hspace{40pt}
\begin{tikzpicture}[every node/.style={midway}]
  \matrix[column sep={8em,between origins}, row sep={2em}] at (0,0) {
    \node(f) {$f$}  ; & \node(h) {$h$}; \\
    \node(Q) {$Q$}; & \node (Q1) {$\lambda Q_1$};\\
  };
  \draw[<-] (Q) -- (f) node[anchor=east]  {$\partial^{\{1\}}_{[1,b]}$};
  \draw[->] (f) -- (h) node[anchor=south] {$\partial^{I_2}\cdots \partial^{I_k}$};
  \draw[->] (h) -- (Q1) node[anchor=west] {$\partial^{\{1\}}_{[1,b]}$};
  \draw[->] (Q) -- (Q1) node[anchor=north] {$\partial^{I_2}\cdots \partial^{I_k}
$};
\end{tikzpicture}
\end{equation*}

Then we have
\begin{equation}\label{hx1=01expresssion}
\left[\begin{matrix}h^{x_1=0}\\h^{x_1=1}\end{matrix}\right]=
\frac{\lambda}{b-a}
\left[\begin{matrix}b & -a\\-1 & 1\end{matrix}\right]
\left[\begin{matrix} P_1\\  Q_1\end{matrix}\right].
\end{equation}
Note that
\begin{equation}\label{hx1=01expresssion2}
\left[\begin{matrix} P_1\\  Q_1\end{matrix}\right]=
\begin{bmatrix}
0 & \ldots & 0 & P_1^{\alpha_1} & 0 & \ldots & 0 &  P_1^{\overline{\alpha_1}} & 0 & \ldots & 0\\
0 & \ldots & 0 & Q_1^{\alpha_1}     & 0 & \ldots & 0 &  Q_1^{\overline{\alpha_1}}     & 0 & \ldots & 0
\end{bmatrix}.
\end{equation}
Let
\begin{equation}\label{hx1=01expresssion3}
\left[\begin{matrix}  \check{P}_1^{\alpha_1} &  \check{P}_1^{\overline{\alpha_1}} \\
 \check{Q}_1^{\alpha_1} & \check{Q}_1^{\overline{\alpha_1}} \end{matrix}\right]
=
\frac{\lambda}{b-a}
\left[\begin{matrix}b & -a\\-1 & 1\end{matrix}\right]
\left[\begin{matrix}  P_1^{\alpha_1} &  P_1^{\overline{\alpha_1}} \\
  Q_1^{\alpha_1} &  Q_1^{\overline{\alpha_1}} \end{matrix}\right].
\end{equation}
Then
\begin{equation}\label{hx1=01expresssion4}
\left[\begin{matrix}h^{x_1=0}\\h^{x_1=1}\end{matrix}\right]=
\begin{bmatrix}
0 & \ldots & 0 & \check{P}_1^{\alpha_1} & 0 & \ldots & 0 &
 \check{P}_1^{\overline{\alpha_1}} & 0 & \ldots & 0\\
0 & \ldots & 0 & \check{Q}_1^{\alpha_1}      & 0 & \ldots & 0 &
 \check{Q}_1^{\overline{\alpha_1}}    & 0 & \ldots & 0
\end{bmatrix}.
\end{equation}

Let $\check{h}$ be the binary signature
$(\check{h}^{00},~ \check{h}^{01},~ \check{h}^{10},~\check{h}^{11})
=(\check{P}_1^{\alpha_1},~  \check{P}_1^{\overline{\alpha_1}},~ \check{Q}_1^{\alpha_1}, ~ \check{Q}_1^{\overline{\alpha_1}})$.
If $|I_1|=1$ then $\check{h}$ is $h$.
If $|I_1| >1$, by combining all but one variable in $I_1$ using $[1, 1]$,
which is just $(=_1)$ present in {Pl-\#CSP},
 we can get $\check{h}$.
If  $\check{h}\notin\mathscr{P}$, then we are done by Lemma~\ref{binary-product-asymmetric-symmetric}.

If $\check{h}\in\mathscr{P}$, then $h\in\mathscr{P}$.
Since $P_1$ and $Q_1$ are linearly independent,
$\det
\begin{bmatrix}  \check{P}_1^{\alpha_1} &  \check{P}_1^{\overline{\alpha_1}} \\
 \check{Q}_1^{\alpha_1} & \check{Q}_1^{\overline{\alpha_1}} \end{bmatrix}
\not =0$.
Hence either $\check{P}_1^{\alpha_1} = \check{Q}_1^{\overline{\alpha_1}}
=0$ or $\check{P}_1^{\overline{\alpha_1}} =  \check{Q}_1^{\alpha_1} =0$.
In either case,
compare (\ref{hx1=01expresssion}) to (\ref{hx1=01expresssion4}) with
\begin{equation*}
\left[
\begin{matrix}
f^{x_1=0}\\
f^{x_1=1}
\end{matrix}
\right]
=
\frac{1}{b-a}
\left[
\begin{matrix}
b & -a\\
-1 & 1
\end{matrix}
\right]
\left[
\begin{matrix}
(0 & \ldots & 0 &  P_1^{\alpha_1} & 0 & \ldots & 0 &  P_1^{\overline{\alpha_1}} & 0 & \ldots & 0) \otimes P_2 \otimes \ldots \otimes P_k\\
(0 & \ldots & 0 &  Q_1^{\alpha_1}     & 0 & \ldots & 0 &  Q_1^{\overline{\alpha_1}}     & 0 & \ldots & 0) \otimes P_2 \otimes \ldots \otimes P_k
\end{matrix}
\right]
\end{equation*}
we have
\[
\left[
\begin{matrix}
f^{x_1=0}\\
f^{x_1=1}
\end{matrix}
\right]
=
\frac{1}{\lambda}
\left[
\begin{matrix}
(0 & \ldots & 0 &  \check{P}_1^{\alpha_1} & 0 & \ldots & 0 &
 \check{P}_1^{\overline{\alpha_1}} & 0 & \ldots & 0) \otimes P_2 \otimes \ldots \otimes P_k\\
(0 & \ldots & 0 &  \check{Q}_1^{\alpha_1}    & 0 & \ldots & 0 &
\check{Q}_1^{\overline{\alpha_1}}    & 0 & \ldots & 0) \otimes P_2 \otimes \ldots \otimes P_k
\end{matrix}
\right]
\]
We conclude that $f\in\mathscr{P}$. But this is a contradiction.

\item If there are some two functions
among $\{P, Q, R\}$  that do not have compatible type,
without loss of generality,
 we assume that $P$ and $Q$ do not have compatible type.
 There exist two partitions $\mathcal I=\{I_1, I_2, \ldots, I_k\}$ and
 $\mathcal J=\{J_1, J_2, \ldots, J_{\ell}\}$ of $[n]\setminus \{1\}$
 and signatures $P_1, P_2, \ldots, P_k$ and $Q_1, Q_2, \ldots, Q_{\ell}$ such that
  \[P(X)=\displaystyle\prod_{i=1}^k P_i(X|_{I_i})
 ~~~~\mbox{and}~~~~Q(X)=\displaystyle\prod_{j=1}^{\ell} Q_j(X|_{J_j})\]
 are the primitive decompositions of $P$ and $Q$ respectively.
\begin{itemize}
 \item If the partitions $\mathcal{I}$ and $\mathcal{J}$ are not compatible,
 then there exist $I_i$ and $J_j$ such that
\[
 |I_i|\geq 2, ~~~~|J_j|\geq 2,~~~~ I_i\cap J_j\neq \emptyset~~~~
\mbox{and}~~~~ I_i\neq J_j.\]
 Without loss of generality, we assume that there exist
$s, t \in [n]\setminus\{1\}$
 such that $\{s, t\}\subseteq I_i$ and $s\in J_j$ but $t\notin J_j$.
 Since $n\geq 4$, we have at least one other variable $x_r$,
where $r \not =1, s,t$.

 Suppose $r\in I_p$ and $r\in J_q$ for some $p\in[k]$ and $q\in[\ell]$.
 We will connect $x_r$ to some unary signature
from $\{[1,a], [1,b], [1,c]\}$  in the following way.
There are at least two unary signatures $u_1, u_2 \in \{[1,a], [1,b], [1,c]\}$
that are not $[1, 0]$, i.e., with both entries nonzero.
 \begin{itemize}
 \item If $|I_p|\geq 2$ and $|J_q|\geq 2$, then we connect $u_1$ to
 $x_r$.
Note that $P_i$ (or $\partial^{\{r\}}_{u_1}(P_i)$ if $i=p$)
still has at least two variables $x_s$ and $x_t$,
and two antipodal points in its support,
regardless of whether $p=i$.
The function $\partial^{\{r\}}_{u_1}(Q_q)$
still has two antipodal points in its support set.
Thus, regardless of whether  $q=j$,
the functions $\partial^{\{r\}}_{u_1}(P)$ and $\partial^{\{r\}}_{u_1}(Q)$
do not have compatible type.
 \item Suppose $|I_p|=1$ and $|J_q|\geq 2$.
Then $P_p$ is a nonzero unary signature $[P_p(0), P_p(1)]$.
For $1 \le k \le 2$, $\partial^{\{r\}}_{u_k} (P_p)$ are constants, and at most one
of them can be zero.
We choose one $u_k$ such that $\partial^{\{r\}}_{u_k} (P_p) \not = 0$,
and connect that $u_k$ to $x_r$.
In the new partition $\mathcal{I'}$ obtained from $\mathcal{I}$
by removing $I_p$,  $I_i$ is unchanged,   still
containing both $s$ and $t$.  $P_i$
still has at least two variables $x_s$ and $x_t$,
and two antipodal points in its support.
The function $\partial^{\{r\}}_{u_k}(Q_q)$
still has two antipodal points in its support set,
because both entries of $u_k$ are nonzero.
Thus, regardless of whether $q=j$,
the functions $\partial^{\{r\}}_{u_k}(P)$ and $\partial^{\{r\}}_{u_k}(Q)$
do not have compatible type.
 \item Suppose $|I_p|\geq 2$ and $|J_q|=1$.
Then $Q_q$ is a nonzero unary signature $[Q_q(0), Q_q(1)]$.
For $1 \le k \le 2$, $\partial^{\{r\}}_{u_k} (Q_q)$ are constants, and at most one
of them can be zero.
We choose one $u_k$ such that $\partial^{\{r\}}_{u_k} (Q_q) \not = 0$,
and connect that $u_k$ to $x_r$.
In the new partition $\mathcal{I'}$, obtained by removing $r$ from $I_p$,
$I_i$ still contains both $s$ and $t$.  $P_i$
still has at least two variables $x_s$ and $x_t$,
and two antipodal points in its support, regardless of whether $i=p$.
The function $Q_j$ is unchanged and
still has two antipodal points in its support set.
Thus the functions $\partial^{\{r\}}_{u_k}(P)$ and $\partial^{\{r\}}_{u_k}(Q)$
do not have compatible type.
 \item For $|I_p|=1$ and $|J_q|=1$, we have $I_p = J_q = \{r\}$. Note that
there exists at least one unary signature $u \in
\{[1,a], [1,b], [1,c]\}$ such that
 $\partial^{\{r\}}_{u}  (P_p) \neq 0$
 and $\partial^{\{r\}}_{u}  (Q_q) \neq 0$.
 Then we connect this $u$ to $x_r$.
Again,
the functions $\partial^{\{r\}}_{u}(P)$ and $\partial^{\{r\}}_{u}(Q)$
do not have compatible type, as both $P_i$ and $Q_j$ are unchanged.
 \end{itemize}
 Thus after connecting $x_r$ to a suitable
unary signature $u$ in this way, we get $P'= \partial^{\{r\}}_{u}  (P)
\in \mathscr{P}$
and $Q' = \partial^{\{r\}}_{u}  (Q) \in \mathscr{P}$, both not
identically zero, and not having compatible type.
 If we connect $x_r$ in $f$ to the unary signature $u$,
 we get $f' = \partial^{\{r\}}_{u}  (f)$.
Note that $\partial^{\{1\}}_{[1,a]}(f')=P'$ and $\partial^{\{1\}}_{[1,b]}(f')=Q'$.
\begin{equation*}
\begin{tikzpicture}[every node/.style={midway}]
  \matrix[column sep={8em,between origins}, row sep={2em}] at (0,0) {
    \node(f) {$f$}  ; & \node(f') {$f'$}; \\
    \node(P) {$P$}; & \node (P') {$P'$};\\
  };
  \draw[<-] (P) -- (f) node[anchor=east]  {$\partial^{\{1\}}_{[1,a]}$};
  \draw[->] (f) -- (f') node[anchor=south] {$\partial_u^{\{r\}}$};
  \draw[->] (f') -- (P') node[anchor=west] {$\partial^{\{1\}}_{[1,a]}$};
  \draw[->] (P) -- (P') node[anchor=north] {$\partial_u^{\{r\}}$};
\end{tikzpicture}
\hspace{40pt}
\begin{tikzpicture}[every node/.style={midway}]
  \matrix[column sep={8em,between origins}, row sep={2em}] at (0,0) {
    \node(f) {$f$}  ; & \node(f') {$f'$}; \\
    \node(Q) {$Q$}; & \node (Q') {$Q'$};\\
  };
  \draw[<-] (Q) -- (f) node[anchor=east]  {$\partial^{\{1\}}_{[1,b]}$};
  \draw[->] (f) -- (f') node[anchor=south] {$\partial_u^{\{r\}}$};
  \draw[->] (f') -- (Q') node[anchor=west] {$\partial^{\{1\}}_{[1,b]}$};
  \draw[->] (Q) -- (Q') node[anchor=north] {$\partial_u^{\{r\}}$};
\end{tikzpicture}
\end{equation*}

By Lemma~\ref{arity-3-not-product}, $f'\notin\mathscr{P}$.
Thus we are done by induction.

 \item If $\mathcal{I}$ and $\mathcal{J}$
are compatible yet $P$ and $Q$ still do not have compatible type,
then without loss of generality,
one of the following holds:
\begin{enumerate}
\item\label{item1-in-thm3.1}
There exist $I_i \in \mathcal{I}$, $J_j \in \mathcal{J}$,
such that $I_i=J_j$ and $|I_i|\geq 2$  but
$\mbox{supp}(P_i) \not = \mbox{supp}(Q_j)$, or
\item\label{item2-in-thm3.1}
There exist $I_i \in \mathcal{I}$ with $|I_i|\geq 2$
and $J_{j_1}, J_{j_2}, \ldots, J_{j_{|I_i|}} \in \mathcal{J}$ with
 $|J_{j_k}|=1$ for $1\leq k\leq |I_i|$,
  such that $I_i=\displaystyle\bigcup_{k=1}^{|I_i|}J_{j_k}$ but
$\mbox{supp}~(\displaystyle\prod_{k=1}^{|I_i|}Q_{j_k})$
is not a singleton subset of $\mbox{supp}(P_i)$.
\end{enumerate}

In case \ref{item1-in-thm3.1},
supp$(P_i) = \{\alpha, \overline{\alpha}\}$,
and
supp$(Q_j) = \{\beta, \overline{\beta}\}$,
for some $\alpha, \beta \in \{0, 1\}^{|I_i|}$, because both are factors
in a primitive decomposition and $|I_i| = |J_j| \geq 2$.
Being both antipodal pairs, and
$\alpha \not = \beta$ and $\alpha \not = \overline{\beta}$,
 it follows that
$\mbox{supp}(P_i) \cap \mbox{supp}(Q_j) = \emptyset$.
In case \ref{item2-in-thm3.1},
supp$(P_i) = \{\alpha, \overline{\alpha}\}$,
and we have $\mbox{supp}~(\displaystyle\prod_{k=1}^{|I_i|}Q_{j_k})
\not \subseteq \mbox{supp}(P_i)$.
To see that, if any $Q_{j_k} \not = \lambda[1,0]$ or $\lambda[0,1]$
($\lambda \in \mathbb{C}$), then
$\mbox{supp}~(\displaystyle\prod_{k=1}^{|I_i|}Q_{j_k})$ is clearly
not a subset of any set of the form $\{\alpha, \overline{\alpha}\}$.
If all $Q_{j_k}$ are of this form, then
$\mbox{supp}~(\displaystyle\prod_{k=1}^{|I_i|}Q_{j_k})$ is a singleton
set, but not a subset of $\mbox{supp}(P_i)$.
Hence,
there exists some $\beta \in \{0, 1\}^{|I_i|}$
such that
$\beta \in  \mbox{supp}~(\displaystyle\prod_{k=1}^{|I_i|}Q_{j_k})
\setminus \mbox{supp}(P_i)$.
Thus
we have $\alpha \not = \beta$ and $\alpha \not = \overline{\beta}$
as well.
This is equivalent to the existence of some $s, t \in I_i$, $s \not = t$,
such that $\alpha_s = \beta_s$ and $\alpha_t = \overline{\beta_t}$.

Aside from $x_1$,  $x_s$ and $x_t$, there exists another variable $x_r$,
since $n \ge 4$.

Suppose $r \in I_i$.
In case \ref{item1-in-thm3.1},
 we take any $u \in \{[1,a], [1,b], [1,c]\}$
with two nonzero entries, and connect it to $x_r$.
We get
$\partial^{\{r\}}_u(P_i)$ and $\partial^{\{r\}}_u(Q_j)$ with
support $\{\alpha', \overline{\alpha'}\}$
and $\{\beta', \overline{\beta'}\}$,
where $\alpha'$ and $\overline{\alpha'}$ are obtained
from $\alpha$ and $\overline{\alpha}$ by removing the $r$-th bit,
and similarly for $\beta'$ and $\overline{\beta'}$.
Since $r \not = s,t$, we still have $\alpha'_s = \beta'_s$ and $\alpha'_t
= \overline{\beta'_t}$, and
thus $\alpha' \not = \beta'$ and $\alpha' \not = \overline{\beta'}$.
Also $|I_i \setminus \{r\}| \ge 2$.
Hence
$\partial^{\{r\}}_u(P)$ and $\partial^{\{r\}}_u(Q)$ do not have
compatible type.
The proof for case \ref{item2-in-thm3.1} is similar;
we pick $u \in \{[1,a], [1,b], [1,c]\}$
with two nonzero entries as well as satisfying
$\partial^{\{r\}}_u(Q_{j_k}) \not =0$ for that (nonzero) unary signature
$Q_{j_k}$, where $J_{j_k} = \{r\}$.

If $r \not \in I_i$, then for some $i'$ and $j'$ such that
$r \in I_{i'}$ and $r \in J_{j'}$.
If $|I_{i'}|=1$, then $P_{i'}$ is a nonzero unary
function, and at most one $u \in \{[1,a], [1,b], [1,c]\}$
satisfies $\partial^{\{r\}}_{u}(P_{i'}) = 0$; if so
we exclude this $u$.  If $|I_{i'}|\ge 2$,
then there are two antipodal support points in supp$(\partial^{\{r\}}_{u}(P_{i'}))$
for any $u \in \{[1,a], [1,b], [1,c]\}$ with two nonzero entries,
which again excludes at most one unary, namely $[1,0]$.
Thus in either case we exclude at most one  $u \in \{[1,a], [1,b], [1,c]\}$
on account of $I_{i'}$.  Similarly we exclude at most one  $u$
on account of $J_{j'}$. Pick one $u \in \{[1,a], [1,b], [1,c]\}$
not excluded, and form $P' =\partial^{\{r\}}_u(P)$ and $Q'
=\partial^{\{r\}}_u(Q)$.
These do not have compatible type.

By connecting $u$ to $x_r$ in $f$, we get $f' = \partial^{\{r\}}_u(f)$
with arity $n-1$, and
\[\partial^{\{1\}}_{[1,a]}(f')= \partial^{\{1\}}_{[1,a]}(\partial^{\{r\}}_u(f))
=\partial^{\{r\}}_u(\partial^{\{1\}}_{[1,a]}(f))
=\partial^{\{r\}}_u(P) = P'\]
and similarly
\[
\partial^{\{1\}}_{[1,b]}(f')= \partial^{\{1\}}_{[1,b]}(\partial^{\{r\}}_u(f))
=\partial^{\{r\}}_u(\partial^{\{1\}}_{[1,b]}(f))
= \partial^{\{r\}}_u(Q) = Q'.
\]
\begin{equation*}
\begin{tikzpicture}[every node/.style={midway}]
  \matrix[column sep={8em,between origins}, row sep={2em}] at (0,0) {
    \node(f) {$f$}  ; & \node(f') {$f'$}; \\
    \node(P) {$P$}; & \node (P') {$P'$};\\
  };
  \draw[<-] (P) -- (f) node[anchor=east]  {$\partial^{\{1\}}_{[1,a]}$};
  \draw[->] (f) -- (f') node[anchor=south] {$\partial_u^{\{r\}}$};
  \draw[->] (f') -- (P') node[anchor=west] {$\partial^{\{1\}}_{[1,a]}$};
  \draw[->] (P) -- (P') node[anchor=north] {$\partial_u^{\{r\}}$};
\end{tikzpicture}
\hspace{40pt}
\begin{tikzpicture}[every node/.style={midway}]
  \matrix[column sep={8em,between origins}, row sep={2em}] at (0,0) {
    \node(f) {$f$}  ; & \node(f') {$f'$}; \\
    \node(Q) {$Q$}; & \node (Q') {$Q'$};\\
  };
  \draw[<-] (Q) -- (f) node[anchor=east]  {$\partial^{\{1\}}_{[1,b]}$};
  \draw[->] (f) -- (f') node[anchor=south] {$\partial_u^{\{r\}}$};
  \draw[->] (f') -- (Q') node[anchor=west] {$\partial^{\{1\}}_{[1,b]}$};
  \draw[->] (Q) -- (Q') node[anchor=north] {$\partial_u^{\{r\}}$};
\end{tikzpicture}
\end{equation*}
This implies that $f'\notin\mathscr{P}$ by Lemma~\ref{arity-3-not-product}.
Thus we are done by induction.

 \end{itemize}

\end{itemize}
\end{proof}

\subsection{Arity Reduction for Non-Affine Signatures}
\begin{lemma}\label{combine-affine-support-[1,x]}
Let $f$ be a signature of arity $n$ with affine support of dimension $k<n$,
and let $S = \{{i_1}, {i_2}, \ldots,  {i_k}\}$ be the
indices of a set of $k$ free variables.
Let $f'=\partial_{[1, a]}^{[n]\setminus S}(f)$,
 where $a^4=1$, then $f \in \mathscr{A}$ iff $f'\in \mathscr{A}$.
\end{lemma}
\begin{proof}
Define
\[\widetilde{f}(x_1, x_2, \ldots, x_n)=
a^{\sum_{i \in [n]\setminus S} x_i}
f(x_1, x_2, \ldots, x_n),\]
then $f \in \mathscr{A}$  iff $\widetilde{f} \in \mathscr{A}$,
as the modifier $a^{\sum_{i \in [n]\setminus S} x_i}$
is a power of $\frak{i}$ raised to a  linear sum, and the inverse transformation
is of the same kind.

Note that
 $f'=
\partial_{[1, 1]}^{[n]\setminus S}(\widetilde{f})$.
It follows that
$\widetilde{f} \in \mathscr{A}$
iff $f' \in \mathscr{A}$, by Corollary~\ref{f-affine-iff-f*-affine}
 since $f'$ is just the compressed signature
of $\widetilde{f}$ for $X$.
\end{proof}

\begin{lemma}\label{argue-a1-0-b1-1}
Assume that $S\subseteq \mathbb{Z}_2^n$ is not a linear subspace but for all $i\in[n]$, $S^{x_i=0}$ is a linear subspace of
$\mathbb{Z}_2^{n-1}$.
Then there exist ${\bf a}=a_1a_2\cdots a_n$, ${\bf b}=b_1b_2\cdots b_n$, ${\bf c}={\bf a}\oplus {\bf b}=c_1c_2\cdots c_n$
such that ${\bf a, b}\in S$ and ${\bf c}\notin S$,
and there exists $i\in[n]$ such that $a_i \not = b_i$.
\end{lemma}
\begin{proof}
Since $S\subseteq \mathbb{Z}_2^n$ is not a linear subspace, there exist
${\bf a}=a_1a_2\cdots a_n$, ${\bf b}=b_1b_2\cdots b_n$, ${\bf c}={\bf a}\oplus {\bf b}=c_1c_2\cdots c_n$
such that ${\bf a, b}\in S$ and ${\bf c}\notin S$.

If there exists $i$ such that $a_i=b_i=0$, then $c_i=0$.
The following table shows that $S^{x_i=0}$ is not a linear subspace.
 This is a contradiction.
 \[
{\begin{array}{llllllllllllllll}
        & a_1  \cdots  \hat{a}_i  \cdots  a_n & = & a_1 \cdots  \hat{0}  \cdots  a_n &  \in S^{x_i=0} \\
\oplus  & b_1  \cdots  \hat{b}_i  \cdots  b_n & = & b_1 \cdots  \hat{0}  \cdots  b_n &  \in S^{x_i=0}\\
\hline
        & c_1  \cdots  \hat{c}_i  \cdots  c_n & = & c_1 \cdots  \hat{0}  \cdots  c_n &  \notin S^{x_i=0}\\
\end{array} }
\]
where $\hat{~}$ means deleting this bit from the bit string.
So for any $i\in[n]$, there is at least one of $a_i, b_i$ is 1.

Moreover,
 if $a_1=a_2=\cdots=a_n=b_1=b_2=\cdots=b_n=1$,
then $c_1=c_2=\cdots=c_n=0$.
It implies that $(0, \ldots, 0)\notin S^{x_1=0}$.
This contradicts that $S^{x_1=0}$ is a linear subspace of
$\mathbb{Z}_2^{n-1}$.
Thus there exists $i\in[n]$ such that $a_i=0$ or $b_i=0$.
On the other hand, from the previous claim, there is at least one of $a_i$ or $b_i$ that is $1$.
So $a_i \not = b_i$.
\end{proof}
For the proof of the following lemma, we will use our Tableau Calculus.
It will be used several times in this paper with some variations.

\begin{lemma}\label{arity-reduction-affine}
Fix any $x\neq 0$. If $\widehat{\mathcal{F}}$ contains a signature
$f\notin\mathscr{A}$,
then there exists  a unary signature $u \notin\mathscr{A}$,
such that
\[\operatorname{Pl-Holant}(\widehat{\mathcal{EQ}}, [0, 1], [1, x], \widehat{\mathcal{F}}, u)
\le_{\rm T}
\operatorname{Pl-Holant}(\widehat{\mathcal{EQ}}, [0, 1], [1, x], \widehat{\mathcal{F}}).\]
\end{lemma}
\begin{proof}
Let $f$ have arity $n \ge 1$. If $n=1$ we can  choose $u=f$.
If $x^4\neq  1$,
we can  choose $u=[1, x] \notin \mathscr{A}$, as $x \neq 0$ by assumption.
So  we may assume  $n\geq 2$ and $x^4 =1$.
 We prove the Lemma by constructing some signature $g\notin\mathscr{A}$
of arity less than $n$, such that
\[\operatorname{Pl-Holant}(\widehat{\mathcal{EQ}}, [0, 1], [1, x], \widehat{\mathcal{F}}, g)
\le_{\rm T}
\operatorname{Pl-Holant}(\widehat{\mathcal{EQ}}, [0, 1], [1, x], \widehat{\mathcal{F}}).\]
 By Lemma~\ref{[0,1]-EQ-hat-wight-0-neq-0}, we can assume that $f_{00\ldots 0}=1$.
 We have $[0,1]$ explicitly given, as well as
$[1,0] \in \widehat{\mathcal{EQ}}$.
 If there exists $i\in[n]$ such that $f^{x_i=0}
\notin\mathscr{A}$ or $f^{x_i=1} \notin\mathscr{A}$,
then we can choose $g$ to be one of
these, which has arity $n-1$.
 So we may assume that
  both $f^{x_i=0} \in \mathscr{A}$ and $f^{x_i=1}  \in \mathscr{A}$
 for all $i\in[n]$.

We first prove that if supp$(f)$ is not an affine subspace,
then we can construct some signature $g \notin\mathscr{A}$
of arity less than $n$ in $\operatorname{Pl-Holant}(\widehat{\mathcal{EQ}}, [0, 1], [1, x], \widehat{\mathcal{F}})$.
 Suppose supp$(f)$ is not affine.
Note that a subset of $\mathbb{Z}_2^n$ containing $(0, 0, \ldots, 0)$
is affine iff it is a linear subspace.
 Since $(0, 0, \ldots, 0)\in {\rm supp}(f)$,
By Lemma~\ref{argue-a1-0-b1-1},
  there exist ${\bf a} =a_1a_2\ldots a_n,
{\bf b} =b_1b_2\ldots b_n \in {\rm supp}(f)$,
such that ${\bf c} =  {\bf a} \oplus {\bf b}
=  c_1 c_2 \ldots c_n  \notin {\rm supp}(f)$,
and there exists  $i\in[n]$ such that $a_i \not = b_i$.
Without loss of generality, we assume that $i=1$, $a_1=0$,
$b_1=1$, and $c_1=1$.
   We denote  ${\bf a}'=a_2\ldots a_n, {\bf b}'=b_2\ldots b_n$
and $ {\bf c}'=c_2\ldots c_n$.

 \[
{\begin{array}{lllllll}
 &{\bf a} & = & a_1  {\bf a}' & =  & 0~  a_2\ldots a_n & \in {\rm supp}(f)  \\
 \oplus &{\bf b} & = & b_1 {\bf b}' & = & 1~ b_2\ldots b_n
 & \in {\rm supp}(f)  \\
\hline
 & {\bf c}  & = & c_1   {\bf c}'  & =  & 1~ c_2\ldots c_n  & \not \in {\rm supp}(f) \\
\end{array} }
\]

By connecting the unary signature $[1, x]$ to the first variable of $f$ we get
 $h=\partial^{\{1\}}_{[1, x]}(f)$, which has arity $n-1$.
 If $h \notin\mathscr{A}$, then we are done.
 Therefore we may assume $h \in \mathscr{A}$.
 Note that
 $h_{\alpha}=f_{0\alpha}+xf_{1\alpha}$ for all $\alpha\in\{0, 1\}^{n-1}$.
The next Claim will be used several times in the following proof.
\vspace{.1in}

\noindent
{\bf Claim}:
If there exists $\alpha\in\{0, 1\}^{n-1}$
 such that
 $h_{\alpha}=0$ and $f_{0\alpha}\neq 0$, then we can construct $[1, -x]$.

\vspace{.1in}
To prove this Claim,
 we can first obtain the unary
signature $[f_{0\alpha}, f_{1\alpha}]$ from $f$ by pinning on all
variables $x_2, \ldots, x_n$ according to $\alpha$,
using  $[1, 0] \in \widehat{\mathcal{EQ}}$ and the explicitly given
$[0, 1]$.
We have $x \in \{ \pm 1, \pm {\frak i}\}$.
 If $x=\pm 1$, then $f_{1\alpha} = -x f_{0\alpha}$ from $h_{\alpha}=0$,
and so $[f_{0\alpha}, f_{1\alpha}]=f_{0\alpha}[1, -x]$.
 Thus we have $[1, -x]$ up to the nonzero scalar $f_{0\alpha}$.
 If $x=\pm {\frak i}$, then from $[1, 0, 1, 0] \in \widehat{\mathcal{EQ}}$
 we have $\partial_{[0, 1]}([1, 0, 1, 0])=[0, 1 ,0]$ and $\partial_{[1, x]}([0, 1, 0])=[x, 1]=x[1, -x]$.

Once we have $[1, -x]$, we can construct
another signature $\widetilde{h}=\partial_{[1, -x]}(f)$
in addition to $h$.
The analysis below will use both $h$ and $\widetilde{h}$.
\begin{eqnarray*}
h_{\alpha} &=&  f_{0 \alpha}+ xf_{1\alpha}\\
\widetilde{h}_{\alpha} &=& f_{0 \alpha}-xf_{1\alpha}
\end{eqnarray*}
In the following we
consider various cases according to the membership of
$\bar{a}_1 \bf{a}'$ and $\bar{b}_1\bf{b}'$
in ${\rm supp}(f)$.

 \begin{itemize}
 \item Suppose $\bar{a}_1{\bf a}'\in {\rm supp}(f)$
and $ \bar{b}_1{\bf b}'\in {\rm supp}(f)$.

Note that supp$(f^{x_1=0})$ is a
linear subspace since $f^{x_1=0}$ is affine and $f_{00\ldots 0}\neq 0$.
By ${\bf a} = a_1 {\bf a}' = 0 {\bf a}'
\in {\rm supp}(f)$, we have ${\bf a}' \in {\rm supp}(f^{x_1=0})$.
 By $\bar{b}_1{\bf b}' = 0 {\bf b}'\in {\rm supp}(f)$,
we have ${\bf b}' \in {\rm supp}(f^{x_1=0})$.
By definition, ${\bf a}' \oplus {\bf b}' = {\bf c}'$
and thus ${\bf c}'\in {\rm supp}(f^{x_1=0})$.
 This implies that $f_{\bar{c}_1{\bf c}'}\neq 0$.

 Since all of $f_{00\ldots 0}, f_{a_1{\bf a}'}, f_{\bar{b}_1{\bf b}'}, f_{\bar{c}_1{\bf c}'}$ are nonzero entries of $f^{x_1=0} \in \mathscr{A}$,
 they are all powers of ${\frak i}$ as $f_{00\ldots 0}=1$.
We have $h_{{\bf c}'}=xf_{c_1{\bf c}'}+f_{\bar{c}_1{\bf c}'}
= f_{\bar{c}_1{\bf c}'}$ since $c_1 =1$ and $f_{c_1{\bf c}'}=0$.
 Hence $|h_{{\bf c}'}|=1$ since it is a power of ${\frak i}$.
Moreover,  since both $f_{a_1{\bf a}'}$ and $f_{\bar{a}_1{\bf a}'}$ are nonzero entries of $f^{x_2=a_2} \in \mathscr{A}$ and $f_{a_1{\bf a}'}$ is a
power of ${\frak i}$, so is
 $f_{\bar{a}_1{\bf a}'}$.
Any nonzero sum of two quantities that are both powers of ${\frak i}$
must have norm either $2$ or $\sqrt{2}$.
 This implies that if $h_{{\bf a}'}=f_{a_1{\bf a}'}+xf_{\bar{a}_1{\bf a}'}$
is nonzero, then $|h_{{\bf a}'}|=2$ or $\sqrt{2}$.
 This implies that $|h_{{\bf a}'}|\neq |h_{{\bf c}'}|$ and both are nonzero. This
 contradicts that $h \in \mathscr{A}$, by Proposition~\ref{A-has-same-norm-etc}.

Therefore $h_{{\bf a}'}=f_{a_1{\bf a}'}+xf_{\bar{a}_1{\bf a}'}= 0$.
 Then we have $[1, -x]$ and  obtain $\widetilde{h}$ by the Claim.
 If $\widetilde{h} \not \in \mathscr{A}$, then we are done since the arity of $\widetilde{h}$ is $n-1$.
 Therefore we may assume $\widetilde{h} \in \mathscr{A}$.
 We have $|\widetilde{h}_{{\bf c}'}|=|f_{\bar{c}_1{\bf c}'}-xf_{c_1{\bf c}'}|=1$ since $c_1=1$,  $f_{c_1{\bf c}'}=0$ and $f_{\bar{c}_1{\bf c}'}$ is
a power of ${\frak i}$.

We already have $h_{{\bf a}'}=0$. If additionally
 $\widetilde{h}_{{\bf a}'}=f_{a_1{\bf a}'}-xf_{\bar{a}_1{\bf a}'}=0$, then we have
\begin{eqnarray*}
 & & f_{a_1{\bf a}'}+xf_{\bar{a}_1{\bf a}'}=0\\
 & & f_{a_1{\bf a}'}-xf_{\bar{a}_1{\bf a}'}=0
\end{eqnarray*}
 This implies that $f_{a_1{\bf a}'}=0$ and it is a  contradiction to
${\bf a} = a_1{\bf a}' \in {\rm supp}(f)$.
Therefore
 $\widetilde{h}_{{\bf a}'}=f_{a_1{\bf a}'}-xf_{\bar{a}_1{\bf a}'}\neq 0$.
Since both $f_{a_1{\bf a}'}$ and $xf_{\bar{a}_1{\bf a}'}$ are powers
of ${\frak i}$, the norm $|\widetilde{h}_{{\bf a}'}|$  is either
2 or $\sqrt{2}$.
 This implies that
 $|\widetilde{h}_{{\bf a}'}|\neq |\widetilde{h}_{{\bf c}'}|$, and
both are nonzero. This
 contradicts that $\widetilde{h} \in \mathscr{A}$, by Proposition~\ref{A-has-same-norm-etc}.

 \item Suppose $\bar{a}_1{\bf a}'\notin {\rm supp}(f)$
and $ \bar{b}_1{\bf b}'\notin {\rm supp}(f)$.

 We have $h_{{\bf a}'}=f_{a_1{\bf a}'}+ x f_{\bar{a}_1{\bf a}'}\neq 0$
and $h_{{\bf b}'}=f_{\bar{b}_1{\bf b}'} + x f_{b_1{\bf b}'}\neq 0$,
 by $f_{\bar{a}_1{\bf a}'}=f_{\bar{b}_1{\bf b}'}=0$ and
 $f_{a_1{\bf a}'}\neq 0$, $f_{b_1{\bf b}'}\neq 0$, and also $x \not =0$.
We show next that $h_{0\ldots 0} =0$. Suppose
for a contradiction that $h_{0\ldots 0}\neq 0$.
Since $h \in \mathscr{A}$, and $0\ldots 0 \in {\rm supp}(h)$,
supp$(h)$ is a  linear subspace.
As ${\bf a}', {\bf b}' \in {\rm supp}(h)$ we have ${\bf c}'\in{\rm supp}(h)$.
 \[
{\begin{array}{*{3}c}
 ~ &   {\bf a}' ~~~\in {\rm supp}(h) \\
 \oplus &   {\bf b}' ~~~\in {\rm supp}(h)  \\
\hline
 ~ &   {\bf c}' ~~~\in {\rm supp}(h)  \\
\end{array} }
\]
 By $h_{{\bf c}'}=xf_{c_1{\bf c}'}+f_{\bar{c}_1{\bf c}'}\neq 0$, we have $f_{\bar{c}_1{\bf c}'}\neq 0$ since $f_{c_1{\bf c}'}=0$.
Thus ${\bf c}' \in {\rm supp}(f^{x_1=0})$, as $\bar{c}_1 =0$.
 Since ${\bf a}' \in {\rm supp}(f^{x_1=0})$,
and the support of $f^{x_1=0}$ is a
linear subspace, we have ${\bf b}'\in{\rm supp}(f^{x_1=0})$.
 \[
{\begin{array}{*{3}c}
 ~ & {\bf a}' ~~~\in {\rm supp}(f^{x_1=0}) \\
 \oplus  &    {\bf c}' ~~~\in {\rm supp}(f^{x_1=0})  \\
\hline
 ~ &    {\bf b}' ~~~\in {\rm supp}(f^{x_1=0})  \\
\end{array} }
\]
 This contradicts that $f_{\bar{b}_1{\bf b}'}=0$.

 Therefore $h_{0\ldots 0}=f_{00\ldots 0}+xf_{10\cdots 0}=0$. Then we
can obtain $[1, -x]$ and $\widetilde{h}$ by the Claim.
 If $\widetilde{h} \not \in \mathscr{A}$ then we are done.
 Therefore we may assume $\widetilde{h} \in \mathscr{A}$.
 Moreover, $\widetilde{h}_{0\ldots 0}=f_{00\ldots 0}-xf_{10\cdots 0}\neq 0$ by $f_{00\ldots 0}+xf_{10\ldots 0}=0$ and $f_{00\ldots 0}\neq 0$.
 Thus supp$(\widetilde{h})$ is a linear subspace.

Note that
$\widetilde{h}_{{\bf a}'}=f_{a_1{\bf a}'} - xf_{\bar{a}_1{\bf a}'}\neq 0$
and $\widetilde{h}_{{\bf b}'}=f_{\bar{b}_1{\bf b}'}
- x f_{b_1{\bf b}'} \neq 0$,
 by $f_{\bar{a}_1{\bf a}'}=f_{\bar{b}_1{\bf b}'}=0$ and $f_{a_1{\bf a}'}\neq 0$, $f_{b_1{\bf b}'}\neq 0$, $x \not =0$.
Thus ${\bf a}' \in{\rm supp}(\widetilde{h})$
and ${\bf b}' \in{\rm supp}(\widetilde{h})$. It follows that
 ${\bf c}' ={\bf a}'\oplus {\bf b}' \in{\rm supp}(\widetilde{h})$.
 This implies that $\widetilde{h}_{{\bf c}'}=f_{\bar{c}_1{\bf c}'}-
xf_{c_1{\bf c}'}\neq 0$.
 So $f_{\bar{c}_1{\bf c}'}\neq 0$ since $f_{c_1{\bf c}'}=0$.
We have ${\bf a}' \in{\rm supp}(f^{x_1=0})$ and ${\bf c}'
\in{\rm supp}(f^{x_1=0})$.
 Because the support of $f^{x_1=0}$ is a linear
subspace, it follows that ${\bf b}'\in{\rm supp}(f^{x_1=0})$.
 \[
{\begin{array}{*{3}c}
 ~ &  {\bf a}' ~~~\in {\rm supp}(f^{x_1=0})\\
 \oplus  &    {\bf c}'~~~\in {\rm supp}(f^{x_1=0}) \\
\hline
 ~ &   {\bf b}'~~~\in {\rm supp}(f^{x_1=0}) \\
\end{array} }
\]
 This contradicts that $f_{\bar{b}_1{\bf b}'}=0$.

 \item Suppose $\bar{a}_1{\bf a}'\in {\rm supp}(f)$
and $\bar{b}_1{\bf b}'\notin {\rm supp}(f)$.

 Note that $h_{{\bf b}'}=f_{\bar{b}_1{\bf b}'}+xf_{b_1{\bf b}'}
= xf_{b_1{\bf b}'}$,  since $f_{\bar{b}_1{\bf b}'} =0$.
We have $h_{{\bf b}'} \not =0$ since $b_1{\bf b}' \in {\rm supp}(f)$,
and $x \not =0$.
As $f_{00\ldots0} =1$, and $f_{a_1 {\bf a}'} \not =0$, by $f^{x_1=0}
\in \mathscr{A}$,
and both $0\ldots0$ and ${\bf a}' \in {\rm supp}(f^{x_1=0})$,
we have $f_{a_1 {\bf a}'}$ is a power of ${\frak i}$.
By hypothesis $f_{\bar{a}_1{\bf a}'} \not =0$. By pinning $x_2$
to the same value $a_2 = 0$ or 1 in both $f_{a_1 {\bf a}'} \not =0$
and $f_{\bar{a}_1{\bf a}'} \not =0$,  and by $f^{x_2 = a_2}
\in  \mathscr{A}$,
we conclude that the value  $f_{\bar{a}_1{\bf a}'}$ is a power of $i$.
As $\bar{a}_1 = b_1 =1$, by pinning $x_1$ to 1,
and $f^{x_1=1} \in \mathscr{A}$, we conclude that the nonzero value
$f_{b_1{\bf b}'}$ is a power of ${\frak i}$.
By $x^4 =1$,  $h_{{\bf b}'}= xf_{b_1{\bf b}'}$
is also a power of ${\frak i}$.

To recap,
we have $f_{a_1 {\bf a}'}$, $f_{\bar{a}_1{\bf a}'}$, $f_{b_1{\bf b}'}$
and $h_{{\bf b}'}$ are all powers of ${\frak i}$.  In particular,
$|h_{{\bf b}'}| = 1$.
Moreover $h_{{\bf a}'} = f_{a_1 {\bf a}'} + x f_{\bar{a}_1{\bf a}'}$
is a sum of two quantities that are both powers of ${\frak i}$.
  If $h_{{\bf a}'}\neq 0$, then its norm is 2 or $\sqrt{2}$.
 This is a contradiction to $h \in \mathscr{A}$,
 by Proposition~\ref{A-has-same-norm-etc}.

Thus $h_{{\bf a}'}= 0$. Then  we can construct $[1, -x]$ and $\widetilde{h}$ by
the Claim.
 If $\widetilde{h} \not \in \mathscr{A}$, then we are done.
Otherwise, we have
\begin{eqnarray*}
h_{{\bf a}'}  & = & f_{a_1{\bf a}'} +xf_{\bar{a}_1{\bf a}'} = 0\\
\widetilde{h}_{{\bf a}'}  & = & f_{a_1{\bf a}'} -xf_{\bar{a}_1{\bf a}'}
\end{eqnarray*}
If $\widetilde{h}_{{\bf a}'}  =0$, then we would have  $f_{a_1{\bf a}'} =0$,
 a contradiction. Hence
$\widetilde{h}_{{\bf a}'} \not =0$ and is the sum of two
quantities that are both powers of ${\frak i}$. Hence $|\widetilde{h}_{{\bf a}'}|$
is  2 or $\sqrt{2}$.
Yet,
$\widetilde{h}_{{\bf b}'}  =  f_{\bar{b}_1{\bf b}'}-xf_{{b}_1{\bf b}'}
= -xf_{{b}_1{\bf b}'}$ is a power of ${\frak i}$, as
$f_{\bar{b}_1{\bf b}'} =0$ by hypothesis.
Thus $|\widetilde{h}_{{\bf b}'}|=1$.
 This is a contradiction to $\widetilde{h} \in \mathscr{A}$, by Proposition~\ref{A-has-same-norm-etc}.

\item
Suppose $\bar{a}_1{\bf a}'\notin {\rm supp}(f)$
and $\bar{b}_1{\bf b}'\in {\rm supp}(f)$.

Consider $h_{{\bf a}'}=f_{a_1{\bf a}'}+xf_{\bar{a}_1{\bf a}'}$.
Since $f_{\bar{a}_1{\bf a}'} =0$
and ${\bf a} = a_1{\bf a}' \in {\rm supp}(f)$,
$h_{{\bf a}'}= f_{a_1{\bf a}'} \not =0$.
As $f_{00\ldots0} =1$, $a_1 =0$ and $f_{a_1{\bf a}'} \not =0$,
by $f^{x_1=0}
\in \mathscr{A}$,
$f_{a_1{\bf a}'}$ is a power of ${\frak i}$, and so is $h_{{\bf a}'}$.
In particular $|h_{{\bf a}'}| =1$.

Also by hypothesis, $f_{\bar{b}_1{\bf b}'} \not =0$.
As $\bar{b}_1 =0$ and $f_{00\ldots0} =1$, by pinning $x_1$ to $0$,
and $f^{x_1=0}
\in \mathscr{A}$, we have $f_{\bar{b}_1{\bf b}'}$ is a power of ${\frak i}$.
Then pinning $x_2$ to $b_2$ in $f_{\bar{b}_1{\bf b}'}$
and the nonzero value $f_{{b}_1{\bf b}'}$ we have
 $f_{{b}_1{\bf b}'}$ is also  a power of ${\frak i}$.

We have $h_{{\bf b}'}=f_{\bar{b}_1{\bf b}'} + x f_{b_1{\bf b}'}$,
which is a sum of two quantities both a power of ${\frak i}$.
If $h_{{\bf b}'} \not = 0$, it would have norm $2$ or $\sqrt{2}$.
As $|h_{{\bf a}'}| =1$, and $h \in \mathscr{A}$, this is a contradiction.
Hence $h_{{\bf b}'} =0$.

Then  we can construct $[1, -x]$ and $\widetilde{h}$ by
the Claim.
If $\widetilde{h} \not \in \mathscr{A}$, then we are done.
Otherwise, we have
\begin{eqnarray*}
h_{{\bf b}'}  & = & f_{\bar{b}_1{\bf b}'} + x f_{b_1{\bf b}'} =0\\
\widetilde{h}_{{\bf b}'}  & = &  f_{\bar{b}_1{\bf b}'}  -x f_{b_1{\bf b}'}
\end{eqnarray*}
If $\widetilde{h}_{{\bf b}'}  =0$, then we would have $f_{\bar{b}_1{\bf b}'}
=0$,  a contradiction. Hence
$|\widetilde{h}_{{\bf b}'}| = 2$ or $\sqrt{2}$.
Yet,
$\widetilde{h}_{{\bf a}'}  = f_{a_1{\bf a}'} -xf_{\bar{a}_1{\bf a}'}
= f_{a_1{\bf a}'}$ is a power of ${\frak i}$, hence of norm 1.
This is a contradiction to $\widetilde{h} \in \mathscr{A}$, by Proposition~\ref{A-has-same-norm-etc}.
 \end{itemize}
The above argument is what we call the Tableau Calculus.

\vspace{.1in}
Now we can assume that supp$(f)$ is an affine subspace,
indeed a linear subspace since $f_{00\ldots0} =1$.
Suppose it has dimension $k$.
If $k=0$, then $f \in \mathscr{A}$. This is a contradiction.
For $k\geq 1$, let $S=\{x_{i_1}, x_{i_2}, \ldots, x_{i_k}\}$ be
a set of free variables defining the linear subspace supp$(f)$.
Let $\check{f}$ be obtained from $f$ by connecting
$[1, x]$ to every variable outside $S$ (if there is any),
$\check{f}=\partial_{[1, x]}^{[n] \setminus S}(f)$.
Note that for every assignment to $S$,
the sum in the expression defining $\check{f}$ has exactly one nonzero
entry of $f$, multiplied by a suitable power of $x \not =0$.
Hence all  entries of $\check{f}$ are nonzero.
By Lemma~\ref{combine-affine-support-[1,x]}, $\check{f} \not \in \mathscr{A}$.
If there exists $j\in[k]$ such that $\check{f}^{x_{i_j}=0}$
or $\check{f}^{x_{i_j}=1}$
is not affine, then we get a signature not in $\mathscr{A}$
 with arity $k-1<n$. This completes the proof.
Therefore, we may assume
 both  $\check{f}^{x_{i_j}=0}$ and $\check{f}^{x_{i_j}=1}$ are affine for all
 $j\in[k]$.
In the following we will rename the variables of $\check{f}$
as $x_1,  \ldots, x_k$.

We claim that if we have $[1, {\frak i}^r]$, where $r\in\{0, 1, 2, 3\}$,
then we can construct $[1, {\frak i}^{-r}]$.
If $r=0$ or $2$, then ${\frak i}^r={\frak i}^{-r}$ and we are done.
If $r=1$ or $3$, we have $\partial_{[0, 1]}([1, 0, 1, 0])=[0, 1, 0]$,
where $[1, 0, 1, 0] \in  \widehat{\mathcal{EQ}}$
and $[0, 1]$ is given in the lemma.
Then we can obtain $\partial_{[1, {\frak i}^r]}([0, 1, 0])=[{\frak i}^r, 1]={\frak i}^r[1, {\frak i}^{-r}]$,
a nonzero multiple of $[1, {\frak i}^{-r}]$.

Now we finish the proof by constructing a  signature
not in $\mathscr{A}$ with arity less than $n$.
\begin{enumerate}
\item If $k=1$, then $\check{f}$ is the desired signature
not in $\mathscr{A}$.

\item If $k=2$,
then all four values
$\check{f}_{00}$, $\check{f}_{01}$, $\check{f}_{10}$ and $\check{f}_{11}$
 are powers of $i$.
This can be seen by noting that
$\check{f}_{00}=1$,  and the signatures
 $\check{f}^{x_1=0}=[\check{f}_{00}, \check{f}_{01}]$,
$\check{f}^{x_2=0}=[\check{f}_{00}, \check{f}_{10}]$ and
 $\check{f}^{x_1=1}=[\check{f}_{10}, \check{f}_{11}]$
all belong to $\mathscr{A}$.
So there exist $r, s, t\in\{0, 1, 2, 3\}$ such that
\[\check{f}=(\check{f}_{00}, \check{f}_{01}, \check{f}_{10}, \check{f}_{11})
=(1, {\frak i}^r, {\frak i}^s, {\frak i}^{t}).\]
If $r+s\equiv t\mod 2$, then $\check{f} \in \mathscr{A}$
 by Lemma~\ref{binary-affine-compressed function}.
This is a contradiction. Therefore, we have ${\frak i}^t=\pm {\frak i}^{r+s+1}$.
Note that we have $\check{f}^{x_1=0}=[1, {\frak i}^r]$. Thus we can
construct  $[1, {\frak i}^{-r}]$ by the claim and we can obtain
 $\partial_{[1, {\frak i}^{-r}]}^{\{2\}}(\check{f})=[2, {\frak i}^s(1\pm {\frak i})]$.
 However $|{\frak i}^s(1\pm {\frak i})|=\sqrt{2}$. Thus $[2, {\frak i}^s(1\pm {\frak i})]\notin\mathscr{A}$, by Proposition~\ref{A-has-same-norm-etc}.

\item If $k=3$, then  there exist $r, s, t\in\{0, 1, 2, 3\}$
and $\epsilon_j\in\{1, -1\}$ for $1 \le j \le 4$ such that
\[M_{x_1, x_2x_3}(\check{f})
=
\left[\begin{matrix}
f^{000} & f^{001} & f^{010} & f^{011}\\
f^{100} & f^{101} & f^{110} & f^{111}
\end{matrix}\right]
=
\left[\begin{matrix}
1 & {\frak i}^r & {\frak i}^s & \epsilon_1 {\frak i}^{r+s} \\
{\frak i}^t &  \epsilon_2{\frak i}^{r+t} & \epsilon_3{\frak i}^{s+t} & \epsilon_4 {\frak i}^{r+s+t}
\end{matrix}\right].
\]
This can be seen by observing that
all signatures
 $\check{f}^{x_k=0}$ for $k=1, 2, 3$ and $\check{f}^{x_1=1}$ are affine.

If $\epsilon_1\epsilon_2\epsilon_3\epsilon_4=1$, then $\check{f}
\in\mathscr{A}$ by Lemma~\ref{binary-affine-compressed function}.
This is a contradiction.
Therefore, $\epsilon_4=-\epsilon_1\epsilon_2\epsilon_3$.

Since we have $\partial_{[1, 0]}^{\{2, 3\}}=
[f^{000}, f^{100}] = [1, {\frak i}^t]$ and $\partial_{[1, 0]}^{\{1, 2\}}(\check{f})=
[f^{000}, f^{001}] = [1, {\frak i}^r]$, by the claim we also have
$[1, {\frak i}^{-t}]$ and $[1, {\frak i}^{-r}]$.

We have
\[\partial_{[1, {\frak i}^{-t}]}^{\{1\}}(\check{f})
=(2,~~ (1+\epsilon_2){\frak i}^r,~~ (1+\epsilon_3){\frak i}^s,~~ \epsilon_1 (1 - \epsilon_2
\epsilon_3){\frak i}^{r+s})\]

\begin{itemize}
\item If $\epsilon_2 = - \epsilon_3$ or $\epsilon_2=\epsilon_3=1$,
then $\partial_{[1, {\frak i}^{-t}]}^{\{1\}}(\check{f})$
is not affine since its support is not affine.
Thus we are done. So we may assume in the following $\epsilon_2 =
\epsilon_3=-1$, and
\[M_{x_1, x_2x_3}(\check{f})
=
\left[\begin{matrix}
f^{000} & f^{001} & f^{010} & f^{011}\\
f^{100} & f^{101} & f^{110} & f^{111}
\end{matrix}\right]
=
\left[\begin{matrix}
1 & {\frak i}^r & {\frak i}^s & \epsilon_1 {\frak i}^{r+s} \\
{\frak i}^t &  -{\frak i}^{r+t} & -{\frak i}^{s+t} & -\epsilon_1 {\frak i}^{r+s+t}
\end{matrix}\right].
\]

\item If $\epsilon_1=1$ and $\epsilon_2=\epsilon_3=-1$, then
\[\partial_{[1, {\frak i}^{-r}]}^{\{3\}}(\check{f})=
(2,~~ 2 {\frak i}^s,~~ 0,~~ -2{\frak i}^{r+s})\]
is not affine since its support is not affine.
Thus we are done.
So we may assume in the following $\epsilon_1= \epsilon_2 =
\epsilon_3=-1$.
\[M_{x_1, x_2x_3}(\check{f})
=
\left[\begin{matrix}
f^{000} & f^{001} & f^{010} & f^{011}\\
f^{100} & f^{101} & f^{110} & f^{111}
\end{matrix}\right]
=
\left[\begin{matrix}
1 & {\frak i}^r & {\frak i}^s & - {\frak i}^{r+s} \\
{\frak i}^t &  -{\frak i}^{r+t} & -{\frak i}^{s+t} & {\frak i}^{r+s+t}
\end{matrix}\right].
\]

\item For $\epsilon_1=\epsilon_2=\epsilon_3=-1$,
we take two copies of $\check{f}$ and connect
the variables $x_2$ and $x_3$ of one copy with the
 variables $x_3$ and $x_2$ of the other copy, creating
a planar binary gadget with a symmetric signature
$g(y, z)=\displaystyle\sum_{x_2, x_3\in\{0, 1\}}\check{f}(y, x_2, x_3)\check{f}(z, x_3, x_2)$. Notice the reversal of the order
of  $x_2$ and $x_3$ in the second copy of $\check{f}$; this
amounts to a
cyclic permutation of its inputs and is necessary
in order to make a planar connection in the gadget construction.
 The signature of $g$ can be computed as a matrix product
$\left[\begin{smallmatrix}
f^{000} & f^{001} & f^{010} & f^{011}\\
f^{100} & f^{101} & f^{110} & f^{111}
\end{smallmatrix}\right]
\left[\begin{smallmatrix}
f^{000} & f^{100} \\
f^{010} & f^{110}\\
f^{001} & f^{101} \\
f^{011} & f^{111}
\end{smallmatrix}\right]$.
In symmetric signature notation
\[g=[1+2{\frak i}^{r+s}+(-1)^{r+s},~~~ {\frak i}^t(1-2{\frak i}^{r+s}-(-1)^{r+s}),~~~ (-1)^t(1+2{\frak i}^{r+s}+(-1)^{r+s})].\]
Further, we have the unary signature $\partial_{[1, 0]}(g)
=[1+2{\frak i}^{r+s}+(-1)^{r+s},~~ {\frak i}^t(1-2{\frak i}^{r+s}-(-1)^{r+s})]$.

If $r+s$ is odd, then the norm $|1+2{\frak i}^{r+s}+(-1)^{r+s}| =2$ and
the norm $|{\frak i}^t(1-2{\frak i}^{r+s}-(-1)^{r+s})| = |2 \pm 2{\frak i}| = 2 \sqrt{2}$,
and hence $\partial_{[1, 0]}(g) \not \in \mathscr{A}$, by Proposition~\ref{A-has-same-norm-etc}, and we are done.
For even $r+s = 2k$, the norms of the entries of $\partial_{[1, 0]}(g)$
are $2 + 2(-1)^k$ and $2$ respectively. Hence if $k$ is even then
 $\partial_{[1, 0]}(g) \not \in \mathscr{A}$, by Proposition~\ref{A-has-same-norm-etc}, and we are done.
Hence we may assume $k$ is odd, and
\begin{equation}\label{rs-eqn}
r+s  \equiv 2 \bmod 4.
\end{equation}

By symmetry of argument,
 we have
\begin{equation}\label{st-eqn}
s+t\equiv 2\bmod 4
\end{equation}
and
\begin{equation}\label{tr-eqn}
t +r\equiv 2\bmod 4
\end{equation}
From (\ref{rs-eqn}) and (\ref{st-eqn}) we get $r \equiv t \bmod 4$,
and by symmetry,
\begin{equation}\label{rst-equal-eqn}
r \equiv  s \equiv  t \bmod 4.
\end{equation}
Also by (\ref{rst-equal-eqn}) and (\ref{rs-eqn}) we have
\begin{equation}\label{rst-equal-1-or-3-eqn}
r \equiv  s \equiv  t \equiv 1 \bmod 4
~~~~\mbox{or}~~~~
r \equiv  s \equiv  t \equiv 3 \bmod 4
\end{equation}

If $r \equiv  s \equiv  t \equiv 1 \bmod 4$, then
\[M_{x_1, x_2x_3}(\check{f})
=
\left[\begin{matrix}
f^{000} & f^{001} & f^{010} & f^{011}\\
f^{100} & f^{101} & f^{110} & f^{111}
\end{matrix}\right]
=
\left[\begin{matrix}
1 & {\frak i} & {\frak i} & 1  \\
{\frak i} & 1 & 1 & -{\frak i}
\end{matrix}\right].
\]
This is the symmetric ternary signature $[1, {\frak i}, 1, -{\frak i}] \not \in
 \mathscr{A}$.
Having $[1,0] \in \widehat{\mathcal{EQ}}$
we can get $[1,{\frak i}]$ from $[1, {\frak i}, 1, -{\frak i}]$.
Then $\partial_{[1,{\frak i}]}([1, {\frak i}, 1, -{\frak i}]) = [0, 2{\frak i}, 2]$.
Once again $\partial_{[1,{\frak i}]}([0, 2{\frak i}, 2]) = [-2, 4{\frak i}]  \not \in
 \mathscr{A}$.

Similarly if $r \equiv  s \equiv  t \equiv 3 \bmod 4$, then
\[M_{x_1, x_2x_3}(\check{f})
=
\left[\begin{matrix}
f^{000} & f^{001} & f^{010} & f^{011}\\
f^{100} & f^{101} & f^{110} & f^{111}
\end{matrix}\right]
=
\left[\begin{matrix}
1 & -{\frak i} & -{\frak i} & 1  \\
-{\frak i} & 1 & 1 & {\frak i}
\end{matrix}\right].
\]
This is  the symmetric ternary signature $[1, -{\frak i}, 1, {\frak i}] \not \in
 \mathscr{A}$, and we can get $[-2, -4{\frak i}]  \not \in
 \mathscr{A}$ in a similar way.
\end{itemize}

\item If $k\geq 4$, then $\check{f}$ is affine by Lemma~\ref{[1,0]-[0,1]-pinning-implies-affine-arity-4}.
This is a contradiction.
\end{enumerate}
This completes the proof of Lemma~\ref{arity-reduction-affine}.
\end{proof}

The next lemma says that generally we can construct a unary
signature $[1,a]$,
with $a \not = 0, 1$, in
 $\operatorname{Pl-\#CSP}([1, 0], \mathcal{F})$.
The condition on $\mathcal{F}$ is satisfied as long as
not every signature in  $\mathcal{F}$ is $\{0,1\}$-valued
up to a constant.

\begin{lemma}\label{[1,0]-[1,1]-[0,1]-csp-new-unary}
Suppose $\mathcal{F}$ contains a signature $f$ of arity $n \ge 1$,
that has two distinct nonzero values:
$f_{\alpha}\neq 0, f_{\beta}\neq 0$ and $f_{\alpha}\neq f_{\beta}$
 for some $\alpha, \beta\in\{0, 1\}^n$.
Then
there is a unary signature $[1, a]$, where $a\neq 0, 1$, such that
\[\operatorname{Pl-\#CSP}([1, 0], [1, a], \mathcal{F})
\le_{\rm T}
\operatorname{Pl-\#CSP}([1, 0], \mathcal{F}).\]
The statement is also valid if we replace $[1, 0]$ by $[0, 1]$.
\end{lemma}
\begin{proof}
We prove the lemma for $\operatorname{Pl-\#CSP}([1, 0], \mathcal{F})$. The proof for $\operatorname{Pl-\#CSP}([0, 1], \mathcal{F})$
 is symmetric.

Since $f_{\alpha}\neq f_{\beta}$,
at least one value  of $f_{\alpha}$ or $f_{\beta}$
is not equal to $f_{00\cdots 0}$.
Without loss of generality, we assume that
 $f_{\alpha}\neq f_{00\cdots 0}$.
Then we have
$f'=\partial_{[1, 0]}^{S}(f)=(f_{00\cdots 0}, \ldots, f_{\alpha})$,
which is a signature of arity ${\rm wt}(\alpha)$,
where $S=\{k \mid \mbox{the $k$-th bit of $\alpha$ is $0$}\}$.
We only care about the value of $f'$ at $0\ldots 0
\in \{0,1\}^{{\rm wt}(\alpha)}$ and $1\ldots 1
\in \{0,1\}^{{\rm wt}(\alpha)}$, as specified.
By connecting $f'$ to an {\sc Equality}  ($=_{{\rm wt}(\alpha)+1}$) of arity
${\rm wt}(\alpha)+1$ in a planar fashion,
 the gadget gives the unary signature $\partial_{f'}(=_{{\rm wt}(\alpha)+1})=[f_{00\cdots 0}, f_{\alpha}]$.
If $f_{00\cdots 0}\neq 0$, then we have $f_{00\cdots 0}[1, a]$, where $a=\frac{f_{\alpha}}{f_{00\cdots 0}}$.
Then we get $[1, a]$ up to the nonzero scalar $f_{00\cdots 0}$. This finishes the proof.

Otherwise,  $f_{00\cdots 0}=0$.
This implies that $[f_{00\cdots 0}, f_{\alpha}]=f_{\alpha}[0, 1]$.
Then we have $[0, 1]$ up to the nonzero scalar $f_{\alpha}$,
and can now use both pinning signatures $[1,0]$ and $[0, 1]$.

Since the set $T=\{(\xi, \eta) \mid f_{\xi}\neq 0, f_{\eta}\neq 0, f_{\xi}\neq f_{\eta}\}$ is nonempty
by $(\alpha, \beta)\in T$,
there exists some $(\xi, \eta)\in T$
with minimum Hamming distance, i.e.,
$f_{\xi}\neq 0, f_{\eta}\neq 0$, $f_{\xi}\neq f_{\eta}$ and
\[{\rm wt}(\xi\oplus \eta)=\displaystyle\min_{\xi',\eta'\in\{0, 1\}^n}\{{\rm wt}(\xi'\oplus\eta') \mid
f_{\xi'}\neq 0, f_{\eta'}\neq 0, f_{\xi'}\neq f_{\eta'}\}.\]
For $b \in \{0, 1\}$, let
$S_b=\{k \mid \mbox{the $k$-th bits of both $\xi$ and $\eta$
are $b$}\}$.
Then we can construct
$f''=\partial_{[1, 0]}^{S_0}[\partial_{[0, 1]}^{S_1}(f)]$.

Denote wt$(\xi\oplus\eta)$ by $d$.
Note that  $f''$ has arity  $d$.
Let $\check{\xi} \in \{0, 1\}^d$ denote
the $d$-bit string obtained from $\xi$ by deleting all bits in $S_0\cup S_1$.
Similarly define $\check{\eta} \in \{0, 1\}^d$.
Clearly
$f''_{\check{\xi}}=f_{\xi}$ and $f''_{\check{\eta}}=f_{\eta}$.
If $d=1$, i.e., $f''=[f_{\xi}, f_{\eta}]$ or
$f''=[f_{\eta}, f_{\xi}]$,
then we are done by normalizing.
Therefore, we may assume that $d\geq 2$.

All entries of $f''$ are zero  except for $f''_{\check{\xi}}, f''_{\check{\eta}}$.
To see this,
if there is another nonzero entry of $f''$
it has the form $f''_{\check{\gamma}} = f_{\gamma}$
for some $\gamma \in \{0, 1\}^n$ and $\check{\gamma}\in\{0, 1\}^d$,
where $\gamma$ has the same bits as $\xi$ on $S_0\cup S_1$,
and $\check{\gamma}$ is obtained from $\gamma$
by deleting all bits in $S_0\cup S_1$.
Then $f''_{\check{\gamma}}\neq 0$ implies that $f_{\gamma}\neq 0$.
Both ${\rm wt}(\xi\oplus\gamma) < {\rm wt}(\xi \oplus \eta)$,
${\rm wt}(\gamma\oplus\eta) < {\rm wt}(\xi \oplus \eta)$,
and either $f_{\gamma}\neq f_{\xi}$ or $f_{\gamma}\neq f_{\eta}$.
This contradicts the minimality of $d$.

By using $[1, 1]$,
we have $\partial_{[1, 1]}^{\{2, \ldots, d\}}(f'')$ that is
$[f_{\xi}, f_{\eta}]$ or
$[f_{\eta}, f_{\xi}]$,
 since all  entries of $f''$ are zero
 except for $f''_{\check{\xi}} = f_{\xi}$
and $f''_{\check{\eta}} = f_{\eta}$.
Thus we are done by normalizing $\partial_{[1, 1]}^{\{2, \ldots, d\}}(f'')$.
\end{proof}

The next lemma handles $\{0, 1\}$-valued function sets,
and more generally it applies to any $\mathcal{F}$
where  every $f \in \mathcal{F}$ has at most
one nonzero value. Notice that in this $\{0, 1\}$-valued case,
any function $f \in \mathcal{F} \cap \mathscr{P}$ is also in $\mathscr{A}$.
(And since $\widehat{\mathcal{F}}$ does not satisfy
the Parity Condition, $\mathcal{F} \subseteq \widehat{\mathscr{M}}$
is not feasible. Therefore the statement of Lemma~\ref{0-1-value-signature}
is in accordance with Theorem~\ref{main-dichotomy-thm}.)
When used in Theorem~\ref{main-theorem-for-no-parity}
 the auxiliary unary signatures
$[1,0]$ or $[0,1]$ will be provided by $[1, \omega] = [1, \pm 1]$
from the $\operatorname{Pl-Holant}$ side,
as constructed in Lemma~\ref{constructing-[1,a]}.
The proof of the following lemma
  will again use our Tableau Calculus.
\begin{lemma}\label{0-1-value-signature}
If each  $f\in\mathcal{F}$ takes values in $\{0, 1\}$, then
either $\operatorname{Pl-\#CSP}(\mathcal{F}, [1, 0])$
is \#P-hard or $\mathcal{F}\subseteq\mathscr{A}$.
The statement is also true if we replace $[1, 0]$ by $[0, 1]$.
\end{lemma}
\begin{proof}
Suppose $\mathcal{F}\nsubseteq\mathscr{A}$.
We show that $\operatorname{Pl-\#CSP}(\mathcal{F}, [1, 0])$
is \#P-hard.
The statement for $[0, 1]$ is symmetric.

As $\mathcal{F}\nsubseteq\mathscr{A}$,
there exists $f\in\mathcal{F}$ such that $f\notin\mathscr{A}$.
Firstly, we claim that $f\notin\mathscr{P}$.
By definition, $\mathscr{P}=\langle \mathcal{E}\rangle$.
Note that all signatures in $\mathcal{E}$ have affine support.
Thus all signatures in $\mathscr{P}$ have affine support.
On the other hand,
 a $\{0, 1\}$-valued signature is in $\mathscr{A}$
iff its support is affine.
 So supp$(f)$ is not affine since $f$ takes
 values in $\{0, 1\}$  and $f\notin\mathscr{A}$.
 This implies that $f\notin\mathscr{P}$.

We prove the lemma by induction on the arity $n$ of $f$.
Note that $n\geq 2$ since supp$(f)$ is not affine.

For $n=2$, there is exactly one entry of $f$ that is
 0 since supp$(f)$ is not affine.
So $f=(1, 1, 1, 0)$, or $f=(1, 1, 0, 1)$, or $f=(1, 0, 1, 1)$ or $f=(0, 1, 1, 1)$.
In each case, we take 3 copies of $f$ and connect the first input of each $f$
to an edge of $=_3$ and leave the second input as dangling edges.
The resulting signature $g$ is
$[1, 1]^{\otimes 3}+[1, 0]^{\otimes 3}$ or $[1, 1]^{\otimes 3}+[0, 1]^{\otimes 3}$.
Both $[1, 1]^{\otimes 3}+[1, 0]^{\otimes 3}$
and $[1, 1]^{\otimes 3}+[0, 1]^{\otimes 3}$
are not in $\mathscr{P}\cup\mathscr{A}\cup\widehat{\mathscr{M}}$ but are symmetric.
By Theorem~\ref{heng-tyson-dichotomy-pl-csp}, $\operatorname{Pl-\#CSP}(g)$
is \#P-hard.
By \[\operatorname{Pl-\#CSP}(g)
\le_{\rm T}
\operatorname{Pl-\#CSP}([1, 0], \mathcal{F}).\]
$\operatorname{Pl-\#CSP}([1, 0], \mathcal{F})$
is \#P-hard.

In the following, by induction   we assume that the lemma is true for $n-1$,
where $n\ge 3$.

We will  prove that we can construct
a signature $f' \not \in \mathscr{A}$ with arity $<n$ by a gadget construction in $\operatorname{Pl-\#CSP}([1, 0], \mathcal{F})$.
If $f'$ takes values in $\{0, 1\}$ up to a scalar,
then the induction is finished.
Otherwise, there exist two distinct
 nonzero entries $f'_{\alpha}$ and $f'_{\beta}$ of $f'$,
namely $f'_{\alpha}\neq 0$, $f'_{\beta} \neq 0$,
and  $f'_{\alpha}\neq f'_{\beta}$.
By Lemma~\ref{[1,0]-[1,1]-[0,1]-csp-new-unary}, we can construct $[1, a]$ with $a\neq 0, 1$ from $f'$ by gadget construction.
Because signatures in $\mathcal{F}$ take values in $\{0, 1\}$,
all nonzero entries of
the signature of any gadget construction in
  $\operatorname{Pl-\#CSP}([1, 0], \mathcal{F})$
take positive rational values, in $\mathbb{Q}^+$,
 after normalization.  This implies that $|a|\neq 1$.
So we can get the unary signatures $[1, 2], [1, 3], [1, 4]$
(in fact, we can get any constantly many unary signatures)
using interpolation, by Lemma~\ref{interpolation-unary}.
Then by Theorem~\ref{arity-reduction-product} and $f\notin\mathscr{P}$, we can get a symmetric signature $f''$ that is not in $\mathscr{P}$.
Note that the symmetric signature set $\{[1, 2], f''\}$ satisfies
\[\{[1, 2], f''\}\nsubseteq\mathscr{P},~~~~
\{[1, 2], f''\}\nsubseteq\mathscr{A}, ~~~~
\{[1, 2], f''\}\nsubseteq\widehat{\mathscr{M}}.\]
By Theorem~\ref{heng-tyson-dichotomy-pl-csp}, $\operatorname{Pl-\#CSP}([1, 2], f'')$
is \#P-hard.
By \[\operatorname{Pl-\#CSP}([1, 2], f'')
\le_{\rm T}
\operatorname{Pl-\#CSP}([1, 0], \mathcal{F}).\]
 $\operatorname{Pl-\#CSP}([1, 0], \mathcal{F})$
is \#P-hard.

Thus we only need to construct a signature $f'
\not \in \mathscr{A}$ with arity
$<n$ by gadget construction to finish the proof.
If there exists $i\in[n]$ such that $f^{x_i=0} \not \in \mathscr{A}$,
then we are done since we have $[1, 0]$.
Therefore, we may assume $f^{x_i=0} \in \mathscr{A}$  for all $i\in[n]$.
By connecting the unary signature $[1, 1]$ to the first variable of $f$ we get
 $h=\partial^{\{1\}}_{[1,1]}(f)$, which has arity $n-1$. Note that
 \[h_{\alpha}=f_{0\alpha}+f_{1\alpha}\] for all $\alpha\in\{0, 1\}^{n-1}$.
 If $h \notin\mathscr{A}$, then we are done.
 Therefore we may assume $h \in \mathscr{A}$.

\begin{itemize}
\item Suppose $f_{00\cdots 0}=1$.

Since supp$(f)$ is not affine, it is also not a linear subspace.
Thus there exist ${\bf a}, {\bf b}\in{\rm supp}(f)$ such that ${\bf a}\oplus {\bf b}\notin{\rm supp}(f)$.
Let ${\bf a}=a_1a_2\cdots a_n$, ${\bf b}=b_1b_2\cdots b_n$ and ${\bf c}={\bf a}\oplus {\bf b}=c_1c_2\cdots c_n$
and ${\bf a'}=a_2\cdots a_n$, ${\bf b'}=b_2\cdots b_n$ and ${\bf c'}=c_2\cdots c_n$.
If $a_1=a_2=\cdots =a_n=1$ and $b_1=b_2=\cdots =b_n=1$, then $c_1=c_2=\cdots=c_n=0$.
This contradicts that $00\cdots 0 \in {\rm supp}(f)$.
Thus there is at least one $i\in[n]$ such that $a_i=0$ or $b_i=0$.
Without loss of generality, we assume that $a_1=0$.
If $b_1=0$ as well, then $c_1=0$.
This implies that supp$(f^{x_1=0})$
 is not a linear subspace of $\mathbb{Z}_2^{n-1}$.
However $(0, \ldots, 0)\in {\rm supp}(f^{x_1=0})$,
being not a linear subspace implies that supp$(f^{x_1=0})$
is not an affine subspace. This
 contradicts that $f^{x_1=0} \in \mathscr{A}$.
 Hence $b_1=1$ and it follows that $c_1=1$.

\[
{\begin{array}{lllllll}
 &{\bf a} & = & a_1  {\bf a}' & =  & 0~  a_2\ldots a_n  & \in {\rm supp}(f) \\
 \oplus &{\bf b} & = & b_1 {\bf b}' & = & 1~ b_2\ldots b_n  & \in {\rm supp}(f)  \\
\hline
 & {\bf c}  & = & c_1   {\bf c}'  & =  & 1~ c_2\ldots c_n  & \not \in {\rm supp}(f) \\
\end{array} }
\]

We have
\[h_{{\bf a}'}=f_{0{\bf a}'}+f_{1{\bf a}'}\neq 0,\]
\[h_{{\bf b}'}=f_{0{\bf b}'}+f_{1{\bf b}'}\neq 0\]
since $f_{0{\bf a}'}=f_{1{\bf b}'}=1$ and $f_{1{\bf a}'}\geq 0, f_{0{\bf b}'}\geq 0$.
Note that $h_{00\cdots 0}=f_{00\cdots 0}+f_{10\cdots 0}\neq 0$ since $f_{00\cdots 0}=1, f_{10\cdots 0}\geq 0$.
So supp$(h)$ is a linear space since $h$ is affine.
As ${\bf a}', {\bf b}' \in {\rm supp}(h)$ we have ${\bf c}'\in{\rm supp}(h)$.
 \[
{\begin{array}{*{3}c}
 ~ &   {\bf a}' ~~~\in {\rm supp}(h) \\
 \oplus &   {\bf b}' ~~~\in {\rm supp}(h)  \\
\hline
 ~ &   {\bf c}' ~~~\in {\rm supp}(h)  \\
\end{array} }
\]
 This implies that $h_{{\bf c}'}=f_{c_1{\bf c}'}+f_{\bar{c}_1{\bf c}'}\neq 0$. So we have $f_{\bar{c}_1{\bf c}'}\neq 0$ since $f_{c_1{\bf c}'}=0$.
Thus ${\bf c}' \in {\rm supp}(f^{x_1=0})$, as $\bar{c}_1 =0$.
As $f$ takes values in $\{0, 1\}$, $f_{\bar{c}_1{\bf c}'} =1$.
 Since ${\bf a}' \in {\rm supp}(f^{x_1=0})$,
and the support of $f^{x_1=0}$ is a
linear subspace, we have ${\bf b}'\in{\rm supp}(f^{x_1=0})$,
since ${\bf b}' = {\bf a}' \oplus {\bf c}'$.
 \[
{\begin{array}{*{3}c}
 ~ & {\bf a}' ~~~\in {\rm supp}(f^{x_1=0}) \\
 \oplus  &    {\bf c}' ~~~\in {\rm supp}(f^{x_1=0})  \\
\hline
 ~ &    {\bf b}' ~~~\in {\rm supp}(f^{x_1=0})  \\
\end{array} }
\]
 This implies that
 $\bar{b}_1{\bf b}'\in$supp$(f)$.
 Thus we have
 $f_{\bar{b}_1{\bf b}'}=1$, as $f$ takes values in $\{0, 1\}$.
 But we also have $f_{b_1{\bf b}'}=1$, thus
 $h_{\bf{b'}}=f_{b_1{\bf b}'}+f_{\bar{b}_1{\bf b}'}=2$. However
 $h_{{\bf c}'}=f_{c_1{\bf c}'}+f_{\bar{c}_1{\bf c}'}=1$
since $f_{c_1{\bf c}'}=0$ and $f_{\bar{c}_1{\bf c}'}=1$.
  Thus $h \not \in  \mathscr{A}$. This is a contradiction.

\item Suppose $f_{00\cdots 0}=0$.

Since $f$ is not identically 0,
there exists $\beta\in{\rm supp}(f)$ with wt$(\beta)$ minimum among
all nonzero entries.
Then we have $\partial_{[1, 0]}^{S}(f)$,
where $S=\{k \mid \mbox{the $k$-th bit of $\beta$ is 0}\}$.
$\partial_{[1, 0]}^{S}(f)$
is the symmetric signature $[0, \ldots, 0, 1]$
of arity wt$(\beta)$.
This gives $[0, 1]=\partial_{[1, 1]}^{\{2, \ldots, {\rm wt}(\beta)\}}([0, \cdots, 0, 1])$.
If there exists $i$ such that
$f^{x_i=1} \not \in \mathscr{A}$, then we are done since we have $[0, 1]$ now.
Therefore, we may assume $f^{x_i=1} \in \mathscr{A}$ as well
  as $f^{x_i=0} \in \mathscr{A}$ for all $i\in[n]$, since
we have $[1,0]$ explicitly.

Since supp$(f)$ is not affine,
there exist ${\bf a}, {\bf b}, {\bf c}\in{\rm supp}(f)$ such that ${\bf a}\oplus {\bf b}\oplus{\bf c}\notin{\rm supp}(f)$.
Let ${\bf a}=a_1a_2\cdots a_n$, ${\bf b}=b_1b_2\cdots b_n$, ${\bf c}=c_1c_2\cdots c_n$ and
${\bf d}={\bf a\oplus  b\oplus  c}=d_1d_2\cdots d_n$,
and we denote ${\bf a'}=a_2\cdots a_n$, ${\bf b'}=b_2\cdots b_n$,
  ${\bf c'}=c_2\cdots c_n$ and ${\bf d'}=d_2\cdots d_n$.

 \[
{\begin{array}{lllllllll}
 &{\bf a} & = & a_1  {\bf a}' & = & a_1 a_2 \ldots a_n ~~~ &\in {\rm supp}(f)\\
 &{\bf b} & = & b_1  {\bf b}' & = & b_1 b_2 \ldots b_n ~~~ &\in {\rm supp}(f)\\
 \oplus & {\bf c}  & = & c_1   {\bf c}'  & = & c_1 c_2  \ldots c_n
~~ &\in {\rm supp}(f)\\
\hline
 &{\bf d} & = & d_1  {\bf d}' & = & d_1 d_2 \ldots d_n & \not \in {\rm supp}(f) \\
\end{array} }
\]

If $a_1=b_1=c_1$, then it follows that $a_1=b_1=c_1=d_1$.
This implies that supp$(f^{x_1=a_1})$
 is not an affine subspace of $\mathbb{Z}_2^{n-1}$.
 This
 contradicts that $f^{x_1=a_1} \in \mathscr{A}$.
 Hence without loss of generality, we can assume that
 $a_1=b_1=\bar{c}_1$, it follows that $a_1=b_1=\bar{c}_1=\bar{d}_1$.

We have
\[h_{\bf {a'}}=f_{a_1{\bf a'}}+f_{\bar{a}_1{\bf a'}}\neq 0,\]
\[h_{\bf {b'}}=f_{b_1{\bf b'}}+f_{\bar{b}_1\bf {b'}}\neq 0,\]
\[h_{\bf {c'}}=f_{c_1{\bf c'}}+f_{\bar{c}_1{\bf c'}}\neq 0,\]
since
$f_{a_1{\bf a'}}=f_{b_1{\bf b'}}=f_{c_1{\bf c'}}=1$ and $f_{\bar{a}_1{\bf a'}}\geq 0, f_{\bar{b}_1{\bf b'}}\geq 0, f_{\bar{c}_1{\bf c'}}\geq 0$.
Since supp$(h)$ is affine
and ${\bf a}', {\bf b}' , {\bf c}'\in {\rm supp}(h)$ we have ${\bf d}'\in{\rm supp}(h)$, as ${\bf d}' = {\bf a}' \oplus {\bf b}'  \oplus {\bf c}'$.
 \[
{\begin{array}{*{3}c}
 ~ &   {\bf a}' ~~~\in {\rm supp}(h) \\
 ~ &   {\bf b}' ~~~\in {\rm supp}(h)  \\
 ~ \oplus  &   {\bf c}' ~~~\in {\rm supp}(h)  \\
\hline
 ~ &   {\bf d}' ~~~\in {\rm supp}(h)  \\
\end{array} }
\]
 By $h_{{\bf d}'}=f_{d_1{\bf d}'}+f_{\bar{d}_1{\bf d}'}\neq 0$, we have $f_{\bar{d}_1{\bf d}'}\neq 0$ since $f_{d_1{\bf d}'}=0$.
Thus ${\bf d}' \in {\rm supp}(f^{x_1=\bar{d}_1})$.
As $f$ takes values in $\{0, 1\}$, $f_{\bar{d}_1{\bf d}'} =1$.
 Recall that $a_1 = b_1 = \bar{c}_1 =  \bar{d}_1$,
we have ${\bf d}' \in {\rm supp}(f^{x_1=\bar{c}_1})$, and also
 ${\bf a}', {\bf b}' \in {\rm supp}(f^{x_1=\bar{c}_1})$.
The support of $f^{x_1=\bar{c}_1}$ is an
affine subspace, and so we have ${\bf c}'\in{\rm supp}(f^{x_1=\bar{c}_1})$,
as ${\bf c}' = {\bf a}' \oplus {\bf b}'  \oplus {\bf d}'$.
 \[
{\begin{array}{*{3}c}
 ~ & {\bf a}' ~~~\in {\rm supp}(f^{x_1=\bar{c}_1}) \\
 ~ &    {\bf b}' ~~~\in {\rm supp}(f^{x_1=\bar{c}_1})  \\
 \oplus  &    {\bf d}' ~~~\in {\rm supp}(f^{x_1=\bar{c}_1})  \\
\hline
 ~ &    {\bf c}' ~~~\in {\rm supp}(f^{x_1=\bar{c}_1})  \\
\end{array} }
\]
 Thus $\bar{c}_1{\bf c'}\in {\rm supp}(f)$.
 So
 $h_{\bf{c'}}= f_{c_1\bf{c'}} +  f_{\bar{c}_1\bf{c'}} = 2$. But
 $h_{{\bf d}'}=f_{d_1{\bf d}'}+f_{\bar{d}_1{\bf d}'}=1$
since $f_{d_1{\bf d}'}=0$ and $f_{\bar{d}_1{\bf d}'}=1$.
  Thus $h \not \in \mathscr{A}$. This is a contradiction.
\end{itemize}

\end{proof}

\subsection{Dichotomy When $\widehat{\mathcal{F}}$ Does Not Satisfy Parity}

\begin{theorem}\label{main-theorem-for-no-parity}
If the signature set $\widehat{\mathcal{F}}$
contains a signature that does not satisfy the parity condition,
then either  $\operatorname{Pl-Holant}(\widehat{\mathcal{EQ}}, \widehat{\mathcal{F}})$
is \#P-hard, or $\widehat{\mathcal{F}}\subseteq\widehat{\mathscr{P}}$, or $\widehat{\mathcal{F}}\subseteq\mathscr{A}$,
in which case the problem is tractable in P.
\end{theorem}
\begin{proof}
Let $\mathcal{F}=H_2\widehat{\mathcal{F}}$,
where $H_2=\left[\begin{smallmatrix}
1 & 1\\
1 & -1
\end{smallmatrix}
\right]$.
For any $\widehat{f} \in \widehat{\mathcal{F}}$ of arity $n$,
let
 $H_2^{\otimes{n}}\widehat{f}=f$.
One can translate the theorem statement to an equivalent statement
in the Pl-\#CSP setting, i.e.,
either $\operatorname{Pl-\#CSP}(\mathcal{F})$
is \#P-hard, or $\mathcal{F}\subseteq\mathscr{P}$, or $\mathcal{F}\subseteq\mathscr{A}$.
Recall that $\widehat{\mathscr{A}} = H_2 \mathscr{A} = \mathscr{A}$.

 If $\widehat{\mathcal{F}}\subseteq\widehat{\mathscr{P}}$ or $\widehat{\mathcal{F}}\subseteq\mathscr{A}$,
equivalently if $\mathcal{F}\subseteq\mathscr{P}$ or
$\mathcal{F}\subseteq\mathscr{A}$,
then the problem  $\operatorname{Pl-\#CSP}(\mathcal{F})$
 is tractable in P by Theorem~\ref{non-planar-csp-dichotomy}.
Otherwise, there exist $\widehat{f}, \widehat{g}\in\widehat{\mathcal{F}}$ such that $\widehat{f}\notin \widehat{\mathscr{P}}$
 and $\widehat{g}\notin\mathscr{A}$.
Translating to the Pl-\#CSP setting,
 there exist $f, g\in\mathcal{F}$
such that $f\notin \mathscr{P}$ and $g\notin\mathscr{A}$
 in Pl-\#CSP$(\mathcal{F})$.

Moreover, by Lemma~\ref{constructing-[1,a]},
we can construct $[1, w]$ in $\operatorname{Pl-Holant}(\widehat{\mathcal{EQ}}, \widehat{\mathcal{F}})$, where $w\neq 0$,
such that
\[\operatorname{Pl-Holant}(\widehat{\mathcal{EQ}}, [1, w], \widehat{\mathcal{F}})
\le_{\rm T}
\operatorname{Pl-Holant}(\widehat{\mathcal{EQ}}, \widehat{\mathcal{F}}).\]
This implies that
\[\operatorname{Pl-\#CSP}(H_2[1,w], \mathcal{F})
\le_{\rm T}
\operatorname{Pl-\#CSP}(\mathcal{F}).\]
We can finish the proof by the following two alternatives:
\begin{description}
\item{(A)} Construct a unary signature $[1, b]$ with $b^4\neq 0, 1$ in $\operatorname{Pl-\#CSP}(\mathcal{F})$, i.e.,
\[\operatorname{Pl-\#CSP}([1, b], \mathcal{F})
\le_{\rm T}
\operatorname{Pl-\#CSP}(\mathcal{F}).\]
In this case we have unary signatures
 $\partial_{[1, b]}^{\{1,2\}}(=_3)=[1, b^2]$ and $\partial_{[1, b]}^{\{1,2,3\}}(=_4)=[1, b^3]$.
So \[\operatorname{Pl-\#CSP}([1, b], [1, b^2], [1, b^3], \mathcal{F})
\le_{\rm T}
\operatorname{Pl-\#CSP}(\mathcal{F}).\]
Note that $[1, b], [1, b^2], [1, b^3]$ are pairwise linearly independent
 since $b^4\neq 0, 1$.
We have $[1, b] \not \in \mathscr{A} \cup \widehat{\mathscr{M}}$
by Propositions~\ref{A-has-same-norm-etc} and \ref{matchgate:affine:hat}.
Then by Theorem~\ref{arity-reduction-product}
and $f\in\mathcal{F} \setminus \mathscr{P}$,
 there exists a symmetric signature $f'\notin\mathscr{P}$ such that
 \[\operatorname{Pl-\#CSP}(f', [1, b], [1, b^2], [1, b^3], \mathcal{F})
\le_{\rm T}
\operatorname{Pl-\#CSP}([1, b], [1, b^2], [1, b^3], \mathcal{F}).\]
 Note that the symmetric signature set $\{[1, b], f'\}$
satisfies
\[\{[1, b], f'\} \nsubseteq \mathscr{P}, ~~~~
\{[1, b], f'\} \nsubseteq \mathscr{A}, ~~~~
\{[1, b], f'\} \nsubseteq \widehat{\mathscr{M}}.\]
So  $\operatorname{Pl-\#CSP}(f', [1, b])$ is \#P-hard by Theorem~\ref{heng-tyson-dichotomy-pl-csp}.
Thus
$\operatorname{Pl-\#CSP}(\mathcal{F})$
is \#P-hard.

\item{(B)} Construct $[0, 1]$ and $[1, x]$ with $x^4=1$ in $\operatorname{Pl-Holant}(\widehat{\mathcal{EQ}}, \widehat{\mathcal{F}})$, i.e.,
\[\operatorname{Pl-Holant}(\widehat{\mathcal{EQ}}, [0, 1], [1, x], \widehat{\mathcal{F}})
\le_{\rm T}
\operatorname{Pl-Holant}(\widehat{\mathcal{EQ}}, \widehat{\mathcal{F}}).\]
Then
by $\widehat{g} \in \widehat{\mathcal{F}} \setminus \mathscr{A}$ and
Lemma~\ref{arity-reduction-affine}, we have a unary signature $[y, z]$
in $\operatorname{Pl-Holant}(\widehat{\mathcal{EQ}}, \widehat{\mathcal{F}})$
that is not in $\mathscr{A}$.
Translating into
$\operatorname{Pl-\#CSP}(\mathcal{F})$,
this means that we have a unary signature $H_2[y, z] \not \in \mathscr{A}$,
because $\mathscr{A}$ is invariant under $H_2$.
    Then we are done by the previous case.
\end{description}

 Now we prove the theorem according to the value of $w$.
\begin{enumerate}
\item Suppose $w^4\neq 1$.
As $w \not =0$ is given, by Proposition~\ref{A-has-same-norm-etc} we have
$[1,w]  \not \in \mathscr{A}$, and thus
we have $H_2[1,w] \not \in \mathscr{A}$ in
$\operatorname{Pl-\#CSP}(\mathcal{F})$.
Thus we are done by alternative (A).

\item Suppose $w=\pm 1$.

In this case, we have $H_2[1,w]= [1+w, 1-w]=2[1, 0]$ if
$w=1$ or $2[0,1]$ if $w=-1$,
 in $\operatorname{Pl-\#CSP}(\mathcal{F})$.
If, up to a scalar,
 each  signature in $\mathcal{F}$ takes value
in  $\{0, 1\}$, then we are done by Lemma~\ref{0-1-value-signature}.
Otherwise, we can get a unary signature $[1, c]$ with $c\neq 0, 1$ by Lemma~\ref{[1,0]-[1,1]-[0,1]-csp-new-unary}
in $\operatorname{Pl-\#CSP}(\mathcal{F})$.
\begin{itemize}
\item If $c^4\neq  1$, then we are done by alternative (A),
 as $c \not = 0$ is given by Lemma~\ref{[1,0]-[1,1]-[0,1]-csp-new-unary}.

\item If $c=\pm {\frak i}$,
 we translate into $\operatorname{Pl-Holant}(\widehat{\mathcal{EQ}}, \widehat{\mathcal{F}})$
by $H_2^{-1} = \frac{1}{2} H_2$.
So we have
$H_2^{-1}[1, c] = \frac{1+c}{2}[1, -c]$,
and therefore also
 $\partial_{[1, -c]}^{\{1,2\}}([1, 0, 1, 0])=-2c[0, 1]$
in $\operatorname{Pl-Holant}(\widehat{\mathcal{EQ}}, \widehat{\mathcal{F}})$,
where the signature $[1, 0, 1, 0] \in
\widehat{\mathcal{EQ}}$.
 Thus we are done by alternative (B).
    \item If $c=-1$, again we
translate into $\operatorname{Pl-Holant}(\widehat{\mathcal{EQ}}, \widehat{\mathcal{F}})$
by $H_2^{-1}$.
In addition to $[1, w]$ where $w = \pm 1$, and so $w^4=1$,
we also have $H_2^{-1}[1, -1] = [0, 1]$ in
$\operatorname{Pl-Holant}(\widehat{\mathcal{EQ}}, \widehat{\mathcal{F}})$.
Thus we are done by alternative (B).
\end{itemize}

\item For $w=\pm {\frak i}$,
in addition to $[1, w]$ with  $w^4=1$, we also have
$\partial_{[1, w]}^{\{1,2\}}([1, 0, 1, 0])=2w[0, 1]$  in $\operatorname{Pl-Holant}(\widehat{\mathcal{EQ}}, \widehat{\mathcal{F}})$.
Thus we are done by alternative (B).
\end{enumerate}
\end{proof}

Theorem~\ref{main-theorem-for-no-parity} is a dichotomy for
$\PlCSP(\mathcal{F})$ in the case when
$\widehat{\mathcal{F}}$ does not satisfy the Parity Condition.
It conforms to the final form Theorem~\ref{main-dichotomy-thm}.
Note that since some signature in $\widehat{\mathcal{F}}$ violates
the Parity Condition, the possible tractability  condition
$\widehat{\mathcal{F}} \subseteq \mathscr{M}$ does not
appear in the statement of Theorem~\ref{main-theorem-for-no-parity}.

%% file: 4pl-csp2.tex
\section{A Dichotomy Theorem of Pl-CSP$^2(\widehat{\mathcal{EQ}}, \widehat{\mathcal{F}})$}\label{sec-csp2}
In this section we prove a dichotomy theorem for
 Pl-CSP$^2(\widehat{\mathcal{EQ}}, \widehat{\mathcal{F}})$,
Theorem~\ref{dichotomy-csp-2},
where all signatures in $\widehat{\mathcal{F}}$ satisfy the
Parity Condition.
By (\ref{eqn:prelim:PlCSPd_equiv_Holant}), we have
\[\operatorname{Pl-CSP}^2(\widehat{\mathcal{EQ}}, \widehat{\mathcal{F}})
\equiv_{\rm T}
\operatorname{Pl-Holant}( \mathcal{EQ}_2,\widehat{\mathcal{EQ}}, \widehat{\mathcal{F}}).\]
Theorem~\ref{dichotomy-csp-2} will be used later in Section~\ref{sec:F-has-parity}
in the situation when we can construct $(=_4)$ in
$\operatorname{Pl-Holant}(\widehat{\mathcal{EQ}}, \widehat{\mathcal{F}})$.
 Then we have the following chain of equivalent problems
\begin{eqnarray*}
\operatorname{Pl-CSP}(\mathcal{F})
& \equiv_{\rm T} & \operatorname{Pl-Holant}(\widehat{\mathcal{EQ}},
 \widehat{\mathcal{F}}) \\
& \equiv_{\rm T} &
\operatorname{Pl-Holant}(\widehat{\mathcal{EQ}},  (=_4),
\widehat{\mathcal{F}}) ~~~~\mbox{(when we can construct $(=_4)$)}\\
& \equiv_{\rm T} &
\operatorname{Pl-Holant}(\widehat{\mathcal{EQ}}, \mathcal{EQ}_2,
\widehat{\mathcal{F}}) ~~~~\mbox{(by Lemma~\ref{equality-4-to-all-even-equality})}\\
& \equiv_{\rm T} &
\operatorname{Pl-CSP}^2(\widehat{\mathcal{EQ}}, \widehat{\mathcal{F}}).
\end{eqnarray*}

By Proposition~\ref{A-has-same-norm-etc},
a binary signature $[1, 0, x]$ is not in $\mathscr{A}$ iff
$x^4\neq 0, 1$.
Suppose we have some $[1, 0, x]$.
The following lemma says that if
$[1, 0, x]\notin\mathscr{A}$,
then we can get $[1, 0, z]$ for any $z\in\mathbb{C}$,
as well as $[0, 1]^{\otimes 2}$,
in  $\operatorname{Pl-Holant}(\widehat{\mathcal{EQ}}, \widehat{\mathcal{F}})$.
Moreover, even if $x=\pm {\frak i}$, we still can get $[0, 1]^{\otimes 2}$
and $[1, 0, -1]$ from $[1, 0, x]$ in $\operatorname{Pl-Holant}(\widehat{\mathcal{EQ}}, \widehat{\mathcal{F}})$.
The proof of the lemma is the same as
 the proof of Lemma~8.3 in \cite{Guo-Williams}.

\begin{lemma}\label{constructing-[1,0,x]-heng}
If $[1, 0, x]\notin\mathscr{A}$, then
for any $z\in\mathbb{C}$,
\[\operatorname{Pl-Holant}([1, 0, z], [0, 1]^{\otimes 2},
\widehat{\mathcal{EQ}}, \widehat{\mathcal{F}})
\le_{\rm T}
\operatorname{Pl-Holant}([1, 0, x],
\widehat{\mathcal{EQ}}, \widehat{\mathcal{F}}).\]

Moreover, if  $x= \pm {\frak i}$, we have
\[\operatorname{Pl-Holant}( [1, 0, -1],[0, 1]^{\otimes 2},
\widehat{\mathcal{EQ}}, \widehat{\mathcal{F}})
\le_{\rm T}
\operatorname{Pl-Holant}([1, 0, x],
\widehat{\mathcal{EQ}}, \widehat{\mathcal{F}}). \]
\end{lemma}
\begin{proof}
We will use the following gadget (Figure~\ref{fig:heng})
from \cite{Guo-Williams}, where
circle vertices are assigned $[1, 0, 1, 0] \in \widehat{\mathcal{EQ}}$
and square vertices are assigned $[1, 0, x]$.  It
 has signature $[1+x^2, 0, 2x]$.

\begin{figure}[ht]
 \centering
 \def\capWidth{3.5cm}
 \captionsetup[subfigure]{width=\capWidth}
 \tikzstyle{entry} = [internal, inner sep=2pt]
  \makebox[\capWidth][c]{
   \begin{tikzpicture}[scale=\scale,transform shape,node distance=\nodeDist,semithick]
    \node[external] (0)              {};
    \node[internal] (1) [right of=0] {};
    \node[external] (2) [right of=1] {};
    \node[external] (3) [right of=2] {};
    \node[internal] (4) [right of=3] {};
    \node[external] (5) [right of=4] {};
    \path (0) edge                          node[near end]   (e1) {}               (1)
          (1) edge[out= 45, in= 135]        node[square]     (e2) {}               (4)
                            edge[out=-45, in=-135]        node[square]     (e3) {}               (4)
          (4) edge                          node[near start] (e4) {}               (5);
    \begin{pgfonlayer}{background}
     \node[draw=\borderColor,thick,rounded corners,fit = (e1) (e2) (e3) (e4),transform shape=false] {};
    \end{pgfonlayer}
  \end{tikzpicture}}
 \caption{A gadget with signature $[1+x^2, 0, 2x]$.}
 \label{fig:heng}
\end{figure}
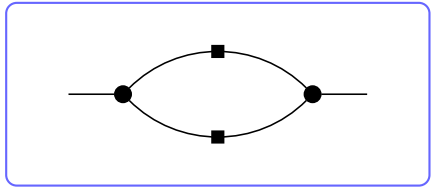

\begin{itemize}
\item
For  $[1, 0, x] \not \in \mathscr{A}$ we have $x^4\neq 0, 1$ by
 Proposition~\ref{A-has-same-norm-etc}.
If $|x|\neq 1$, by combining $k$ copies of $[1, 0, x]$, we have $[1, 0, x^k]$.
Then we can use polynomial interpolation to get
a reduction
$\operatorname{Pl-Holant}([y, 0, z],
\widehat{\mathcal{EQ}}, \widehat{\mathcal{F}})
\le_{\rm T}
\operatorname{Pl-Holant}([1, 0, x],
\widehat{\mathcal{EQ}}, \widehat{\mathcal{F}})$, for any
$y, z \in\mathbb{C}$. In particular we can get
 $[1, 0, z]$ for any $z\in\mathbb{C}$ and $[0, 0, 1]=[0, 1]^{\otimes 2}$.

Otherwise, $|x|=1$. In $[1+x^2, 0, 2x]$,
$0 < |1+x^2| < 2$ by $x \not = \pm {\frak i}$
and $x\neq \pm 1$.
However $|2x|=2$.
Therefore, after normalizing, the signature $[1, 0, \frac{2x}{1+x^2}]$ can interpolate
$[1, 0, z]$ for any $z\in\mathbb{C}$ and $[0, 0, 1]=[0, 1]^{\otimes 2}$.
\item For $x=\pm {\frak i}$, $[1+x^2, 0, 2x]=2{\frak i}[0, 0, 1]=2{\frak i}[0, 1]^{\otimes 2}$ and by combining two copies of $[1, 0, x]$, we have $[1, 0, -1]$.
\end{itemize}
\end{proof}

A very desirable tool is to pin a variable to 0 or 1.
This means we would like to have $[1,0]$ and $[0, 1]$.
We do have $[1,0] \in \widehat{\mathcal{EQ}}$. However,
if all signatures in $\widehat{\mathcal{F}}$ satisfy the even Parity
Condition, namely $f_{\alpha} =0$ for all $\alpha$ of odd weight,
then  every signature constructed in
$\operatorname{Pl-Holant}(
\widehat{\mathcal{EQ}}, \widehat{\mathcal{F}})$
has even Parity
Condition. Therefore
it is impossible to construct
$[0, 1]$.
But it is possible to construct $[0, 1]^{\otimes 2}$.
The next lemma shows that with $\widehat{\mathcal{EQ}}$,
getting $[0, 1]^{\otimes 2}$ is almost as good as $[0, 1]$.
\begin{lemma}\label{[0,1]-tensor-2-equal-[0,1]}
For $\mathscr{C} = \mathscr{A}$ or $\mathscr{M}$,
if there exists $f\in\widehat{\mathcal{F}}$ of arity $n\geq 2$
such that $f^{x_i=1}\notin\mathscr{C}$ for some $i\in[n]$,
then there exists a signature $g\notin\mathscr{C}$ with
 ${\rm arity}(g)=n-1$ such that
\[\operatorname{Pl-Holant}(\widehat{\mathcal{EQ}}, [0, 1]^{\otimes 2}, g, \widehat{\mathcal{F}})
\le_{\rm T}
\operatorname{Pl-Holant}(\widehat{\mathcal{EQ}}, [0, 1]^{\otimes 2}, \widehat{\mathcal{F}}).\]
Furthermore if $f$ satisfies the even Parity Condition, so does $g$.
\end{lemma}
\begin{proof}
We have $[1, 0, 1, 0] \in \widehat{\mathcal{EQ}}$.
By having $[0, 1]^{\otimes 2}$,
and the fact that $\partial_{[0, 1]}([1, 0, 1, 0])=[0, 1, 0]$,
 we get $[0, 1, 0]^{\otimes 2}$ by applying
$\partial_{[0, 1]^{\otimes 2}}$ on $[1, 0, 1, 0]^{\otimes 2}$.
Then by having $[1, 0] \in \widehat{\mathcal{EQ}}$, we
get $h(x_1, x_2, x_3)=[0, 1, 0]\otimes [0, 1]$,
where the binary disequality is on $x_1, x_2$ and the unary $[0, 1]$
is on $x_3$.
By connecting the variable $x_2$ of $h$ to the variable $x_{i+1}$
(if $i=n$, then let $x_{i+1}=x_1$) of $f$
and connecting the variable $x_3$ of $h$ to the
the variable $x_i$ of $f$ (See Figure~\ref{fig:[0,1]=[1,0]}),
the gadget gives an $(n-1)$-ary signature $g$, such that
\[g(x_1, \ldots, x_{i-1}, x_{i+1}, \ldots, x_n)=\displaystyle\sum_{x'_{i}, x'_{i+1}\in\{0, 1\}}
f(x_1,\ldots, x'_{i}, x'_{i+1}, \ldots,  x_n)h(x_{i+1}, x'_{i+1}, x'_{i}).\]
Notice that the variables of $f$ and $h$ are counterclockwise
ordered, and connections respect this order in a planar fashion.

\begin{figure}[ht]
 \centering
 \def\capWidth{3.5cm}
 \captionsetup[subfigure]{width=\capWidth}
 \tikzstyle{entry} = [internal, inner sep=2pt]
  \makebox[\capWidth][c]{
   \begin{tikzpicture}[scale=\scale,transform shape,node distance=\nodeDist,semithick]
    \node[internal] (0)              {};
    \node[external] (1) [above of=0] {};
    \node[external] (2) [below of=0] {};
    \node[external] (3) [left of=1] {};
    \node[external] (4) [left of=0] {};
    \node[external] (5) [left of=2] {};
    \node[external] (6) [right of=1] {};
    \node[external] (7) [right of=0] {};
    \node[internal] (8) [right of=2, square] {};
    \node[external] (9) [right of=6] {};
    \path (0) edge [bend right]                                       (3)
          (0) edge [bend left]                                        (5)
          (0) edge [white]    node[black]           {\Huge $\vdots$}                                    (4)
          (0) edge [bend right]                                            (8)
          (0) edge [bend left]   node[triangle]           (10) {}                                    (9)
          (8) edge [dashed]                                       (10);
              \begin{pgfonlayer}{background}
     \node[draw=\borderColor,thick,rounded corners,fit = (0),transform shape=false] {};
    \end{pgfonlayer}
  \end{tikzpicture}}
 \caption{The circle vertex is assigned $f$, the square vertex denotes $[0, 1]$
 and the triangle vertex denotes $[0, 1, 0]$.
 The two nodes connected by the dash line is a single
signature $[0, 1, 0]\otimes [0, 1]$.}
\label{fig:[0,1]=[1,0]}
\end{figure}
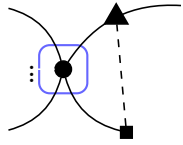

We have \[g(x_1, \ldots, x_{i-1}, x_{i+1}, \ldots, x_n)=f^{x_i=1}(x_1, \ldots, x_{i-1}, \overline{x_{i+1}}, \ldots, x_n).\]
Thus $g \in \mathscr{A}$ iff $f^{x_i=1} \in \mathscr{A}$,
and  $g \in \mathscr{M}$  iff $f^{x_i=1} \in \mathscr{M}$,
 by Lemma~\ref{[0,1,0]-not-change-tractable}.
\end{proof}

The next lemma shows that for an $n$-ary  signature with affine support and a set of free variables
$X=\{x_{i_1}, x_{i_2}, \ldots, x_{i_k}\}$,
if two consecutive variables $x_s, x_{s+1}\notin X$,
then we can combine the two variables
  to one new variable using $[1, 0, 1, 0]$, and
do not change the compressed signature.

\begin{lemma}\label{shrinking-non-free-variables-by-[1,0,1,]}
Let $f$ be an $n$-ary signature with affine support,
$X=\{x_{i_1}, x_{i_2}, \ldots, x_{i_k}\}$ is a set of free variables.
If there exists $s\in[n]$ such that $X$ does not include
$x_s, x_{s+1}$ (if $s=n$, then $x_{s+1}=x_1$),
letting
\[g(x_1\ldots, x_{s-1}, x', x_{s+2}, \ldots,  x_n)=\displaystyle\sum_{x_s, x_{s+1}\in\{0, 1\}}f(x_1\ldots, x_{s-1}, x_s, x_{s+1}, \ldots,  x_n)[1, 0, 1, 0](x_{s+1}, x_s, x'),\]
then $g$ has affine support, $X$ is a set of free variables of $g$
and $\underline{f_X}=\underline{g_X}$.
On its support, $x' = x_s \oplus x_{s+1}$.
Furthermore if $f$ satisfies the even Parity Condition, so does $g$.
\end{lemma}
\begin{proof}
We first note that $g$ as defined is the signature of
a planar gadget using $f$ and $[1, 0, 1, 0]$;
the order of variable connections respects  planarity.
Since $X$ is a set of free variables for $f$, if we fix an assignment
on $X$, then there exist unique
 $x_s=x_s(X)$, $x_{s+1}=x_{s+1}(X)$
as affine linear functions
such that $f$ is nonzero.
Moreover, in
\[
\displaystyle\sum_{x_s, x_{s+1}\in\{0, 1\}}f(x_1\ldots, x_{s-1}, x_s, x_{s+1}, \ldots,  x_n)[1, 0, 1, 0](x_{s+1}, x_s, x'),\]
if $[1, 0, 1, 0]$ takes value $1$, $x'$ must be $x_{s+1}\oplus x_s$.
This implies that  $g$ has affine support,
and  $X$ is a set of free variables.
It follows that  $\underline{f_X}=\underline{g_X}$.
\end{proof}

The following two lemmas are about how to find a set of free variables that includes an adjacent pair of variables.

\begin{lemma}\label{rearrange-the-order-of-the-variables}
Let $f$ be a signature with affine support of dimension $k\geq 2$,
and suppose there are no
variables that take
a constant value in the support.
Then there exists a set of free variables $X=\{x_{i_1}, x_{i_2}, \ldots, x_{i_k}\}$
where some two variables are adjacent.
\end{lemma}
\begin{proof}
Let $X=\{x_{i_1}, x_{i_2}, \ldots, x_{i_k}\}$
be the set of free variables which is minimum in the lexicographic order.
If $i_2=i_1+1$, then we are done.
Otherwise, $x_{i_2-1}$ is not in $X$.
If its dependency $x_{i_2-1}=\displaystyle\sum_{j=1}^{k}a_{j}x_{i_j}
+b$ involves any variable other than $x_{i_1}$, namely if $a_j \not =0$
for some $1 < j \le k$, then by switching $x_{i_2-1}$ with the variable
$x_{i_j}$, we get another set of free variables which is
lexicographically smaller than $X$, a contradiction.
Thus we have
$x_{i_2-1}=ax_{i_1}+b$.
Since there are no variables that take a
 constant value in the support, we have $a\neq 0$.
So $X'=\{x_{i_2-1}, x_{i_2}, \ldots, x_{i_k}\}$ is a set of free variables that includes $x_{i_2-1}, x_{i_2}$.
\end{proof}

\begin{lemma}\label{rearrange-the-order-of-the-variables-5}
Let $f$ be a 5-ary signature with affine support of dimension $3$.
Suppose  there are no variables that take a constant value in the support.
Then there exists a set of free variables $X$ such that the variables in $X$ are consecutive in a cyclic sense.
\end{lemma}
\begin{proof}
By Lemma~\ref{rearrange-the-order-of-the-variables}, without loss of generality, we can assume that
there exists a set of free variables including $x_1, x_2$.
If the other free variable is $x_3$ or $x_5$, then we are done.
Otherwise, $\{x_1, x_2, x_4\}$ is a set of free variables and
$x_3=a_1x_1+a_2x_2+a_4x_4+c$ and $x_5=b_1x_1+b_2x_2+b_4x_4+d$, where $a_i, b_i, c, d\in\mathbb{Z}_2$ for $i\in\{1, 2, 4\}$.
\begin{itemize}
\item If $a_1\neq 0$, then $\{x_2, x_3, x_4\}$ is a set of free variables.
\item If $a_4\neq 0$, then $\{x_1, x_2, x_3\}$ is a set of free variables.
\item If $b_2\neq 0$, then $\{x_1, x_4, x_5\}$ is a set of free variables.
\item If $b_4\neq 0$, then $\{x_1, x_2, x_5\}$ is a set of free variables.
\item If $a_1=a_4=0$ and $b_2=b_4=0$, then $a_2\neq 0$ and $b_1\neq 0$
since there are no variables that take a  constant value in the support.
 So $\{x_3, x_4, x_5\}$ is a set of free variables.
\end{itemize}
This finishes the proof.
\end{proof}

If $\widehat{\mathcal{F}}\subseteq\mathscr{A}$, then $\operatorname{Pl-Holant}(\widehat{\mathcal{EQ}}, \widehat{\mathcal{F}})$
is tractable. Otherwise, there exists $f\in\widehat{\mathcal{F}}
\setminus \mathscr{A}$.
The following three lemmas are about reducing the arity of $f$.
Since all signatures in $\widehat{\mathcal{EQ}}  \cup \widehat{\mathcal{F}}$
satisfy the Parity Condition,
any constructible unary signature in
$\operatorname{Pl-Holant}(\widehat{\mathcal{EQ}}, \widehat{\mathcal{F}})$
also satisfies the Parity Condition, and therefore is in $\mathscr{A}$.
So the lowest arity a constructible non-affine signature can have
 is 2.
Furthermore a binary signature satisfying the even Parity Condition
is symmetric and takes the form $[a,0,b]$.
By Proposition~\ref{A-has-same-norm-etc},
 $[1, 0, x] \notin \mathscr{A}$ iff
$x^4 \not = 0,1$.
The next lemma implies that in $\operatorname{Pl-Holant}(\widehat{\mathcal{EQ}}, [0, 1]^{\otimes 2}, [1, 0, -1], \widehat{\mathcal{F}})$,
from any $f\in\widehat{\mathcal{F}}
\setminus \mathscr{A}$  we can construct some $[1, 0, x]\not \in \mathscr{A}$.

\begin{lemma}\label{[1,0]-[0,1]-[1,0,-1]-EQ-hat-affine-reduction}
If all signatures in $\widehat{\mathcal{F}}$
 satisfy the Parity Condition and
 $\widehat{\mathcal{F}}\nsubseteq\mathscr{A}$, then there exists
$[1, 0, x] \not \in \mathscr{A}$ such that
\[\operatorname{Pl-Holant}(\widehat{\mathcal{EQ}}, [0, 1]^{\otimes 2}, [1, 0, -1], [1, 0, x], \widehat{\mathcal{F}})
\le_{\rm T}
\operatorname{Pl-Holant}(\widehat{\mathcal{EQ}}, [0, 1]^{\otimes 2}, [1, 0, -1], \widehat{\mathcal{F}}).\]
\end{lemma}
\begin{proof}
Since $\widehat{\mathcal{F}}\nsubseteq\mathscr{A}$,
 there exists $f\in\widehat{\mathcal{F}} \setminus \mathscr{A}$.
By Lemma~\ref{[0,1]-EQ-hat-wight-0-neq-0}, we can assume that $f_{00\cdots 0}=1$
and $f$ satisfies the even Parity Condition.
If $f$ has arity 1, then $f=[1, 0]\in\mathscr{A}$. This is a contradiction.
If $f$ has arity 2, then $f=[1, 0, x]$ with $x^4\neq 0, 1$ and
we are done.
In the following, we assume that
$f$ has arity $\geq 3$.
If we can construct a non-affine signature with arity $\leq n-1$, then we are done by induction.

If there exists $i\in[n]$ such that $f^{x_i=0}\notin\mathscr{A}$, then we are done by induction since we have $[1, 0]\in\widehat{\mathcal{EQ}}$.
If there exists $i\in[n]$
 such that $f^{x_i=1}\notin\mathscr{A}$,
then we are done by induction and Lemma~\ref{[0,1]-tensor-2-equal-[0,1]}.
So in the following we assume that both $f^{x_i=0}$ and $f^{x_i=1}$ are affine signatures for any $i\in[n]$.

Firstly, we claim that if supp$(f)$ is not affine, then we can construct a signature that is not in $\mathscr{A}$ with arity $\leq n-1$.
Certainly supp$(f)$ is not a linear subspace.
Note that $(0, 0, \ldots, 0)\in {\rm supp}(f)$.
A subset of $\mathbb{Z}_2^n$ containing $(0, 0, \ldots, 0)$ is affine iff  it is a linear subspace.
So supp$(f^{x_i=0})$ is a linear subspace of $\mathbb{Z}_2^{n-1}$ since $f^{x_i=0}$ is affine and
$f^{x_i=0}_{0\cdots 0}=1$.
By Lemma~\ref{argue-a1-0-b1-1},
 there exist
 ${\bf a}=a_1a_2\cdots a_n, {\bf b}=b_1b_2\cdots b_n,$
such that  ${\bf a}, {\bf b}\in{\rm supp}(f)$, ${\bf c}={\bf a}\oplus {\bf b}=c_1c_2\cdots c_n\notin{\rm supp}(f)$
and there exists $i\in[n]$ such that $a_i \not = b_i$.
Without loss of generality, we assume that $a_1=0, b_1=1$.
It follows that $c_1=1$.
Let ${\bf a}'=a_3\cdots a_n, {\bf b}'=b_3\cdots b_n, {\bf c}'=c_3\cdots c_n$.

By connecting one variable of $[1, 0, -1]$ to the first variable of $f$, we
get a gadget that gives
\[\bar{f}(x_1, x_2, \ldots, x_n)=\displaystyle\sum_{x'_1\in\{0, 1\}}[1, 0, -1](x_1, x'_1)f(x'_1, x_2, \ldots, x_n).\]
Note that \[\bar{f}(x_1, x_2, \ldots, x_n)=(-1)^{x_1}f(x_1, x_2, \ldots, x_n).\]
Moreover, by connecting the variables $x_2, x_1$ of $[1, 0, 1, 0]$ to the variables $x_1, x_2$ of $f$ respectively,
\begin{figure}[ht]
 \centering
 \def\capWidth{3.5cm}
 \captionsetup[subfigure]{width=\capWidth}
 \tikzstyle{entry} = [internal, inner sep=2pt]
  \makebox[\capWidth][c]{
   \begin{tikzpicture}[scale=\scale,transform shape,node distance=\nodeDist,semithick]
    \node[external] (0)              {};
    \node[external] (6) [above of=0] {};
    \node[external] (7) [below of=0] {};
    \node[external] (8) [left of=0] {};
    \node[internal] (1) [right of=0] {};
    \node[external] (2) [right of=1] {};
    \node[external] (3) [right of=2] {};
    \node[internal, square] (4) [right of=3] {};
    \node[external] (5) [right of=4] {};
    \path (1) edge[bend right]                                        (6)
          (1) edge[bend left]                                       (7)
          (1) edge[white]          node[black] {{\Huge $\vdots$}}                             (8)
          (1) edge[out= 45, in= 135]                       (4)
                            edge[out=-45, in=-135]                      (4)
          (4) edge                          node[near start, white] (e4) {}               (5);
    \begin{pgfonlayer}{background}
     \node[draw=\borderColor,thick,rounded corners,fit = (e1) (e2) (e3) (e4),transform shape=false] {};
    \end{pgfonlayer}
  \end{tikzpicture}}
 \caption{The circle vertex is assigned $f$ and the square vertex
is assigned $[1, 0, 1, 0]$. }
 \label{fig:[1,0,1,0]}
\end{figure}
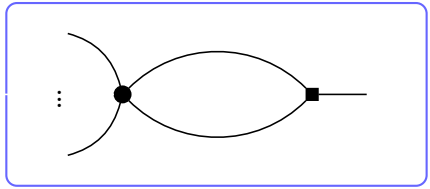
the planar gadget in Figure~\ref{fig:[1,0,1,0]} gives the signature
\[h(x', x_3, x_4, \ldots, x_n)=\displaystyle\sum_{x_1, x_2\in\{0, 1\}}[1, 0, 1, 0](x_2, x_1, x')f(x_1, x_2, \ldots, x_n).\]
This way of connecting the variables satisfies planarity.
Note that
\[h(x_1\oplus x_2, x_3, \ldots, x_n)=f(x_1, x_2, x_3, \ldots, x_n)+f(\overline{x_1}, \overline{x_2}, x_3, \ldots, x_n).\]
Similarly,  by connecting the variables $x_2, x_1$ of $[1, 0, 1, 0]$ to the variables $x_1, x_2$ of $\bar{f}$ respectively,
the planar gadget gives the signature
 \[\bar{h}(x', x_3, x_4, \ldots, x_n)=\displaystyle\sum_{x_1, x_2\in\{0, 1\}}[1, 0, 1, 0](x_2, x_1, x')\bar{f}(x_1, x_2, \ldots, x_n),\]
and we have
\[\bar{h}(x_1\oplus x_2, x_3, \ldots, x_n)=(-1)^{x_1}f(x_1, x_2, x_3, \ldots, x_n)+(-1)^{\overline{x_1}}f(\overline{x_1}, \overline{x_2}, x_3, \ldots, x_n).\]
If $h\notin\mathscr{A}$ or $\bar{h}\notin\mathscr{A}$, then we are done since both  $h$ and $\bar{h}$ have arity $n-1$.
In the following, we assume that both of $h$ and $\bar{h}$ are affine.

\vspace{.1in}
\noindent
Now we use the Tableau Calculus.
\begin{itemize}
\item If both $\bar{a}_1\bar{a}_2{\bf a}'$ and $\bar{b}_1\bar{b}_2{\bf b}'$ are not in supp$(f)$,
then
\begin{align*}
h_{(a_1\oplus a_2)\bf {a}'}
&=f_{a_1a_2\bf a'}+f_{\bar{a}_1\bar{a}_2\bf a'}\neq 0, \\
h_{(b_1\oplus b_2)\bf {b}'}
&=f_{b_1b_2\bf a'}+f_{\bar{b}_1\bar{b}_2\bf b'}\neq 0
\end{align*}
since $f_{a_1a_2\bf a'}\neq 0, f_{b_1b_2\bf a'}\neq 0$ and $f_{\bar{a}_1\bar{a}_2\bf a'}=f_{\bar{b}_1\bar{b}_2\bf b'}=0$.

Similarly, we have
\begin{align*}
\bar{h}_{(a_1\oplus a_2)\bf {a}'}
&=f_{a_1a_2\bf a'}-f_{\bar{a}_1\bar{a}_2\bf a'}\neq 0, \\
\bar{h}_{(b_1\oplus b_2)\bf {b}'}
&=-f_{b_1b_2\bf a'}+f_{\bar{b}_1\bar{b}_2\bf b'}\neq 0.
\end{align*}
Note that
\begin{align*}
h_{0\cdots 0}
&=f_{000\cdots 0}+f_{110\cdots 0}, \\
\bar{h}_{0\cdots 0}
&=f_{000\cdots 0}-f_{110\cdots 0}.
\end{align*}
Then by $f_{00\cdots 0}\neq 0$,
one of $h_{00\cdots 0}$ or $\bar{h}_{00\cdots 0}$ is nonzero.
We assume that $h_{00\cdots 0}\neq 0$.
Otherwise, we can assume that $\bar{h}_{00\cdots 0}\neq 0$ and the following proof can apply to $\bar{h}$, too.

Note that $h$ is affine, so supp$(h)$ is a linear subspace.
Then by
\[
{\begin{array}{llll}
        &  (a_1\oplus a_2)\bf a' & \in {\rm supp}(h)\\
 \oplus &  (b_1\oplus b_2)\bf b' & \in {\rm supp}(h)\\
\hline
        &  (c_1\oplus c_2)\bf c' & \\
\end{array} }
\]
we have $(c_1\oplus c_2){\bf c}'\in {\rm supp}(h)$.
This implies that
 $h_{(c_1\oplus c_2){\bf c}'}=f_{c_1c_2{\bf c}'}+f_{\bar{c}_1 \bar{c}_2{\bf c}'}\neq 0$.
 Since $f_{c_1 c_2{\bf c}'}=0$, we have
  $f_{\bar{c}_1 \bar{c}_2{\bf c}'}\neq 0$.
  This implies that $\bar{c}_1 \bar{c}_2{\bf c}'\in{\rm supp}(f)$.
  Then $\bar{c}_2{\bf c}'\in{\rm supp}(f^{x_1=0})$
  since $\bar{c}_1=0$.
  Note that ${\rm supp}(f^{x_1=0})$ is a linear subspace and
   \[
{\begin{array}{llll}
        &  a_2\bf a' & \in {\rm supp}(f^{x_1=0})\\
 \oplus &  \bar{c}_2\bf c' & \in {\rm supp}(f^{x_1=0})\\
\hline
        &  \bar{b}_2 \bf b' & \\
\end{array} }
\]
 we have $\bar{b}_2 {\bf b'} \in {\rm supp}(f^{x_1=0})$ .
 This implies that $\bar{b}_1\bar{b}_2 {\bf b'}\in{\rm supp}(f)$.
 This contradicts the hypothesis  $\bar{b}_1\bar{b}_2 {\bf b'}\notin{\rm supp}(f)$.

\item If both $\bar{a}_1\bar{a}_2\bf a'$ and $\bar{b}_1\bar{b}_2\bf b'$ are in supp$(f)$, then
by \[
{\begin{array}{llll}
        &  a_2\bf a' & \in {\rm supp}(f^{x_1=0})\\
 \oplus &  \bar{b}_2\bf b' & \in {\rm supp}(f^{x_1=0})\\
\hline
        &  \bar{c}_2 \bf c' & \\
\end{array} }
\]
we have $\bar{c}_2 {\bf c'} \in {\rm supp}(f^{x_1=0})$. Thus $\bar{c}_1\bar{c}_2 {\bf c'} \in {\rm supp}(f)$
since $\bar{c}_1=0$.

We claim that all of $f_{a_1a_2{\bf a}'}, f_{\bar{a}_1\bar{a}_2{\bf a}'}, f_{b_1b_2{\bf b}'}, f_{\bar{b}_1\bar{b}_2{\bf b}'}, f_{\bar{c}_1\bar{c}_2{\bf c}'}$
are powers of ${\frak i}$.
Firstly,
since $f^{x_1=0}$ is affine and all of $f_{a_1a_2{\bf a}'}, f_{\bar{b}_1\bar{b}_2{\bf b}'}, f_{\bar{c}_1\bar{c}_2{\bf c}'}$
are nonzero entries of $f^{x_1=0}$,
 we derive that
all of $f_{a_1a_2{\bf a}'}$, $f_{\bar{b}_1\bar{b}_2{\bf b}'}$, $f_{\bar{c}_1\bar{c}_2{\bf c}'}$ are powers of ${\frak i}$ by $f_{00\cdots 0}=1$.
Secondly, since  $f^{x_3=a_3}$ is affine, both $f_{a_1a_2{\bf a}'}$ and $f_{\bar{a}_1\bar{a}_2{\bf a}'}$ are nonzero entries of $f^{x_3=a_3}$
and $f_{a_1a_2{\bf a}'}$ is a power of ${\frak i}$,
we derive that $f_{\bar{a}_1\bar{a}_2{\bf a}'}$
is a power of ${\frak i}$.
Finally, since  $f^{x_3=b_3}$ is affine, both $f_{b_1b_2{\bf b}'}$ and $f_{\bar{b}_1\bar{b}_2{\bf b}'}$ are nonzero entries of $f^{x_3=b_3}$
 and $f_{\bar{b}_1\bar{b}_2{\bf b}'}$  is a power of ${\frak i}$,
we derive that  $f_{b_1b_2{\bf b}'}$ is a power of ${\frak i}$.

By
\begin{align*}
h_{(c_1\oplus c_2){\bf c}'}
&=f_{c_1c_2{\bf c}'}+f_{\bar{c}_1 \bar{c}_2{\bf c}'},\\
\bar{h}_{(c_1\oplus c_2){\bf c}'}
&=-f_{c_1c_2{\bf c}'}+f_{\bar{c}_1 \bar{c}_2{\bf c}'},
\end{align*}
we have  $|h_{(c_1\oplus c_2){\bf c}'}|=|\bar{h}_{(c_1\oplus c_2){\bf c}'}|=1$
since $f_{\bar{c}_1 \bar{c}_2{\bf c}'}$ is a power of ${\frak i}$ and $f_{c_1c_2{\bf c}'}=0$.
On the other hand,
\begin{equation}\label{linear-system-2by2-in-lm4.7}
\begin{aligned}
h_{(a_1\oplus a_2)\bf {a}'}
&=f_{a_1a_2\bf a'}+f_{\bar{a}_1\bar{a}_2\bf a'},  \\
\bar{h}_{(a_1\oplus a_2)\bf {a}'}
&=f_{a_1a_2\bf a'}-f_{\bar{a}_1\bar{a}_2\bf a'},
\end{aligned}
\end{equation}
both $h_{(a_1\oplus a_2){\bf a}'}$ and $\bar{h}_{(a_1\oplus a_2){\bf a}'}$
are sums of two quantities, each a power of ${\frak i}$.
If at least one of them is not zero, then it has norm $2$ or $\sqrt{2}$.
This implies that $h$ or $\bar{h}$ is not affine. This is a contradiction.

On the other hand, if both
   $h_{(a_1\oplus a_2){\bf a}'}$ and $\bar{h}_{(a_1\oplus a_2){\bf a}'}$ are zero,
   then $f_{a_1a_2{\bf a}'}=0$,
by treating (\ref{linear-system-2by2-in-lm4.7})
as a linear system. This contradicts that $a_1a_2{\bf a}'\in$supp$(f)$.

\item If $\bar{a}_1\bar{a}_2{\bf a}'\in {\rm supp}(f)$ and $\bar{b}_1\bar{b}_2{\bf b}'\notin {\rm supp}(f)$,
we claim that all of $f_{a_1a_2{\bf a}'}, f_{\bar{a}_1\bar{a}_2{\bf a}'}, f_{b_1b_2{\bf b}'}$ are  powers of ${\frak i}$.
 Firstly, since $f^{x_1=0}$ is affine and both $f_{a_1a_2{\bf a}'}$ and $f_{00\cdots 0}$ are nonzero entries of $f^{x_1=0}$, $f_{a_1a_2{\bf a}'}$
is a power of ${\frak i}$ by $f_{00\cdots 0}=1$.
Secondly, since $f^{x_3=a_3}$ is affine and both $f_{a_1a_2{\bf a}'}$ and $f_{\bar{a}_1\bar{a}_2{\bf a}'}$ are nonzero entries of $f^{x_3=a_3}$,
$f_{\bar{a}_1\bar{a}_2{\bf a}'}$ is a power of ${\frak i}$.
Finally, since $f^{x_1=1}$ is affine and both $f_{b_1b_2{\bf b}'}$ and $f_{\bar{a}_1\bar{a}_2{\bf a}'}$ are nonzero entries of $f^{x_1=1}$,
$f_{b_1b_2{\bf b}'}$ is a power of ${\frak i}$.

By
\begin{align*}
h_{(b_1\oplus b_2)\bf {b}'}
&=f_{b_1b_2\bf b'}+f_{\bar{b}_1\bar{b}_2\bf b'}, \\
\bar{h}_{(b_1\oplus b_2)\bf {b}'}
&=-f_{b_1b_2\bf b'}+f_{\bar{b}_1\bar{b}_2\bf b'},
\end{align*}
we have $|h_{(b_1\oplus b_2){\bf b}'}|=|\bar{h}_{(b_1\oplus b_2){\bf b}'}|=1$
since $f_{b_1b_2\bf a'}$ is a power of ${\frak i}$ and
$f_{\bar{b}_1\bar{b}_2\bf b'}=0$.
Moreover, by
\begin{align*}
h_{(a_1\oplus a_2)\bf {a}'}
&=f_{a_1a_2\bf a'}+f_{\bar{a}_1\bar{a}_2\bf a'}, \\
\bar{h}_{(a_1\oplus a_2)\bf {a}'}
&=f_{a_1a_2\bf a'}-f_{\bar{a}_1\bar{a}_2\bf a'},
\end{align*}
and the similar argument with the previous case,
 at least one of $h_{(a_1\oplus a_2){\bf a}'}$ or $\bar{h}_{(\bar{a}_1\oplus \bar{a}_2){\bf a}'}$ has norm $2$ or $\sqrt{2}$.
  This implies that $h$ or $\bar{h}$ is not affine. This is a contradiction.

  \item If $\bar{a}_1\bar{a}_2{\bf a}'\notin {\rm supp}(f)$ and $\bar{b}_1\bar{b}_2{\bf b}'\in {\rm supp}(f)$,
  the proof is symmetric by reversing the order of ${\bf a}$ and ${\bf b}$ in the previous item.
\end{itemize}

Now we can assume that supp$(f)$ is affine and has dimension $k$.

If $k=0$, then $f\in\mathscr{A}$. This is a contradiction.

If $k=1$, then there exist exactly one $\alpha\in\{0, 1\}^n$ such that $f_{\alpha}\neq 0$ other than $f_{00\cdots 0}=1$.
Note that wt$(\alpha)$ is even since $f$ satisfies
the  even Parity Condition.
Thus we have $\partial_{=_2}^{\left(\frac{\operatorname{wt}{(\alpha)}}{2}-1\right)}[\partial_{[1, 0]}^{S}(f)]
=[1, 0, f_{\alpha}]$, where $S=\{k|$ the $k$-th bit of $\alpha$ is $0\}$. If $f_{\alpha}^4\neq 1$, then we are done
as $[1, 0, f_{\alpha}]\notin\mathscr{A}$, by
Proposition~\ref{A-has-same-norm-etc}.
Otherwise, $f$ is affine. This is a contradiction.

If $k\geq 4$, then
since both $f^{x_i=0}$ and $f^{x_i=1}$ are affine for all $i\in[n]$,
  we get
$f \in \mathscr{A}$  by Lemma~\ref{[1,0]-[0,1]-pinning-implies-affine-arity-4}.
This is a contradiction.

Thus we only need to consider $k=2$ or $k= 3$.
If on its support some variable $x_i$ is a constant $c$,
then $f = f^{x_i =c} \otimes [1,0](x_i)$ or $f^{x_i =c} \otimes [0,1]$
depending on whether $c=0$ or $1$ respectively. Then $f$ would be affine,
a contradiction.
So no variable of $f$  takes a constant value on its support.
\begin{itemize}
\item For $k=2$, by Lemma~\ref{rearrange-the-order-of-the-variables},
without loss of generality, we can assume that $\{x_1, x_2\}$ is a set of free variables.
Then
apply Lemma~\ref{shrinking-non-free-variables-by-[1,0,1,]} repeatedly, we can get a ternary signature $\hat{f}$ such that
\[\operatorname{Pl-Holant}(\widehat{\mathcal{EQ}}, [0, 1]^{\otimes 2}, [1, 0, -1], \hat{f}, \widehat{\mathcal{F}})
\le_{\rm T}
\operatorname{Pl-Holant}(\widehat{\mathcal{EQ}}, [0, 1]^{\otimes 2}, [1, 0, -1], \widehat{\mathcal{F}}),\]
where the compressed signatures of $\hat{f}$ and $f$ for $\{x_1, x_2\}$ are the same. By Lemma~\ref{shrinking-non-free-variables-by-[1,0,1,]}
$\hat{f}$ satisfies the even Parity Condition and $\hat{f}_{000}
= f_{00 \ldots 0} = 1$.
By Corollary~\ref{f-affine-iff-f*-affine},
the compressed signature of $f$ is not affine.
Thus the compressed signature of $\hat{f}$ is not affine.
So
$\hat{f}$ is not affine.

If there exists $i\in[3]$ such that $\hat{f}^{x_i=0}\notin\mathscr{A}$, then we are done by $[1, 0]\in\widehat{\mathcal{EQ}}$.
If there exists $i\in[3]$ such that $\hat{f}^{x_i=1}\notin\mathscr{A}$, then we are done by
Lemma~\ref{[0,1]-tensor-2-equal-[0,1]} and $[0, 1]^{\otimes 2}$.
Therefore, we may assume that $\hat{f}^{x_i=0}$, $\hat{f}^{x_1=1}$ are affine for all $i\in[3]$.
Then there exist $r, s, t\in\{0, 1, 2, 3\}$ such that
$\hat{f}^{x_1=0}=[1, 0, {\frak i}^r]$, $\hat{f}^{x_2=0}=[1, 0, {\frak i}^s]$, $\hat{f}^{x_3=0}=[1, 0, {\frak i}^t]$.
So
\[M_{x_1, x_2x_3}(\hat{f})
=\begin{bmatrix}
\hat{f}_{000} & \hat{f}_{001} & \hat{f}_{010} & \hat{f}_{011}\\
\hat{f}_{100} & \hat{f}_{101} & \hat{f}_{110} & \hat{f}_{111}
\end{bmatrix}
=\begin{bmatrix}
1 & 0 & 0 & {\frak i}^r\\
0 & {\frak i}^s & {\frak i}^t & 0
\end{bmatrix}.\]

Note that the compressed signature of $\hat{f}$ for the free variable set $\{x_1, x_2\}$ is $[1, {\frak i}^r, {\frak i}^s, {\frak i}^t]$.
It is affine iff ${\frak i}^t=\pm {\frak i}^{r+s}$ by Lemma~\ref{binary-affine-compressed function}.
Since $\hat{f} \not \in \mathscr{A}$,
we have ${\frak i}^t=\pm {\frak i}^{r+s+1}$, i.e.,
\[M_{x_1, x_2x_3}(\hat{f})
=\begin{bmatrix}
1 & 0 & 0 & {\frak i}^r\\
0 & {\frak i}^s & \pm {\frak i}^{r+s+1} & 0
\end{bmatrix}.\]
By   connecting three copies of $\hat{f}^{x_2=0}=[1, 0, {\frak i}^s]$ consecutively,
the gadget gives $[1, 0, {\frak i}^{3s}]
= [1, 0, {\frak i}^{-s}]$. Then we have
(see (\ref{ternary-sig-matrix-modification-f1}))
  \[\hat{g}(x_1, x_2, x_3)=\displaystyle\sum_{x'_1\in\{0, 1\}}[1, 0, {\frak i}^{-s}](x_1, x'_1)\hat{f}(x'_1, x_2, x_3),\]
with signature matrix
\[M_{x_1, x_2x_3}(\hat{g})
=\begin{bmatrix}
1 & 0 & 0 & {\frak i}^r\\
0 & 1 & \pm {\frak i}^{r+1} & 0
\end{bmatrix}.\]
By connecting variables $x_2, x_3$ of  $\hat{g}$
to variables $x_3, x_2$ of $[1, 0, 1, 0]$,
the planar gadget gives the signature
\[\displaystyle\sum_{x_1, x_3\in\{0, 1\}}\hat{g}(x_1, x_2, x_3)
[1, 0, 1, 0](x'_1, x_3, x_2)=[2, 0, (1\pm {\frak i}){\frak i}^r](x_1, x'_1).\]
 This is not affine since the norms of $2$ and $(1\pm {\frak i}){\frak i}^r$ are different. Thus we are done.

 \item If the dimension $k$ of the support is 3, then $n\geq 4$ since $f$ satisfies the even Parity Condition.
 Firstly, we claim that we can get a 4-ary signature $\tilde{f}$ that has the same compressed signature with $f$.
 If $n=4$, then we set $\tilde{f} = f$. Otherwise, $n\geq 5$.
 By Lemma~\ref{rearrange-the-order-of-the-variables},
without loss of generality, we can assume that there exists a set of free variables including $\{x_1, x_2\}$.
If $\{x_1, x_2, x_3\}$ or $\{x_n, x_1, x_2\}$ is a set of free variables, then
by Lemma~\ref{shrinking-non-free-variables-by-[1,0,1,]} repeatedly, we can
shrink the variables that are not in the free variable set to one variable while keep planarity.
Then we
get a signature $\tilde{f}$ that has arity 4
and has the same compressed signature with $f$.

Otherwise, there exists $k$ such that $\{x_1, x_2, x_k\}$ is a set of free variables, where $4\leq k\leq n-1$.
By Lemma~\ref{shrinking-non-free-variables-by-[1,0,1,]}, we can
shrink the variables indexed by $3\leq i\leq k-1$ to one variable $x'$, and
shrink the variables indexed by $k+1\leq i\leq n$ to one variable $x''$ while keep planarity.
Then we
get a signature $g(x_1, x_2, x', x_k, x'')$ with arity 5
and has the same compressed signature with $f$.
Then by Lemma~\ref{rearrange-the-order-of-the-variables-5}, there exists a set of free variables that are consecutive. Then
by Lemma~\ref{shrinking-non-free-variables-by-[1,0,1,]},
we can get a signature $\tilde{f}$ with arity 4 and has the same compressed signature with $f$.

Since $f$ is not affine, $\tilde{f}$ is not affine.
By Lemma~\ref{[0,1]-EQ-hat-wight-0-neq-0}
we can assume that
\[M_{x_1x_2, x_4x_3}(\tilde{f})=\begin{bmatrix}
\tilde{f}_{0000} & \tilde{f}_{0010} & \tilde{f}_{0001} & \tilde{f}_{0011}\\
\tilde{f}_{0100} & \tilde{f}_{0110} & \tilde{f}_{0101} & \tilde{f}_{0111}\\
\tilde{f}_{1000} & \tilde{f}_{1010} & \tilde{f}_{1001} & \tilde{f}_{1011}\\
\tilde{f}_{1100} & \tilde{f}_{1110} & \tilde{f}_{1101} & \tilde{f}_{1111}\\
\end{bmatrix}
=\begin{bmatrix}
1 & 0 & 0 & b\\
0 & \alpha & \beta & 0\\
0 & \gamma & \delta & 0\\
c & 0 & 0 & d\\
\end{bmatrix}.\]
Let $\underline{\tilde{f}}$ be the compressed signature of $\tilde{f}$ for $\{x_1, x_2, x_3\}$,
then
\begin{equation}\label{compressed-signature-of-arity-3}
M_{x_1, x_2x_3}(\underline{\tilde{f}})=\begin{bmatrix}
\tilde{f}_{0000} & \tilde{f}_{0011} & \tilde{f}_{0101} & \tilde{f}_{0110}\\
\tilde{f}_{1001} & \tilde{f}_{1010} & \tilde{f}_{1100} & \tilde{f}_{1111}\\
\end{bmatrix}
=\begin{bmatrix}
1 & b & \beta & \alpha\\
\delta & \gamma & c & d\\
\end{bmatrix}.
\end{equation}

If there exists $i\in[4]$ such that $\tilde{f}^{x_i=0}$ is not affine, then we are done by $[1, 0]\in\widehat{\mathcal{EQ}}$.
If there exists $i\in[4]$ such that $\tilde{f}^{x_i=1}$ is not affine, then we are done by Lemma~\ref{[0,1]-tensor-2-equal-[0,1]}
and  $[0, 1]^{\otimes 2}$.
Thus we may assume that both $\tilde{f}^{x_i=0}$ and $\tilde{f}^{x_i=1}$ are affine for $i\in[4]$.
Since $\tilde{f}^{x_i=0}$ is affine for $i=1, 2, 3$,
there exists $r, s, t\in\{0, 1, 2, 3\}$ and $\epsilon_1, \epsilon_2, \epsilon_3\in\{1, -1\}$
such that
\[M_{x_2, x_4x_3}(\tilde{f}^{x_1=0})=
\begin{bmatrix}
1 & 0 & 0 & b\\
0 & \alpha & \beta & 0\\
\end{bmatrix}=
\begin{bmatrix}
1 & 0 & 0 & {\frak i}^r\\
0 & \epsilon_1 {\frak i}^{r+s} & {\frak i}^s & 0\\
\end{bmatrix},
\]
\[M_{x_1, x_4x_3}(\tilde{f}^{x_2=0})=
\begin{bmatrix}
1 & 0 & 0 & b\\
0 & \gamma & \delta & 0\\
\end{bmatrix}=
\begin{bmatrix}
1 & 0 & 0 & {\frak i}^r\\
0 & \epsilon_2 {\frak i}^{r+t} & {\frak i}^t & 0\\
\end{bmatrix},
\]
\[M_{x_1x_2, x_4}(\tilde{f}^{x_3=0})=
\begin{bmatrix}
1 & 0 \\
0 &  \beta \\
0 &  \delta\\
c &  0 \\
\end{bmatrix}=
\begin{bmatrix}
1 & 0 \\
0 &  {\frak i}^s \\
0 &  {\frak i}^t\\
\epsilon_3 {\frak i}^{s+t} & 0\\
\end{bmatrix}.
\]
Moreover, since $\tilde{f}^{x_4=1}$ is affine, there exists $\epsilon_4 \in \{1, -1\}$, such that
\[M_{x_1x_2, x_3}(\tilde{f}^{x_4=1})=
\begin{bmatrix}
0 & b \\
\beta &  0 \\
\delta &  0\\
  0 & d\\
\end{bmatrix}=
\begin{bmatrix}
0 & {\frak i}^r \\
{\frak i}^s &  0 \\
{\frak i}^{s+t} &  0\\
  0 & \epsilon_4{\frak i}^{r+s+t}\\
\end{bmatrix}.
\]
So we have
\begin{equation}\label{M-of-fhat-in-both-letter-and-i-powers}
M_{x_1x_2, x_4x_3}(\tilde{f})
=\begin{bmatrix}
1 & 0 & 0 & b\\
0 & \alpha & \beta & 0\\
0 & \gamma & \delta & 0\\
c & 0 & 0 & d\\
\end{bmatrix}
=\begin{bmatrix}
1 & 0 & 0 & {\frak i}^r\\
0 & \epsilon_1 {\frak i}^{r+s} & {\frak i}^s & 0\\
0 & \epsilon_2 {\frak i}^{r+t} & {\frak i}^t & 0\\
\epsilon_3 {\frak i}^{s+t} & 0 & 0 & \epsilon_4 {\frak i}^{r+s+t}\\
\end{bmatrix}.
\end{equation}

Note that we have $\partial_{[1, 0]}^{\{3, 4\}}(\tilde f)=[1, 0, \epsilon_3 {\frak i}^{s+t}]$, $\partial_{[1, 0]}^{\{1, 2\}}(\tilde f)=[1, 0, {\frak i}^{r}]$.
By connecting consecutively three copies of
$[1, 0, \epsilon_3 {\frak i}^{s+t}]$, respectively
 of $[1, 0, {\frak i}^{r}]$,
we have
 $[1, 0, \epsilon^3_3 {\frak i}^{3s+3t}]
= [1, 0, \epsilon_3 ({\frak i}^{s+t})^{-1}]$,
respectively $[1, 0, {\frak i}^{3r}]
= [1, 0, {\frak i}^{-r}]$.
 Let
\[h(x_1, x_2, x_3, x_4)=\displaystyle\sum_{x'_2, x'_4\in\{0, 1\}}
\tilde{f}(x_1, x'_2, x_3, x'_4)[1, 0, \epsilon_3 ({\frak i}^{s+t})^{-1}](x'_2, x_2)[1, 0, {\frak i}^{-r}](x'_4, x_4).\]
Then (see (\ref{4-ary-sig-matrix-modification-f2}) and
(\ref{4-ary-sig-matrix-modification-f4}))
\[M_{x_1x_2, x_4x_3}(h)=\begin{bmatrix}
1 & 0 & 0 & 1\\
0 & \epsilon_1\epsilon_3 {\frak i}^{r-t} & \epsilon_3{\frak i}^{-r-t} & 0\\
0 & \epsilon_2 {\frak i}^{r+t} & {\frak i}^{t-r} & 0\\
1 & 0 & 0 & \epsilon_3\epsilon_4\\
\end{bmatrix}=
\begin{bmatrix}
1 & 0 & 0 & 1\\
0 & \epsilon_1\epsilon_3 (-1)^t a & \epsilon_3 (-1)^{r+t} a & 0\\
0 & \epsilon_2 a & (-1)^r a & 0\\
1 & 0 & 0 & \epsilon_3\epsilon_4\\
\end{bmatrix},\] where $a={\frak i}^{r+t}$.

Take  two copies of $h$, connect the $3, 4$-th variables of one copy to the $4, 3$-th variables of another copy,
the planar
gadget (see Figure~\ref{fig:gadget:connecting two h}) gives the signature
 \[
h'(x_1, x_2, x_3, x_4))=\displaystyle\sum_{x'_3, x'_4\in\{0, 1\}}h(x_1, x_2, x'_3, x'_4)h(x_3, x_4, x'_4, x'_3).
\]

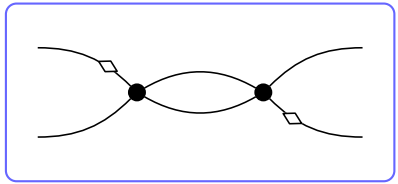
\begin{figure}[htpb]
 \centering
 \begin{tikzpicture}[scale=\scale,transform shape,node distance=\nodeDist,semithick]
  \node[internal]  (0)                    {};
  \node[external]  (1) [above left  of=0] {};
  \node[external]  (2) [below left  of=0] {};
  \node[external]  (3) [left        of=1] {};
  \node[external]  (4) [left        of=2] {};
  \node[external]    (5) [right       of=0] {};
  \node[internal]  (6) [right       of=5] {};
  \node[external]  (7) [above right of=6] {};
  \node[external]  (8) [below right of=6] {};
  \node[external]  (9) [right       of=7] {};
  \node[external] (10) [right       of=8] {};
  \path (0) edge[in=   0, out=135, postaction={decorate, decoration={
                                                           markings,
                                                           mark=at position 0.4   with {\arrow[>=diamond, white] {>}; },
                                                           mark=at position 0.4   with {\arrow[>=open diamond]   {>}; },
                                                           mark=at position 2 with {\arrow[>=diamond, white] {>}; },
                                                           mark=at position 1.0   with {\arrow[>=open diamond, white]   {>}; } } }] (3)
            edge[out=-135, in=   0]  (4)
            edge[bend right]                    (6)
            edge[bend left]                    (6)
        (6) edge[out=  45, in= 180]  (9)
            edge[in=   180, out=-45, postaction={decorate, decoration={
                                                           markings,
                                                           mark=at position 0.4   with {\arrow[>=diamond, white] {>}; },
                                                           mark=at position 0.4   with {\arrow[>=open diamond]   {>}; },
                                                           mark=at position 2 with {\arrow[>=diamond, white] {>}; },
                                                           mark=at position 1.0   with {\arrow[>=open diamond, white]   {>}; } } }] (10);
  \begin{pgfonlayer}{background}
   \node[inner sep=0pt,transform shape=false,draw=\borderColor,thick,rounded corners,fit = (1) (2) (7) (8)] {};
  \end{pgfonlayer}
 \end{tikzpicture}
 \caption{The two vertices are assigned $h$. The edges with diamond
indicate the first variable. Other variables are ordered counterclockwise.}
 \label{fig:gadget:connecting two h}
\end{figure}

Note that (see Figure~(\ref{fig:rotate_asymmetric_signature});
for the second copy of $h$ we rotate $180^{\circ}$)
\begin{align*}
M_{x_1x_2, x_4x_3}(h')
&=\begin{bmatrix}
1 & 0 & 0 & 1\\
0 & \epsilon_1\epsilon_3 (-1)^t a & \epsilon_3 (-1)^{r+t} a & 0\\
0 & \epsilon_2 a & (-1)^r a & 0\\
1 & 0 & 0 & \epsilon_3\epsilon_4\\
\end{bmatrix}
\begin{bmatrix}
1 & 0 & 0 & 1\\
0 & (-1)^r a  & \epsilon_3 (-1)^{r+t} a & 0\\
0 & \epsilon_2 a & \epsilon_1\epsilon_3 (-1)^t a& 0\\
1 & 0 & 0 & \epsilon_3\epsilon_4\\
\end{bmatrix}\\
&=
\begin{bmatrix}
2 & 0 & 0 & 1+\epsilon_3\epsilon_4\\
0 & \epsilon_3(\epsilon_1+\epsilon_2)  & 2\epsilon_1 (-1)^t & 0\\
0 & 2\epsilon_2 (-1)^{t}  & \epsilon_3(\epsilon_1+\epsilon_2)& 0\\
1+\epsilon_3\epsilon_4 & 0 & 0 & 2\\
\end{bmatrix}.
\end{align*}
Let $h''=\partial_{[1, 0]}^{\{1\}}(h')$,
then \[M_{x_2, x_4x_3}(h'')=
\begin{bmatrix}
2 & 0 & 0 & 1+\epsilon_3\epsilon_4\\
0 & \epsilon_3(\epsilon_1+\epsilon_2)  & 2\epsilon_1 (-1)^t  & 0\\
\end{bmatrix}.\]

Thus we have
\begin{itemize}
\item if $\epsilon_1\epsilon_2\epsilon_3\epsilon_4=-1$, then
exactly one of $\epsilon_1 = - \epsilon_2$ or $\epsilon_3 = - \epsilon_4$
holds.
This implies that $h''$ does not have
affine support. Thus $h''$ is not affine.
Then we are done by induction.
\item if $\epsilon_1\epsilon_2\epsilon_3\epsilon_4=1$, by (\ref{compressed-signature-of-arity-3}) and (\ref{M-of-fhat-in-both-letter-and-i-powers}),
\[
M_{x_1, x_2x_3}(\underline{\tilde{f}})=\begin{bmatrix}
1  &  {\frak i}^r  &  {\frak i}^s  &\epsilon_1 {\frak i}^{r+s}\\
{\frak i}^t&\epsilon_2 {\frak i}^{r+t} & \epsilon_3 {\frak i}^{s+t} & \epsilon_4 {\frak i}^{r+s+t}
\end{bmatrix},
\]
where $\underline{\tilde{f}}$ is the compressed signature of $f$ for $\{x_1, x_2, x_3\}$.
Note that $\underline{\tilde{f}}$
is affine by Lemma~\ref{binary-affine-compressed function}.
Thus $f$ is affine. This is a contradiction.
\end{itemize}

\end{itemize}
\end{proof}

The next lemma shows how to reduce the arity of a
non-affine signature  in
$\operatorname{Pl-Holant}(\widehat{\mathcal{EQ}}, \widehat{\mathcal{F}})$
when all signatures in $\widehat{\mathcal{F}}$ take  values
in  $\{0, 1\}$.
\begin{lemma}\label{0,1-valued}
Suppose all signatures in $\widehat{\mathcal{F}}$ take values in $\{0, 1\}$
and satisfy the Parity Condition.
If $\widehat{\mathcal{F}}$ contains a signature $f \not \in \mathscr{A}$
of arity $n \ge 3$,
then there exists a signature $g\notin\mathscr{A}$ of arity $<n$,
such that
\[\operatorname{Pl-Holant}(\widehat{\mathcal{EQ}}, g, \widehat{\mathcal{F}})
\le_{\rm T}
\operatorname{Pl-Holant}(\widehat{\mathcal{EQ}}, \widehat{\mathcal{F}}).\]
Furthermore, if $f$ satisfies the even Parity Condition, so does $g$.
\end{lemma}
\begin{proof}
The proof is by induction on $n$, and uses the Tableau Calculus.

If there exists $i\in[n]$ such that $f^{x_i=0}  \not \in \mathscr{A}$,
then we are done since we have $[1, 0]\in\widehat{\mathcal{EQ}}$.
In the following, we assume that $f^{x_i=0} \in \mathscr{A}$
for $1\leq i\leq n$.
By Lemma~\ref{[0,1]-EQ-hat-wight-0-neq-0}, we can assume that $f_{00\cdots 0}=1$ and $f$ satisfies the even Parity Condition.

For a $\{0,1\}$-valued signature  $f$, $f  \in \mathscr{A}$  iff supp$(f)$
 is an affine subspace.
Thus  supp$(f)$ is not an affine subspace, in particular
not a linear subspace.
A subset of $\mathbb{Z}_2^n$ containing $(0, 0, \ldots, 0)$ is affine iff  it is a linear subspace.
By $(0, 0, \ldots, 0)\in {\rm supp}(f)$, and $f^{x_i=0} \in \mathscr{A}$,
supp$(f^{x_i=0})$ is a linear subspace of $\mathbb{Z}_2^{n-1}$.
By Lemma~\ref{argue-a1-0-b1-1},
 there exist
 ${\bf a}=a_1a_2\cdots a_n, {\bf b}=b_1b_2\cdots b_n,$
such that  ${\bf a}, {\bf b}\in{\rm supp}(f)$, ${\bf c}={\bf a}\oplus {\bf b}=c_1c_2\cdots c_n\notin{\rm supp}(f)$
and there exists $i\in[n]$ such that $a_i \not = b_i$.
Without loss of generality, we assume that $a_1=0, b_1=1$.
It follows that $c_1=1$.
Let ${\bf a}'=a_3\cdots a_n, {\bf b}'=b_3\cdots b_n, {\bf c}'=c_3\cdots c_n$.

By connecting the variables $x_2, x_1$ of $[1, 0, 1, 0]$ to the variables $x_1, x_2$ of $f$ respectively,
the planar gadget gives the signature
\[h(x', x_3, x_4, \ldots, x_n)=\displaystyle\sum_{x_1, x_2\in\{0, 1\}}[1, 0, 1, 0](x_2, x_1, x')f(x_1, x_2, \ldots, x_n).\]
Note that
\[h(x_1\oplus x_2, x_3, \ldots, x_n)=f(x_1, x_2, x_3, \ldots, x_n)+f(\bar{x}_1, \bar{x}_2, x_3, \ldots, x_n).\]
If $h \notin\mathscr{A}$, then we are done since $h$ has arity $n-1$.
So in the following, we assume that $h \in\mathscr{A}$.
Now comes the Tableau Calculus.

\begin{itemize}
\item For $\bar{a}_1\bar{a}_2{\bf a}'\in {\rm supp}(f)$, $\bar{b}_1\bar{b}_2{\bf b}'\in {\rm supp}(f)$,
 since supp$(f^{x_1=0})$ is a linear subspace and
 \[
{\begin{array}{llllllllllllllll}
        & a_2{\bf a}'& \in {\rm supp}(f^{x_i=0}) \\
 \oplus & \bar{b}_2{\bf b}' &  \in {\rm supp}(f^{x_i=0})\\
\hline
        & \bar{c}_2{\bf c}' &  \\
\end{array} }
\]
we have $\bar{c}_2{\bf c}'\in{\rm supp}(f^{x_i=0})$.
This means that
$f_{\bar{c}_1\bar{c}_2{\bf c}'}\neq 0$ since $\bar{c}_1=0$.
Note that
\[h_{(c_1\oplus c_2){\bf c}'}=f_{c_1c_2\bf{c}'}+f_{\bar{c}_1\bar{c}_2\bf{c}'}.\]
Thus $|h_{(c_1\oplus c_2){\bf c}'}|=1$
since $f_{c_1c_2\bf{c}'}=0$ and $f_{\bar{c}_1\bar{c}_2\bf{c}'}=1$
(as $f$ is $\{0,1\}$-valued).
Moreover, by
\[h_{(a_1\oplus a_2){\bf a}'}=f_{a_1a_2\bf{a}'}+f_{\bar{a}_1\bar{a}_2\bf{a}'},\]
 we have $|h_{(a_1\oplus a_2){\bf a}'}|=2,$
 since $f_{a_1a_2\bf{a}'}=f_{\bar{a}_1\bar{a}_2\bf{a}'}=1$. This implies that
$h$ is not affine, as nonzero values of $h$ have different norms
(Proposition~\ref{A-has-same-norm-etc}).
This is a contradiction.

\item For $\bar{a}_1\bar{a}_2{\bf a}'\notin {\rm supp}(f)$, $\bar{b}_1\bar{b}_2{\bf b}'\notin {\rm supp}(f)$,
by
\[h_{00\cdots 0}=f_{000\cdots 0}+f_{110\cdots 0},\]
we have $h_{00\cdots 0}\neq 0$  since $f_{000\cdots 0}=1$
and $f_{110\cdots 0}\in\{0, 1\}$.
This implies that supp$(h)$ is a linear subspace.
By \[h_{(a_1\oplus a_2){\bf a}'}=f_{a_1 a_2{\bf a}'}+f_{\bar{a}_1 \bar{a}_2{\bf a}'},\]
\[h_{(b_1\oplus b_2){\bf b}'}=f_{b_1 b_2{\bf b}'}+f_{\bar{b}_1 \bar{b}_2{\bf b}'},\]
we have
$h_{(a_1\oplus a_2){\rm a}'}=1,$ $h_{(b_1\oplus b_2){\bf b}'}=1$
since  $f_{a_1 a_2{\bf a}'}= f_{b_1 b_2{\bf b}'}=1$ and $f_{\bar{a}_1 \bar{a}_2{\bf a}'}=f_{\bar{b}_1 \bar{b}_2{\bf b}'}=0$.
This implies that $(a_1\oplus a_2){\bf a}', (b_1\oplus b_2){\bf b}'\in{\rm supp}(h)$.
Then by
\[
{\begin{array}{llllllllllllllll}
        & (a_1\oplus a_2){\bf a}' & \in{\rm supp}(h) \\
 \oplus & (b_1\oplus b_2){\bf b}' & \in{\rm supp}(h) \\
\hline
        & (c_1\oplus c_2){\bf c}' &  \\
\end{array} }
\]
we have $(c_1\oplus c_2){\bf c}' \in{\rm supp}(h)$.
This implies that
\[h_{(c_1\oplus c_2){\bf c}'}=f_{c_1 c_2{\bf c}'}+f_{\bar{c}_1 \bar{c}_2{\bf c}'}\neq 0.\]
Thus $f_{\bar{c}_1 \bar{c}_2{\bf c}'}\neq 0$ since $f_{c_1 c_2{\bf c}'}=0$.
So $\bar{c}_1 \bar{c}_2{\bf c}'\in{\rm supp}(f)$.
Hence, $\bar{c}_2{\bf c}'\in{\rm supp}(f^{x_1=0})$ since $\bar{c}_1=0$.
By
\[
{\begin{array}{llllllllllllllll}
        & a_2{\bf a}' & \in{\rm supp}(f^{x_1=0}) \\
 \oplus & \bar{c}_2{\bf c}' & \in{\rm supp}(f^{x_1=0}) \\
\hline
        & \bar{b}_2{\bf b}' &  \\
\end{array} }
\]
we have $\bar{b}_2{\bf b}'\in{\rm supp}(f^{x_1=0})$.
Thus $\bar{b}_1\bar{b}_2{\bf b}'\in{\rm supp}(f)$ as $\bar{b}_1=0$.
This is a contradiction.

\item If $\bar{a}_1\bar{a}_2{\bf a'}\in {\rm supp}(f)$
 and $\bar{b}_1\bar{b}_2{\bf b}'\notin {\rm supp}(f)$, then
we consider
 \begin{equation}\label{lemma-0-1-reduce-1}
\begin{aligned}
 h_{(a_1\oplus a_2){\bf a}'}
&=f_{a_1a_2{\bf a}'}+f_{\bar{a}_1\bar{a}_2{\bf a'}},\\
 h_{(b_1\oplus b_2){\bf b}'}
&=f_{b_1b_2{\bf b}'}+f_{\bar{b}_1\bar{b}_2{\bf b'}}.
\end{aligned}
 \end{equation}
 This implies that $h_{(a_1\oplus a_2){\bf a}'}=2$
and $h_{(b_1\oplus b_2){\bf b}'}=1$. Thus
$h \notin\mathscr{A}$. This is a contradiction.

 \item If $\bar{a}_1\bar{a}_2{\bf a'}\notin {\rm supp}(f)$
 and $\bar{b}_1\bar{b}_2{\bf b}'\in {\rm supp}(f)$,
 then (\ref{lemma-0-1-reduce-1}) implies that
 $h_{(a_1\oplus a_2){\bf a}'}=1$ and $h_{(b_1\oplus b_2){\bf b}'}=2$. Thus
 $h \notin\mathscr{A}$. This is a contradiction.
\end{itemize}
\end{proof}

The next lemma shows how to reduce the arity of a
non-affine signature with the help of an additional binary
signature $[1, 0, -1]$ in
$\operatorname{Pl-Holant}(\widehat{\mathcal{EQ}}, [1, 0, -1],
 \widehat{\mathcal{F}})$,
when all signatures in $\widehat{\mathcal{F}}$ take $\{0, 1, -1\}$ values.

\begin{lemma}\label{0,1,-1-valued}
Suppose all signatures in $\widehat{\mathcal{F}}$ take values in $\{0, 1, -1\}$
and satisfy the Parity Condition.
If $\widehat{\mathcal{F}}$ contains a signature $f \not \in \mathscr{A}$
of arity $n \ge 3$,
then there exists a signature $g\notin\mathscr{A}$ of arity $<n$,
such that
\[\operatorname{Pl-Holant}(\widehat{\mathcal{EQ}}, g, [1, 0, -1], \widehat{\mathcal{F}})
\le_{\rm T}
\operatorname{Pl-Holant}(\widehat{\mathcal{EQ}}, [1, 0, -1], \widehat{\mathcal{F}}).\]
Furthermore, if $f$ satisfies the even Parity Condition, so does $g$.
\end{lemma}
\begin{proof}
This proof is a bit more involved;
it is also by induction on $n$, and uses the Tableau Calculus.

If there exists $i\in[n]$ such that $f^{x_i=0} \notin\mathscr{A}$,
 then we are done since we have $[1, 0] \in \widehat{\mathcal{EQ}}$.
In the following, we assume that $f^{x_i=0} \in \mathscr{A}$
for $1\leq i\leq n$.
By Lemma~\ref{[0,1]-EQ-hat-wight-0-neq-0}, we may assume that $f_{00\cdots 0}=1$ and $f$ satisfies the even Parity Condition.

Firstly, we claim that if supp$(f)$ is not
an  affine subspace, then we can construct a signature that is not in $\mathscr{A}$ and has arity $\leq n-1$.
Suppose supp$(f)$ is not an  affine subspace,
then it is not a linear subspace.
But supp$(f^{x_i=0})$ is a linear subspace of $\mathbb{Z}_2^{n-1}$
 since $f^{x_i=0} \in \mathscr{A}$  and
$f^{x_i=0}_{0\cdots 0}=1$.
By Lemma~\ref{argue-a1-0-b1-1},
 there exist
 ${\bf a}=a_1a_2\cdots a_n, {\bf b}=b_1b_2\cdots b_n,$
such that  ${\bf a}, {\bf b}\in{\rm supp}(f)$, ${\bf c}={\bf a}\oplus {\bf b}=c_1c_2\cdots c_n\notin{\rm supp}(f)$
and there exists $i\in[n]$ such that $a_i \not = b_i$.
Without loss of generality, we assume that $a_1=0, b_1=1$.
It follows that $c_1=1$.
Let ${\bf a}'=a_3\cdots a_n, {\bf b}'=b_3\cdots b_n, {\bf c}'=c_3\cdots c_n$.

Connecting the first variable of $[1, 0, -1]$ to the first variable of $f$, the gadget gives the signature
\[\bar{f}(x_1, x_2, \ldots, x_n)=\displaystyle\sum_{x'_1\in\{0, 1\}}[1, 0, -1](x'_1, x_1)f(x'_1, x_2, \ldots, x_n).\]
Then
 \[\bar{f}(x_1, x_2, \ldots, x_n)=(-1)^{x_1}f(x_1, x_2, \ldots, x_n).\]
Moreover, by connecting the variables $x_2, x_1$ of $[1, 0, 1, 0]$ to the  variable $x_1, x_2$ of $f$ respectively,
the planar gadget gives the signature
\[h(x', x_3, x_4, \ldots, x_n)=\displaystyle\sum_{x_1, x_2\in\{0, 1\}}[1, 0, 1, 0](x_2, x_1, x')f(x_1, x_2, \ldots, x_n).\]
We have
\[h(x_1\oplus x_2, x_3, \ldots, x_n)=f(x_1, x_2, x_3, \ldots, x_n)+f(\bar{x_1}, \bar{x_2}, x_3, \ldots, x_n).\]
Similarly,  by connecting the variables $x_2, x_1$ of $[1, 0, 1, 0]$ to the variable $x_1, x_2$ of $\bar{f}$ respectively,
the planar gadget gives the signature
 \[\bar{h}(x', x_3, x_4, \ldots, x_n)=\displaystyle\sum_{x_1, x_2\in\{0, 1\}}[1, 0, 1, 0](x_2, x_1, x')\bar{f}(x_1, x_2, \ldots, x_n),\]
and we have
\[\bar{h}(x_1\oplus x_2, x_3, \ldots, x_n)=(-1)^{x_1}f(x_1, x_2, x_3, \ldots, x_n)+(-1)^{ \bar{x}_1}f(\bar{x}_1, \bar{x}_2, x_3, \ldots, x_n).\]
If at least one of $\{h, \bar{h} \}$ is not affine, then we are done since
both  $h$ and $\bar{h}$ have arity $n-1$.
In the following, we assume that both  $h$ and $\bar{h}$ are affine.

Next comes the Tableau Calculus.

\begin{itemize}
\item For $\bar{a}_1\bar{a}_2{\bf a}'\in {\rm supp}(f)$, $\bar{b}_1\bar{b}_2{\bf b}'\in {\rm supp}(f)$,
note that $(0, 0, \ldots, 0)\in{\rm supp}(f)$ and $f^{x_1=0} \in \mathscr{A}$.
Thus
 supp$(f^{x_1=0})$ is a linear
 subspace of $\mathbb{Z}_2^{n-1}$.
 By
 \[
{\begin{array}{llllllllllllllll}
        & a_2{\bf a}'& \in {\rm supp}(f^{x_i=0}) \\
 \oplus & \bar{b}_2{\bf b}' &  \in {\rm supp}(f^{x_i=0})\\
\hline
        & \bar{c}_2{\bf c}' &  \\
\end{array} }
\]
we have $\bar{c}_2{\bf c}'\in{\rm supp}(f^{x_i=0})$.
This means that
$f_{\bar{c}_1\bar{c}_2{\bf c}'}\neq 0$ since $\bar{c}_1=0$.
Note that
\begin{align*}
h_{(c_1\oplus c_2){\bf c}'}
&=\hspace{.12in}f_{c_1c_2\bf{c}'}+f_{\bar{c}_1\bar{c}_2\bf{c}'},\\
\bar{h}_{(c_1\oplus c_2){\bf c}'}
&=-f_{c_1c_2\bf{c}'}+f_{\bar{c}_1\bar{c}_2\bf{c}'}.
\end{align*}
Thus $|h_{(c_1\oplus c_2){\bf c}'}|=|\bar{h}_{(c_1\oplus c_2){\bf c}'}|=1$
since $f_{c_1c_2\bf{c}'}=0$ and $f_{\bar{c}_1\bar{c}_2\bf{c}'}=\pm 1$.
Moreover, by
\begin{equation}\label{system-as2by2-sec4}
\begin{aligned}
h_{(a_1\oplus a_2){\bf a}'}
&=f_{a_1a_2\bf{a}'}+f_{\bar{a}_1\bar{a}_2\bf{a}'},\\
\bar{h}_{(a_1\oplus a_2){\bf a}'}
&=f_{a_1a_2\bf{a}'}-f_{\bar{a}_1\bar{a}_2\bf{c}'},
\end{aligned}
\end{equation}
if both of $h_{(a_1\oplus a_2){\bf a}'}$ and $\bar{h}_{(a_1\oplus a_2){\bf a}'}$ are zero, then
$f_{a_1a_2\bf{a}'}=0$, by treating (\ref{system-as2by2-sec4})
as a linear equation system. This contradicts that ${\bf a}\in{\rm supp}(f)$.
Thus
 we have $|h_{(a_1\oplus a_2){\bf a}'}|\neq 0$ or $|\bar{h}_{(a_1\oplus a_2){\bf a}'}|\neq 0$.
 This implies that one of $|h_{(a_1\oplus a_2){\bf a}'}|$ or $|\bar{h}_{(a_1\oplus a_2){\bf a}'}|$ is 2
 since $f_{a_1a_2\bf{a}'}, f_{\bar{a}_1\bar{a}_2\bf{a}'}\in\{1, -1\}$. So
$h$ or $\bar{h}$ is not affine, because at least one of them
has nonzero values of unequal norms.
This is a contradiction.

\item For $\bar{a}_1\bar{a}_2{\bf a}'\notin {\rm supp}(f)$, $\bar{b}_1\bar{b}_2{\bf b}'\notin {\rm supp}(f)$,
by treating the following as a linear system
\begin{equation}\label{system-as2by2-sec4-another-bullet}
\begin{aligned}
h_{00\cdots 0}
&=f_{000\cdots 0}+f_{110\cdots 0},\\
\bar{h}_{00\cdots 0}
&=f_{000\cdots 0}-f_{110\cdots 0},
\end{aligned}
\end{equation}
we have $h_{00\cdots 0}\neq 0$ or $\bar{h}_{00\cdots 0}\neq 0$
 since $f_{000\cdots 0}\neq 0$.
Without loss of generality, we assume that $h_{00\cdots 0}\neq 0$.
The same argument can be applied to $\bar{h}$ if $\bar{h}_{00\cdots 0}\neq 0$.
Then supp$(h)$ is a linear subspace.
By \[h_{(a_1\oplus a_2){\bf a}'}=f_{a_1 a_2{\bf a}'}+f_{\bar{a}_1 \bar{a}_2{\bf a}'},\]
\[h_{(b_1\oplus b_2){\bf b}'}=f_{b_1 b_2{\bf b}'}+f_{\bar{b}_1 \bar{b}_2{\bf b}'},\]
we have
$|h_{(a_1\oplus a_2){\rm a}'}|=1,$ $|h_{(b_1\oplus b_2){\bf b}'}|=1$
since $f_{\bar{a}_1 \bar{a}_2{\bf a}'}=f_{\bar{b}_1 \bar{b}_2{\bf b}'}=0$
and $f_{a_1 a_2{\bf a}'}, f_{b_1 b_2{\bf b}'}\in\{1, -1\}$.
This implies that $(a_1\oplus a_2){\bf a}', (b_1\oplus b_2){\bf b}'\in{\rm supp}(h)$.
Then by
\[
{\begin{array}{llllllllllllllll}
        & (a_1\oplus a_2){\bf a}' & \in{\rm supp}(h) \\
 \oplus & (b_1\oplus b_2){\bf b}' & \in{\rm supp}(h) \\
\hline
        & (c_1\oplus c_2){\bf c}' & \\
\end{array} }
\]
we have $(c_1\oplus c_2){\bf c}' \in{\rm supp}(h)$.
This implies that
\[h_{(c_1\oplus c_2){\bf c}'}=f_{c_1 c_2{\bf c}'}+f_{\bar{c}_1 \bar{c}_2{\bf c}'}\neq 0.\]
Thus $f_{\bar{c}_1 \bar{c}_2{\bf c}'}\neq 0$ since $f_{(c_1 c_2){\bf c}'}=0$.
So $\bar{c}_1 \bar{c}_2{\bf c}'\in{\rm supp}(f)$.
Therefore, $\bar{c}_2{\bf c}'\in{\rm supp}(f^{x_1=0})$ since $\bar{c}_1=0$.
By
\[
{\begin{array}{llllllllllllllll}
        & a_2{\bf a}' & \in{\rm supp}(f^{x_1=0}) \\
 \oplus & \bar{c}_2{\bf c}' & \in{\rm supp}(f^{x_1=0}) \\
\hline
        & \bar{b}_2{\bf b}' &  \\
\end{array} }
\]
we have $\bar{b}_2{\bf b}'\in{\rm supp}(f^{x_1=0})$.
Thus $\bar{b}_1\bar{b}_2{\bf b}'\in{\rm supp}(f)$ as $\bar{b}_1=0$.
This is a contradiction.

\item If $\bar{a}_1\bar{a}_2{\bf a}'\in {\rm supp}(f)$,
$\bar{b}_1\bar{b}_2{\bf b}'\notin {\rm supp}(f)$, by
\begin{align*}
h_{(b_1\oplus b_2){\bf b}'}
&=\hspace{.12in}f_{b_1 b_2{\bf b}'}+f_{\bar{b}_1 \bar{b}_2{\bf b}'},\\
\bar{h}_{(b_1\oplus b_2){\bf b}'}
&=-f_{b_1 b_2{\bf b}'}+f_{\bar{b}_1 \bar{b}_2{\bf b}'},
\end{align*}
 we have
  $|h_{(b_1\oplus b_2){\rm b}'}|=|\bar{h}_{(b_1\oplus b_2){\bf b}'}|=1$ since
 $f_{b_1 b_2{\bf b}'}=\pm 1$ and $f_{\bar{b}_1 \bar{b}_2{\bf b}'}=0$.
 Then by
\begin{align*}
h_{(a_1\oplus a_2){\bf a}'}
&=f_{a_1 a_2{\bf a}'}+f_{\bar{a}_1 \bar{a}_2{\bf a}'},\\
 \bar{h}_{(a_1\oplus a_2){\bf a}'}
&=f_{a_1 a_2{\bf a}'}-f_{\bar{a}_1 \bar{a}_2{\bf a}'},
\end{align*}
  and $f_{a_1 a_2{\bf a}'}, f_{\bar{a}_1 \bar{a}_2{\bf a}'}\in\{1, -1\}$,
 we have $|h_{(a_1\oplus a_2){\bf a'}}|=2$ or $|\bar{h}_{(a_1\oplus a_2){\bf a'}}|=2$.
 This implies that $h$ or $\bar{h}$ is not affine.
 This is a contradiction.

 \item If $\bar{a}_1\bar{a}_2{\bf a}'\not \in {\rm supp}(f)$, $\bar{b}_1\bar{b}_2{\bf b}'\in {\rm supp}(f)$,
 the proof is symmetric by reversing the order of ${\bf a}$ and ${\bf b}$ in the previous item.
 \end{itemize}

 Now we can assume that supp$(f)$ is affine with dimension $k$.
 Let $X=\{y_1, y_2, \ldots, y_k\}$ be a set of free variables,
 where $X\subseteq\{x_1, x_2, \ldots, x_n\}$, and
let $\underline{f}$ be the compressed signature of $f$ for $X$.
 If $k\leq 2$, then $\underline{f}$ is affine
 by Lemma~\ref{binary-affine-compressed function}. So $f$ is affine
by Corollary~\ref{f-affine-iff-f*-affine}. This is a contradiction.

 In the following, we assume that $k\geq 3$.
By Lemma~\ref{multilinear-polynomial-unique-affine-definition},
and since  $\underline{f}$ takes   values in  $\{1, -1\}$,
 there exists a unique multilinear polynomial $Q(y_1, y_2, \ldots, y_k)\in\mathbb{Z}_2[X]$
 such that $\underline{f}(y_1, y_2, \ldots, y_k)=(-1)^{Q(y_1, y_2, \ldots, y_k)}$.
 $\underline{f}$ is affine iff $Q(y_1, y_2, \ldots, y_k)$ is a quadratic multilinear polynomial.

 If $k \ge 4$ and there exists a  term $y_{i_1}y_{i_2}\cdots y_{i_s}$
with nonzero coefficient, where $3\leq s< k$, then there exists some $y_j
\not \in \{y_{i_1}, y_{i_2}, \ldots, y_{i_s}\}$,
such that
\[\underline{f}^{y_{j}=0}(y_{1},  \ldots, \widehat{y_{j}},  \ldots, y_{k})
=(-1)^{Q'(y_{1},\ldots, \widehat{y_j},  \ldots, y_{k})},\]
where $Q'$ is a polynomial on $k-1$ variables, where the term
$y_{i_1}y_{i_2}\cdots y_{i_s}$ of degree $s >2$ still appears.
 This implies that
 $f^{y_j=0}$
  is not affine. This is a contradiction.
 Thus we may assume that $Q(y_1, y_2, \ldots, y_k)=P(y_1, y_2, \ldots, y_k)+ay_1y_2\cdots y_k$,
 where $P(y_1, y_2, \ldots, y_k)\in\mathbb{Z}_2[X]$ is a
multilinear polynomial of total degree at most 2, and $a\in\mathbb{Z}_2$.
Note that this statement is also vacuously true if $k=3$.
 If $a=0$, then $f$ is affine. This is a contradiction.
 Otherwise, $Q(y_1, y_2, \ldots, y_k)=P(y_1, y_2, \ldots, y_k)+y_1y_2\cdots y_k$.
 Moreover, 
 by connecting
 the first variable of $[1, 0, -1]$ to $y_i$ of $f$, the gadget gives the signature
 $f'$ such that \[f'(x_1, x_2, \ldots, x_n)=(-1)^{y_i}f(x_1, x_2, \ldots, x_n).\]
 This implies that $f'$ has the same support of $f$
and
 \[\underline{f'}(y_1, y_2, \ldots, y_k)=(-1)^{y_i+P(y_1, y_2, \ldots, y_k)+y_1y_2\cdots y_k},\]
 where $\underline{f'}$ is the compressed signature of $f'$ for $X$.
Thus $f' \not \in \mathscr{A}$.
 This implies that we can add a linear term to $P(y_1, y_2, \ldots, y_k)$ freely.

 In the following, we connect all variables of $f$ except for $y_1$ to $n-1$ variables of
 $\frac{1}{2}\{[1, 1]^{\otimes n}+[1, -1]^{\otimes n}\}=[1, 0, 1, \ldots,
0 ~(\operatorname{or} 1)]\in\widehat{\mathcal{EQ}}$ to get the binary signature $\hat{f}=[\hat{f}_{00}, 0, \hat{f}_{11}]$.
 Note that
 \begin{equation}\label{1-0--1-reduction-1}
\begin{aligned}
 \hat{f}_{00}
&=&\displaystyle\sum_{y_{2}, y_{3}, \ldots, y_k\in\{0, 1\}}\underline{f}^{y_1=0}(y_{2}, y_{3}, \ldots, y_k),\\
 \hat{f}_{11}
&=&\displaystyle\sum_{y_{2}, y_{3}, \ldots, y_k\in\{0, 1\}}\underline{f}^{y_1=1}(y_{2}, y_{3}, \ldots, y_k).
\end{aligned}
 \end{equation}

 Firstly, we consider the special case that  the coefficient of $y_iy_j$ in  $P(y_1, y_2, \ldots, y_k)$ is nonzero
 for all $1\leq i<j\leq k$.
  \begin{itemize}
\item  If $k=3$, we may assume that $P(y_1, y_2, y_3)
=y_1+y_2+y_3+y_1y_2+y_1y_3+y_2y_3$
since we can add linear terms to $P(y_1, y_2, \ldots, y_k)$ at will.
Then we have
$\underline{f}(y_1, y_2, y_3)=(-1)^{P(y_1, y_2, y_3)+y_1y_2y_3}$.
The polynomial $P(y_1, y_2, y_3)+y_1y_2y_3 = 1 + (1+y_1)(1+y_2)(1+y_3) \in
\mathbb{Z}_2[y_1, y_2, y_3]$
corresponds to the {\sc Or} function on 3 bits, $y_1 \vee y_2 \vee y_3$.
Thus
 $M_{y_1, y_2y_3}(\underline{f})=\left[\begin{smallmatrix}
~1 & -1 & -1 & -1\\
-1 & -1 & -1 & -1\\
\end{smallmatrix}\right]$. Thus $\hat{f}=[-2, 0, -4]$, which
has nonzero terms of unequal norms, thus  not in $\mathscr{A}$ and we are done.

\item For $k\geq 4$,
we may assume that
  $P(y_1, y_2, \ldots, y_k)$ has no linear terms since we can add linear terms to $P(y_1, y_2, \ldots, y_k)$ freely.
  Since $P$ has all terms $y_iy_j$,
 both $\underline{f}^{y_1=0}$ and $\underline{f}^{y_1=1}$ are symmetric signatures.
  For $\underline{f}^{y_1=0}$,
the entry of Hamming weight $\ell$ is $(\underline{f}^{y_1=0})_{\ell}
=(-1)^{\frac{\ell(\ell-1)}{2}}$ for $0\leq \ell\leq k-1$.
  For $\underline{f}^{y_1=1}$,
we have $(\underline{f}^{y_1=1})_{\ell}=(-1)^{\frac{\ell(\ell+1)}{2}}$ for $0\leq \ell\leq k-2$
  and $(\underline{f}^{y_1=1})_{k-1}=(-1)^{\frac{k(k-1)}{2}+1}$.
  This implies that
\begin{eqnarray*}
 \underline{f}^{y_1=0}
&=&[1, 1, -1, -1, \ldots, (-1)^{\frac{(k-1)(k-2)}{2}}]\\
&=&\frac{1}{1+{\frak i}}\left\{[1, {\frak i}]^{\otimes k-1}
 +{\frak i}[1, -{\frak i}]^{\otimes k-1}\right\},\\
\underline{f}^{y_1=1}
&=&[1, -1, -1, 1, \ldots, (-1)^{\frac{k(k-1)}{2}}]-2(-1)^{\frac{k(k-1)}{2}}[0, 1]^{\otimes k-1}\\
&=&\frac{1}{1-{\frak i}}\left\{[1, {\frak i}]^{\otimes k-1}-{\frak i}[1, -{\frak i}]^{\otimes k-1}\right\}-2(-1)^{\frac{k(k-1)}{2}}[0, 1]^{\otimes k-1}.
\end{eqnarray*}
Thus
\begin{eqnarray*}
\hat{f}_{00}
&=&\displaystyle\sum_{\beta\in\{0, 1\}^{k-1}}(\underline{f}^{y_1=0})_{\beta}\\
&=&\frac{1}{1+{\frak i}}\displaystyle\sum_{w=0}^{k-1}{k-1 \choose w}[{\frak i}^w+{\frak i}(-{\frak i})^w]\\
&=&\frac{1}{1+{\frak i}}[(1+{\frak i})^{k-1}+{\frak i}(1-{\frak i})^{k-1}]\\
&=&(1+{\frak i})^{k-2}+(1-{\frak i})^{k-2}\\
&=&2^{\frac{k}{2}}\cos((k-2)\pi/{4}),\\
\hat{f}_{11}
&=&\displaystyle\sum_{\beta\in\{0, 1\}^{k-1}}(\underline{f}^{y_1=1})_{\beta}\\
&=&\frac{1}{1-{\frak i}}\displaystyle\sum_{w=0}^{k-1}{k-1 \choose w}[{\frak i}^w-{\frak i}(-{\frak i})^w]-2(-1)^{\frac{k(k-1)}{2}}\\
&=&\frac{1}{1-{\frak i}}[(1+{\frak i})^{k-1}-{\frak i}(1-{\frak i})^{k-1}]-2(-1)^{\frac{k(k-1)}{2}}\\
&=&-2^{\frac{k}{2}}\sin((k-2)\pi/{4})-2(-1)^{\frac{k(k-1)}{2}}.
\end{eqnarray*}

For $k\equiv 1\mod 2$,
$|\hat{f}_{00}|=2^{\frac{k-1}{2}}$, and
 $|\hat{f}_{11}|=2^{\frac{k-1}{2}}\pm 2$
(since $k\geq 4$), we have $\hat{f}_{11}\hat{f}_{00}\neq 0$ and
$|\hat{f}_{11}|\neq |\hat{f}_{00}|$.
Thus $\hat{f}\notin\mathscr{A}$ and we are done.

For $k\equiv 0\mod 4$, $\hat{f}_{00}=0$, $|\hat{f}_{11}|=2^{\frac{k}{2}}\pm 2\neq 0$.
This implies that we have $\hat{f}=\hat{f}_{11}[0, 1]^{\otimes 2}$.
By $[1, 0, -1]$, $[0, 1]^{\otimes 2}$ and $f\notin\mathscr{A}$,
we can get a binary signature that is not in $\mathscr{A}$
 by Lemma~\ref{[1,0]-[0,1]-[1,0,-1]-EQ-hat-affine-reduction}. Thus we are done.

For $k\equiv 2\mod 4$, $|\hat{f}_{00}|=2^{\frac{k}{2}}\geq 4$ since $k\geq 4$,
and  $|\hat{f}_{11}|=2$,
so $\hat{f}\notin\mathscr{A}$ and we are done.

\end{itemize}

\vspace{.1in}

 Now we assume that there exist $i\neq j\in[k]$ such that the coefficient of $y_iy_j$ is 0 in $P(y_1, y_2, \ldots, y_k)$.
 For notational simplicity, without loss of generality we assume that $i=k-1, j=k$.
 Then we can assume that (with the linear term $y_{k-1}$ and $y_k$ removed
if needed)
 \[P(y_1, y_2, \ldots, y_k)=y_1(L_1+\epsilon_1)+y_2(L_2+\epsilon_2)+\ldots +y_{k-2}(L_{k-2}+\varepsilon_{k-2}),\]
 where
 \[L_1=\displaystyle\sum_{i=2}^ka_{1i}y_1,~~~~L_2=\displaystyle\sum_{i=3}^ka_{2i}y_i,~~~~\ldots,~~~~L_{k-2}=\displaystyle\sum_{i=k-1}^ka_{(k-2)i}y_i\]
 and $a_{ji}\in\mathbb{Z}_2$ are fixed, but we can choose  $\epsilon_i\in\mathbb{Z}_2$ freely since we can add linear terms freely.
 Let $F_{(0)}=\underline{f}$, $F_{(i)}=\underline{f}^{y_1=0, y_2=0, \ldots, y_i=0}$ for $i\in[k-2]$.
 We claim that there exist $\epsilon_1, \epsilon_2, \ldots, \epsilon_{k-2}
\in \mathbb{Z}_2$
 such that
 \[\displaystyle\sum_{y_{i+1}, y_{i+2}, \ldots, y_k\in\{0, 1\}}F_{(i)}(y_{i+1}, y_{i+2}, \ldots, y_k)\geq 4,\]
 for all $1\leq i\leq k-2$. We prove this claim by induction.
 The base case is for $F_{(k-2)}$. Note that
 $P(0, \ldots, 0, y_{k-1}, y_k)$ is identically 0. Thus
  \[\displaystyle\sum_{y_{k-1}, y_k\in\{0, 1\}}F_{(k-2)}(y_{k-1}, y_k)=4.\]
 By induction, we may assume that \[\displaystyle\sum_{y_{i+1}, y_{i+2}, \ldots, y_k\in\{0, 1\}}F_{(i)}(y_{i+1}, y_{i+2}, \ldots, y_k)\geq 4\]
 and prove that
 \[\displaystyle\sum_{y_{i}, y_{i+1}, \ldots, y_k\in\{0, 1\}}F_{(i-1)}(y_{i}, y_{i+1}, \ldots, y_k)\geq 4\]
for $i\leq k-2$.
 Note that
 \begin{equation}\label{affine-reduction-0-1--1-pin-0}
 \begin{split}
 &F_{(i-1)}(1, y_{i+1}, y_{i+2}, \ldots, y_k)=
 (-1)^{(L_i+\epsilon_{i})+y_{i+1}(L_{i+1}+\epsilon_{i+1})
 +\cdots +y_{k-2}(L_{k-2}+\epsilon_{k-2})},\\
 &F_{(i-1)}(0, y_{i+1}, y_{i+2}, \ldots, y_k)=
 (-1)^{y_{i+1}(L_{i+1}+\epsilon_{i+1})
 +\cdots +y_{k-2}(L_{k-2}+\epsilon_{k-2})}.
 \end{split}
 \end{equation}
By inductive hypothesis,
 \[\displaystyle\sum_{y_{i+1}, \ldots, y_k\in\{0, 1\}}F_{(i-1)}(0, y_{i+1}, \ldots, y_k)
 =
 \displaystyle\sum_{y_{i+1}, \ldots, y_k\in\{0, 1\}}F_{(i)}(y_{i+1}, \ldots, y_k)\geq 4.\]
 If $L_i$ is identically 0, then we set $\epsilon_{i}=0$.
 It follows that
 \[\displaystyle\sum_{y_{i+1}, \ldots, y_k\in\{0, 1\}}F_{(i-1)}(1, y_{i+1}, \ldots, y_k)=
 \displaystyle\sum_{y_{i+1}, \ldots, y_k\in\{0, 1\}}F_{(i-1)}(0, y_{i+1}, \ldots, y_k).\]
 Thus we have
 \[\displaystyle\sum_{y_{i}, y_{i+1}, \ldots, y_k\in\{0, 1\}}F_{(i-1)}(y_{i}, y_{i+1}, \ldots, y_k)=
 2\displaystyle\sum_{y_{i+1}, \ldots, y_k\in\{0, 1\}}F_{(i-1)}(0, y_{i+1}, \ldots, y_k)\geq 8\]
 and we are done.

 Otherwise, $L_i=0$ defines a subspace $V$ of $\mathbb{Z}_2^{k-i}$ which has dimension
 $k-i-1$. Let $V'=\mathbb{Z}_2^{k-i}\setminus V$, then $V'$ is a affine space defined by $L_i=1$
 which has dimension $k-i-1$.

 Let
 \[a=\displaystyle\sum_{y_{i+1} y_{i+2} \cdots y_k\in V}F_{(i-1)}(0, y_{i+1}, y_{i+2}, \ldots, y_k),\]
 \[b=\displaystyle\sum_{y_{i+1} y_{i+2} \cdots y_k\in V'}F_{(i-1)}(0, y_{i+1}, y_{i+2}, \ldots, y_k),\]
 \[a'=\displaystyle\sum_{y_{i+1} y_{i+2} \cdots y_k\in V}F_{(i-1)}(1, y_{i+1}, y_{i+2}, \ldots, y_k),\]
 and
 \[b'=\displaystyle\sum_{y_{i+1} y_{i+2} \cdots y_k\in V'}F_{(i-1)}(1, y_{i+1}, y_{i+2}, \ldots, y_k).\]
 Then
 \[\displaystyle\sum_{y_{i}, y_{i+1}, \ldots, y_k\in\{0, 1\}}F_{(i-1)}(y_{i}, y_{i+1}, \ldots, y_k)=a+b+a'+b'.\]
  By induction, we have $a+b\geq 4$. Thus $a\geq 2$ or $b\geq 2$.
 If $a\geq 2$, we choose $\epsilon_i=0$, then $a=a'$, $b=-b'$ by
 (\ref{affine-reduction-0-1--1-pin-0}).
 Thus $a+a'+b+b'=2a\geq 4$.
 If $b\geq 2$, we choose $\epsilon_i=1$, then $a=-a'$, $b=b'$
 by (\ref{affine-reduction-0-1--1-pin-0}).
 Thus $a+a'+b+b'=2b\geq 4$.
 This finishes the proof of the claim.

 The claim shows that
 \[\hat{f}_{00}=\displaystyle\sum_{y_{2}, y_{3}, \ldots, y_k\in\{0, 1\}}\underline{f}^{y_1=0}(y_{2}, y_{3}, \ldots, y_k)
 =\displaystyle\sum_{y_{2}, y_{3}, \ldots, y_k\in\{0, 1\}}F_{(1)}(y_{2}, y_{3}, \ldots, y_k)\geq 4.\]
Let $g$ be the $n$-ary signature with the same support as $f$ (thus
satisfies
the even Parity Condition) and on its support
\[g(x_1, x_2, \ldots, x_n)=(-1)^{P(y_1, y_2, \ldots, y_k)},\]
then $\underline{f}_{\beta}=\underline{g}_{\beta}$ for any $\beta\in\{0, 1\}^{k}$ other than $\beta=11\cdots 1$.
For $\beta=11\cdots 1$, \[\underline{f}_{\beta}=(-1)^{P(1, 1, \ldots, 1)+1\cdot 1\cdots 1}=-(-1)^{P(1, 1, \ldots, 1)}=-\underline{g}_{\beta}.\]
This implies that
\begin{equation}\label{1-0--1-reduction-3}
\underline{f}=\underline{g}\pm 2[0, 1]^{\otimes k}.
\end{equation}
(Note that we do not really construct $g$. We just use $g$ to argue that $\hat{f}\notin\mathscr{A}$.)
Since both $\frac{1}{2}\{[1, 1]^{\otimes n}+[1, -1]^{\otimes n}\}$ and $g$ are affine signatures,
the following construction would produce an affine signature:
Connect all variables of $g$ other than $y_1$ to $n-1$ variables of
 $\frac{1}{2}\{[1, 1]^{\otimes n}+[1, -1]^{\otimes n}\}=[1, 0, 1, \ldots, 0 ~(\operatorname{or} 1)]$. This construction  gives a binary signature
  $\hat{g}=[\hat{g}_{00}, 0, \hat{g}_{11}]$.
(By the even Parity Condition, the weight 1 entry must be 0.)
 Note that \[\hat{g}_{00}=\displaystyle\sum_{y_{2}, y_{3}, \ldots, y_k\in\{0, 1\}}\underline{g}^{y_1=0}(y_{2}, y_{3}, \ldots, y_k),\]
 \[\hat{g}_{11}=\displaystyle\sum_{y_{2}, y_{3}, \ldots, y_k\in\{0, 1\}}\underline{g}^{y_1=1}(y_{2}, y_{3}, \ldots, y_k).\]
 Thus by (\ref{1-0--1-reduction-1}) and (\ref{1-0--1-reduction-3}), we have
 \[\hat{f}_{00}=\hat{g}_{00},~~~~ \hat{f}_{11}=\hat{g}_{11}\pm 2.\]
  Since $\hat{g}$ is an affine signature, we must have
  either $\hat{g}_{00} =0$ or $\hat{g}_{11}=0$ or $\hat{g}^4_{00}=\hat{g}^4_{11}$.
  Since we have $\hat{g}_{00}=\hat{f}_{00}\neq 0$ and both $\hat{g}_{00}$ and $\hat{g}_{11}$
are real numbers,
  we must have $\hat{g}_{11}=0$ or $\hat{g}_{11}=\pm \hat{g}_{00}$.
Recall that $\hat{f}_{00} \ge 4$.
If  $\hat{g}_{11}=0$ then $\hat{f}_{11} = \pm 2$ has
a different nonzero norm than $\hat{f}_{00}$.
If $\hat{g}_{11}=\pm \hat{g}_{00}$, then $\hat{g}_{11}=\pm \hat{f}_{00}$
has norm at least 4, and thus
 $\hat{f}_{11} = \hat{g}_{11} \pm 2$ has norm $|\hat{g}_{11}| \pm 2
= |\hat{f}_{00}|  \pm 2$.
And so in this case $\hat{f}_{11}$ also  has a different
nonzero norm than $\hat{f}_{00}$.
In each case, $|\hat{f}_{00}|\neq |\hat{f}_{11}|$ and $\hat{f}_{00}\hat{f}_{11}\neq 0$.
This implies that $\hat{f}\notin\mathscr{A}$ and we are done.
\end{proof}

Now we give the main dichotomy theorem of this section.
By (\ref{eqn:prelim:PlCSPd_equiv_Holant}), we have
\begin{equation*}
 \PlCSP^2(\widehat{\mathcal{EQ}}, \widehat{\mathcal{F}})
\equiv_T \PlHolant(\mathcal{EQ}_2, \widehat{\mathcal{EQ}}, \widehat{\mathcal{F}}).
\end{equation*}
 We will prove a dichotomy theorem for $\PlCSP^2(\widehat{\mathcal{EQ}}, \widehat{\mathcal{F}})$
when every signature in $\widehat{\mathcal{F}}$ satisfies the Parity
 Condition.

\begin{theorem}\label{dichotomy-csp-2}
If all signatures in $\widehat{\mathcal{F}}$ satisfy
the Parity  Condition, then
$\PlHolant(\mathcal{EQ}_2, \widehat{\mathcal{EQ}}, \widehat{\mathcal{F}})$
(equivalently
$\PlCSP^2(\widehat{\mathcal{EQ}}, \widehat{\mathcal{F}})$)
is \#$\operatorname{P}$-hard, or $\widehat{\mathcal{F}}\subseteq\mathscr{A}$,
in which case the problem is in P.
\end{theorem}
\begin{proof}
If $\widehat{\mathcal{F}} \subseteq\mathscr{A}$, then
$\PlHolant(\mathcal{EQ}_2, \widehat{\mathcal{EQ}}, \widehat{\mathcal{F}})$
 is tractable in P by Theorem~\ref{non-planar-csp-dichotomy}', since
$\widehat{\mathcal{EQ}} \subset \mathscr{A}$ as well.

If $\widehat{\mathcal{F}}\nsubseteq\mathscr{A}$, then there exists an $n$-ary signature $f\in\widehat{\mathcal{F}} \setminus \mathscr{A}$.
By Lemma~\ref{[0,1]-EQ-hat-wight-0-neq-0}, we can assume that $f_{00\cdots 0}=1$ and $f$ satisfies the even Parity Condition.
If $f$ is a unary signature then $f = [1,0] \in \mathscr{A}$,
a contradiction.
If $f$ has arity 2, then $f$ must be symmetric, and has the form
$f=[1, 0, x]$, where $x^4\neq 0, 1$, by Proposition~\ref{A-has-same-norm-etc}.
Note that $\widehat{\mathcal{EQ}}$ contains the
symmetric signatures $[1, 0]$ and $[1, 0, 1, 0]$,
and
 $[1, 0]\notin\widehat{\mathscr{M}} \cup \widehat{\mathscr{M}}^{\dagger}$,
$[1, 0, 1, 0]\notin\mathscr{P} \cup \mathscr{A}^{\dagger}$.
We also have $f = [1, 0, x]\notin\mathscr{A}$.
By Theorem~\ref{heng-tyson-dichotomy-pl-csp2},
$\PlHolant(\mathcal{EQ}_2, [1, 0], [1, 0, 1, 0], [1, 0, x])$ is \#P-hard.
Thus $\PlHolant(\mathcal{EQ}_2, \widehat{\mathcal{EQ}}, \widehat{\mathcal{F}})$ is \#P-hard
since
\[\operatorname{Pl-Holant}(\mathcal{EQ}_2, [1, 0], [1, 0, x], [1, 0, 1, 0])
\le_{\rm T}
\operatorname{Pl-Holant}(\mathcal{EQ}_2, \widehat{\mathcal{EQ}}, \widehat{\mathcal{F}}).\]
So in the following, we assume that the arity of $f$ is $n\geq 3$.

\begin{description}
\item{(A)}
 If there exists $\alpha\in\{0, 1\}^n$ such that $f_{\alpha}^4\neq 0, 1$, then we can get $[1, 0, f_{\alpha}]$ in the following way:
firstly, using $[1, 0] \in \widehat{\mathcal{EQ}}$,
 we can get $\partial_{[1, 0]}^S(f)=(1, \ldots, f_{\alpha})$, where $S=\{k
\mid $the $k$-th bit of $\alpha$ is $0\}$.
 Note that the arity of $(1, \ldots, f_{\alpha})$ is wt$(\alpha)$ which is even and we have $(=_{{\rm wt}(\alpha)+2})\in\mathcal{EQ}_2$.
 So we have \[\partial_{(1, \ldots, f_{\alpha})}(=_{{\rm wt}(\alpha)+2})=[1, 0, f_{\alpha}] \not \in \mathscr{A}.\]
 Then by Theorem~\ref{heng-tyson-dichotomy-pl-csp2},
 $\PlHolant(\mathcal{EQ}_2, [1, 0], [1, 0, 1, 0], [1, 0, f_{\alpha}])$ is \#P-hard.
It follows that
 $\operatorname{Pl-Holant}(\mathcal{EQ}_2, \widehat{\mathcal{EQ}}, \widehat{\mathcal{F}})$ is \#P-hard.

Now we may assume that all nonzero entries of $f$ are powers of $\frak i$.

 \item {(B)}
If there exists $\alpha\in\{0, 1\}^n$ such that $f_{\alpha}=\pm {\frak i}$, then by the same way as Item (A), we have
$[1, 0, {\frak i}]$ or $[1, 0, -{\frak i}]$. In each case, we
 have $[1, 0, -1]$ and $[0, 1]^{\otimes 2}$ by Lemma~\ref{constructing-[1,0,x]-heng}.
 Then by $f$ and Lemma~\ref{[1,0]-[0,1]-[1,0,-1]-EQ-hat-affine-reduction}, we can get $[1, 0, x] \not \in \mathscr{A}$
 Then we are done by Theorem~\ref{heng-tyson-dichotomy-pl-csp2}.

Now we may assume that all nonzero entries of $f$ are $1$ or $-1$.

\item {(C)} If $f$ takes values in $\{0, 1, -1\}$ and there exists at least one $\alpha\in\{0, 1\}^n$
such that $f_{\alpha}=-1$,
then we can get  $[1, 0, -1]$ in the same way as (A).
Now we prove the lemma by induction on the arity $n\geq 3$ of $f$.

For $n=3$,
 by Lemma~\ref{0,1,-1-valued}, we can get
 a signature $g \not \in \mathscr{A}$ with arity $<3$.
 Note that $g$ also satisfies the even Parity Condition.
 If $g$ has arity 1, then $g\in\mathscr{A}$.
 This is a contradiction.
 If $g$ has arity 2, then it must be of the form
$[x,0,y]$, and $xy \not =0, (x/y)^4 \not = 1$ lest
$g \in\mathscr{A}$.
In particular $g$ is symmetric.
Then by Theorem~\ref{heng-tyson-dichotomy-pl-csp2},
$\operatorname{Pl-Holant}(\mathcal{EQ}_2, [1, 0], [1, 0, 1, 0], g)$ is \#P-hard.
Thus $\operatorname{Pl-Holant}(\mathcal{EQ}_2, \widehat{\mathcal{EQ}}, \widehat{\mathcal{F}})$
is \#P-hard.

For $n\geq 4$, by Lemma~\ref{0,1,-1-valued}, we can get
 a signature $g  \not \in \mathscr{A}$ with arity $<n$.
 If $g$ takes values in $\{0, 1, -1\}$, up to a nonzero factor, then we are done by induction.
 Otherwise, we are done by Items (A) and Item (B).

Now we may assume that $f$ takes values in $\{0, 1\}$.

 \item {(D)} For a $\{0, 1\}$-valued signature satisfying the even
Parity Condition, if it has arity $\le 2$ then it is affine.
Hence $f$ has arity$\geq 3$ since $f\notin\mathscr{A}$.
  We now  induct on the arity $n\geq 3$ of $f$ in this case.

For $n=3$,
 by Lemma~\ref{0,1-valued}, we can get
 a non-affine signature $h$ with arity $<3$.
 Note that $h$ satisfies the even Parity Condition.
 If $h$ has arity 1, then $h\in\mathscr{A}$.
 This is a contradiction.
 If $h$ has arity 2,
then $h$ has the form
  $[1, 0, x] \not \in \mathscr{A}$ up to a nonzero factor.
Then by Theorem~\ref{heng-tyson-dichotomy-pl-csp2},
$\operatorname{Pl-Holant}(\mathcal{EQ}_2, [1, 0], [1, 0, 1, 0], [1, 0, x])$ is \#P-hard.
Thus $\operatorname{Pl-Holant}(\mathcal{EQ}_2, \widehat{\mathcal{EQ}}, \widehat{\mathcal{F}})$
is \#P-hard.

For $n\geq 4$, by Lemma~\ref{0,1-valued}, we can get
 a non-affine signature $h$ with arity $<n$.
 If $h$ takes values in $\{0, 1\}$, then we are done by induction.
 Otherwise, we are done by Item (A), Item (B), and Item (C).
\end{description}
\end{proof}

%% file: 5parity.tex
\section{When $\widehat{\mathcal{F}}$ Satisfies Parity}\label{sec:F-has-parity}
In this section, we  give a dichotomy for $\operatorname{Pl-Holant}(\widehat{\mathcal{EQ}}, \widehat{\mathcal{F}})$,
where all  signatures in $\widehat{\mathcal{F}}$ satisfy
the Parity Condition.
In this case, $\widehat{\mathcal{F}}$ will involve matchgate signatures.
If $\widehat{\mathcal{F}} \subseteq \mathcal{M}$ then the problem is
tractable in P. Assume  $\widehat{\mathcal{F}} \not \subseteq \mathcal{M}$.
General matchgate signatures are governed by the Matchgate Identities.
For asymmetric signatures of high arities, these are intricate and
difficult to handle.
So we first try to reduce the arity of a non-matchgate signature.

\subsection{Arity Reduction of Non-Matchgate Signatures}
The following lemma is from \cite{jinyi-aaron}.
\begin{lemma}\label{construct-g-by-weight-0-2}
For any signature $f$ of arity $n \ge 2$ with
$f_{00\cdots 0}=1$,
there exists a matchgate signature $g$ of arity $n$
 such that $g_{00\cdots 0}=f_{00\cdots 0}=1$
and $g_{00\cdots 0\oplus e_i\oplus e_j}=f_{00\cdots 0\oplus e_i\oplus e_j}$,
where $i, j \in [n]$ and $i<j$.
\end{lemma}

\begin{theorem}\label{arity-reduction-matchgate-signature}
If all signatures in $\widehat{\mathcal{F}}$
satisfy the Parity Condition, and
$\widehat{\mathcal{F}}\nsubseteq\mathscr{M}$,
then there exists $h\notin\mathscr{M}$ of arity 4
such that
\[\operatorname{Pl-Holant}(\widehat{\mathcal{EQ}}, h, \widehat{\mathcal{F}})
\le_{\rm T}
\operatorname{Pl-Holant}(\widehat{\mathcal{EQ}}, \widehat{\mathcal{F}}).\]
\end{theorem}
\begin{proof}
Since $\widehat{\mathcal{F}}\nsubseteq\mathscr{M}$, there exists $f\in\widehat{\mathcal{F}} \setminus \mathscr{M}$.
By Lemma~\ref{matchgate-identity-for-arity-4}
 $f$ has arity $n \geq 4$ since $f$ satisfies the Parity Condition.
Moreover,
by Lemma~\ref{[0,1]-EQ-hat-wight-0-neq-0}, we can assume that $f_{00\cdots 0}=1$ and $f$ satisfies the even Parity Condition.
By Lemma~\ref{construct-g-by-weight-0-2}, there exists  $g \in \mathscr{M}$ such that
$g_{00\cdots 0}=1$ and $g_{00\cdots 0\oplus e_i\oplus e_j}=f_{00\cdots 0\oplus e_i\oplus e_j}$ for any  $i, j \in [n]$ and $i<j$.

We will prove the theorem by induction on $n$.
If $n= 4$, then we are done. Now we assume that the theorem is true for $\leq n-1$ and prove the theorem for $n\geq 5$.
If there exists $i \in [n]$ such that $f^{x_i=0}\notin\mathscr{M}$, then we are done by induction since we have $[1, 0]\in\widehat{\mathcal{EQ}}$.
Therefore, we may assume that $f^{x_i=0}\in\mathscr{M}$ for $1\leq i\leq n$.

With $f^{x_i=0}\in\mathscr{M}$ for $1\leq i\leq n$,
we claim that $f_{\alpha}=g_{\alpha}$ for any $\alpha\in\{0, 1\}^n$ with wt$(\alpha)<n$.
If wt$(\alpha)$ is odd, then $f_{\alpha}=g_{\alpha}=0$ since both  $f$ and $g$
satisfy the even Parity Condition.
If wt$(\alpha)$ is even, we prove the claim by induction on $k= $ wt$(\alpha)$.
For $k=0, 2$, $f_{\alpha}=g_{\alpha}$ by the definition of $g$.
By induction, we may assume that
 $f_{\beta}=g_{\beta}$ for any wt$(\beta)<k$.
 For wt$(\alpha)=k\geq 4$,
Let $P=\{p_1, p_2, \ldots, p_{{\rm wt}(\alpha)}\}\subseteq[n]$ be
such that $\alpha_{p_i} =1$ and all other bits of $\alpha$ are 0.
Then there exists $\ell\in[n]$ such that $\ell\notin P$ since $k<n$.
Since $f^{x_{\ell}=0}$ is a matchgate signature, by the Matchgate Identities,
we have
\begin{equation}\label{MGI-1'}
\displaystyle\sum_{j=1}^{k}(-1)^j f_{e_{p_1}\oplus e_{p_j}}f_{\alpha\oplus e_{p_1}\oplus e_{p_j}}=0,
\end{equation}
where the position vector is $P$ and the pattern is $e_{p_1}$.
Note that all entries of $f$ that appear in (\ref{MGI-1'})
are indeed entries of $f^{x_{\ell}=0}$.

The first term of (\ref{MGI-1'}) is $-f_{\alpha}$ since $f_{00\cdots 0}=1$.
Thus
 \begin{equation}\label{MGI-1}
 f_{\alpha}=\displaystyle\sum_{j=2}^{k}(-1)^j f_{e_{p_1}\oplus e_{p_j}}f_{\alpha\oplus e_{p_1}\oplus e_{p_j}}.
 \end{equation}
Similarly,  since $g$ is a matchgate signature and $g_{00\cdots 0}=f_{00\cdots 0}=1$,
we have
\begin{equation}\label{MGI-2}
g_{\alpha}=\displaystyle\sum_{j=2}^{k}(-1)^j g_{e_{p_1}\oplus e_{p_j}}g_{\alpha\oplus e_{p_1}\oplus e_{p_j}}.
\end{equation}
Since  wt$(e_{p_1}\oplus e_{p_j})=2$, wt$(\alpha\oplus e_{p_1}\oplus e_{p_j})=k-2$ for $2\leq j\leq k$,
the right hand side expressions
 of (\ref{MGI-1}) and (\ref{MGI-2}) are equal by induction, and thus
$f_{\alpha} = g_{\alpha}$.
This finishes the proof of the claim.

By the claim, if $n$ is odd, then $f$ is identically equal to $g$.
Thus $f \in\mathscr{M}$. This is a contradiction.

If $n$ is even, then $n\geq 6$.
Since $f_{\alpha}=g_{\alpha}$ for wt$(\alpha)<n$, there exists $x\in\mathbb{C}$ such that
 \begin{equation}\label{matchgate-reduce-1}
 f=g+x[0, 1]^{\otimes n}.
 \end{equation}
If $x=0$, then $f \in\mathscr{M}$, a contradiction.
Thus $x\neq 0$.
Since we have $(=_2) \in \widehat{\mathcal{EQ}}$, we can construct
$f'=\partial_{(=_2)}^{\{1, 2\}}(f)$.
Let $g'=\partial_{(=_2)}^{\{1, 2\}}(g)$.
Note that $g' \in\mathscr{M}$. A matchgate
for $g'$ is obtained from a matchgate for $g$
by forming a loop
 on the two adjacent external dangling
edges corresponding to $x_1$ and $x_2$.
On the other hand, since the operator $\partial$ is linear, by (\ref{matchgate-reduce-1}) we have
\[f'=g'+x[0, 1]^{\otimes n-2}.\]
This implies that
\begin{equation}\label{matchgate-reduce-2}
f'_{\beta}=g'_{\beta}
\end{equation}
 for any $\beta\in\{0, 1\}^{n-2}$ with wt$(\beta)<n-2$.
If $f'\notin\mathscr{M}$, then we are done by induction.
Otherwise, $f'\in\mathscr{M}$.

If $f'_{00\cdots 0} \not=0$,
consider the Matchgate Identities for $f'$ and $g'$
determined by the position vector $P'=\{1, 2, \ldots, n-2\}$
and the pattern  $e_1=10\cdots 0\in\{0, 1\}^{n-2}$, then
\[f'_{00\cdots 0}f'_{11\cdots 1}=\displaystyle\sum_{j=2}^{n-2}(-1)^j f'_{e_1\oplus e_{j}}f'_{11\cdots 1\oplus e_{1}\oplus e_{j}},\]
\[g'_{00\cdots 0}g'_{11\cdots 1}=\displaystyle\sum_{j=2}^{n-2}(-1)^j g'_{e_1\oplus e_{j}}g'_{11\cdots 1\oplus e_{1}\oplus e_{j}}.\]
Note that  wt$(e_1\oplus e_{j})=2$ and wt$(11\cdots 1\oplus e_{1}\oplus e_{j})=n-4$ for $2\leq j\leq n-2$
in the above expressions.
Thus by (\ref{matchgate-reduce-2}) we have \[f'_{00\cdots 0}f'_{11\cdots 1}=g'_{00\cdots 0}g'_{11\cdots 1}.\]
By $f'_{00\cdots 0}=g'_{00\cdots 0}\neq 0$, we have
\[f'_{11\cdots 1}=g'_{11\cdots 1}.\]
This contradicts that $x\neq 0$.

If $f'_{00\cdots 0}=0$, i.e., $f_{000\cdots 0}+f_{110\cdots 0}=0$, then
$f_{110\cdots 0}=-1$
and we can construct
 $\partial_{[1, 0]}^{\{3, 4, \ldots, n\}}(f)=[1, 0, -1]$,
since $[1,0] \in \widehat{\mathcal{EQ}}$.
Then we can construct
$f''=\partial_{[1, 0, -1]}^{\{1, 2\}}(f)$,
 and define $g''=\partial_{[1, 0, -1]}^{\{1, 2\}}(g)$.
It follows from (\ref{matchgate-reduce-1}) that
\begin{equation}\label{matchgate-reduce-3}
f''=g''-x[0, 1]^{\otimes n-2}.
\end{equation}
 Also  $f''_{0\cdots 0}= f_{000\cdots 0} - f_{110\cdots 0}
= 2$.
Note that $g''$ is a matchgate signature, because
$[1,0,-1] \in  \mathscr{M}$.
If $f''\notin\mathscr{M}$, then we are done by induction.
Otherwise, $f''\in\mathscr{M}$.
Consider the Matchgate Identities for $f''$ and $g''$ with
the  position vector $P'=\{1, 2, \ldots, n-2\}$
and the pattern  $e_1=10\cdots 0\in\{0, 1\}^{n-2}$, then
\[f''_{00\cdots 0}f''_{11\cdots 1}=\displaystyle\sum_{j=2}^{n-2}(-1)^j f''_{e_1\oplus e_{j}}f''_{11\cdots 1\oplus e_{1}\oplus e_{j}},\]
\[g''_{00\cdots 0}g''_{11\cdots 1}=\displaystyle\sum_{j=2}^{n-2}(-1)^j g''_{e_1\oplus e_{j}}g''_{11\cdots 1\oplus e_{1}\oplus e_{j}}.\]
By (\ref{matchgate-reduce-3}), we have
\begin{equation*}
f''_{\beta}=g''_{\beta}
\end{equation*}
 for any $\beta\in\{0, 1\}^{n-2}$ with wt$(\beta)<n-2$.
Moreover, by wt$(e_1\oplus e_{p_j})=2\leq n-2$ and wt$(11\cdots 1\oplus e_{p_1}\oplus e_{p_j})=n-4<n-2$ for $2\leq j\leq n-2$, we have
\[f''_{00\cdots 0}f''_{11\cdots 1}=g''_{00\cdots 0}g''_{11\cdots 1}.\]
Since $f''_{00\cdots 0}=g''_{00\cdots 0}=2$, we have
 $f''_{11\cdots 1}=g''_{11\cdots 1}$. This contradicts
(\ref{matchgate-reduce-3}) and  $x\neq 0$.
\end{proof}

\subsection{A Dichotomy Theorem for $\operatorname{Pl-Holant}(\widehat{\mathcal{EQ}}, [1, 0, x], \widehat{\mathcal{F}})$ with $[1, 0, x]\not \in \mathscr{A}$}

A signature $f$ of arity 4 satisfying the even Parity Condition
has signature matrix of the form
$M_{x_1x_2, x_4x_3}(f)=\left[\begin{smallmatrix}
a & 0 & 0 & b\\
0 & \alpha & \beta & 0\\
0 & \gamma & \delta & 0\\
c & 0 & 0 & d
\end{smallmatrix}\right]$.
For such signatures we call $\left[\begin{smallmatrix}
a & b\\
c & d
\end{smallmatrix}\right]$ the outer matrix and
$\left[\begin{smallmatrix}
\alpha & \beta\\
\gamma & \delta
\end{smallmatrix}\right]$ the inner matrix.
The following lemma implies that we can switch
the outer matrix and the inner matrix,
and also reverse the order of the columns.

\begin{lemma}\label{inner-outer}
If $\widehat{\mathcal{F}}$ contains an $f$ with  signature matrix
\begin{equation}\label{even-arity-4-sec5.2}
M_{x_1x_2, x_4x_3}(f)
=\begin{bmatrix}
f_{0000} & f_{0010} & f_{0001} & f_{0011}\\
f_{0100} & f_{0110} & f_{0101} & f_{0111}\\
f_{1000} & f_{1010} & f_{1001} & f_{1011}\\
f_{1100} & f_{1110} & f_{1101} & f_{1111}
\end{bmatrix}
=\begin{bmatrix}
a & 0 & 0 & b\\
0 & \alpha & \beta & 0\\
0 & \gamma & \delta & 0\\
c & 0 & 0 & d
\end{bmatrix},
\end{equation}
then
we can construct $g$ and $h$, where 
\[M_{x_1x_2, x_4x_3}(g)=\begin{bmatrix}
\alpha & 0 & 0 & \beta\\
0 & a & b & 0\\
0 & c & d & 0\\
\gamma & 0 & 0 & \delta
\end{bmatrix}, ~~~
\mbox{and}~~~  M_{x_1x_2, x_4x_3}(h)=
\begin{bmatrix}
\beta & 0 & 0 & \alpha\\
0 & b & a & 0\\
0 & d & c & 0\\
\delta & 0 & 0 & \gamma
\end{bmatrix}\]
such that
\[\operatorname{Pl-Holant}(\widehat{\mathcal{EQ}}, g, h, [0, 1]^{\otimes 2}, \widehat{\mathcal{F}})
\le_{\rm T}
\operatorname{Pl-Holant}(\widehat{\mathcal{EQ}}, [0, 1]^{\otimes 2}, \widehat{\mathcal{F}}).\]
\end{lemma}
\begin{proof}
We have $[0, 1], [1, 0, 1, 0] \in \widehat{\mathcal{EQ}}$
and $[0, 1]^{\otimes 2}$.
Lemma~\ref{how-to-flip-two-bits-by-[0,1,0]-tensor-2}
shows that we can  flip any two variables in $f$.
If we flip variables $x_2, x_3$ of $f$ we get $g$.
If we flip variables  $x_2, x_4$ of $f$ we get  $h$.
\end{proof}

If $f_{0000} = a  \not =0$ we can normalize it to 1.
The next lemma deals with signatures of arity 4 that
just ``miss'' to be matchgate signatures. Note that
for a signature of the form (\ref{even-arity-4-sec5.2}),
 it is a matchgate signature iff the determinants of  the inner
matrix and the outer matrix are
equal (Lemma~\ref{matchgate-identity-for-arity-4}).
Lemma~\ref{how-to-vanish-b-c} shows how to
clear some entries of (\ref{even-arity-4-sec5.2}).

\begin{lemma}\label{how-to-vanish-b-c}
Suppose $[1, 0, x] \not \in \mathscr{A}$, and
$\widehat{\mathcal{F}}$ contains a signature $f$ of arity 4 such that
\[M_{x_1x_2, x_4x_3}(f)=\begin{bmatrix}
1 & 0 & 0 & b\\
0 & \alpha & \beta & 0\\
0 & \gamma & \delta & 0\\
c & 0 & 0 & d
\end{bmatrix},\]
satisfying
$\det
\left[\begin{smallmatrix}
1 & b\\
c & d
\end{smallmatrix}\right]
=-\det
\left[\begin{smallmatrix}
\alpha & \beta\\
\gamma & \delta
\end{smallmatrix}\right]
\neq 0$.
Then we can construct $f'$ such that
\[\operatorname{Pl-Holant}(\widehat{\mathcal{EQ}}, [1, 0, x], f', \widehat{\mathcal{F}})
\le_{\rm T}
\operatorname{Pl-Holant}(\widehat{\mathcal{EQ}}, [1, 0, x], \widehat{\mathcal{F}}),\]
where
 \[M_{x_1x_2, x_4x_3}(f')=\begin{bmatrix}
1 & 0 & 0 & 0\\
0 & \alpha' & \beta' & 0\\
0 & \gamma' & \delta' & 0\\
0 & 0 & 0 & d'
\end{bmatrix},\]
and $f'$ satisfies the following conditions
\begin{itemize}
\item $\det\left[\begin{smallmatrix}
1 & 0\\
0 & d'
\end{smallmatrix}\right]=-\det\left[\begin{smallmatrix}
\alpha' & \beta'\\
\gamma' & \delta'
\end{smallmatrix}\right]\neq 0$.
\item If
$\left[\begin{smallmatrix}
\alpha & \beta\\
\gamma & \delta
\end{smallmatrix}\right]$
is a diagonal (resp. anti-diagonal) matrix,
then $\left[\begin{smallmatrix}
\alpha' & \beta'\\
\gamma' & \delta'
\end{smallmatrix}\right]$ is also a diagonal (resp. anti-diagonal)
matrix.
\end{itemize}
\end{lemma}
\begin{proof}
By Lemma~\ref{constructing-[1,0,x]-heng}, we can
get $[1, 0, z]$ for all $z\in\mathbb{C}$
from the given $[1, 0, x]  \not \in \mathscr{A}$.
In the following proof, firstly, we will prove the lemma for a special case that $b=0$ or $c=0$ in Item (A).
Then we finish the proof in Item (B)
 by reducing the general case to the special case in Item (A).

\begin{description}
\item {(A)}
Suppose $bc=0$.
If $c=0$, we can rotate the signature by 180$^\circ$ to get $b=0$
(Figure~\ref{fig:rotate_asymmetric_signature}).
  So we assume that $b=0$. If both $b=c=0$,
 then we are done by letting $f'=f$.
  Otherwise, $c\neq 0$.
  Note that $d\neq 0$ since $\det\left[\begin{smallmatrix}
1 & 0\\
c & d
\end{smallmatrix}\right]\neq 0.$
Thus we may assume that $cd\neq 0$.

    We use two binary signatures $[1, 0, u], [1, 0, v]$, where $u, v\in\mathbb{C}$ and $uv\neq 0$.
  Connect the first variables of $[1, 0, u]$ and $[1, 0, v]$
  to the variables $x_2, x_4$ of $f$ respectively,
(see (\ref{4-ary-sig-matrix-modification-f2})
and (\ref{4-ary-sig-matrix-modification-f4})) we
get a gadget with signature
  \[
  h(x_1, x_2, x_3, x_4)=
 \displaystyle\sum_{x'_2, x'_4\in\{0, 1\}}
 f(x_1, x'_2, x_3, x'_4)[1, 0, u](x'_2, x_2)[1, 0, v][x'_4, x_4],
  \]
  and
  \[
 M_{x_1x_2, x_4x_3}(h)=\begin{bmatrix}
1 & 0 & 0 & 0\\
0 & \alpha u & \beta uv & 0\\
0 & \gamma & \delta v & 0\\
cu & 0 & 0 & duv
\end{bmatrix}.
\]

 \begin{figure}[htpb]
 \centering
 \begin{tikzpicture}[scale=\scale,transform shape,node distance=\nodeDist,semithick]
  \node[internal]  (0)                    {};
  \node[external]  (1) [above left  of=0] {};
  \node[external]  (2) [below left  of=0] {};
  \node[external]  (3) [left        of=1] {};
  \node[external]  (4) [left        of=2] {};
  \node[external]    (5) [right       of=0] {};
  \node[square]  (6) [right       of=5] {};
  \node[external]  (12) [right       of=6] {};
  \node[internal]  (11) [right       of=12] {};
  \node[external]  (7) [above right of=11] {};
  \node[external]  (8) [below right of=11] {};
  \node[external]  (9) [right       of=7] {};
  \node[external] (10) [right       of=8] {};
  \path (0) edge[in=   0, out=135, postaction={decorate, decoration={
                                                           markings,
                                                           mark=at position 0.4   with {\arrow[>=diamond, white] {>}; },
                                                           mark=at position 0.4   with {\arrow[>=open diamond]   {>}; },
                                                           mark=at position 2 with {\arrow[>=diamond, white] {>}; },
                                                           mark=at position 1.0   with {\arrow[>=open diamond, white]   {>}; } } }] (3)
            edge[out=-135, in=   0]  (4)
           (6) edge[bend right, postaction={decorate, decoration={
                                                           markings,
                                                           mark=at position 0.4   with {\arrow[>=diamond, white] {>}; },
                                                           mark=at position 0.4   with {\arrow[>=open diamond]   {>}; },
                                                           mark=at position 2 with {\arrow[>=diamond, white] {>}; },
                                                           mark=at position 1.0   with {\arrow[>=open diamond, white]   {>}; } } }]                    (0)
           (0) edge[bend right]                    (6)
          (6)   edge[bend right]                    (11)
            (11)edge[bend right, postaction={decorate, decoration={
                                                           markings,
                                                           mark=at position 0.4   with {\arrow[>=diamond, white] {>}; },
                                                           mark=at position 0.4   with {\arrow[>=open diamond]   {>}; },
                                                           mark=at position 2 with {\arrow[>=diamond, white] {>}; },
                                                           mark=at position 1.0   with {\arrow[>=open diamond, white]   {>}; } } }]                   (6)
        (11) edge[out=  45, in= 180]  (9)
            edge[in=   180, out=-45] (10);
  \begin{pgfonlayer}{background}
   \node[inner sep=0pt,transform shape=false,draw=\borderColor,thick,rounded corners,fit = (1) (2) (7) (8)] {};
  \end{pgfonlayer}
 \end{tikzpicture}
 \caption{The two circle vertices are assigned $f$ and the square vertex
is assigned $h$.}
 \label{fig:gadget:connecting two f and one h}
\end{figure}
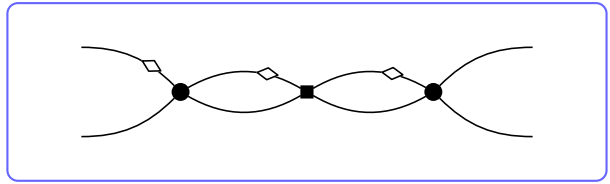
Then connect in a chain of $f$, $h$ and then $f$
as depicted in Figure~\ref{fig:gadget:connecting two f and one h}
we get the signature
  \[f'(x_1, x_2, x_3, x_4)=
 \displaystyle\sum_{y'_1, y'_2, x'_3, x'_4\in\{0, 1\}}
 f(x_1, x_2, x'_3, x'_4)h(x'_4, x'_3, y'_2, y'_1)f(y'_1, y'_2, x_3, x_4).
 \]
Note that
 \[
 M_{x_1x_2, x_4x_3}(f')=\begin{bmatrix}
1 & 0 & 0 & 0\\
0 & \alpha & \beta & 0\\
0 & \gamma & \delta & 0\\
c & 0 & 0 & d
\end{bmatrix}
\begin{bmatrix}
1 & 0 & 0 & 0\\
0 & \alpha u & \beta uv & 0\\
0 & \gamma & \delta v & 0\\
cu & 0 & 0 & duv
\end{bmatrix}
\begin{bmatrix}
1 & 0 & 0 & 0\\
0 & \alpha & \beta & 0\\
0 & \gamma & \delta & 0\\
c & 0 & 0 & d
\end{bmatrix}.
 \]

Thus the outer matrix is
\begin{equation}\label{det-1}\begin{bmatrix}
f'_{0000} & f'_{0011}\\
f'_{1100} & f'_{1111}
\end{bmatrix}=
\begin{bmatrix}
1 & 0\\
c & d
\end{bmatrix}
\begin{bmatrix}
1 & 0\\
cu & duv
\end{bmatrix}
\begin{bmatrix}
1 & 0\\
c & d
\end{bmatrix}
\end{equation}
which is
$
\left[\begin{smallmatrix}
1 & 0\\
c[1+du(1+dv)] & d^3uv
\end{smallmatrix}\right]
$,
and
the inner matrix is
\begin{equation}\label{det-2}
\begin{bmatrix}
f'_{0110} & f'_{0101}\\
f'_{1010} & f'_{1001}
\end{bmatrix}=
\begin{bmatrix}
\alpha & \beta\\
\gamma & \delta
\end{bmatrix}\begin{bmatrix}
\alpha u & \beta uv\\
\gamma & \delta v
\end{bmatrix}\begin{bmatrix}
\alpha & \beta\\
\gamma & \delta
\end{bmatrix}.
\end{equation}
By
(\ref{det-1}), (\ref{det-2}), we have
\begin{equation*}
\det
\begin{bmatrix}
f'_{0000} & f'_{0011}\\
f'_{1100} & f'_{1111}
\end{bmatrix}=d^3uv
~~~\mbox{and}~~~
\det
\begin{bmatrix}
f'_{0110} & f'_{0101}\\
f'_{1010} & f'_{1001}
\end{bmatrix}=(\alpha\delta-\beta\gamma)^3uv,
\end{equation*}
Since $\alpha\delta-\beta\gamma=-d\neq 0$, we have
\[\det
\begin{bmatrix}
f'_{0000} & f'_{0011}\\
f'_{1100} & f'_{1111}
\end{bmatrix}
=-\det
\begin{bmatrix}
f'_{0110} & f'_{0101}\\
f'_{1010} & f'_{1001}
\end{bmatrix}\neq 0.
\]
Moreover, by (\ref{det-2}),
if $\left[\begin{smallmatrix}
\alpha & \beta\\
\gamma &\delta
\end{smallmatrix}\right]$
is diagonal (respectively, anti-diagonal) then
$\left[\begin{smallmatrix}
f'_{0110} & f'_{0101}\\
f'_{1010} & f'_{1001}
\end{smallmatrix}\right]$ is diagonal (respectively, anti-diagonal).

Now we choose suitable $u, v$ such that $uv\neq 0$ and
 $f'_{1100}=0$.
We let $u=\frac{1}{d}, v=-\frac{2}{d}$, then $1+du(1+dv)=0$
and therefore  $f'_{1100}=c[1+du(1+dv)]=0$.
 This implies that
  \[
 M_{x_1x_2, x_4x_3}(f')=\begin{bmatrix}
1 & 0 & 0 & 0\\
0 & \alpha' & \beta' & 0\\
0 & \gamma' & \delta' & 0\\
0 & 0 & 0 & d'
\end{bmatrix}\]
which satisfies the requirements of the lemma.

\item {(B)}
For $bc\neq 0$, we reduce  the proof to
Item (A) by constructing $f''$ such that
\[M_{x_1x_2, x_4x_3}(f'')=\begin{bmatrix}
f''_{0000} & 0 & 0 & 0\\
0 & f''_{0110} & f''_{0101} & 0\\
0 & f''_{1010} & f''_{1001} & 0\\
f''_{1100} & 0 & 0 & f''_{1111}
\end{bmatrix}\]
with $f''_{0000}\neq 0$,
$\det\left[\begin{smallmatrix}
f''_{0000} & 0\\
f''_{1100} & f''_{1111}
\end{smallmatrix}\right]=-
\det\left[\begin{smallmatrix}
f''_{0110} & f''_{0101}\\
 f''_{1010} & f''_{1001}
\end{smallmatrix}\right]\neq 0$,
and
if $\left[\begin{smallmatrix}
\alpha & \beta\\
\gamma &\delta
\end{smallmatrix}\right]$
is diagonal (respectively, anti-diagonal) then
$\left[\begin{smallmatrix}
f''_{0110} & f''_{0101}\\
f''_{1010} & f''_{1001}
\end{smallmatrix}\right]$ is also diagonal (respectively, anti-diagonal).

    We use two binary signatures $[1, 0, u], [1, 0, v]$, where $u, v\in\mathbb{C}$ and $uv\neq 0$.
  Connect the first variables of $[1, 0, u]$ and $[1, 0, v]$
  to the variables $x_2, x_4$ of $f$ respectively,
(see (\ref{4-ary-sig-matrix-modification-f2})
and (\ref{4-ary-sig-matrix-modification-f4})) we
get a gadget with signature
  \[
  h'(x_1, x_2, x_3, x_4)=
 \displaystyle\sum_{x'_2, x'_4\in\{0, 1\}}
 f(x_1, x'_2, x_3, x'_4)[1, 0, u](x'_2, x_2)[1, 0, v][x'_4, x_4],
  \]
  and
  \[
 M_{x_1x_2, x_4x_3}(h')=\begin{bmatrix}
1 & 0 & 0 & bv\\
0 & \alpha u & \beta uv & 0\\
0 & \gamma & \delta v & 0\\
cu & 0 & 0 & duv
\end{bmatrix}.
\]
In Figure~\ref{fig:gadget:connecting two f and one h},
 by assigning $f$ to the circle vertices and assigning $h'$ to the square
vertex,
we get a gadget
with the signature
  \[f''(x_1, x_2, x_3, x_4)=
 \displaystyle\sum_{y'_1, y'_2, x'_3, x'_4\in\{0, 1\}}
 f(x_1, x_2, x'_3, x'_4)h'(x'_4, x'_3, y'_2, y'_1)f(y'_1, y'_2, x_3, x_4).
 \]
We have
 \[
 M_{x_1x_2, x_4x_3}(f'')=\begin{bmatrix}
1 & 0 & 0 & b\\
0 & \alpha & \beta & 0\\
0 & \gamma & \delta & 0\\
c & 0 & 0 & d
\end{bmatrix}
\begin{bmatrix}
1 & 0 & 0 & bv\\
0 & \alpha u & \beta uv & 0\\
0 & \gamma & \delta v & 0\\
cu & 0 & 0 & duv
\end{bmatrix}
\begin{bmatrix}
1 & 0 & 0 & b\\
0 & \alpha & \beta & 0\\
0 & \gamma & \delta & 0\\
c & 0 & 0 & d
\end{bmatrix}.
 \]
Thus
 \begin{equation}\label{det-3}
 \begin{split}
 \begin{bmatrix}
f''_{0000} & f''_{0011}\\
f''_{1100} & f''_{1111}
\end{bmatrix}&=
\begin{bmatrix}
1 & b\\
c & d
\end{bmatrix}
\begin{bmatrix}
1 & bv\\
cu & duv
\end{bmatrix}
\begin{bmatrix}
1 & b\\
c & d
\end{bmatrix}\\
&=\begin{bmatrix}
1 + bc(u + v + duv) & b[1+bcu+dv(1+du)]\\
c(1+du)+cv(bc+d^2u) & bc(1+du)+dv(bc+d^2u)
\end{bmatrix},
\end{split}
\end{equation}
and \begin{equation}\label{det-4}
\begin{bmatrix}
f''_{0110} & f''_{0101}\\
f''_{1010} & f''_{1001}
\end{bmatrix}=
\begin{bmatrix}
\alpha & \beta\\
\gamma & \delta
\end{bmatrix}\begin{bmatrix}
\alpha u & \beta uv\\
\gamma & \delta v
\end{bmatrix}\begin{bmatrix}
\alpha & \beta\\
\gamma & \delta
\end{bmatrix}.
\end{equation}

By (\ref{det-3}), (\ref{det-4}), we have
\begin{equation*}
\det
\begin{bmatrix}
f''_{0000} & f''_{0011}\\
f''_{1100} & f''_{1111}
\end{bmatrix}=(d-bc)^3uv
~~~\mbox{and}~~~
\det
\begin{bmatrix}
f''_{0110} & f''_{0101}\\
f''_{1010} & f''_{1001}
\end{bmatrix}=(\alpha\delta-\beta\gamma)^3uv.
\end{equation*}
Then by $\alpha\delta-\beta\gamma=-(d-bc)\neq 0$, we have
\[\det
\begin{bmatrix}
f''_{0000} & f''_{0011}\\
f''_{1100} & f''_{1111}
\end{bmatrix}
=-\det
\begin{bmatrix}
f''_{0110} & f''_{0101}\\
f''_{1010} & f''_{1001}
\end{bmatrix}\neq 0.
\]
Moreover,
if $\left[\begin{smallmatrix}
\alpha & \beta\\
\gamma &\delta
\end{smallmatrix}\right]$
is diagonal (respectively, anti-diagonal) then
$\left[\begin{smallmatrix}
f''_{0110} & f''_{0101}\\
f''_{1010} & f''_{1001}
\end{smallmatrix}\right]$ is also diagonal (respectively, anti-diagonal).

Now  we choose suitable $u, v$ such that
 $f''_{0000}\neq 0, f''_{0011}=0$ and $uv\neq 0$.
Let $\Delta=u+v+duv$, then
\begin{align*}
f''_{0000}
&=1+bc\Delta,\\
f''_{0011}
&=b[1+(bc-d)u+d\Delta].
\end{align*}
\begin{itemize}
\item For $bc\neq 2d$, let $u=\frac{1}{d-bc}$ and $v=\frac{1}{bc-2d}$,
then $\Delta=0$ and $1+(bc-d)u=0$. Thus
 $f''_{0000}=1, f''_{0011}=0$.

\item For $bc=2d$, let $u=\frac{1}{\sqrt{2}d}$ and $v=-\frac{\sqrt{2}}{d}$,
 then $\Delta=-\frac{1+\sqrt{2}}{\sqrt{2}d}$ and
\begin{align*}
f''_{0000}
&=1+2d\Delta=-1-\sqrt{2}\neq 0, \\
f''_{0011}
&=b[1+du+d\Delta]=0.
\end{align*}
\end{itemize}
 \end{description}
\end{proof}

Now we prove  that if all
signatures in $\widehat{\mathcal{F}}$ satisfy the Parity Condition,
and contains a binary non-affine signature $[1, 0, x]$,
then either $\widehat{\mathcal{F}} \subseteq\mathscr{M}$,
or the problem $\operatorname{Pl-Holant}(\widehat{\mathcal{EQ}},
[1, 0, x], \widehat{\mathcal{F}})$ is \#$\operatorname{P}$-hard.
Note that this is consistent with the final
dichotomy Theorem~\ref{main-dichotomy-thm}. If
$\widehat{\mathcal{F}}$ satisfies the Parity Condition,
then  $\mathcal{F} \subseteq \mathscr{P}$
would imply
$\mathcal{F}  \subseteq \mathscr{A}$
(see Proposition~\ref{parity-product-affine}).
But it contains $[1, 0, x] \not \in \mathscr{A}$,
and also $[1, 0, x] \in \mathscr{M}$ and $\widehat{\mathcal{EQ}}
\subset \mathscr{M}$, therefore
the only tractable case is $\widehat{\mathcal{F}} \subseteq
 \mathscr{M}$.

\begin{theorem}\label{With-arity-4-non-matchgate-signature-and-[1,0,x]}
Suppose all signatures in $\widehat{\mathcal{F}}$
satisfy the Parity Condition, and suppose $[1, 0, x] \not \in \mathscr{A}$.
Then either $\widehat{\mathcal{F}}\subseteq\mathscr{M}$,
or $\operatorname{Pl-Holant}(\widehat{\mathcal{EQ}}, [1, 0, x], \widehat{\mathcal{F}})$ is \#$\operatorname{P}$-hard.
\end{theorem}
\begin{proof}
By Lemma~\ref{constructing-[1,0,x]-heng}, we can construct
 $[0, 1]^{\otimes 2}$ and  $[1, 0, z]$  from
$[1, 0, x]$  for all $z\in\mathbb{C}$.

Suppose $\widehat{\mathcal{F}}\nsubseteq\mathscr{M}$.
By Theorem~\ref{arity-reduction-matchgate-signature},
 we can construct $f\notin\mathscr{M}$ and $f$ has arity 4,
such that
\[\operatorname{Pl-Holant}(\widehat{\mathcal{EQ}}, [1, 0, x], f,
 \widehat{\mathcal{F}})
\le_{\rm T}
\operatorname{Pl-Holant}(\widehat{\mathcal{EQ}}, [1, 0, x],
 \widehat{\mathcal{F}}).\]
Moreover, by Lemma~\ref{[0,1]-EQ-hat-wight-0-neq-0}, we can assume
 that $f$ satisfies the even Parity Condition and $f_{0000}=1$, i.e.,
\[M_{x_1x_2, x_4x_3}(f)=\begin{bmatrix}
1 & 0 & 0 & b\\
0 & \alpha & \beta & 0\\
0 & \gamma & \delta & 0\\
c & 0 & 0 & d
\end{bmatrix}.\]
Let
$A=\left[\begin{smallmatrix}
1 & b\\
c & d
\end{smallmatrix}\right],$
$B=\left[\begin{smallmatrix}
\alpha & \beta\\
\gamma & \delta
\end{smallmatrix}\right],$
then $f\notin\mathscr{M}$ iff
$\det A \neq \det B$
by Lemma~\ref{matchgate-identity-for-arity-4}.
We may assume  that $\det A \not =0$.
If $\det A=0$, then $\det B\neq 0$, which implies that
 $\alpha\neq 0$ or $\beta\neq 0$.
By Lemma~\ref{inner-outer}, we may switch the inner and outer matrices,
and reverse the order of the columns if necessary.
Hence  we may assume that $\det A \not =0$.

\noindent
{\bf Claim}:  We can construct $f' \not \in \mathscr{M}$ which
has the form
\[M_{x_1x_2, x_4x_1}(f')=\begin{bmatrix}
1 & 0 & 0 & 0\\
0 & \alpha' & \beta' & 0\\
0 & \gamma' & \delta' & 0\\
0 & 0 & 0 & d'
\end{bmatrix},\]
such that $d' \not =0$ and
\[\operatorname{Pl-Holant}(\widehat{\mathcal{EQ}}, [1, 0, x], f',
 \widehat{\mathcal{F}})
\le_{\rm T}
\operatorname{Pl-Holant}(\widehat{\mathcal{EQ}}, [1, 0, x],
 \widehat{\mathcal{F}}).\]
If $\det A =-\det B$, then since $\det A \not =0$, we may
apply Lemma~\ref{how-to-vanish-b-c}, and the Claim is proved.
So we may assume $\det A \not = -\det B$.
 Together with the non-matchgate condition, we may assume that
$\det A
\neq  \pm \det B$.
 If $d=0$, then we have $bc\neq 0$ by $\det A\neq 0$.
Use a binary $[1, 0, -2/(bc)]$ to modify the second variable of $f$ (as in
the proof of Lemma~\ref{how-to-vanish-b-c}, see
(\ref{4-ary-sig-matrix-modification-f2}))
and then connect a copy of $f$ with the modified $f$, we get
 \[h(x_1, x_2, x_3, x_4)=\displaystyle\sum_{x'_3, x''_3, x'_4\in\{0, 1\}}f(x_1, x_2, x'_3, x'_4)f(x'_4, x''_3, x_3, x_4)[1, 0, -\frac{2}{bc}](x'_3, x''_3).\]
Then \[M_{x_1x_2, x_4x_3}(h)=\begin{bmatrix}
1 & 0 & 0 & b\\
0 & \alpha & \beta & 0\\
0 & \gamma & \delta & 0\\
c & 0 & 0 & 0
\end{bmatrix}
\begin{bmatrix}
1 & 0 & 0 & b\\
0 & -\frac{2\alpha}{bc} & -\frac{2\beta}{bc} & 0\\
0 & \gamma & \delta & 0\\
-\frac{2}{b} & 0 & 0 & 0
\end{bmatrix}=
\begin{bmatrix}
-1 & 0 & 0 & b\\
0 & \alpha'' & \beta'' & 0\\
0 & \gamma'' & \delta'' & 0\\
c & 0 & 0 & bc
\end{bmatrix}.\]
where \[\begin{bmatrix}
\alpha'' & \beta''\\
\gamma'' & \delta''\end{bmatrix}=\begin{bmatrix}
\alpha & \beta\\
\gamma & \delta\end{bmatrix}\begin{bmatrix}
-\frac{2\alpha}{bc} & -\frac{2\beta}{bc}\\
\gamma & \delta\end{bmatrix}.\]
 Thus
$\det\left[\begin{smallmatrix}
\alpha'' & \beta''\\
\gamma'' & \delta''\end{smallmatrix}\right]=-\frac{2(\alpha\delta-\beta\gamma)^2}{bc}$.
So \[\det\begin{bmatrix}
\alpha'' & \beta''\\
\gamma'' & \delta''\end{bmatrix}\neq \det\begin{bmatrix}
-1 & b\\
c & bc\end{bmatrix}\]
by $b^2c^2\neq (\alpha\delta-\beta\gamma)^2$,
which is the same as $\det A \not = \pm \det B$.
This implies that $h\notin\mathscr{M}$.
If
\begin{equation}\label{inner-outer-not-eq-to-neg-for-h}
\det\begin{bmatrix}
\alpha'' & \beta''\\
\gamma'' & \delta''\end{bmatrix}=-\det\begin{bmatrix}
-1 & b\\
c & bc\end{bmatrix},
\end{equation}
then this quantity in (\ref{inner-outer-not-eq-to-neg-for-h})
is also nonzero  because $bc \not =0$.
Then  we can finish the proof of the Claim by Lemma~\ref{how-to-vanish-b-c}.
Therefore we may assume \[\det\begin{bmatrix}
\alpha'' & \beta''\\
\gamma'' & \delta''\end{bmatrix}
\neq
\pm \det\begin{bmatrix}
-1 & b\\
c & bc\end{bmatrix}.\]
To summarize on $h$ when $d=0$ in $f$, we have $h \notin\mathscr{M}$
satisfying the even Parity Condition,
its outer determinant is nonzero,
$h_{0000}\neq 0, h_{1111}\neq 0$,
 and the squares of the determinants of the
outer matrix and inner matrix are not equal.
We may substitute $h$ in place of $f$. To simplify notations,
we may assume that in the expression for $f$,
we have $\det A \not = 0$,  $\det A \neq  \pm \det B$ and $d\neq 0$.

Using  the same construction as above but with the binary $[1, 0, -1/d]$
instead, we get
\[f'(x_1, x_2, x_3, x_4)=\displaystyle\sum_{x'_3, x''_3, x'_4\in\{0, 1\}}f(x_1, x_2, x'_3, x'_4)f(x'_4, x''_3, x_3, x_4)
[1, 0, -\frac{1}{d}](x'_3, x''_3).\]
Then \[M_{x_1x_2, x_4x_3}(f')=\begin{bmatrix}
1 & 0 & 0 & b\\
0 & \alpha & \beta & 0\\
0 & \gamma & \delta & 0\\
c & 0 & 0 & d
\end{bmatrix}
\begin{bmatrix}
1 & 0 & 0 & b\\
0 & -\frac{\alpha}{d} & -\frac{\beta}{d} & 0\\
0 & \gamma & \delta & 0\\
-\frac{c}{d} & 0 & 0 & -1
\end{bmatrix}=
\begin{bmatrix}
1-\frac{bc}{d} & 0 & 0 & 0\\
0 & \alpha' & \beta' & 0\\
0 & \gamma' & \delta' & 0\\
0 & 0 & 0 & bc-d
\end{bmatrix},\]
where \[\begin{bmatrix}
\alpha' & \beta'\\
\gamma' & \delta'\end{bmatrix}=\begin{bmatrix}
\alpha & \beta\\
\gamma & \delta\end{bmatrix}\begin{bmatrix}
-\frac{\alpha}{d} & -\frac{\beta}{d}\\
\gamma & \delta\end{bmatrix}.\]
 Thus
$\det\left[\begin{smallmatrix}
\alpha' & \beta'\\
\gamma' & \delta'\end{smallmatrix}\right]=-\frac{(\alpha\delta-\beta\gamma)^2}{d}$.
So \[\det\begin{bmatrix}
\alpha' & \beta'\\
\gamma' & \delta'\end{bmatrix}\neq \det\begin{bmatrix}
1-\frac{bc}{d} & 0\\
0 & bc-d\end{bmatrix}\]
by $\det A\neq  \pm \det B$.
This implies that $f'\notin\mathscr{M}$.
Note that $f'_{0000} =1-\frac{bc}{d} = \frac{\det A}{d} \neq 0$,
so we can renormalize $f'_{0000}$ to 1.
Also $f'_{1111} = bc-d = -\det A \not =0$.
 Thus $f'$ satisfies the requirement of the Claim.
This finishes the proof of the Claim.
The Claim shows that we may assume that $b=c=0$ and
$d \not =0$ in
\[M_{x_1x_2, x_4x_3}(f)=\begin{bmatrix}
1 & 0 & 0 & 0\\
0 & \alpha & \beta & 0\\
0 & \gamma & \delta & 0\\
0 & 0 & 0 & d
\end{bmatrix}.\]

In the following.
we can finish the proof of the theorem by two alternatives:
\begin{description}
\item{(A)} If we can construct the crossover function $\mathfrak{X}$ such that
\[\operatorname{Pl-Holant}(\widehat{\mathcal{EQ}}, [1, 0, x], \mathfrak{X}, \widehat{\mathcal{F}})
\le_{\rm T}
\operatorname{Pl-Holant}(\widehat{\mathcal{EQ}}, [1, 0, x], \widehat{\mathcal{F}}),\]
then the presence of $\mathfrak{X}$ reduces a general (non-planar)
Holant problem to a $\operatorname{Pl-Holant}$ problem
\[\operatorname{Holant}(\widehat{\mathcal{EQ}}, [1, 0, x])
\le_{\rm T}
\operatorname{Pl-Holant}(\widehat{\mathcal{EQ}}, [1, 0, x], \mathfrak{X}, \widehat{\mathcal{F}}).\]
So
we have \[\operatorname{Holant}(\widehat{\mathcal{EQ}}, [1, 0, x])
\le_{\rm T}
\operatorname{Pl-Holant}(\widehat{\mathcal{EQ}}, [1, 0, x], \widehat{\mathcal{F}}),\]
and we can apply Theorem~\ref{non-planar-csp-dichotomy}$'$.
We have $[1, 0, x]\notin\widehat{\mathscr{P}}$
(this can be directly verified
or use Proposition~\ref{parity-product-affine}),
 and $[1, 0, x]\notin\mathscr{A}$.
 Thus by Theorem~\ref{non-planar-csp-dichotomy}$'$,  $\operatorname{Holant}(\widehat{\mathcal{EQ}}, [1, 0, x])$ is \#P-hard.
Hence $\operatorname{Pl-Holant}(\widehat{\mathcal{EQ}}, [1, 0, x], \widehat{\mathcal{F}})$ is \#P-hard.
\item{(B)} If we can construct $(=_4)$ such that
\[\operatorname{Pl-Holant}(\widehat{\mathcal{EQ}}, [1, 0, x], (=_4), \widehat{\mathcal{F}})
\le_{\rm T}
\operatorname{Pl-Holant}(\widehat{\mathcal{EQ}}, [1, 0, x], \widehat{\mathcal{F}}),\]
then by
Lemma~\ref{constructing-even-equality-by-equality-4},
\[\operatorname{Pl-Holant}(\widehat{\mathcal{EQ}}, \mathcal{EQ}_2, [1, 0, x], \widehat{\mathcal{F}})
\le_{\rm T}
\operatorname{Pl-Holant}(\widehat{\mathcal{EQ}}, [1, 0, x], \widehat{\mathcal{F}}).\]
It follows that
\[\operatorname{Pl-Holant}(\mathcal{EQ}_2, [1, 0], [1, 0, 1, 0], [1, 0, x])
\le_{\rm T}
\operatorname{Pl-Holant}(\widehat{\mathcal{EQ}}, [1, 0, x], \widehat{\mathcal{F}})\]
since $[1, 0], [1, 0, 1, 0]\in\widehat{\mathcal{EQ}}$.
Note that $[1, 0] \not \in \widehat{\mathscr{M}} \cup
\widehat{\mathscr{M}}^{\dagger}$ (Proposition~\ref{matchgate:affine:hat}),
$[1,0,1,0]  \not \in \mathscr{P} \cup
\mathscr{A}^{\dagger}$ (Corollary~\ref{[1,0,1,0]:not:prod}
and Proposition~\ref{matchgate:affine:hat}),
and we are given $[1, 0, x]  \not \in \mathscr{A}$, thus,
for the symmetric signature set
$\mathcal{G} = \{[1, 0], [1, 0, 1, 0], [1, 0, x]\}$
\[ \mathcal{G} \nsubseteq \mathscr{P},
~~\mathcal{G} \nsubseteq\mathscr{A},
~~\mathcal{G} \nsubseteq\mathscr{A}^{\dagger},
~~\mathcal{G} \nsubseteq\widehat{\mathscr{M}},
~~\mathcal{G} \nsubseteq\widehat{\mathscr{M}}^{\dagger}.\]
By
Theorem~\ref{heng-tyson-dichotomy-pl-csp2}$'$
on Pl-CSP$^2$ problems for symmetric
signatures, $\operatorname{Pl-Holant}(\mathcal{EQ}_2, \mathcal{G})$
is \#P-hard.
It follows that
 $\operatorname{Pl-Holant}(\widehat{\mathcal{EQ}}, [1, 0, x], \widehat{\mathcal{F}})$ is \#P-hard.
\end{description}

In the following,  in Case 1.  we  prove the theorem
  when $\det A = -\det B$.
Then in Case 2. we prove the theorem
  when $\det A \neq -\det B$.
Since $f \not \in \mathscr{M}$ we are given
$\det A \neq \det B$. So Case 2. is equivalent to
$\det A \neq  \pm \det B$.

\begin{enumerate}
\item Suppose $\det A=-\det B$. Since we have $\det A \neq 0$,
both $\det A$ and $\det B  \neq 0$.
At least one of $\alpha$ or $\beta$ is nonzero by $\det B \neq 0$.
\begin{itemize}
\item Suppose  $\alpha\neq 0$.
By Lemma~\ref{inner-outer}, we also have the 4-ary signature $g$ such that
\[M_{x_1x_2, x_4x_3}(g)=\begin{bmatrix}
\alpha & 0 & 0 & \beta\\
0 & 1 & 0 & 0\\
0 & 0 & d & 0\\
\gamma & 0 & 0 & \delta
\end{bmatrix}.\]
Since $\alpha\neq 0$, we can assume that $\alpha=1$ by normalizing.
Then we may write
\[M_{x_1x_2, x_4x_3}(g)=\begin{bmatrix}
1 & 0 & 0 & \beta\\
0 & a & 0 & 0\\
0 & 0 & d & 0\\
\gamma & 0 & 0 & \delta
\end{bmatrix},\] where $\det\left[\begin{smallmatrix}
a & 0 \\
0 & d
\end{smallmatrix}\right]=-\det
\left[\begin{smallmatrix}
1  & \beta\\
\gamma  & \delta
\end{smallmatrix}\right]\neq 0$.
Then by Lemma~\ref{how-to-vanish-b-c}, we have $g'$ such that
\[M_{x_1x_2, x_4x_3}(g')=\begin{bmatrix}
1 & 0 & 0 & 0\\
0 & a' & 0 & 0\\
0 & 0 & d' & 0\\
0 & 0 & 0 & \delta'
\end{bmatrix},\]
where $\det\left[\begin{smallmatrix}
a' & 0 \\
0 & d'
\end{smallmatrix}\right]=-\det
\left[\begin{smallmatrix}
1  & 0\\
0  & \delta'
\end{smallmatrix}\right]\neq 0$.
Thus $\left[\begin{smallmatrix}
1 & d'\\
a' & \delta'
\end{smallmatrix}\right]$ has full rank. Then
by
 Lemma~\ref{interpolation-equality-4},
we have $(=_4)$ by interpolating and we are done by alternative (B).

\item Suppose $\beta\neq 0$,  by Lemma~\ref{inner-outer}, we have $h$ such that \[M_{x_1x_2, x_4x_3}(h)=\begin{bmatrix}
\beta & 0 & 0 & \alpha\\
0 & 0 & 1 & 0\\
0 & d & 0 & 0\\
\delta & 0 & 0 & \gamma
\end{bmatrix}.\]
Since $\beta\neq 0$, we can assume that $\beta=1$ by normalizing.
Then we may write
\[M_{x_1x_2, x_4x_3}(h)=\begin{bmatrix}
1 & 0 & 0 & \alpha\\
0 & 0 & a & 0\\
0 & d & 0 & 0\\
\delta & 0 & 0 & \gamma
\end{bmatrix},\] where $\det\left[\begin{smallmatrix}
0 & a \\
d & 0
\end{smallmatrix}\right]=-\det
\left[\begin{smallmatrix}
1  & \alpha\\
\delta  & \gamma
\end{smallmatrix}\right] \not = 0$.
Then by Lemma~\ref{how-to-vanish-b-c}, we have $h'$ such that
\[M_{x_1x_2, x_4x_3}(h')=\begin{bmatrix}
1 & 0 & 0 & 0\\
0 & 0 & a' & 0\\
0 & d' & 0 & 0\\
0 & 0 & 0 & \gamma'
\end{bmatrix},\]
where $\det\left[\begin{smallmatrix}
0 & a' \\
d' & 0
\end{smallmatrix}\right]=-\det
\left[\begin{smallmatrix}
1  & 0\\
0 & \gamma'
\end{smallmatrix}\right]\neq 0$.
By Lemma~\ref{constructing-[1,0,x]-heng}
and using $[1, 0, x]$, we have
$[1, 0, (d')^{-1}]$ and  $[1, 0, (a')^{-1}]$.
Modifying $h'$  on the first and second variables by
$[1, 0, (d')^{-1}]$ and  $[1, 0, (a')^{-1}]$ respectively
((see (\ref{ternary-sig-matrix-modification-f1})
and (\ref{ternary-sig-matrix-modification-f2}))
gives the crossover function since $\gamma' = a' d'$:
\[\mathfrak{X}=\displaystyle\sum_{x'_1, x'_2\in\{0, 1\}}h'(x'_1, x'_2, x_3, x_4)[1, 0, (d')^{-1}](x'_1, x_1)[1, 0, (a')^{-1}](x'_2, x_2).\]
Then we are done by alternative (A).
\end{itemize}

\item Suppose $\det A\neq  -\det B$.
Since $f \not \in \mathscr{A}$ we have $\det A\neq  \pm \det B$.
So $d^2\neq (\alpha\delta-\beta\gamma)^2$.

\begin{itemize}
\item If $\alpha=\delta=0$, then $d^2 \neq \beta^2\gamma^2$.
We can construct
\[\tilde{f}(x_1, x_2, x_3, x_4)=\displaystyle\sum_{x'_3, x'_4\in\{0, 1\}}f(x_1, x_2, x'_3, x'_4)f(x'_3, x'_4, x_3, x_4).\]
Then \[M_{x_1x_2, x_4x_3}(\tilde{f})=\begin{bmatrix}
1 & 0 & 0 & 0\\
0 & 0 & \beta & 0\\
0 & \gamma & 0 & 0\\
0 & 0 & 0 & d
\end{bmatrix}
\begin{bmatrix}
1 & 0 & 0 & 0\\
0 & 0 & \beta & 0\\
0 & \gamma & 0 & 0\\
0 & 0 & 0 & d
\end{bmatrix}=
\begin{bmatrix}
1 & 0 & 0 & 0\\
0 & \beta\gamma & 0 & 0\\
0 & 0 & \beta\gamma & 0\\
0 & 0 & 0 & d^2
\end{bmatrix}.\]
Note that $\left[\begin{smallmatrix}1 & \beta\gamma\\
\beta\gamma & d^2\end{smallmatrix}\right]$
has full rank.
Then by  Lemma~\ref{interpolation-equality-4}, we have $(=_4)$ by interpolation. So we are done by  alternative (B).

\item If $\alpha\neq 0$ and $\delta\neq 0$, for any $u \in \mathbb{C}$
 we can construct
\[\hat{h}(x_1, x_2, x_3, x_4)=\displaystyle\sum_{x'_2\in\{0, 1\}}f(x_1, x'_2, x_3, x_4)
[1, 0, u](x'_2, x_2).\]
Then \[M_{x_1x_2, x_4x_3}(\hat{h})=\begin{bmatrix}
1 & 0 & 0 & 0\\
0 & \alpha u & \beta u & 0\\
0 & \gamma & \delta  &  0\\
0 & 0 & 0 & du
\end{bmatrix}.\]
Then we  can construct
\[\hat{f}(x_1, x_2, x_3, x_4)=\displaystyle\sum_{x'_3, x'_4\in\{0, 1\}}f(x_1, x_2, x'_3, x'_4)\hat{h}(x'_3, x'_4, x_3, x_4).\]
Then \begin{equation*}
\begin{split}
M_{x_1x_2, x_4x_3}(\hat{f})&=\begin{bmatrix}
1 & 0 & 0 & 0\\
0 & \alpha & \beta & 0\\
0 & \gamma & \delta & 0\\
0 & 0 & 0 & d
\end{bmatrix}
\begin{bmatrix}
1 & 0 & 0 & 0\\
0 & \alpha u & \beta u & 0\\
0 & \gamma & \delta  &  0\\
0 & 0 & 0 & du
\end{bmatrix}\\
&=
\begin{bmatrix}
1 & 0 & 0 & 0\\
0 & \alpha^2u+\beta\gamma & \beta (\delta+\alpha u) & 0\\
0 & \gamma(\delta+\alpha u) & \beta\gamma u+\delta^2 & 0\\
0 & 0 & 0 & d^2u
\end{bmatrix}.
\end{split}
\end{equation*}
Choose $u=-\frac{\delta}{\alpha} \not = 0$, then $\beta(\delta+\alpha u)=\gamma(\delta+\alpha u)=0$ and
\[M_{x_1x_2, x_4x_3}(\hat{f})=\begin{bmatrix}
1 & 0 & 0 & 0\\
0 & -(\alpha\delta -\beta\gamma) & 0 & 0\\
0 & 0 & \frac{\delta}{\alpha}(\alpha\delta-\beta\gamma) & 0\\
0 & 0 & 0 & -\frac{d^2\delta}{\alpha}
\end{bmatrix}.\]
Note that $\left[\begin{smallmatrix}
1 & \frac{\delta}{\alpha}(\alpha\delta-\beta\gamma)\\
-(\alpha\delta-\beta\gamma) & -\frac{d^2\delta}{\alpha}
\end{smallmatrix}\right]$ has full rank since
$d^2\neq (\alpha\delta-\beta\gamma)^2$ by $(\det A)^2\neq (\det B)^2$.
Then by Lemma~\ref{interpolation-equality-4}
we have $(=_4)$ by interpolation and we are done by alternative (B).

\item If $\alpha\neq 0$ and $\delta=0$, then after a rotation (Figure~\ref{fig:rotate_asymmetric_signature}) clockwise by 90$^\circ$ we have
\[M_{x_2x_3, x_1x_4}(f)=\begin{bmatrix}
1 & 0 & 0 & 0\\
0 & 0 & \gamma & 0\\
0 & \beta & 0 & 0\\
\alpha & 0 & 0 & d
\end{bmatrix}.\]
We can construct
 \[\bar{h}(x_1, x_2, x_3, x_4)=\displaystyle\sum_{ x'_3\in\{0, 1\}}f(x_1, x_2, x'_3, x_4)
[1, 0, -\frac{1}{d}](x'_3, x_3).\]
Then
(see (\ref{4-ary-sig-matrix-modification-f2}))
 \[M_{x_2x_3, x_1x_4}(\bar{h})=\begin{bmatrix}
1 & 0 & 0 & 0\\
0 & 0 & -\frac{\gamma}{d} & 0\\
0 & \beta & 0 & 0\\
-\frac{\alpha}{d} & 0 & 0 & -1
\end{bmatrix}.
\]
With this we can further construct
\[\bar{f}(x_1, x_2, x_3, x_4)=\displaystyle\sum_{x'_1, x'_4\in\{0, 1\}}f(x'_1, x_2, x_3, x'_4)\bar{h}(x_1, x'_1, x'_4, x_4),\]
with \[M_{x_2x_3, x_1x_4}(\bar{f})=\begin{bmatrix}
1 & 0 & 0 & 0\\
0 & 0 & \gamma & 0\\
0 & \beta & 0 & 0\\
\alpha & 0 & 0 & d
\end{bmatrix}
\begin{bmatrix}
1 & 0 & 0 & 0\\
0 & 0 & -\frac{\gamma}{d} & 0\\
0 & \beta & 0 & 0\\
-\frac{\alpha}{d} & 0 & 0 & -1
\end{bmatrix}=
\begin{bmatrix}
1 & 0 & 0 & 0\\
0 & \beta\gamma & 0 & 0\\
0 & 0 & -\frac{\beta\gamma}{d} & 0\\
0 & 0 & 0 & -d
\end{bmatrix}.\]
Note that
$\left[\begin{smallmatrix}
 1 & \beta\gamma\\
 -\frac{\beta\gamma}{d} & -d
\end{smallmatrix}\right]$
has full rank since $d^2\neq \beta^2\gamma^2$ by $(\det A)^2\neq (\det B)^2$.
By Lemma~\ref{interpolation-equality-4},
 we have $(=_4)$ by interpolation  and we are done by  alternative (B).

\item If $\alpha=0$ and $\delta\neq 0$, then
the proof is symmetric by first rotating $f$ counterclockwise by 90$^\circ$
(Figure~\ref{fig:rotate_asymmetric_signature}) and then
switching the roles of $\alpha$ and $\delta$ in the previous item.
\end{itemize}
\end{enumerate}
\end{proof}

\subsection{A Dichotomy When $\widehat{\mathcal{F}}$ Satisfies
Parity}
\begin{lemma}\label{from-[1,0,-1]-to-[0,1]-tensor-2}
If ${\mathcal{F}}$ contains a 4-ary signature $f \not \in \mathscr{M}$
of the form
\[M_{x_1x_2, x_4x_3}(f)=\begin{bmatrix}
1 & 0 & 0 & b\\
0 & \alpha & \beta & 0\\
0 & \gamma & \delta & 0\\
c & 0 & 0 & d
\end{bmatrix}.\]
Furthermore,
suppose that at least one of $\{b, c, \alpha, \delta\}$ is nonzero.
Then we can construct $[0, 1]^{\otimes 2}$ such that
\[\operatorname{Pl-Holant}(\widehat{\mathcal{EQ}},[0, 1]^{\otimes 2}, [1, 0, -1], \widehat{\mathcal{F}})
\le_{\rm T}
\operatorname{Pl-Holant}(\widehat{\mathcal{EQ}}, [1, 0, -1], \widehat{\mathcal{F}}).\]
\end{lemma}
\begin{proof}
By a rotation (Figure~\ref{fig:rotate_asymmetric_signature}),
 without loss of generality we can assume that $b\neq 0$.
We have $\partial_{[1, 0]}^{\{1, 2\}}(f)=[1, 0, b]$.
If $b^2\neq 1$, then we have $[0, 1]^{\otimes 2}$ by Lemma~\ref{constructing-[1,0,x]-heng}.
Therefore we may assume that $b=\pm 1$.

By a planar gadget we can construct the signature
\[h(x_1, x_2, x_3, x_4)=\displaystyle\sum_{x'_3, x''_3, x'_4\in\{0, 1\}}f(x_1, x_2, x'_3, x'_4)f(x_3, x_4, x'_4, x''_3)[1, 0, -1](x''_3, x'_3).\]
Then (see Figure~\ref{fig:rotate_asymmetric_signature}
and  (\ref{4-ary-sig-matrix-modification-f2}))
\[M_{x_1x_2, x_4x_3}(h)=\begin{bmatrix}
1 & 0 & 0 & b\\
0 & \alpha & \beta & 0\\
0 & \gamma & \delta & 0\\
c & 0 & 0 & d
\end{bmatrix}
\begin{bmatrix}
1 & 0 & 0 & c\\
0 & -\delta & -\beta & 0\\
0 & \gamma & \alpha & 0\\
-b & 0 & 0 & -d
\end{bmatrix}=
\begin{bmatrix}
0 & 0 & 0 & c-bd\\
0 & \alpha' & \beta' & 0\\
0 & \gamma' & \delta' & 0\\
c-bd & 0 & 0 & c^2-d^2
\end{bmatrix}.
\]
If $c-bd\neq 0$, then we have
$\partial_{[1, 0]}^{\{1,2\}}(h)=(c-bd)[0, 1]^{\otimes 2}$,
 a nonzero multiple of $[0, 1]^{\otimes 2}$.
Otherwise, we have \[d-bc=d-b^2d=0,\]
by $c=bd$ and  $b^2=1$.
This implies that $\det\left[\begin{smallmatrix}1 & b\\ c & d\end{smallmatrix}\right]=0$.
So $\det\left[\begin{smallmatrix}\alpha & \beta\\ \gamma & \delta\end{smallmatrix}\right]\neq 0$ by Lemma~\ref{matchgate-identity-for-arity-4}
and $f\notin\mathscr{M}$.
Thus $\det\left[\begin{smallmatrix}-\delta & -\beta\\ \gamma & \alpha\end{smallmatrix}\right]\neq 0$.
So $\left[\begin{smallmatrix}\alpha' & \beta'\\ \gamma' & \delta'\end{smallmatrix}\right]=
\left[\begin{smallmatrix}\alpha & \beta\\ \gamma & \delta\end{smallmatrix}\right]
\left[\begin{smallmatrix}-\delta & -\beta\\ \gamma & \alpha\end{smallmatrix}\right]$
has full rank.
This implies that at least one of $\alpha', \beta', \gamma', \delta'$ is nonzero.
Then we have $[0, 1]^{\otimes 2}$ by $[1, 0]$ and $h$.
For example, if $\alpha' \not =0$, then
$\partial_{[1,0]}^{\{1,4\}}(h) = [0, 0, \alpha']
= \alpha' [0,1]^{\otimes 2}$.
\end{proof}

\begin{theorem}\label{With-arity-4-non-matchgate-signature}
If all signatures in $\widehat{\mathcal{F}}$ satisfy
the Parity Condition, then the following dichotomy holds:
If $\widehat{\mathcal{F}}\subseteq\mathscr{A}$, or $\widehat{\mathcal{F}}\subseteq\widehat{\mathscr{P}}$, or $\widehat{\mathcal{F}}\subseteq\mathscr{M}$,
then
$\operatorname{Pl-Holant}(\widehat{\mathcal{EQ}}, \widehat{\mathcal{F}})$ is
tractable in P, otherwise it is
  \#$\operatorname{P}$-hard.
\end{theorem}
\begin{proof}
Clearly if $\widehat{\mathcal{F}}\subseteq\mathscr{A}$, or $\widehat{\mathcal{F}}\subseteq\widehat{\mathscr{P}}$, or $\widehat{\mathcal{F}}\subseteq\mathscr{M}$, then  the problem
$\operatorname{Pl-Holant}(\widehat{\mathcal{EQ}}, \widehat{\mathcal{F}})$ is
tractable in P.
(Since all signatures in $\widehat{\mathcal{F}}$ satisfy the Parity Condition,
Proposition~\ref{parity-product-affine} implies that
if $\widehat{\mathcal{F}}\subseteq\widehat{\mathscr{P}}$ then in fact
$\widehat{\mathcal{F}}\subseteq\mathscr{A}$. But the proof below will
not use this fact.)

Now suppose $\widehat{\mathcal{F}} \not \subseteq\mathscr{A}$, and
 $\widehat{\mathcal{F}} \not \subseteq\widehat{\mathscr{P}}$,
and  $\widehat{\mathcal{F}} \not \subseteq\mathscr{M}$.
Since  $\widehat{\mathcal{F}}\nsubseteq\mathscr{M}$, by Theorem~\ref{arity-reduction-matchgate-signature},
we can construct a 4-ary signature $f \not \in \mathscr{M}$
from $\widehat{\mathcal{F}}$.
By Lemma~\ref{[0,1]-EQ-hat-wight-0-neq-0}, we can assume that
 \[M_{x_1x_2, x_4x_3}(f)=\begin{bmatrix}
1 & 0 & 0 & b\\
0 & \alpha & \beta & 0\\
0 & \gamma & \delta & 0\\
c & 0 & 0 & d
\end{bmatrix}.\]

We  can finish the proof by the following four alternatives:
\begin{description}
\item {(A)} If we can get $[1, 0, x] \not \in \mathscr{A}$ such that
\[\operatorname{Pl-Holant}(\widehat{\mathcal{EQ}}, [1, 0, x], \widehat{\mathcal{F}})
\le_{\rm T}
\operatorname{Pl-Holant}(\widehat{\mathcal{EQ}}, \widehat{\mathcal{F}}),\]
then $\operatorname{Pl-Holant}(\widehat{\mathcal{EQ}}, \widehat{\mathcal{F}})$ is \#P-hard
since $ \operatorname{Pl-Holant}(\widehat{\mathcal{EQ}}, [1, 0, x], \widehat{\mathcal{F}})$ is \#P-hard by Theorem~\ref{With-arity-4-non-matchgate-signature-and-[1,0,x]}.
\item {(B)} If we can get $[1, 0, -1]$ and $[0, 1]^{\otimes 2}$ such that
\[\operatorname{Pl-Holant}(\widehat{\mathcal{EQ}}, [1, 0, -1], [0, 1]^{\otimes 2}, \widehat{\mathcal{F}})
\le_{\rm T}
\operatorname{Pl-Holant}(\widehat{\mathcal{EQ}}, \widehat{\mathcal{F}}),\]
then since $\widehat{\mathcal{F}}\nsubseteq\mathscr{A}$,
 we can get $[1, 0, x] \not \in \mathscr{A}$ by Lemma~\ref{[1,0]-[0,1]-[1,0,-1]-EQ-hat-affine-reduction},
such that
\[\operatorname{Pl-Holant}(\widehat{\mathcal{EQ}},[1, 0, x], [1, 0, -1], [0, 1]^{\otimes 2}, \widehat{\mathcal{F}})
\le_{\rm T}
\operatorname{Pl-Holant}(\widehat{\mathcal{EQ}}, [1, 0, -1], [0, 1]^{\otimes 2}, \widehat{\mathcal{F}}).\]
$\operatorname{Pl-Holant}(\widehat{\mathcal{EQ}},[1, 0, x], [1, 0, -1], [0, 1]^{\otimes 2}, \widehat{\mathcal{F}})$
is \#$\operatorname{P}$-hard
 by Theorem~\ref{With-arity-4-non-matchgate-signature-and-[1,0,x]}.
Thus $\operatorname{Pl-Holant}(\widehat{\mathcal{EQ}}, \widehat{\mathcal{F}})$ is \#$\operatorname{P}$-hard.

\item {(C)} If we can construct the crossover function $\mathfrak{X}$
(Definition~\ref{def:crossover})
such that
\[\operatorname{Pl-Holant}(\widehat{\mathcal{EQ}}, \mathfrak{X}, \widehat{\mathcal{F}})
\le_{\rm T}
\operatorname{Pl-Holant}(\widehat{\mathcal{EQ}}, \widehat{\mathcal{F}}),\]
then note that the Holant problem (on general, not necessarily
planar, instances) can be reduced to the planar one
\[
\operatorname{Holant}(\widehat{\mathcal{EQ}}, \widehat{\mathcal{F}})
\equiv_{\rm T}
\operatorname{Pl-Holant}(\widehat{\mathcal{EQ}}, \mathfrak{X}, \widehat{\mathcal{F}}).\]
Since $\widehat{\mathcal{F}}\nsubseteq\widehat{\mathscr{P}}$ and
 $\widehat{\mathcal{F}}\nsubseteq\mathscr{A}$, we have
$\operatorname{Pl-Holant}(\widehat{\mathcal{EQ}}, \widehat{\mathcal{F}})$ is \#P-hard
by Theorem~\ref{non-planar-csp-dichotomy}$'$.

\item {(D)}  If we have $(=_4)$, by Lemma~\ref{constructing-even-equality-by-equality-4}, we have
\[\operatorname{Pl-Holant}(\widehat{\mathcal{EQ}}, \mathcal{EQ}_2, \widehat{\mathcal{F}})
\le_{\rm T}
\operatorname{Pl-Holant}(\widehat{\mathcal{EQ}}, \widehat{\mathcal{F}}).\]
Then by Theorem~\ref{dichotomy-csp-2}, $\operatorname{Pl-Holant}(\widehat{\mathcal{EQ}}, \mathcal{EQ}_2, \widehat{\mathcal{F}})$
is \#P-hard, since $\widehat{\mathcal{F}}\nsubseteq\mathscr{A}$.
Thus
 $\operatorname{Pl-Holant}(\widehat{\mathcal{EQ}}, \widehat{\mathcal{F}})$
is \#P-hard.
\end{description}

Note that for any $x\in\{b, c, \alpha, \beta, \gamma, \delta\}$, we have $[1, 0, x]$
by taking $\partial^{\{i,j\}}_{[1,0]} (f)$ on some two variables
$x_i$ and $x_j$.
If there exists $x\in\{b, c, \alpha, \beta, \gamma, \delta\}$ such that $x^4\neq 0, 1$.
Then $[1, 0, x] \not \in \mathscr{A}$
by Proposition~\ref{A-has-same-norm-etc},  we are done by alternative (A).
If there exists $x\in\{b, c, \alpha, \beta, \gamma, \delta\}$ such that $x^2=-1$, then we
have $[1, 0, -1]$ and $[0, 1]^{\otimes 2}$ by Lemma~\ref{constructing-[1,0,x]-heng}
such that
\[\operatorname{Pl-Holant}(\widehat{\mathcal{EQ}}, [1, 0, -1], [0, 1]^{\otimes 2}, \widehat{\mathcal{F}})
\le_{\rm T}
\operatorname{Pl-Holant}(\widehat{\mathcal{EQ}}, \widehat{\mathcal{F}}).\]
Thus we are done by alternative  (B).
So in the following, we may assume that
\[\{b, c, \alpha, \beta, \gamma, \delta\}\subseteq\{0, 1, -1\}.\]
Now we finish the proof by a case analysis of $\{b, c, \alpha, \beta, \gamma, \delta\}$.

 If $b=c=\alpha=\delta=0$, then
\[M_{x_1x_2, x_4x_3}(f)=\begin{bmatrix}
1 & 0 & 0 & 0\\
0 & 0 & \beta & 0\\
0 & \gamma & 0 & 0\\
0 & 0 & 0 & d
\end{bmatrix}.\]
In this case,
\begin{itemize}
 \item if $\beta\gamma\neq 0$, then we have
(see (\ref{ternary-sig-matrix-modification-f1})
 and (\ref{ternary-sig-matrix-modification-f2}),
and note that $\beta = \beta^{-1}$ and $\gamma = \gamma^{-1}$)
 \[
 h(x_1, x_2, x_3,x_4)=\displaystyle\sum_{x'_1, x'_2\{0, 1\}}f(x'_1, x'_2, x_3, x_4)[1, 0, \gamma](x'_1, x_1)[1, 0, \beta](x'_2, x_2)
 \]
and
\[M_{x_1x_2, x_4x_3}(h)=
\begin{bmatrix}
1 & 0 & 0 & 0\\
0 & 0 & 1 & 0\\
0 & 1 & 0 & 0\\
0 & 0 & 0 & {d\beta\gamma}\\
\end{bmatrix}.
\]
Note that we have $\partial_{(=_2)}^{\{1, 2\}}(f)=[1, 0, {d\beta\gamma}]$.
\subitem
If $(d\beta\gamma)^4\neq 0, 1$, then
we are done by alternative (A).

\subitem
If $(d\beta\gamma)^2=-1$, then we have $[1, 0, -1]$ and $[0, 1]^{\otimes 2}$
by Lemma~\ref{constructing-[1,0,x]-heng}.
Then we are done by alternative (B).

\subitem
If $d\beta\gamma=-1$, then $f\in\mathscr{M}$ by Lemma~\ref{matchgate-identity-for-arity-4}. This is a contradiction.

\subitem
If $d\beta\gamma=1$, then $h$ is the crossover function
$\mathfrak{X}$ (Definition~\ref{def:crossover}).
Thus we are done by alternative (C).

\subitem
If $d\beta\gamma=0$, then
\[M_{x_1x_2, x_4x_3}(h)=
\begin{bmatrix}
1 & 0 & 0 & 0\\
0 & 0 & 1 & 0\\
0 & 1 & 0 & 0\\
0 & 0 & 0 & 0\\
\end{bmatrix}.
\]
Take two copies of $h$ and connect $x_2, x_3, x_4$ of the first copy
with $x_4, x_3, x_2$ of the second copy, the planar gadget
has the signature
\[
 \tilde{h}(x_1, x_2)=\displaystyle\sum_{x'_2, x'_3, x'_4\{0, 1\}}h(x_1, x'_2, x'_3, x'_4)h(x_2, x'_4, x'_3, x'_2),
 \]
then $\tilde{h}= [2,0,1] = 2[1, 0, \frac{1}{2}]$. Thus we are done by  alternative (A).

 \item If $\beta\gamma=0$, then $d\neq 0$
by Lemma~\ref{matchgate-identity-for-arity-4}, since $f \not \in
 \mathscr{A}$.
 By connecting two copies of $f$, we get
 \[h'(x_1, x_2, x_3, x_4)=\displaystyle\sum_{x'_3, x'_4\in\{0, 1\}} f(x_1, x_2, x'_3, x'_4)f(x'_4, x'_3, x_3, x_4),\]
 where \[M_{x_1x_2, x_4x_3}(h'')=
 \begin{bmatrix}
1 & 0 & 0 & 0\\
0 & 0 & \beta & 0\\
0 & \gamma & 0 & 0\\
0 & 0 & 0 & d\\
\end{bmatrix}
\begin{bmatrix}
1 & 0 & 0 & 0\\
0 & 0 & \beta & 0\\
0 & \gamma & 0 & 0\\
0 & 0 & 0 & d\\
\end{bmatrix}=
\begin{bmatrix}
1 & 0 & 0 & 0\\
0 & 0 & 0 & 0\\
0 & 0 & 0 & 0\\
0 & 0 & 0 & d^2\\
\end{bmatrix}.
\]
Then by Lemma~\ref{interpolation-equality-4} we can get $(=_4)$
by interpolation.
Thus we are done by alternative (D).
 \end{itemize}

Now we may assume that at least  one of $\{b, c, \alpha, \delta\}$ is nonzero.
By a rotation, without loss of generality, we assume that $ b\neq 0$.
In this case,
if anyone of $\{b, c, \alpha, \beta, \gamma, \delta\}$ is $-1$,
 then we have $[1, 0, -1]$, and then also have  $[0, 1]^{\otimes 2}$
by Lemma~\ref{from-[1,0,-1]-to-[0,1]-tensor-2}. Then
 we are done by alternative (B).

Now we may assume that
 \[\{b, c, \alpha, \beta, \gamma, \delta\}\subseteq\{0, 1\} ~~~~\mbox{and}~~~~
b=1.\]
 In this case,
\begin{itemize}
\item For $c=1$, note that we have
\begin{align*}
M_{x_2, x_4x_3}(f^{x_1=0})
&=\begin{bmatrix}
f_{0000} & f_{0010} & f_{0001} & f_{0011}\\
f_{0100} & f_{0110} & f_{0101} & f_{0111}
\end{bmatrix}=
\begin{bmatrix}
1 & 0 & 0 & 1\\
0 & \alpha & \beta & 0
\end{bmatrix},\\
M_{x_1, x_4x_3}(f^{x_2=0})
&=\begin{bmatrix}
f_{0000} & f_{0010} & f_{0001} & f_{0011}\\
f_{1000} & f_{1010} & f_{1001} & f_{1011}
\end{bmatrix}=
\begin{bmatrix}
1 & 0 & 0 & 1\\
0 & \gamma & \delta & 0
\end{bmatrix},\\
M_{x_1x_2, x_4}(f^{x_3=0})
&=\begin{bmatrix}
f_{0000} & f_{0001} \\
f_{0100} & f_{0101}\\
f_{1000} & f_{1001} \\
f_{1100} & f_{1101}
\end{bmatrix}=
\begin{bmatrix}
1 & 0 \\
0 & \beta\\
0 & \delta \\
1 & 0
\end{bmatrix},\\
M_{x_1x_2, x_3}(f^{x_4=0})
&=\begin{bmatrix}
f_{0000} & f_{0010} \\
f_{0100} & f_{0110}\\
f_{1000} & f_{1010} \\
f_{1100} & f_{1110}
\end{bmatrix}=
\begin{bmatrix}
1 & 0 \\
0 & \alpha\\
0 & \gamma \\
1 & 0
\end{bmatrix}.
\end{align*}
If  $\{\alpha, \beta, \gamma, \delta\}$ has
at least one  $0$ and one  1,
then there exists $i\in[4]$ such that the support of $f^{x_i=0}$ is not affine.
Then by Lemma~\ref{0,1-valued}, we can construct $g\notin\mathscr{A}$ and
arity$(g)<$ arity$(f^{x_i=0})=3$,
and $g$ satisfies the even Parity Condition. Thus $g$ has arity 2.
This implies that, up to a nonzero factor,
$g$ has the form  $[1, 0, x] \not \in \mathscr{A}$.
Thus we are done by  alternative (A).

So we may assume that
 $\alpha=\beta=\gamma=\delta=0$
  or
 $\alpha=\beta=\gamma=\delta=1$.
 For $\alpha=\beta=\gamma=\delta=0$,
 \[M_{x_1x_2, x_4x_3}(f)=\begin{bmatrix}
1 & 0 & 0 & 1\\
0 & 0 & 0 & 0\\
0 & 0 & 0 & 0\\
1 & 0 & 0 & d
\end{bmatrix}.\]
Since $f\notin\mathscr{M}$, by Lemma~\ref{matchgate-identity-for-arity-4}
$\det \left[\begin{smallmatrix}
1 & 1\\
1 & d
\end{smallmatrix}\right] \not =0$.
By Lemma~\ref{interpolation-equality-4}, we have $(=_4)$ and we are done by
 alternative (D).

For $\alpha=\beta=\gamma=\delta=1$, $f$ is symmetric,
namely $f =[1, 0, 1, 0, d]$.
If $d=1$, then $f\in\mathscr{M}$. This is a contradiction.
Otherwise, $d\neq 1$, then $f\notin\widehat{\mathscr{P}} \cup
\mathscr{A} \cup \mathscr{M}$.
 Thus Pl-Holant$(\widehat{\mathcal{EQ}}, f)$ is \#P-hard by Theorem~\ref{heng-tyson-dichotomy-pl-csp}$'$.
\item If $c=0, d=0$, then the outer matrix of $f$ is degenerate.
Thus the inner matrix has full rank by $f\notin\mathscr{M}$ (Lemma~\ref{matchgate-identity-for-arity-4}).
This implies that either $\alpha \not = \beta$ or $\gamma \not = \delta$.
Being 0-1 valued, this implies that either
 supp$(f^{x_1=0})$ is not affine or supp$(f^{x_2=0})$ is not affine.
By Lemma~\ref{0,1-valued}, we can construct $[1, 0, x] \not
\in \mathscr{A}$ and we are done by  alternative (A).

\item If $c=0, d\neq 0$, then we have $\partial_{(=_2)}^{\{1, 2\}}=[1, 0, 1+d]$ and $\partial_{=_2}^{\{3, 4\}}=2[1, 0, \frac{d}{2}]$.

If $(\frac{d}{2})^4\neq 0, 1$, then  we are done by alternative (A) and $[1, 0, \frac{d}{2}]$.

Otherwise,  $d=\pm 2$ or $d=\pm 2\frak{i}$.
If $d\neq -2$, then $(1+d)^4\neq 0, 1$ and we are done by  alternative (A) and $[1, 0, 1+d]$.
If $d=-2$, then  $[1, 0, 1+d]=[1, 0, -1]$. By Lemma~\ref{from-[1,0,-1]-to-[0,1]-tensor-2},
 we have $[1, 0, -1]$ and $[0, 1]^{\otimes 2}$ and we are done by
 alternative (B).
\end{itemize}
\end{proof}

%% file: 6main.tex
\section{Main Theorem}
By Theorem~\ref{main-theorem-for-no-parity} and Theorem~\ref{With-arity-4-non-matchgate-signature}, we have the following
dichotomy theorem for Pl-Holant$(\widehat{\mathcal{EQ}}, \widehat{\mathcal{F}})$.
\begin{theorem}\label{main-dichotomy-thm}
 Let $\widehat{\mathcal{F}}$ be any set of complex-valued signatures in Boolean variables.
 Then $\operatorname{Pl-Holant}(\widehat{\mathcal{EQ}}, \widehat{\mathcal{F}})$ is $\SHARPP$-hard unless
 $\widehat{\mathcal{F}} \subseteq \mathscr{A}$,
 $\widehat{\mathcal{F}} \subseteq \widehat{\mathscr{P}}$, or
 $\widehat{\mathcal{F}} \subseteq \mathscr{M}$,
 in which case the problem is computable in polynomial time.
\end{theorem}

After the holographic transformation by $\left[\begin{smallmatrix} 1  &  1  \\ 1  &  -1\end{smallmatrix}\right]$, 
we have the following
dichotomy theorem for
planar \#CSP over the Boolean domain.
\begin{specialtheorem}
 Let $\mathcal{F}$ be any set of complex-valued signatures in Boolean variables.
 Then $\PlCSP(\mathcal{F})$ is $\SHARPP$-hard unless
 $\mathcal{F} \subseteq \mathscr{A}$,
 $\mathcal{F} \subseteq \mathscr{P}$, or
 $\mathcal{F} \subseteq \widehat{\mathscr{M}}$,
 in which case the problem is computable in polynomial time.
\end{specialtheorem}

%% file: 7ref.tex
\renewcommand{\refname}{References}

%% file: 8appendix.tex
\section{Appendix}
\subsection{Ternary Non-product Type Under Unary Actions}
In this subsection, we will show how to construct a binary non-product signature or a symmetric non-product signature by a
non-product signature of arity 3 with some unary signatures.
This is the base case of the induction in the proof of Theorem~\ref{arity-reduction-product}.

\input{nonproduct}

\subsection{Parity Condition Implies That $\mathscr{P}\subseteq\mathscr{A}$}
In this subsection, we give the following proposition that implies that if
a signature satisfies parity condition after the holographic transformation by
$H=\left[\begin{smallmatrix} 1 & 1 \\
1 & -1\end{smallmatrix}\right]$, then it is of product type implies that it is of affine type.

\begin{proposition}\label{parity-product-affine}
Let $f\in\mathscr{P}$ be a signature of arity $n$ and $\hat{f}=H^{\otimes n}f$, where
$H=\left[\begin{smallmatrix} 1 & 1 \\
1 & -1\end{smallmatrix}\right]$.
If $\hat{f}$ satisfies parity condition, then $f\in\mathscr{A}$.
\end{proposition}
\begin{proof}
Since $f\in\mathscr{P}$, there exist $f^i$ of arity $n_i$ for $1\leq i\leq s$ such that
$f=f^1\otimes f^2\otimes \cdots \otimes f^s$, where $f^i\in\mathcal{E}$.
Thus $\hat{f}=\hat{f}^1\otimes \hat{f}^2\otimes \cdots \otimes \hat{f}^s$, where $\hat{f}^i=H^{\otimes n_i}f_i$.
Since $\hat{f}$ satisfies parity condition, all of $\hat{f}^i$ satisfy parity condition.
Note that there exists $\alpha_i\in\{0, 1\}^{n_i}$ such that  supp$(f_i)\subseteq \{\alpha_i, \bar{\alpha_i}\}$ for $1\leq i\leq s$.

We claim that $f^i\in\mathscr{A}$ for $1\leq i\leq s$.
Let $f^i_{\alpha_i}=a_i, f^i_{\bar{\alpha_i}}=b_i$.
If $a_i=0$ or $b_i=0$, then $f^i\in\mathscr{A}$.
Otherwise, $a_ib_i\neq 0$. For any $\beta\in\{0, 1\}^{n_i}$,
if wt$(\beta)$ is even, then $\hat{f}^i_{\beta}=\pm (a_i+b_i)$.
If wt$(\beta)$ is odd, then $\hat{f}^i_{\beta}=\pm (a_i-b_i)$.
Since $\hat{f^i}$ satisfies parity condition, we have $a_i=\pm b_i$.
Thus $f^i\in\mathscr{A}$.
This finishes the proof of the claim.
Since $f^i\in\mathscr{A}$ for $1\leq i\leq s$, we have $f\in\mathscr{A}$.
\end{proof}

\subsection{Normalizing Signatures by A Binary Signature}
For a ternary signature $f$, where
$M_{x_1, x_2x_3}(f)=\left[\begin{smallmatrix}
f_{000} & f_{001} & f_{010} & f_{011}\\
f_{100} & f_{101} & f_{110} & f_{111}
\end{smallmatrix}\right]$,
 and a binary signature $[1, 0, \textbf{a}]$
 (note that $\textbf{a}$ is a scalar, not a vector,
and is written in bold to highlight the modification
in the following matrices), we often construct new signatures $f_i$
 by connecting one variable
of $[1, 0, \textbf{a}]$ to the variable $x_i$ of $f$, for $1\leq i\leq 3$.
Then
\begin{equation}\label{ternary-sig-matrix-modification-f1}
M_{x_1, x_2x_3}(f_1)=\begin{bmatrix*}[r]
f_{000} & f_{001} & f_{010} & f_{011}\\
\textbf{a}f_{100} & \textbf{a}f_{101} & \textbf{a}f_{110} & \textbf{a}f_{111}
\end{bmatrix*},
\end{equation}
\begin{equation}\label{ternary-sig-matrix-modification-f2}
M_{x_1, x_2x_3}(f_2)=\begin{bmatrix}
f_{000} & f_{001} & \textbf{a}f_{010} & \textbf{a}f_{011}\\
f_{100} & f_{101} & \textbf{a}f_{110} & \textbf{a}f_{111}
\end{bmatrix},
\end{equation}
\begin{equation}\label{ternary-sig-matrix-modification-f3}
M_{x_1, x_2x_3}(f_3)=\begin{bmatrix}
f_{000} & \textbf{a}f_{001} & f_{010} & \textbf{a}f_{011}\\
f_{100} & \textbf{a}f_{101} & f_{110} & \textbf{a}f_{111}
\end{bmatrix}.
\end{equation}

For  signatures of arity  4 we have similar operations.
In the following we list for a general signature of arity  4
 as well as one satisfying the even Parity Condition.
 This is to highlight graphically the locations
where $\textbf{a}$ appears.
(This operation will actually  be performed on
signatures of arity  4 satisfying the even Parity Condition.)
For
\[
M_{x_1x_2, x_4x_3}(f)
=\begin{bmatrix*}[r]
f_{0000} & f_{0010} & f_{0001} & f_{0011}\\
f_{0100} & f_{0110} & f_{0101} & f_{0111}\\
f_{1000} & f_{1010} & f_{1001} & f_{1011}\\
f_{1100} & f_{1110} & f_{1101} & f_{1111}
\end{bmatrix*}
~~~\mbox{or}~~~
\begin{bmatrix*}[r]
f_{0000} &          &          & f_{0011}\\
         & f_{0110} & f_{0101} &         \\
         & f_{1010} & f_{1001} &         \\
f_{1100} &          &          & f_{1111}
\end{bmatrix*},
\]
and a binary signature $[1, 0, \textbf{a}]$, we can construct new signatures
 $f_i$ or $g_i$
 by connecting one variable of $[1, 0, \textbf{a}]$ to the variable $x_i$
of $f$ or $g$, for $1\leq i\leq 4$.
Then
\begin{equation}\label{4-ary-sig-matrix-modification-f1}
M_{x_1x_2, x_4x_3}(f_1) =
\begin{bmatrix*}[r]
f_{0000} & f_{0010} & f_{0001} & f_{0011}\\
f_{0100} & f_{0110} & f_{0101} & f_{0111}\\
\textbf{a} f_{1000} & \textbf{a} f_{1010} & \textbf{a} f_{1001} & \textbf{a} f_{1011}\\
\textbf{a} f_{1100} & \textbf{a} f_{1110} & \textbf{a} f_{1101} & \textbf{a} f_{1111}
\end{bmatrix*}
\mbox{or}
\begin{bmatrix*}[r]
f_{0000} &          &          & f_{0011}\\
         & f_{0110} & f_{0101} &         \\
 & \textbf{a} f_{1010} & \textbf{a} f_{1001} & \\
 \textbf{a} f_{1100} &    &   &  \textbf{a} f_{1111}
\end{bmatrix*}
\end{equation}
\begin{equation}\label{4-ary-sig-matrix-modification-f2}
M_{x_1x_2, x_4x_3}(f_2)
=\begin{bmatrix*}[r]
f_{0000} & f_{0010} & f_{0001} & f_{0011}\\
\textbf{a} f_{0100} & \textbf{a} f_{0110} & \textbf{a} f_{0101} & \textbf{a} f_{0111}\\
f_{1000} & f_{1010} & f_{1001} & f_{1011}\\
\textbf{a} f_{1100} & \textbf{a} f_{1110} & \textbf{a} f_{1101} & \textbf{a} f_{1111}
\end{bmatrix*}
\mbox{or}
\begin{bmatrix*}[r]
f_{0000} &          &          & f_{0011}\\
 & \textbf{a} f_{0110} & \textbf{a} f_{0101} & \\
 & f_{1010} & f_{1001} & \\
\textbf{a} f_{1100} &  &  & \textbf{a} f_{1111}
\end{bmatrix*}
\end{equation}
\begin{equation}\label{4-ary-sig-matrix-modification-f3}
M_{x_1x_2, x_4x_3}(f_3)
=\begin{bmatrix*}[r]
~~f_{0000} & \textbf{a} f_{0010} & ~~f_{0001} & \textbf{a} f_{0011}\\
~~f_{0100} & \textbf{a} f_{0110} & ~~f_{0101} & \textbf{a} f_{0111}\\
~~f_{1000} & \textbf{a} f_{1010} & ~~f_{1001} & \textbf{a} f_{1011}\\
~~f_{1100} & \textbf{a} f_{1110} & ~~f_{1101} & \textbf{a} f_{1111}
\end{bmatrix*}
\mbox{or}
\begin{bmatrix*}[r]
~~f_{0000} &  &  & \textbf{a} f_{0011}\\
 & \textbf{a} f_{0110} & ~~f_{0101} & \\
 & \textbf{a} f_{1010} & ~~f_{1001} & \\
~~f_{1100} &  &  & \textbf{a} f_{1111}
\end{bmatrix*}
\end{equation}
\begin{equation}\label{4-ary-sig-matrix-modification-f4}
M_{x_1x_2, x_4x_3}(f_4)=
\begin{bmatrix*}[r]
~~f_{0000} & ~~f_{0010} & \textbf{a} f_{0001} & \textbf{a} f_{0011}\\
~~f_{0100} & ~~f_{0110} & \textbf{a} f_{0101} & \textbf{a} f_{0111}\\
~~f_{1000} & ~~f_{1010} & \textbf{a} f_{1001} & \textbf{a} f_{1011}\\
~~f_{1100} & ~~f_{1110} & \textbf{a} f_{1101} & \textbf{a} f_{1111}
\end{bmatrix*}
\mbox{or}
\begin{bmatrix*}[r]
~~f_{0000} &  &  & \textbf{a} f_{0011}\\
 & ~~f_{0110} & \textbf{a} f_{0101} & \\
 & ~~f_{1010} & \textbf{a} f_{1001} & \\
~~f_{1100} &  & & \textbf{a} f_{1111}
\end{bmatrix*}
\end{equation}

%% file: nonproduct.tex
We are given a ternany signature $f$ and a finite set of pair-wise
linearly independent unary signatures $[a_j, b_j]$ ($1 \le j \le m$).
We will write $ f_{abc}
= f_{x_1=a, x_2=b, x_3=c}$ the values
of $f$, for $a,b,c \in \{0,1\}$.
Let $\partial^{\{i\}}_{[a_j, b_j]}(f)$ denote the binary signature
obtained by connecting $[a_j, b_j]$ to the $i$-th variable of $f$.
For example, in matrix form, the binary signature  $\partial^{\{1\}}_{[a_j, b_j]}(f)$
takes the form
$\begin{bmatrix}
a_j f_{000} + b_j f_{100} &  a_j f_{001} + b_j f_{101} \\
a_j f_{010} + b_j f_{110} &  a_j f_{011} + b_j f_{111}
\end{bmatrix}$,
where $x_2 =0,1$ is the row index,
and $x_3 = 0,1$ is the column index.
 It is clear that a necessary and sufficient condition for
$\partial^{\{1\}}_{[a_j, b_j]}(f) \in \mathscr{P}$ is
\begin{align*}
&  & a_j f_{000} + b_j f_{100} =0, ~~~ a_j f_{011} + b_j f_{111} =0
~~~~~~~~~~~~~~~~~~~~~~~~~~~~~~~~~ & \mbox{($\partial^{\{1\}} {\rm D}_j$)} \\
& \mbox{or} &
a_j f_{001} + b_j f_{101} =0, ~~~ a_j f_{010} + b_j f_{110} =0
~~~~~~~~~~~~~~~~~~~~~~~~~~~~~~~~~ & \mbox{($\partial^{\{1\}} {\rm E}_j$)} \\
& \mbox{or} &
~~~~~~\left|
\begin{matrix}
a_j f_{000} + b_j f_{100} &  ~a_j f_{001} + b_j f_{101} \\
a_j f_{010} + b_j f_{110} &  ~a_j f_{011} + b_j f_{111}
\end{matrix}
\right| =0
~~~~~~~~~~~~~~~~~~~~~~~~~~~~~~~~~ & \mbox{($\partial^{\{1\}} \det_j$)}
\end{align*}

We can similarly define the conditions
($\partial^{\{i\}} {\rm D}_j$), ($\partial^{\{i\}} {\rm E}_j$), ($\partial^{\{i\}} \det_j$)
for $1 \le i \le 3$, $1 \le j \le m$.

\begin{lemma}\label{separation}
Let
unary signatures $[a_j, b_j]$ ($1 \le j \le 2$) be
linearly independent.
Suppose $\partial^{\{i\}}_{[a_j, b_j]}(f) \in \mathscr{P}$
for all $1 \le i \le 3$ and $1 \le j \le 2$.
If $f(x_1, x_2, x_3) = g(x_r, x_s) h(x_t)$
where $\{r,s,t\} = \{1,2,3\}$, then $f \in \mathscr{P}$.
\end{lemma}
\proof
If $h$ is identically zero then so is $f$ and the claim is trivial.
Otherwise, by linear independence there exists $1 \le j \le 2$
such that $\partial^{\{t\}}_{[a_j, b_j]}(h)$ is a nonzero constant $c$.
Then $g(x_r, x_s) = c^{-1} \partial^{\{t\}}_{[a_j, b_j]}(f) \in \mathscr{P}$.
Hence $f(x_1, x_2, x_3) = g(x_r, x_s) h(x_t) \in \mathscr{P}$.
\qed

\begin{lemma}\label{linear-at-least-2-easy}
Let $m \ge 3$
and let unary signatures $[a_j, b_j]$ ($1 \le j \le m$) be  pair-wise
linearly independent. Suppose $\partial^{\{i\}}_{[a_j, b_j]}(f) \in \mathscr{P}$
for all $1 \le i \le 3$ and $1 \le j \le m$.
If for some $1 \le i \le 3$,
\begin{itemize}
\item
there are two distinct $j$ such that
{\rm ($\partial^{\{i\}} {\rm D}_j$)} hold, or
\item
there are two distinct $j$ such that
{\rm ($\partial^{\{i\}} {\rm E}_j$)} hold
\end{itemize}
then $f \in \mathscr{P}$.
\end{lemma}
\proof
By symmetry of the 3 variables we may assume $i=1$.
By pair-wise linear independence of $[a_j, b_j]$ ($1 \le j \le m$)
we have either
\begin{enumerate}
\item
 $f_{000} = f_{100} = f_{011} = f_{111} =0$  (by
($\partial^{\{1\}} {\rm D}_j$) for two distinct $j$) or
\item
$f_{001} = f_{101} = f_{010} = f_{110} =0$  (by
($\partial^{\{1\}} {\rm E}_j$) for two distinct $j$).
\end{enumerate}
Suppose it is the first case.

By pair-wise linear independence, there exists some $1 \le j \le 3$
such that $a_j \not = 0$ and $b_j \not = 0$. Consider
 ($\partial^{\{2\}} {\rm D}_j$),
($\partial^{\{2\}} {\rm E}_j$) and
($\partial^{\{2\}} \det_j$), namely
\begin{eqnarray*}
& & \cancel{a_j f_{000}} + b_j f_{010} =0, ~~~
    a_j f_{101} + \cancel{b_j f_{111}} =0
~~~~~~~~~~~~~~~~~~~~~~~~~~~~~~~~~\mbox{($\partial^{\{2\}} {\rm D}_j$)} \\
& \mbox{or} &
a_j f_{001} + \cancel{b_j f_{011}} =0, ~~~
\cancel{a_j f_{100}} + b_j f_{110} =0
~~~~~~~~~~~~~~~~~~~~~~~~~~~~~~~~~\mbox{($\partial^{\{2\}} {\rm E}_j$)} \\
& \mbox{or} &
~~~~\left|
\begin{matrix}
\cancel{a_j f_{000}} + b_j f_{010} &  ~a_j f_{001} + \cancel{b_j f_{011}} \\
\cancel{a_j f_{100}} + b_j f_{110} &  ~a_j f_{101} + \cancel{b_j f_{111}}
\end{matrix}
\right| =0
~~~~~~~~~~~~~~~~~~~~~~~~~~~~~~~~~\mbox{($\partial^{\{2\}} \det_j$)}
\end{eqnarray*}
We have $f_{010} = f_{101} = 0$ from ($\partial^{\{2\}} {\rm D}_j$)
or
$f_{001} = f_{110} = 0$ from ($\partial^{\{2\}} {\rm E}_j$)
or
$\left|
\begin{matrix}
f_{010} &  f_{001}  \\
f_{110} &  f_{101}
\end{matrix}
\right| =0
$ from ($\partial^{\{2\}} \det_j$).
When $f_{010} = f_{101} = 0$ the support of $f$ is contained in
the two diagonal points $\{001, 110\}$, and hence $f \in \mathscr{P}$.
When $f_{001} = f_{110} = 0$ the support of $f$ is contained in
the two diagonal points $\{010,101\}$, and again $f \in \mathscr{P}$.
Suppose $\left|
\begin{matrix}
f_{010} &  f_{001}  \\
f_{110} &  f_{101}
\end{matrix}
\right| =0
$. Then $f$ is the product of the functions $(x_2 \not = x_3)$
and the degenerate function $g(x_1, x_3)$ with the signature in matrix form
$\begin{bmatrix}
f_{010} &  f_{001}  \\
f_{110} &  f_{101}
\end{bmatrix}$, where $x_1 =0,1$ is the row index,
and $x_3 = 0,1$ is the column index.

The second case $f_{001} = f_{101} = f_{010} = f_{110} =0$ is similar.
We exchange all the crossed-out terms in
($\partial^{\{2\}} {\rm D}_j$),
($\partial^{\{2\}} {\rm E}_j$) and
($\partial^{\{2\}} \det_j$) with the uncrossed-out terms.
The conclusions from
($\partial^{\{2\}} {\rm D}_j$)
or from ($\partial^{\{2\}} {\rm E}_j$) are still that
 the support of $f$ is contained in
two diagonal points.
From ($\partial^{\{2\}} \det_j$) we get
$\left|
\begin{matrix}
f_{000} &  f_{011} \\
f_{100} &  f_{111}
\end{matrix}
\right| =0$ and we conclude that
$f$
is the product of the functions $(x_2  = x_3)$
and the degenerate function $g(x_1, x_2)$ with the signature in matrix form
$\begin{bmatrix}
f_{000} &  f_{011} \\
f_{100} &  f_{111}
\end{bmatrix}$, where $x_1 =0,1$ is the row index,
and $x_2 = 0,1$ is the column index.
\qed

\begin{lemma}\label{det-at-least-3-easy}
Let $m \ge 3$
and let unary signatures $[a_j, b_j]$ ($1 \le j \le m$) be  pair-wise
linearly independent. Suppose $\partial^{\{i\}}_{[a_j, b_j]}(f) \in \mathscr{P}$
for all $1 \le i \le 3$ and $1 \le j \le m$.
If for some $1 \le i \le 3$, there are three  distinct $j$ such that
{\rm ($\partial^{\{i\}} \det_j$)} hold, then $f \in \mathscr{P}$.
\end{lemma}
\proof
By symmetry of the 3 variables we may assume $i=1$.
Each ($\partial^{\{1\}} \det_j$) is a quadratic form in $a_j$ and $b_j$,
\begin{equation}\label{det-as-lin-system}
\left|
\begin{matrix}
f_{000} & f_{001} \\
f_{010} & f_{011}
\end{matrix}
\right| a_j^2
+
\left(
\left|
\begin{matrix}
f_{000} &  f_{101} \\
f_{010} & f_{111}
\end{matrix}
\right|
+
\left|
\begin{matrix}
f_{100} & f_{001} \\
f_{110} & f_{011}
\end{matrix}
\right|
\right) a_j b_j
+
\left|
\begin{matrix}
f_{100} &  f_{101} \\
f_{110} &  f_{111}
\end{matrix}
\right| b_j^2 =0.
\end{equation}
Assume ($\partial^{\{1\}} \det_j$) holds for 3 distinct values $j, k, \ell$.
By  pair-wise
linear independence, we have a non-singular $3 \times 3$ matrix
$\left|
\begin{matrix}
a_j^2 & a_j b_j & b_j^2 \\
a_k^2 & a_k b_k & b_k^2 \\
a_{\ell}^2 & a_{\ell} b_{\ell} & b_{\ell}^2
\end{matrix}
\right| \not =0$.
Indeed, if all $a_j, a_k, a_{\ell} \not =0$, then
the determinant is
$a_j^2 a_k^2 a_{\ell}^2
\left|
\begin{matrix}
1 &  b_j/a_j & (b_j/a_j)^2 \\
1 &  b_k/a_k & (b_k/a_k)^2 \\
1 &  b_{\ell}/a_{\ell} & (b_{\ell}/a_{\ell})^2
\end{matrix}
\right|$, where the Vandemonde determinant is nonzero
because $b_j/a_j$, $b_k/a_k$ and $b_{\ell}/a_{\ell}$ are pair-wise
distinct.
If any  $a_j, a_k, a_{\ell} =0$, say $a_{\ell}=0$,
then by pair-wise
linear independence $a_j, a_k  \not =0$, and $b_{\ell} \not =0$,
and
$\left|
\begin{matrix}
a_j^2 & a_j b_j & b_j^2 \\
a_k^2 & a_k b_k & b_k^2 \\
a_{\ell}^2 & a_{\ell} b_{\ell} & b_{\ell}^2
\end{matrix}
\right|
= b_{\ell}^2
a_j^2 a_k^2
\left|
\begin{matrix}
1 & b_j/a_j\\
1 & b_k/a_k
\end{matrix}
\right|
\not = 0$.

It follows from (\ref{det-as-lin-system}) that
\begin{equation}\label{det-as-lin-system-coeff-0}
\left|
\begin{matrix}
f_{000} & f_{001} \\
f_{010} & f_{011}
\end{matrix}
\right| =0,~~~~
\left|
\begin{matrix}
f_{000} &  f_{101} \\
f_{010} & f_{111}
\end{matrix}
\right|
+
\left|
\begin{matrix}
f_{100} & f_{001} \\
f_{110} & f_{011}
\end{matrix}
\right| =0,~~~~
\left|
\begin{matrix}
f_{100} &  f_{101} \\
f_{110} &  f_{111}
\end{matrix}
\right| =0.
\end{equation}
There is a transitive group action on the four vectors
\[f_{0 \bullet 0} = \begin{bmatrix} f_{000} \\ f_{010} \end{bmatrix},
~~~~
f_{0 \bullet 1} = \begin{bmatrix} f_{001} \\ f_{011} \end{bmatrix},
~~~~
f_{1 \bullet  0} = \begin{bmatrix} f_{100} \\ f_{110} \end{bmatrix}
~~~~
f_{1 \bullet 1} = \begin{bmatrix} f_{101} \\ f_{111} \end{bmatrix},\]
generated by the permutations
$\sigma$  exchanging $f_{0 \bullet 0} \leftrightarrow f_{1 \bullet  0}$
and $f_{0 \bullet 1} \leftrightarrow f_{1 \bullet 1}$,
and $\tau$  exchanging
 $f_{0 \bullet 0} \leftrightarrow f_{0 \bullet 1}$
and $f_{1 \bullet  0}  \leftrightarrow f_{1 \bullet 1}$.
This group
$\mathbb{Z}_2 \times \mathbb{Z}_2$ preserves the equations
(\ref{det-as-lin-system-coeff-0}).  Thus either all four vectors are
zero in which case $f \in \mathscr{P}$ trivially, or we may
assume $f_{0 \bullet 0} \not =0$.

By the first equation in (\ref{det-as-lin-system-coeff-0}),
there exists $\lambda \in \mathbb{C}$ such that $f_{0 \bullet 1}
= \lambda f_{0 \bullet 0}$. Substituting $f_{0 \bullet 1}$
into the second equation
in (\ref{det-as-lin-system-coeff-0}), we get
$\left|
\begin{matrix}
f_{000} &  f_{101} - \lambda f_{100}  \\
f_{010} & f_{111}  - \lambda f_{110}
\end{matrix}
\right|
=0$ and thus
there exists $\mu \in \mathbb{C}$ such that $f_{1 \bullet  1}
- \lambda f_{1 \bullet  0} =
\mu  f_{0 \bullet 0}$.
If $\mu =0$ then $f_{1 \bullet  1} = \lambda f_{1 \bullet  0}$.
Then $f$ is the product of the unary function $[1, \lambda]$ on $x_3$
and the binary function $g(x_1,x_2)$ with the signature matrix
$
\begin{bmatrix}
f_{000} &  f_{010}  \\
f_{100} &  f_{110}
\end{bmatrix}$, where
 $x_1 =0,1$ is the row index,
and $x_2 = 0,1$ is the column index.
By Lemma~\ref{separation} we are done.

Suppose $\mu \ne 0$.
Substituting $f_{1 \bullet  1} = \lambda f_{1 \bullet  0}
+ \mu  f_{0 \bullet 0}$ into the third  equation
in (\ref{det-as-lin-system-coeff-0}),
there exists $\nu  \in \mathbb{C}$ such that
$f_{1 \bullet  0} =  \nu f_{0 \bullet 0}$,
and $f_{1 \bullet  1} = (\lambda \nu + \mu) f_{0 \bullet 0}$.
Hence, $f$ is the product of the unary function $[f_{000}, f_{010}]$ on $x_2$
and the binary function $g(x_1, x_3)$ with the signature matrix
$
\begin{bmatrix}
1 &  \lambda  \\
\nu &  \lambda \nu + \mu
\end{bmatrix}$, where
 $x_1 =0,1$ is the row index,
and $x_3 = 0,1$ is the column index.
By Lemma~\ref{separation} we are done.
\qed

\begin{lemma}\label{5-unaries-suffice}
Let $m \ge 5$
and let unary signatures $[a_j, b_j]$ ($1 \le j \le m$) be  pair-wise
linearly independent. Suppose $\partial^{\{i\}}_{[a_j, b_j]}(f) \in \mathscr{P}$
for all $1 \le i \le 3$ and $1 \le j \le m$.
Then $f \in \mathscr{P}$.
\end{lemma}
\proof
For all $1 \le j \le m$,
either ($\partial^{\{1\}} {\rm D}_j$) or ($\partial^{\{1\}} {\rm E}_j$) or
 ($\partial^{\{1\}} \det_j$) holds,
since $\partial^{\{1\}}_{[a_j, b_j]}(f) \in \mathscr{P}$.
Since $m \ge 5$, either  ($\partial^{\{1\}} {\rm D}_j$) is satisfied for
at least two distinct $j$, or
 ($\partial^{\{1\}} {\rm E}_j$) is satisfied for
at least two distinct $j$, or
($\partial^{\{1\}} \det_j$) is satisfied for
at least three distinct $j$.
Hence by Lemma~\ref{linear-at-least-2-easy} and
Lemma~\ref{det-at-least-3-easy},
$f \in \mathscr{P}$.
\qed

\vspace{.1in}

Consider the Boolean cube with its four diagonal pairs.
We will name them $a,b,c,d$,
where
\[a=(000,111), ~~~~b=(001,110), ~~~~c=(010,101), ~~~~d=(011,100)\]
respectively.
A consequence of each statement ($\partial^{\{i\}} {\rm D}_j$)
or ($\partial^{\{i\}} {\rm E}_j$)
is that for some two diagonal pairs,
the values of $f$ at the diagonal pairs have the same
product. E.g., for any $j$,
($\partial^{\{1\}}  {\rm D}_j$) implies that
the product $f_{000}f_{111} = f_{011}f_{100}$, i.e., the product
of the values of
 $f$ at the diagonal $a$ is the same as that at the diagonal $d$.
Similarly, the statement ($\partial^{\{1\}}  {\rm E}_j$) implies that
the product $f_{001}f_{110} = f_{010}f_{101}$,
 i.e., the product
of the values of
 $f$ at the diagonal $b$ is the same as that at the diagonal $c$.
For ($\partial^{\{2\}} {\rm D}_j$)
(respectively ($\partial^{\{2\}} {\rm E}_j$)), the implications are
for the diagonals $a$ and $c$ (respectively $b$ and $d$).
For ($\partial^{\{3\}} {\rm D}_j$)
(respectively ($\partial^{\{3\}} {\rm E}_j$)), the implications are
for the diagonals $a$ and $b$ (respectively $c$ and $d$).

Define a graph on the vertex set $\{a,b,c,d\}$ where we
add an edge
whenever the corresponding diagonals have the same product value of $f$,
we get a spanning subgraph of $K_4$ on $\{a,b,c,d\}$.
If there are at least 4 edges in this spanning subgraph, then
it is connected, with the implication that all diagonals have the
same product value of $f$.
If this spanning subgraph  is connected, then
 each statement ($\partial^{\{i\}} \det_j$) takes the form
\begin{equation}\label{det-just-quadratic-term}
D_0 a_j^2 + D_2 b_j^2 =0
\end{equation}
for some coefficients $D_0$ and $D_2$ with a zero coefficient
of the cross term $a_jb_j$.

 Moreover the only disconnected spanning subgraph
with 3 edges in $K_4$ is a triangle, meaning that the 3 statements
among all ($\partial^{\{i\}} {\rm D}_j$) and ($\partial^{\{i\}} {\rm E}_j$)
that hold must be those with implications among only three letters
out of four $\{a,b,c,d\}$.
For example, if the triangle  is on $\{a,b,c\}$, then the 3 statements
must be among ($\partial^{\{1\}}  {\rm E}_j$),
($\partial^{\{2\}}  {\rm D}_j$)
and
($\partial^{\{3\}}   {\rm D}_j$)
 (but not ($\partial^{\{1\}}  {\rm D}_j$),
not ($\partial^{\{2\}}   {\rm E}_j$)
and not
  ($\partial^{\{3\}}   {\rm E}_j$)).
Similarly if the triangle  is on $\{b,c,d\}$, then the 3 statements
must be among ($\partial^{\{1\}}  {\rm E}_j$),
($\partial^{\{2\}}  {\rm E}_j$)
and  ($\partial^{\{3\}}   {\rm E}_j$)
 (but not ($\partial^{\{1\}}  {\rm D}_j$),
not ($\partial^{\{2\}}   {\rm D}_j$)
and not ($\partial^{\{3\}}   {\rm D}_j$)).

\begin{lemma}\label{root-i-suffice}
Suppose $f \not \in \mathscr{P}$.
Let $[1, b_j]$ ($1 \le j \le 4$) be the unary signatures
$[1,1], [1,-1], [1, \mathfrak{i}], [1, -\mathfrak{i}]$ respectively.
Suppose $\partial^{\{i\}}_{[1, b_j]}(f) \in \mathscr{P}$
for all $1 \le i \le 3$ and $1 \le j \le 4$.
Then $f$ is a product of some unary functions $[1, b_j]$  and the symmetric
function $[1,1,-1,-1]$.
\end{lemma}
\proof
By Lemma~\ref{linear-at-least-2-easy} and Lemma~\ref{det-at-least-3-easy},
we may assume that for each  $1 \le i \le 3$
there can be at most one $j$ such that
($\partial^{\{i\}} {\rm D}_j$) holds, and
at most one $j'$ such that
($\partial^{\{i\}} {\rm E}_{j'}$) holds, and
at most two distinct values $k$ and $k'$ such that
($\partial^{\{i\}} \det_k$)  and ($\partial^{\{i\}} \det_{k'}$) hold.
Moreover since there are 4 such requirements that must be
satisfied altogether for the same $i$, there exist exactly
one such $j$, $j'$, $k$ and $k'$ respectively, and
$\{j, j', k, k'\} = \{1,2,3,4\}$.

The spanning subgraph of $K_4$ in this case is the full graph $K_4$,
and all diagonals have the same product value.
For any $1 \le i \le 3$ if the two valid determinantal statements
($\partial^{\{i\}} \det_k$)  and ($\partial^{\{i\}} \det_{k'}$)
are not for the pair $\{[1,1], [1,-1]\}$ or
$\{[1,\mathfrak{i}], [1,-\mathfrak{i}]\}$,
 then we have $b_k^2 \not = b_{k'}^2$.
By the form of (\ref{det-just-quadratic-term}) with the vanished
cross term, it follows that all 4 statements ($\partial^{\{i\}} \det_{\ell}$)
must hold, where $1 \le \ell \le 4$.
 By Lemma~\ref{det-at-least-3-easy}
we are done.
Hence we may assume $b_k = \pm b_{k'}$.
Furthermore $b_k \not = b_{k'}$ because $\{j, j', k, k'\} = \{1,2,3,4\}$,
hence $b_k  = - b_{k'} \in \{1, -1\}$ or
 $b_k = - b_{k'} \in \{ \mathfrak{i}, -\mathfrak{i}\}$.
This also implies that $b_j  = - b_{j'}$.

For the valid ($\partial^{\{1\}} {\rm D}_j$)
let $x = -b_j \in \{1,-1, \mathfrak{i}, -\mathfrak{i}\}$. This is the multiplier for which
$f_{000} = x f_{100}$ and $f_{011} = x f_{111}$.
 The corresponding multiplier for ($\partial^{\{1\}} {\rm E}_{j'}$)
is $-x$ such that
$f_{001} = -x f_{101}$ and $f_{010} = -x f_{110}$.
Similarly for the  valid ($\partial^{\{2\}} {\rm D}_j$)
we define $y \in  \{1,-1, \mathfrak{i}, -\mathfrak{i}\}$ such that
$f_{000} = y f_{010}$ and $f_{101} = y f_{111}$.
Also
$f_{001} = -y f_{011}$ and $f_{100} = -y f_{110}$.
Finally  for the  valid ($\partial^{\{3\}}  {\rm D}_j$)
we define $z  \in  \{1,-1, \mathfrak{i}, -\mathfrak{i}\}$ such that
$f_{000} = z f_{001}$ and $f_{110} = z f_{111}$.
Also
$f_{010} = -z f_{011}$ and $f_{100} = -z f_{101}$.

If we let $g(x_1, x_2, x_3)$ be the product function
 $g_1(x_1) g_2(x_2) g_3(x_3)$ where $g_1(x_1) = [x, 1]_{x_1},
g_2(x_2) = [y, 1]_{x_2}, g_3(x_3) = [z,  1]_{x_3}$, i.e.,
$g = [x, 1]_{x_1}  \otimes
[y, 1]_{x_2} \otimes [z,  1]_{x_3}$, then
$f = g h$ where $h$  is the ternary symmetric function $f_{000}[1,1,-1,-1]$.
\qed

\begin{corollary}\label{cor:root-i-suffice}
Let $\mathcal{F}$ be a set of signatures containing a ternary signature
$f \not \in \mathscr{P}$.
Suppose $\mathcal{F}$ contains the
unary signatures
$\{[1, b_j] \mid 1 \le j \le 4\}
=\{[1,1], [1,-1], [1, \mathfrak{i}], [1, -\mathfrak{i}]\}$,
and
$\partial^{\{i\}}_{[1, b_j]}(f) \in \mathscr{P}$
for all $1 \le i \le 3$ and $1 \le j \le 4$.
Then
\[\text{\rm Pl-}\#{\rm CSP}(\mathcal{F} \cup\{ [1,1,-1,-1]\})
\le_{\tt T} \text{\rm Pl-}\#{\rm CSP}(\mathcal{F}).\]
\end{corollary}

\begin{lemma}\label{3rd-root-suffice}
Suppose $f \not \in \mathscr{P}$.
Let $[1, b_j]$ ($1 \le j \le 3$) be the unary signatures
$[1,1], [1,\omega], [1,\omega^2]$ respectively.
Suppose $\partial^{\{i\}}_{[1, b_j]}(f) \in \mathscr{P}$
for all $1 \le i \le 3$ and $1 \le j \le 3$.
Then $f$ is a product of some unary functions $[1, b_j]$  with
\begin{enumerate}
\item
the symmetric
function $[1, -1, x, -x]$ where $x \in \{\omega, \omega^2\}$, or
\item
the symmetric
function $[-2,1,1,-2]$, or
\item
after a cyclic permutation of its three variables
a ternary function  $g(1-x_1, x_2, x_3)$ where $g(x_1, x_2, x_3)$
is the symmetric function $[-2, 1,1,-2]$.
\end{enumerate}
\end{lemma}
\proof
By Lemma~\ref{linear-at-least-2-easy} and Lemma~\ref{det-at-least-3-easy},
we may assume that for each  $1 \le i \le 3$
there can be at most one $j$ such that
($\partial^{\{i\}} {\rm D}_j$) holds, and
at most one $j'$ such that
($\partial^{\{i\}} {\rm E}_{j'}$) holds, and
at most two distinct values $k$ and $k'$ such that
($\partial^{\{i\}} \det_k$)  and ($\partial^{\{i\}} \det_{k'}$) hold.

Let $N$ be the total number of  valid statements among
($\partial^{\{i\}} {\rm D}_j$) and ($\partial^{\{i\}} {\rm E}_j$).
If $N \le 2$ then for some $1 \le i \le 3$, all three statements
($\partial^{\{i\}} \det_j$) for $1 \le j \le 3$ must hold.
Hence $N \ge 3$.

We first assume that the spanning subgraph of $K_4$ defined
by the valid statements ($\partial^{\{i\}} {\rm D}_j$) and ($\partial^{\{i\}} {\rm E}_j$)
is connected.  In particular if  $N \ge 4$, then this is the case.
In this case, all statements
($\partial^{\{i\}} \det_j$) are of the form (\ref{det-just-quadratic-term})
with a vanishing cross term.
For any two disinct $b_j, b_k \in \{1, \omega, \omega^2\}$,
$b_j^2 \not = b_k^2$, therefore for any $1 \le i \le 3$, if there are two
distinct valid statements
($\partial^{\{i\}} \det_j$) and ($\partial^{\{i\}} \det_{j'}$) ($j \not = j'$),
then ($\partial^{\{i\}} \det_k$) is valid for
all $1 \le k \le 3$. By Lemma~\ref{det-at-least-3-easy} we
reach a contradiction to  $f \not \in \mathscr{P}$.
Therefore there cannot be more than one valid ($\partial^{\{i\}} \det_j$)
for each $1 \le i \le 3$. It follows that for every  $1 \le i \le 3$,
 there is exactly one
valid ($\partial^{\{i\}} {\rm D}_j$), one valid ($\partial^{\{i\}} {\rm E}_k$)
and one valid ($\partial^{\{i\}} \det_{\ell}$), such that
$\{j,k,\ell\} = \{1,2,3\}$.

For the valid ($\partial^{\{1\}} {\rm D}_j$), let $x = b_j
\in \{1, \omega, \omega^2\}$,
then $f_{000} = -x f_{100}$ and $f_{011} = -x f_{111}$.
The corresponding multiplier for the valid ($\partial^{\{1\}} {\rm E}_{k}$)
is $x'$ where $x' = x \omega$ or $x \omega^2$, such that
$f_{001} = -x' f_{101}$ and $f_{010} = -x' f_{110}$.
Similarly for the  valid ($\partial^{\{2\}} {\rm D}_j$)
we define $y \in  \{1, \omega, \omega^2\}$ such that
$f_{000} = -y f_{010}$ and $f_{101} = -y f_{111}$.
Also
$f_{001} = -y' f_{011}$ and $f_{100} = -y' f_{110}$,
where $y' = y \omega$ or $y \omega^2$.
Finally  for the  valid ($\partial^{\{3\}}  {\rm D}_j$)
we define $z  \in   \{1, \omega, \omega^2\}$ such that
$f_{000} = -z f_{001}$ and $f_{110} = -z f_{111}$.
Also
$f_{010} = -z' f_{011}$ and $f_{100} = -z' f_{101}$,
where $z' = z \omega$ or $z \omega^2$.
Clearly if any value $f_{abc} =0$ then $f$ is identically 0,
a contradiction to  $f \not \in \mathscr{P}$.
So we may assume $f$ has no zero values.
By consistency of values, $f_{001} = -x' f_{101} = x'y f_{111}$
and $f_{001} = -y' f_{011} = x y' f_{111}$,
hence $x'/x = y'/y$. Similarly $f_{000} = -z f_{001} = y'z f_{011}$
and $f_{000} = -y f_{010} = y z' f_{011}$,
hence $y'/y = z'/z$.

Let $\rho = x'/x = y'/y = z'/z \in \{\omega, \omega^2\}$.
 Let $g(x_1, x_2, x_3)$ be the product function
 $g_1(x_1) g_2(x_2) g_3(x_3)$ where $g_1(x_1) = [-x, 1]_{x_1},
g_2(x_2) = [-y, 1]_{x_2}, g_3(x_3) = [-z,  1]_{x_3}$, i.e.,
$g = [-x, 1]_{x_1}  \otimes
[-y, 1]_{x_2} \otimes [-z,  1]_{x_3}$. Then
$f = g h$ where $h$ is the ternary symmetric function $f_{111}[-\rho,
\rho, -1, 1]$. Alternatvely we have $h = f_{000}[1,-1, \xi, -\xi]$,
where $\xi =  1/\rho \in  \{\omega, \omega^2\}$.

Now suppose $N=3$ and  the spanning subgraph
with 3 edges in $K_4$ is a triangle.
Then for each  $1 \le i \le 3$ there are exactly
two distinct  values $k$ and $k'$ such that
($\partial^{\{i\}} \det_k$)  and ($\partial^{\{i\}} \det_{k'}$) hold.
It is either the triangle on $\{b,c,d\}$
or a triangle
involving the vertex $a$, in which case by a cyclic permutation
of the 3 variables,
we may assume it is $\{a,b,c\}$.
For the  triangle on $\{b,c,d\}$ the 3 valid statements
among all ($\partial^{\{i\}}  {\rm D}_j$) and ($\partial^{\{i\}}  {\rm E}_j$)
must be among ($\partial^{\{1\}}  {\rm E}_j$),
($\partial^{\{2\}}  {\rm E}_j$)
and  ($\partial^{\{3\}}   {\rm E}_j$).
For the  triangle on $\{a,b,c\}$  the 3 valid statements
among all ($\partial^{\{i\}}  {\rm D}_j$) and ($\partial^{\{i\}}  {\rm E}_j$)
must be among ($\partial^{\{1\}}  {\rm E}_j$),
($\partial^{\{2\}}  {\rm D}_j$)
and
($\partial^{\{3\}}   {\rm D}_j$).

We first consider the triangle  $\{b,c,d\}$ case.
According to the valid ($\partial^{\{1\}}  {\rm E}_j$) we let
$x = b_j \in \{1, \omega, \omega^2\}$, then
$f_{001} = -x f_{101}$ and $f_{010} = -x f_{110}$.
Then the two valid ($\partial^{\{1\}} \det_{k}$)
and ($\partial^{\{1\}} \det_{\ell}$) hold for $b_k = b_j \omega$ and $b_{\ell} =
b_j \omega^2$. Hence
the following equation has two roots $X = x \omega$ and $X = x \omega^2$:
\begin{equation}\label{omega-det-two-eqn}
\left|
\begin{matrix}
f_{000} &  f_{001} \\
f_{010} & f_{011}
\end{matrix}
\right|
+
\left|
\begin{matrix}
f_{000} &  f_{101} \\
f_{010} & f_{111}
\end{matrix}
\right|
X
+
\left|
\begin{matrix}
f_{100} &  f_{101} \\
f_{110} & f_{111}
\end{matrix}
\right|
X^2
=0.
\end{equation}
Notice that we have used the fact that one cross term is zero:
$\left|
\begin{matrix}
f_{100} &  f_{001} \\
f_{110} & f_{011}
\end{matrix}
\right|
=0$, because the diagonal $b$ and $d$ have the same product value
$f_{001}f_{110} = f_{011} f_{100}$.
Subtracting one equation from another in (\ref{omega-det-two-eqn})
with $X = x \omega$ and $X = x \omega^2$ we get
\begin{equation}\label{linear-omega-det-two-eqn-diff}
\left|
\begin{matrix}
f_{000} &  f_{101} \\
f_{010} & f_{111}
\end{matrix}
\right| = x \left|
\begin{matrix}
f_{100} &  f_{101} \\
f_{110} & f_{111}
\end{matrix}
\right|.
\end{equation}

Similarly we have a valid ($\partial^{\{2\}}  {\rm E}_j$) for some $j$
according to which we let $y=b_j \in \{1, \omega, \omega^2\}$, and then
$f_{001} = -y f_{011}$ and $f_{100} = -y f_{110}$.
Also two statements ($\partial^{\{2\}} \det_{k}$)
and ($\partial^{\{2\}} \det_{\ell}$) hold for $b_k = b_j \omega$ and $b_{\ell} =
b_j \omega^2$.
Finally we have a valid ($\partial^{\{3\}}   {\rm E}_j$) for some $j$
according to which we let $z=b_j \in \{1, \omega, \omega^2\}$, and then
$f_{100} = -z f_{101}$ and $f_{010} = -z f_{011}$.
Also two statements ($\partial^{\{3\}} \det_{k}$)
and ($\partial^{\{2\}} \det_{\ell}$) hold for $b_k = b_j \omega$ and $b_{\ell} =
b_j \omega^2$.

It follows that
\[f_{001} = -x f_{101},
~~~~f_{100} = -z f_{101},
~~~~f_{110} = \frac{z}{y} f_{101},
~~~~f_{011} = \frac{x}{y} f_{101},
~~~~f_{010} = -\frac{xz}{y} f_{101}.\]
Let $g(x_1, x_2, x_3) = f(x_1, x_2, x_3)/
\left( [-x, 1]_{x_1} \otimes [-y, 1]_{x_2}
 \otimes [-z, 1]_{x_3} \right)$, then
$g^{000} = f_{000}/(-xyz)$,
$g^{111} = f_{111}$,
 and $g^{001} = f_{001}/(xy) = (-1/y) f_{101}
= g^{101}$. Similarly we can show
\[g^{001} = g^{100} = g^{110} = g^{011} =  g^{010} = g^{101}.\]
If $g^{101} =0$ then $g \in \mathscr{P}$, and so does
$f$, a contradiction to  $f \not \in \mathscr{P}$.
Hence we may normalize by a constant and assume $g^{101} =1$.

After some computation,  equation
(\ref{linear-omega-det-two-eqn-diff}) simplifies to
\begin{equation}\label{linear-omega-det-two-eqn-diff-consequence}
YZ + Z - 2 =0
\end{equation}
 where $Y = g^{000}$ and $Z = g^{111}$.
The equation (\ref{omega-det-two-eqn}) for the root $X =x \omega$
simplifies to
\begin{equation}\label{omega-det-two-eqn-consequence}
Y -1 + (1-YZ) \omega + (Z-1) \omega^2=0
\end{equation}
where $Y$ and $Z$ are as above.
If we substitute $1-YZ = Z-1$ from
(\ref{linear-omega-det-two-eqn-diff-consequence})
to
(\ref{omega-det-two-eqn-consequence})
we get $Y=Z$. Substituting this back in
(\ref{linear-omega-det-two-eqn-diff-consequence})
we get $(Y-1)(Y+2) =0$, and $Y=1$ or $Y=-2$.
The solution $Y=1$ gives a degenerate $g$ and hence $f$,
a contradiction to  $f \not \in \mathscr{P}$.
The solution $Y=Z=-2$  gives $g = [-2, 1,1,-2]$.
This gives
\[f(x_1, x_2, x_3) = \left( [-x, 1]_{x_1} \otimes [-y, 1]_{x_2}
 \otimes [-z, 1]_{x_3} \right)  [-2, 1,1,-2],\]
where $x, y, z  \in \{1, \omega, \omega^2\}$.

The last case is that $N=3$ and the  spanning subgraph
with 3 edges in $K_4$ is the triangle on $\{a,b,c\}$ after
a cyclic permutation
of the 3 variables.
By flipping $x_1$ with its negation
$\overline{x_1}$ we can invoke what has been
proved for the triangle $\{b,c,d\}$ case, and conclude that
$f(x_1, x_2, x_3)$ is a product of some unary functions with the
function $g(\overline{x_1}, x_2, x_3)$ where $g(x_1, x_2, x_3)$
is the symmetric function $[-2, 1,1,-2]$.
\qed

\vspace{.1in}

We remark that since we are interested in a planar \#CSP problems,
we may not use arbitrary permutation of variables.
In the proof above, whenever the conclusion is symmetric in all
three variables, then the argument can apply an arbitrary permutation
in the proof without loss of generality.
However if the conclusion is not symmetric in all
three variables, we may only apply a cyclic permutation
in the proof, as is in the last case in Lemma~\ref{3rd-root-suffice}
 with the triangle on $\{a,b,c\}$.
Notice that the function $g(\overline{x_1}, x_2, x_3)$ has the signature matrix
$\begin{bmatrix}
1 & 1 & 1 & -2\\
-2 & 1 & 1 & 1
\end{bmatrix}$, where
 $x_1 =0,1$ is the row index,
and $x_2x_3 = 00,01,10,11$ is the column index.
If we connect two copies of  $g(\overline{x_1}, x_2, x_3)$
with both $x_1$ as external edges, and the variable $x_2$ of
one copy connected to the $x_3$ of the other copy, for both pairs
of $(x_2, x_3)$, we obtain
a planar gadget with a symmetric signature not in $\mathscr{P}$
 with its signature matrix
\[\begin{bmatrix}
1 & 1 & 1 & -2\\
-2 & 1 & 1 & 1
\end{bmatrix}
\begin{bmatrix}
1 & -2\\
1 & 1\\
1 & 1\\
-2 & 1
\end{bmatrix}
=
\begin{bmatrix}
7 & -2\\
-2 & 7
\end{bmatrix}  \not \in \mathscr{P}.\]

\begin{corollary}\label{cor:root-3-suffice}
Let $\mathcal{F}$ be a set of signatures containing a ternary signature
$f \not \in \mathscr{P}$.
Suppose $\mathcal{F}$ contains the
unary signatures
$\{[1, b_j] \mid 1 \le j \le 3 \}
=\{[1,1], [1,\omega], [1,\omega^2]\}$, and
$\partial^{\{i\}}_{[1, b_j]}(f) \in \mathscr{P}$
for all $1 \le i \le 3$ and $1 \le j \le 3$.
Then
\[\text{\rm Pl-}\#{\rm CSP}(\mathcal{F} \cup\{ g \})
\le_{\tt T} \text{\rm Pl-}\#{\rm CSP}(\mathcal{F}),\]
where $g$ is either the symmetric ternary function
$[1, -1, \omega, -\omega]$, or $[1, -1, \omega^2, -\omega^2]$, or
$[-2,1,1,-2]$, or a symmetric binary function  $[7, -2, 7]$.
\end{corollary}

\begin{lemma}\label{10-11-1n1-suffice}
Suppose $f \not \in \mathscr{P}$.
Let $[1, b_j]$ ($1 \le j \le 3$) be the unary signatures
$[1,0],[1,1], [1,-1]$ respectively.
Suppose $\partial^{\{i\}}_{[1, b_j]}(f) \in \mathscr{P}$
for all $1 \le i \le 3$ and $1 \le j \le 3$.
Then after a cyclic permutation of
its three variables $f$ is a product of some unary functions
$[1, b_j]$  with
the symmetric
function $[1,0,1,0]$ or $[0,1,0,1]$.
\end{lemma}
\proof
The requirements for $\partial^{\{i\}}_{[1, b_j]}(f) \in \mathscr{P}$
for all $1 \le i \le 3$ and $1 \le j \le 3$ are listed below:
{\tiny
\begin{eqnarray*}
& & f_{000} = f_{011} =0 ~~~~~~~~~~~~~~~~~~~~~~~\mbox{($\partial^{\{1\}} {\rm D}_1$)}
~~~~\mbox{or}
~~~~f_{001} = f_{010} =0 ~~~~~~~~~~~~~~~~~~~~~~~~\mbox{($\partial^{\{1\}} {\rm E}_1$)}
~~~~\mbox{or}
~~~~
\left|
\begin{matrix}
f_{000}  &  f_{001}  \\
f_{010}  &  f_{011}
\end{matrix}
\right| =0~~~~~~~~~~~~~~~~~~~~~~~\mbox{($\partial^{\{1\}} \det_1$)}\\
& & f_{000} + f_{100} = f_{011} + f_{111} =0 ~~~~\mbox{($\partial^{\{1\}} {\rm D}_2$)}
~~~~\mbox{or}
~~~~f_{001} + f_{101} = f_{010} + f_{110} =0 ~~~~\mbox{($\partial^{\{1\}} {\rm E}_2$)}
~~~~\mbox{or}
~~~~
\left|
\begin{matrix}
f_{000} + f_{100}  &  f_{001} + f_{101}  \\
f_{010} + f_{110}  &  f_{011} + f_{111}
\end{matrix}
\right| =0~~~~\mbox{($\partial^{\{1\}} \det_2$)}\\
& & f_{000} - f_{100} = f_{011} - f_{111} =0 ~~~~\mbox{($\partial^{\{1\}} {\rm D}_3$)}
~~~~\mbox{or}
~~~~f_{001} - f_{101} = f_{010} - f_{110} =0 ~~~~\mbox{($\partial^{\{1\}} {\rm E}_3$)}
~~~~\mbox{or}
~~~~
\left|
\begin{matrix}
f_{000} - f_{100}  &  f_{001} - f_{101}  \\
f_{010} - f_{110}  &  f_{011} - f_{111}
\end{matrix}
\right| =0~~~~\mbox{($\partial^{\{1\}} \det_3$)}\\
& & f_{000} = f_{101} =0 ~~~~~~~~~~~~~~~~~~~~~~~\mbox{($\partial^{\{2\}} {\rm D}_1$)}
~~~~\mbox{or}
~~~~f_{001} = f_{100} =0 ~~~~~~~~~~~~~~~~~~~~~~~~\mbox{($\partial^{\{2\}} {\rm E}_1$)}
~~~~\mbox{or}
~~~~
\left|
\begin{matrix}
f_{000}  &  f_{001}  \\
f_{100}  &  f_{101}
\end{matrix}
\right| =0~~~~~~~~~~~~~~~~~~~~~~~\mbox{($\partial^{\{2\}} \det_1$)}\\
& & f_{000} + f_{010} = f_{101} + f_{111} =0 ~~~~\mbox{($\partial^{\{2\}} {\rm D}_2$)}
~~~~\mbox{or}
~~~~f_{001} + f_{011} = f_{100} + f_{110} =0 ~~~~\mbox{($\partial^{\{2\}} {\rm E}_2$)}
~~~~\mbox{or}
~~~~
\left|
\begin{matrix}
f_{000} + f_{010}  &  f_{001} + f_{011}  \\
f_{100} + f_{110}  &  f_{101} + f_{111}
\end{matrix}
\right| =0~~~~\mbox{($\partial^{\{2\}} \det_2$)}\\
& & f_{000} - f_{010} = f_{101} - f_{111} =0 ~~~~\mbox{($\partial^{\{2\}} {\rm D}_3$)}
~~~~\mbox{or}
~~~~f_{001} - f_{011} = f_{100} - f_{110} =0 ~~~~\mbox{($\partial^{\{2\}} {\rm E}_3$)}
~~~~\mbox{or}
~~~~
\left|
\begin{matrix}
f_{000} - f_{010}  &  f_{001} - f_{011}  \\
f_{100} - f_{110}  &  f_{101} - f_{111}
\end{matrix}
\right| =0~~~~\mbox{($\partial^{\{2\}} \det_3$)}\\
& & f_{000} = f_{110} =0 ~~~~~~~~~~~~~~~~~~~~~~~\mbox{($\partial^{\{3\}} {\rm D}_1$)}
~~~~\mbox{or}
~~~~f_{010} = f_{100} =0 ~~~~~~~~~~~~~~~~~~~~~~~~\mbox{($\partial^{\{3\}} {\rm E}_1$)}
~~~~\mbox{or}
~~~~
\left|
\begin{matrix}
f_{000}  &  f_{010}  \\
f_{100}  &  f_{110}
\end{matrix}
\right| =0~~~~~~~~~~~~~~~~~~~~~~~\mbox{($\partial^{\{3\}} \det_1$)}\\
& & f_{000} + f_{001} = f_{110} + f_{111} =0 ~~~~\mbox{($\partial^{\{3\}} {\rm D}_2$)}
~~~~\mbox{or}
~~~~f_{010} + f_{011} = f_{100} + f_{101} =0 ~~~~\mbox{($\partial^{\{3\}} {\rm E}_2$)}
~~~~\mbox{or}
~~~~
\left|
\begin{matrix}
f_{000} + f_{001}  &  f_{010} + f_{011}  \\
f_{100} + f_{101}  &  f_{110} + f_{111}
\end{matrix}
\right| =0~~~~\mbox{($\partial^{\{3\}} \det_2$)}\\
& & f_{000} - f_{001} = f_{110} - f_{111} =0 ~~~~\mbox{($\partial^{\{3\}} {\rm D}_3$)}
~~~~\mbox{or}
~~~~f_{010} - f_{011} = f_{100} - f_{101} =0 ~~~~\mbox{($\partial^{\{3\}} {\rm E}_3$)}
~~~~\mbox{or}
~~~~
\left|
\begin{matrix}
f_{000} - f_{001}  &  f_{010} - f_{011}  \\
f_{100} - f_{101}  &  f_{110} - f_{111}
\end{matrix}
\right| =0~~~~\mbox{($\partial^{\{3\}} \det_3$)}
\end{eqnarray*}
}
By Lemma~\ref{linear-at-least-2-easy},
we may assume that for each  $1 \le i \le 3$
there can be at most one $j$ such that
($\partial^{\{i\}} {\rm D}_j$) holds, and
at most one $j'$ such that
($\partial^{\{i\}} {\rm E}_{j'}$) holds. This implies that
for every $i$ there is at least one $j$ such that
($\partial^{\{i\}} \det_j$)  holds.
By Lemma~\ref{det-at-least-3-easy},
we may assume that for each  $1 \le i \le 3$
there are
at most two distinct values $k$ and $k'$ such that
($\partial^{\{i\}} \det_k$)  and ($\partial^{\{i\}} \det_{k'}$) hold.

We first suppose the spanning subgraph of $K_4$ is connected.
This implies that all diagonal pairs have the same product value.
In that case, the statements ($\partial^{\{i\}} \det_2$)  and ($\partial^{\{i\}} \det_{3}$)
are identical. Thus if ($\partial^{\{i\}} \det_1$) holds then
we may assume ($\partial^{\{i\}} \det_2$)  and ($\partial^{\{i\}} \det_{3}$)
do not hold. On the other hand if ($\partial^{\{i\}} \det_1$) does not hold then
($\partial^{\{i\}} \det_2$)  and ($\partial^{\{i\}} \det_{3}$) must hold by
Lemma~\ref{linear-at-least-2-easy}.

\begin{enumerate}
\item
Suppose there exists some $1 \le i \le 3$ such that
either ($\partial^{\{i\}} {\rm D}_1$)  or ($\partial^{\{i\}} {\rm E}_{1}$)
holds. By cyclically permuting the variables we may assume $i=1$.

\begin{itemize}
\item
{($\partial^{\{1\}} {\rm D}_1$) holds.}

If ($\partial^{\{1\}}  \det_2$), which is equivalent to ($\partial^{\{1\}}  \det_3$),
does not hold, then we are done by Lemma~\ref{linear-at-least-2-easy}.
Thus ($\partial^{\{1\}}  \det_2$) holds.

We have $f_{000} = f_{011} =0$ by ($\partial^{\{1\}} {\rm D}_1$).
If we have additionally $f_{001} =0$ and $f_{010} =0$, then
$f^{x_1 =0}$ is identically 0, and $f(x_1 x_2, x_3)
= g(x_2, x_3) [0, 1]_{x_1}$ for some binary function $g$.
By Lemma~\ref{separation} we reach a contradiction to
 $f \not \in \mathscr{P}$.
So we assume $f_{001}$ and $f_{010}$ are not both zero.

If  ($\partial^{\{2\}} \det_1$) holds,
then we must have either
($\partial^{\{2\}} {\rm D}_2$)  and  ($\partial^{\{2\}} {\rm E}_{3}$),
or
($\partial^{\{2\}} {\rm D}_3$)  and  ($\partial^{\{2\}} {\rm E}_{2}$).
In either case, by ($\partial^{\{2\}} {\rm D}_2$) or ($\partial^{\{2\}} {\rm D}_3$)
 it easily follows that $f_{010} =0$ and
by ($\partial^{\{2\}} {\rm E}_{3}$) or  ($\partial^{\{2\}} {\rm E}_{2}$)
that $f_{001} =0$.
Hence ($\partial^{\{2\}}  \det_1$) does not hold.
Similarly if  ($\partial^{\{3\}} \det_1$) holds,
we reach the same contradiction.

Therefore we must have  ($\partial^{\{2\}} \det_2$) which is
 identical  to ($\partial^{\{2\}} \det_{3}$), and also  ($\partial^{\{3\}} \det_2$)
 which is
 identical  to ($\partial^{\{3\}} \det_{3}$), in addition to
($\partial^{\{1\}}  \det_2$).

These statements take the form
{\small
\begin{align}
\left|
\begin{matrix}
f_{000} + f_{100}  &  f_{001} + f_{101}  \\
f_{010} + f_{110}  &  f_{011} + f_{111}
\end{matrix}
\right|
&=
\left|
\begin{matrix}
f_{000} & f_{001} \\
f_{010} &  f_{011}
\end{matrix}
\right|
+
\left|
\begin{matrix}
f_{100}  & f_{101}  \\
f_{110}  & f_{111}
\end{matrix}
\right|
= - f_{001} f_{010} + f_{100} f_{111} -  f_{101} f_{110}  =0
\label{partial1det2-ad=0}
\\
\left|
\begin{matrix}
f_{000} + f_{010}  &  f_{001} + f_{011}  \\
f_{100} + f_{110}  &  f_{101} + f_{111}
\end{matrix}
\right|
&=
\left|
\begin{matrix}
f_{000} &  f_{001} \\
f_{100}  &  f_{101}
\end{matrix}
\right|
+
\left|
\begin{matrix}
f_{010}  & f_{011}  \\
f_{110}  & f_{111}
\end{matrix}
\right|
=
- f_{001} f_{100}  + f_{010} f_{111}
=0 \label{partial2det2-ad=0}
\\
\left|
\begin{matrix}
f_{000} + f_{001}  &  f_{010} + f_{011}  \\
f_{100} + f_{101}  &  f_{110} + f_{111}
\end{matrix}
\right|
&=
\left|
\begin{matrix}
f_{000}  &  f_{010} \\
f_{100}   &  f_{110}
\end{matrix}
\right|
+
\left|
\begin{matrix}
f_{001}  & f_{011}  \\
f_{101}  & f_{111}
\end{matrix}
\right|
=
- f_{010} f_{100} + f_{001} f_{111} =0
\label{partial3det2-ad=0}
\end{align}
}

If $f_{001} =0$ then $f_{010} f_{111}=0$ from (\ref{partial2det2-ad=0})
and $f_{010} f_{100} =0$ from (\ref{partial3det2-ad=0}).
Since  $f_{001}$ and $f_{010}$ are not both zero,
 we have $f_{010} \not = 0$ in this case,
we conclude that $f_{111}=f_{100} =0$.
Then $f$ is the product of $(x_2 \not = x_3)$ with the
degenerate binary function $g(x_1, x_3) =
\begin{bmatrix}
f_{010} & f_{001} \\
f_{110} & f_{101}
\end{bmatrix}$
with row index $x_1 =0,1$ and column index $x_3 =0, 1$,
and
$\left|
\begin{matrix}
f_{010} & f_{001} \\
f_{110} & f_{101}
\end{matrix}
\right|
=
0$ from (\ref{partial2det2-ad=0}).
This is a contradiction to
$f \not \in \mathscr{P}$.

If $f_{010} = 0$ then we also get the same conclusion.
So we assume both $f_{001} \not =0$ and $f_{010} \not =0$.
Then from (\ref{partial2det2-ad=0}) and (\ref{partial3det2-ad=0})
we get $f_{100} = \frac{f_{010}}{f_{001}} f_{111}
= \frac{f_{001}}{f_{010}}  f_{111}$. If $f_{111} =0$ then
so does $f_{100} = 0$ and we have $f = (x_2 \not = x_3) g(x_1, x_3)$
for a degenerate binary function $g$ as before.
Therefore we may assume $f_{111} \not =0$, then $f_{100} \not =0$ as well.
Then $(f_{010})^2 = (f_{001})^2$, thus
$f_{010} = \epsilon f_{001}$, where $\epsilon = \pm 1$.
Also $f_{100} = \epsilon f_{111}$.

Since
 all diagonal pairs have the same product value,
$f_{001} f_{110} = f_{010} f_{101} = f_{000} f_{111} =0$.
As  $f_{001} \not =0$ and $f_{010} \not =0$,
we have $f_{110} = f_{101} =0$.
Then from (\ref{partial1det2-ad=0})
we have
$\left|
\begin{matrix}
f_{100} & f_{001} \\
f_{010} & f_{111}
\end{matrix}
\right|
=
\left|
\begin{matrix}
\epsilon f_{111} & f_{001} \\
\epsilon f_{001} & f_{111}
\end{matrix}
\right| =0$.
It follows that $f_{001} = \epsilon^* f_{111}$,
where $\epsilon^* = \pm 1$.
Thus $f_{010} = \epsilon \epsilon^* f_{111}$.

Hence $f$ is  the product of $[\epsilon, 1]_{x_1}
\otimes [\epsilon \epsilon^*, 1]_{x_2}
\otimes [\epsilon^*, 1]_{x_3}$ with the symmetric
ternary function $f_{111}[0,1,0,1]$.

\item
{($\partial^{\{1\}} {\rm E}_1$) holds.}

This case is similar to the case of when ($\partial^{\{1\}} {\rm D}_1$) holds.
The conclusion is that if $f \not  \in \mathscr{P}$
then $f$ is the product of $[1, \epsilon]_{x_1} \otimes [1, \epsilon^*]_{x_2}
 \otimes [1, \epsilon \epsilon^*]_{x_3}$
and the ternary symmetric function $f_{000}[1,0,1,0]$,
where $\epsilon, \epsilon^* = \pm 1$.
\end{itemize}

\item
For no $1 \le i \le 3$
either ($\partial^{\{i\}} {\rm D}_1$)  or ($\partial^{\{i\}} {\rm E}_{1}$)
holds.
Hence for all $1 \le i \le 3$,
($\partial^{\{i\}}  \det_1$) holds.

Then for all $1 \le i \le 3$,
($\partial^{\{i\}}  \det_2$), which is
equivalent to ($\partial^{\{i\}}  \det_3$),
does not hold, by Lemma~\ref{det-at-least-3-easy}.
Thus for all $1 \le i \le 3$,
either ($\partial^{\{i\}} {\rm D}_2$) and ($\partial^{\{i\}} {\rm E}_3$),
or
($\partial^{\{i\}} {\rm D}_3$) and ($\partial^{\{i\}} {\rm E}_2$)
must hold.

We consider the case ($\partial^{\{1\}} {\rm D}_2$) and ($\partial^{\{1\}} {\rm E}_3$)
hold. The alternative case when
 ($\partial^{\{1\}} {\rm D}_3$) and ($\partial^{\{1\}} {\rm E}_2$) hold
is similar.

By ($\partial^{\{1\}}  \det_1$) we have
$\left|
\begin{matrix}
f_{000} & f_{001} \\
f_{010} &  f_{011}
\end{matrix}
\right|
=0$.
By ($\partial^{\{1\}} {\rm D}_2$) and ($\partial^{\{1\}} {\rm E}_3$)
we have
\[f_{000} = - f_{100}, ~~~~f_{011} = - f_{111},~~~~
f_{001} =  f_{101}, ~~~~f_{010} = f_{110}.\]
(In the case of
($\partial^{\{1\}} {\rm D}_3$) and ($\partial^{\{1\}} {\rm E}_2$),
all four right hand sides are multiplied by an extra $-1$.)

If $f_{000} = 0$, then by ($\partial^{\{1\}}  \det_1$) we have
$f_{001} f_{010} =0$.  If $f_{001} =0$ then $f^{x_2=0}$ is
identically 0, and $f = [0,1]_{x_2} g(x_1, x_3)$
for some binary function $g$.
If $f_{010} =0$ then $f^{x_3=0}$ is
identically 0, and $f = [0,1]_{x_3} g(x_1, x_2)$
for some binary function $g$.
In either case, this is a contradiction to
 $f \not  \in \mathscr{P}$
by Lemma~\ref{separation}.

Thus $f_{000} \not = 0$.
By ($\partial^{\{2\}} \det_1$)
we have
$\left|
\begin{matrix}
f_{000} & f_{001} \\
f_{100} &  f_{101}
\end{matrix}
\right|
=
\left|
\begin{matrix}
f_{000} & f_{001} \\
- f_{000} & f_{001}
\end{matrix}
\right|
=0$,
which implies that $f_{001}=0$.
Similarly by ($\partial^{\{3\}} \det_1$)
we have
$\left|
\begin{matrix}
f_{000} & f_{010}\\
f_{100} &  f_{110}
\end{matrix}
\right|
=
\left|
\begin{matrix}
f_{000} & f_{010}\\
-f_{000} &  f_{010}
\end{matrix}
\right|
=0$,
which implies that $f_{010}=0$.
This implies that
\[f_{001}=0,~~~~f_{101}=0,~~~~f_{010}=0,~~~~f_{110}=0.\]
Hence $f$ is the product of
$(x_2  = x_3)$ and the degenerate
binary function $g(x_1, x_3)
=
\begin{bmatrix}
f_{000} & f_{011}\\
f_{100} & f_{111}
\end{bmatrix}.$
Note that the determinant
$\left|
\begin{matrix}
f_{000} & f_{011}\\
f_{100} &  f_{111}
\end{matrix}
\right|
=0$,
by $f_{001}=f_{010}=f_{101}=f_{110} =0$
and ($\partial^{\{2\}}  \det_2$).

\end{enumerate}

\vspace{.1in}

Now we deal with the case when
the spanning subgraph of $K_4$ is disconnected.
This implies that the  spanning subgraph is a triangle
and the number $N$ of  valid statements among
all ($\partial^{\{i\}} {\rm D}_j$) and ($\partial^{\{i\}} {\rm E}_j$)
is exactly 3. Furthermore,  we may assume either the triangle is
 on $\{b,c,d\}$, and then
the 3 valid  statements among
all ($\partial^{\{i\}} {\rm D}_j$) and ($\partial^{\{i\}} {\rm E}_j$)
are among ($\partial^{\{1\}}  {\rm E}_j$),
($\partial^{\{2\}}  {\rm E}_j$)
and  ($\partial^{\{3\}}   {\rm E}_j$),
or, upto a cyclic permutation of the
3 variables of $f$, the triangle is on $\{a,b,c\}$, and then the 3 statements
must be among ($\partial^{\{1\}}  {\rm E}_j$),
($\partial^{\{2\}}  {\rm D}_j$)
and
($\partial^{\{3\}}   {\rm D}_j$).

\begin{enumerate}
\item
Suppose the triangle is on $\{b,c,d\}$.

Since $N=3$ there are at least 6 valid statements among
($\partial^{\{i\}} \det_j$), where $1 \le i,j \le 3$.
By Lemma~\ref{det-at-least-3-easy}, for every $1 \le i \le 3$,
there must be exactly two valid statements among
($\partial^{\{i\}} \det_j$), for $1 \le j \le 3$.
Since the diagonals $b$, $c$ and $d$ have the
same product value, the statements ($\partial^{\{1\}}   \det_{j}$) take the form
\begin{eqnarray*}
&&
\left|
\begin{matrix}
f_{000}  &  f_{001}  \\
f_{010}  &  f_{011}
\end{matrix}
\right| =0\\
&&
\left|
\begin{matrix}
f_{000} + f_{100}  &  f_{001} + f_{101}  \\
f_{010} + f_{110}  &  f_{011} + f_{111}
\end{matrix}
\right| =
\left|
\begin{matrix}
f_{000}  &  f_{001}   \\
f_{010}  &  f_{011}
\end{matrix}
\right|
+
\left|
\begin{matrix}
f_{000}  &  f_{101}   \\
f_{010}  &  f_{111}
\end{matrix}
\right|
+
\left|
\begin{matrix}
 f_{100}  &   f_{101}  \\
 f_{110}  &   f_{111}
\end{matrix}
\right|
=
0\\
&&
\left|
\begin{matrix}
f_{000} - f_{100}  &  f_{001} - f_{101}  \\
f_{010} - f_{110}  &  f_{011} - f_{111}
\end{matrix}
\right| =
\left|
\begin{matrix}
f_{000}  &  f_{001}   \\
f_{010}  &  f_{011}
\end{matrix}
\right|
-
\left|
\begin{matrix}
f_{000}  &  f_{101}   \\
f_{010}  &  f_{111}
\end{matrix}
\right|
+
\left|
\begin{matrix}
 f_{100}  &   f_{101}  \\
 f_{110}  &   f_{111}
\end{matrix}
\right|
=
0
\end{eqnarray*}
Notice that we used the fact that
$\left|
\begin{matrix}
 f_{100}  &  f_{001}   \\
 f_{110}  &  f_{011}
\end{matrix}
\right| =
0$, because the diagonals $b$ and $c$ have the same product value.

If the two valid statements among ($\partial^{\{1\}} \det_j$) are
for $j=2$ and $j=3$, then
$\left|
\begin{matrix}
f_{000}  &  f_{101}   \\
f_{010}  &  f_{111}
\end{matrix}
\right|
=0$ and we would have all four diagonals with an equal product value.
As the spanning subgraph of $K_4$ is disconnected, the two
 valid statements among ($\partial^{\{1\}} \det_j$) must include $j=1$.
Thus we have
\begin{eqnarray}
&
\left|
\begin{matrix}
f_{000}  &  f_{001}  \\
f_{010}  &  f_{011}
\end{matrix}
\right| &=0 \label{eq:1-inlm-triangle}\\
&
\left|
\begin{matrix}
f_{000}  &  f_{101}   \\
f_{010}  &  f^
{111}
\end{matrix}
\right|
&= - \epsilon_1
\left|
\begin{matrix}
 f_{100}  &   f_{101}  \\
 f_{110}  &   f_{111}
\end{matrix}
\right| \label{eq:2-inlm-triangle}
\end{eqnarray}
where $\epsilon_1 = +1$ if ($\partial^{\{1\}} \det_2$) holds,
and
$\epsilon_1 = -1$ if ($\partial^{\{1\}} \det_3$) holds.

Notice that if ($\partial^{\{1\}} \det_2$) holds,
we must have  a valid ($\partial^{\{1\}}  {\rm E}_3$) and then
$f_{001} =  f_{101}$ and  $f_{010} = f_{110}$.
On the other hand if ($\partial^{\{1\}} \det_3$) holds,
then we must have  a valid ($\partial^{\{1\}}  {\rm E}_2$) and then
$f_{001} =  -f_{101}$ and  $f_{010} = -f_{110}$.
Hence
$f_{001} = \epsilon_1 f_{101}$ and $f_{010} =  \epsilon_1 f_{110}$ hold
in either case.

Similarly by ($\partial^{\{2\}}  {\rm E}_2$)
or ($\partial^{\{2\}}  {\rm E}_3$), one of which must hold,
we have $f_{001} = \epsilon_2 f_{011}$ and $f_{100} = \epsilon_2 f_{110}$,
where $\epsilon_2 = \pm 1$.
By ($\partial^{\{3\}} {\rm E}_2$)
or ($\partial^{\{3\}}  {\rm E}_3$), one of which must hold,
we have $f_{010} = \epsilon_3 f_{011}$ and $f_{100} = \epsilon_3 f_{101}$.

If any of $f_{100}, f_{101}, f_{110}, f_{001}, f_{010}, f_{011}$
equals to 0, then all six quantities equal to 0.
Then the support of $f$ is contained in $\{000, 111\}$,
 and we have a contradiction to
 $f \not  \in \mathscr{P}$.
Thus we may normalize $f_{101}=1$.
Then
\[f_{100} = \epsilon_3,~~~~f_{110} = \epsilon_2 \epsilon_3,~~~~
f_{001} = \epsilon_1,~~~~f_{010} = \epsilon_1 \epsilon_2 \epsilon_3,~~~~
f_{011} = \epsilon_1 \epsilon_2.\]

By ($\partial^{\{1\}} \det_1$),
$\left|
\begin{matrix}
 f_{000}  &   f_{001}  \\
 f_{010}  &   f_{011}
\end{matrix}
\right|
=
\left|
\begin{matrix}
 f_{000}  &   \epsilon_1   \\
 \epsilon_1 \epsilon_2 \epsilon_3  &   \epsilon_1 \epsilon_2
\end{matrix}
\right| =0$
which implies that $f_{000} = \epsilon_1 \epsilon_3$.

By (\ref{eq:2-inlm-triangle}) we get
$\left|
\begin{matrix}
 \epsilon_1 \epsilon_3 &   1 \\
\epsilon_1 \epsilon_2 \epsilon_3  &   f_{111}
\end{matrix}
\right| =
- \epsilon_1
\left|
\begin{matrix}
\epsilon_3  &  1\\
\epsilon_2 \epsilon_3  &  f_{111}
\end{matrix}
\right|$,
which implies that $f_{111} = \epsilon_2$.

It follows that $f$ is simply the function
$f_{000} [1, \epsilon_1]_{x_1} \otimes [1, \epsilon_2]_{x_2}
\otimes [1, \epsilon_3]_{x_3} \in   \mathscr{P}$.
This is a contradiction.

\item
Suppose the triangle is on $\{a,b,c\}$.
We can similarly prove that under this hypothesis
$f \in   \mathscr{P}$, a contradiction.

\end{enumerate}
\qed

\begin{corollary}\label{[1,0]-[1,1]-[1,-1]-suffice}
Let $\mathcal{F}$ be a set of signatures containing a ternary signature
$f \not \in \mathscr{P}$.
Suppose $\mathcal{F}$ contains the
unary signatures
$\{[1, b_j] \mid 1 \le j \le 3 \}
=\{[1,0],[1,1], [1,-1]\}$, and
$\partial^{\{i\}}_{[1, b_j]}(f) \in \mathscr{P}$
for all $1 \le i \le 3$ and $1 \le j \le 3$.
Then
\[\text{\rm Pl-}\#{\rm CSP}(\mathcal{F} \cup\{ g \})
\le_{\tt T} \text{\rm Pl-}\#{\rm CSP}(\mathcal{F}),\]
where $g$ is either the symmetric
function $[1,0,1,0]$ or $[0,1,0,1]$.
\end{corollary}

\begin{theorem}\label{arity-3-nonproduct}
Suppose $\mathcal{F}$  contains a signature $f \not \in \mathscr{P}$
of arity $3$. Let $[1, a], [1,b], [1,c]$ be three unary
signatures that are pairwise linearly independent.
Then  there exists $g  \not \in \mathscr{P}$ such that
\begin{equation}\label{eqn-in-thm3.5}
\operatorname{Pl-\#CSP}(g, [1, a], [1, b], [1, c], \mathcal{F})
\le_{\rm T}
\operatorname{Pl-\#CSP}([1, a], [1, b], [1, c], \mathcal{F}),
\end{equation}
where $g$ has arity 2 or $g$ is a symmetric signature of arity 3.
\end{theorem}
\begin{proof}
In $\operatorname{Pl-\#CSP}([1, a], [1, b], [1, c], \mathcal{F})$,
for any $[1, x]\in\{[1, a], [1, b], [1, c]\}$, we have
$[1, x^k]=\partial^{k}_{[1, x]}(=_{k+1})$ for any $k\in\mathbb{Z}^+$.
Since $[1, a], [1,b], [1,c]$ are pairwise linearly independent,
there is at most one of $a, b, c$ that can be zero.
Without loss of generality, we can assume that $bc\neq 0$.

For $b, c$, if one of them is not a root of unity or is
a root of unity of primitive order greater than 4,
then we can construct 5 unary signatures that are pairwise linearly independent and
 we are done by Lemma~\ref{5-unaries-suffice}.

If one of $b, c$ is a root of unity of primitive order 4,
then we can construct $[1, 1], [1, -1], [1, \frak{i}], [1, -\frak{i}]$ and
we are done by Corollary~\ref{cor:root-i-suffice}.

If one of $b, c$ is a root of unity of primitive order 3, then we
can construct $[1, 1], [1, \omega], [1, \omega^2]$ with $\omega^3=1, \omega\neq 1$ and we
 are done by Corollary~\ref{cor:root-3-suffice}.

If both $b, c$ are roots of unity of order at most 2, then
$\{[1, b], [1, c]\}=\{[1, 1], [1, -1]\}$ since $[1, b], [1, c]$ are linearly independent.
If $a=0$,
then we are done by Corollary~\ref{[1,0]-[1,1]-[1,-1]-suffice}.
If $a\neq 0$, then $a\neq \pm 1$ since $[1, a], [1,b], [1,c]$ are pairwise linearly independent.
Thus $a$ is not a root of unity or $a$ is a root of unity of primitive order greater than 2.
In each case, we are done by Lemma~\ref{5-unaries-suffice}, Corollary~\ref{cor:root-i-suffice} or Corollary~\ref{cor:root-3-suffice}.
\end{proof}